\documentclass{article}
\usepackage[a4paper,top=3cm,bottom=3.5cm,left=3.2cm,right=3.2cm]{geometry}

\usepackage[utf8]{inputenc}
\usepackage{authblk}
\usepackage{graphicx}
\usepackage{subfigure}
\usepackage{amssymb}
\usepackage{lineno}
\usepackage[english]{babel}
\usepackage{palatino, url}
\usepackage{amsmath}
\usepackage{hyperref}

\usepackage{listings}
\usepackage{graphicx}
\usepackage[rightcaption]{sidecap}
\usepackage{textcomp}
\usepackage{array} 
\usepackage{booktabs} 
\usepackage{multirow} 
\usepackage{caption} 
\usepackage{floatflt}
\usepackage{rotating} 
\usepackage{scrextend}
\usepackage{multirow}
\usepackage[numbers]{natbib}
\usepackage{soul}
\usepackage{mathtools}
\usepackage{tikz}
\usetikzlibrary{positioning,shapes,fit,arrows}

\newcommand{\ie}{\emph{i.e.}, }
\newcommand{\eg}{\emph{e.g.}, }

\newtheorem{remark}{Remark}
\newtheorem{example}{Example}
\newtheorem{definition}{Definition}
\newtheorem{proposition}{Proposition}
\newtheorem{proof}{Proof}
\newtheorem{theorem}{Theorem}

\newcommand{\react}[1]{\ensuremath{\xrightharpoonup{\hbox{\makebox[.5cm][c]{\scriptsize $#1$}}}}}
\newcommand{\revreact}[2]{\ensuremath{  \xrightleftharpoons[\hbox{\makebox[.5cm][c]{\scriptsize$#2$}}]{\hbox{\makebox[.5cm][c]{\scriptsize$#1$}}}}}

\title{Robustness and resilience of dynamical networks in biology and epidemiology} 

\author[1]{Daniele Proverbio}
\author[1]{Rami Katz}
\author[1]{Giulia Giordano}

\affil[1]{\small Department of Industrial Engineering, University of Trento, Trento 38123, IT }

\date{}

\begin{document}

\maketitle

\begin{abstract}
Systems in nature are extremely robust and resilient, despite uncertainties and variability: they have evolved to preserve fundamental properties and qualitative behaviours that are crucial for survival, even in hostile and ever-changing environments. Studying their nonlinear dynamic behaviour is challenging, due to their complexity and to the many parameters at play, but is crucial to understand important phenomena at all scales, including for instance cellular dynamics, onset of diseases, and epidemic spreading.

Robustness and resilience capture a system's ability to maintain and recover its functions despite uncertainties, fluctuations and perturbations, both intrinsic and extrinsic. Here, we survey the competing definitions of these concepts adopted in different disciplines. Then, given a family of uncertain dynamical networked systems characterised by a structure (the interconnection topology, along with qualitative features of the individual dynamic units and interconnections) and by uncertain or unknown parameters, we provide an overview of methodologies to assess whether a property is structural (parameter-free) or robust (preserved for parameter variations within an uncertainty bounding set) for the family, and we discuss integrated structural and probabilistic approaches for biological and epidemiological systems. Also, we introduce possible formal definitions of resilience for a family of systems consisting of stochastic perturbations of a nominal deterministic system, to probabilistically quantify the system's ability to preserve a prescribed attractor; our proposed definitions generalise the notion of probabilistic robustness, and offer insight into well-known biological models.
Finally, we provide an overview of resilience indicators and of data-driven approaches to detect resilience loss and regime shifts that build upon bifurcation analysis.

These methodologies can strongly support the analysis and the control of complex uncertain systems in nature, including the analysis and the design of biomolecular feedback systems with a desired behaviour, the identification of therapeutic targets, the prediction and the control of epidemics, and the detection of tipping points and regime shifts in complex systems.
\end{abstract}

\subsubsection*{Published version}
The final publication is available from \textbf{now publishers} via \url{https://doi.org/10.1561/260000003}. 

\noindent \textbf{Suggested citation:} Daniele Proverbio, Rami Katz and Giulia Giordano (2025), “Robustness and Resilience of Dynamical Networks in Biology and Epidemiology”, Foundations and Trends in Systems and Control: Vol. 12, No. 2-3, pp 112–360. DOI: 10.1561/2600000036.

\subsubsection*{Acknowledgements}
This work was funded by the European Union through the ERC project INSPIRE  (project number 101076926). Views and opinions expressed are however those of the authors only and do not necessarily reflect those of the European Union or the European Research Council. Neither the European Union nor the granting authority can be held responsible for them.

\newpage
\tableofcontents

\newpage
\section{Introduction, motivation and preliminaries}
\label{ch:intro}

The importance of understanding the dynamics of life at all scales, from cells to organs, from organisms to
ecosystems, is matched by the challenge posed by their complexity and uncertainty -- paralleled, in turn, by
their astounding robustness and ability to maintain life-preserving properties in the most diverse operating
conditions, thanks to their underlying interconnection structure.
In view of their multifaceted complexity, systems in nature require joint efforts and methodological tools from various disciplines to be understood. In the past decades, the field of systems biology \citep{alon2019introduction, Doyle2006b, Kitano2002} has grown at the crossroads of physics, engineering, mathematics, biomedicine and biology. Its aim is to study natural systems, ranging from biomolecular interactions \citep{alon2019introduction,Wellstead2008} to dynamics in physiology \citep{keener2009mathematical} and medicine \citep{Barabasi2011,Belair1995,Piro2012}, up to population dynamics and epidemics \citep{anderson1991infectious,Brauer2012,edelstein2005}. Building on mathematical disciplines, systems biology promotes quantitative modelling and analysis to address pressing questions on the nature of systems in the life sciences, among which the investigation of their robustness \citep{Khammash2016,Kitano2004b} and resilience \citep{Dai2015} that can be found at all scales, from protein interactions to gene regulatory networks, from single cells to whole organisms, species and ecosystems.

Robustness and resilience concern the ability of a system to preserve its functions despite uncertainties, fluctuations and perturbations, both intrinsic and extrinsic. In addition to the concept of \emph{robustness}, fruitfully introduced and widely investigated in engineering to deal with parametric uncertainties and variability or with unmodelled dynamics, the analysis of systems in the life sciences often requires the complementary concept of \emph{resilience}, which addresses the system's ability to maintain a functional operational regime in spite of exogenous disturbances and noise. In particular, resilience often implies recovery after extreme or unforeseen events, which is conceptually different from robustness.

For their survival, biological systems evolved to preserve several properties \citep{Alon2003,Chen2007}, often mutually intertwined \citep{Sneppen2010}, while maintaining the ability to adapt and thrive in uncertain environments \citep{Siami2020}. Robust and resilient behaviours can be observed at all spatial and temporal scales: cells maintain properties related to, \eg homeostasis \citep{Aoki2019,Giordano2019,Palumbo2013,Tang2016}, signalling \citep{Shinar2007, Steuer2011}, protein transport regulation \citep{Giordano2018}, oscillations and multi-stability \citep{Angeli2008,Atkinson2003,Blanchini2014structural,Blanchini2015structuralclass,Blanchini2018homogeneous,Hori2013,Katz2025osc,Leite2010} and swarming \citep{proverbio2020assessing,proverbio2024chemotaxis,proverbio2025threshold,Reginato2025,romeralo}; complex organisms preserve regulatory functions such as circadian rhythms \citep{Doyle2006,Doyle2008,Goldbeter2012,Li2007,stelling2004robustness}, organism-level homeostasis \citep{Burlando2019} and mechanisms causing pathogenesis \citep{Burlando2020,Burlando2022,Demori2022,Demori2024,Mucci2020}; consistent patterns arise in multicellular organisms \citep{Arcak2013}, in population dynamics \citep{edelstein2005,Giordano2017qualitative} and in pathogen spreading that causes epidemics \citep{Alamo2021,Blanchini2021,Brauer2012,Giordano2020,Giordano2021,HernandezVargas2022}. How can we understand robustness and resilience in systems in the life sciences, and analyse the persistence of repeated patterns and functions? 

Along with experimental works addressing biological mechanisms, mathematical modelling and analysis are pivotal to identify the dynamical mechanisms ensuring robust and resilient behaviours. Crucial theoretical insights drive knowledge discovery (\eg extrapolation from sparse observations), prediction of unexpected or paradoxical outcomes, and synthesis of multiple results from reductionist models. Mathematical models allow us to provide interpretations and insight based on formal proofs and analytical results; they ``bridge the scales'', lifting results from micro-, to meso-, to macroscopic scales \citep{Almocera2018, Axer2022, Hart2020, Tegner2016}; they contribute to hypothesis building and validation \citep{Anderson2009,Bates2011,Giordano2018,proverbio2022classification}; they predict relevant properties and future outcomes, along with the analysis of scenarios and counterfactuals. In addition, models can unravel natural behaviours that are determined by global dynamical interplays, in networks of interactions, rather than by local (\eg biochemical) mechanisms \citep{Conrad2008}. Modelling enables not just the analysis, but also the (optimal) control of systems, to steer them towards the desired function and avoid undesired regime shifts within plausible spaces of operation \citep{CalaCampana2024,Cosentino2012,Doyle2008,Giordano2016convex,Hernandez-Vargas2010,Lenhart2007,Morris2021}.

Complex biological systems have an underlying \emph{network structure}, \ie can be seen as ensembles of entities (nodes, vertices) interacting through suitable interconnections (links, arcs, edges); consider for instance gene regulatory networks \citep{Andrecut2011a,Pio2022}, differentiation pathways \citep{Ghaffarizadeh2014}, metabolic regulation \citep{Amara2022, Arcak2007, Liu2017} and proteomics \citep{Hoffman2017}, neural wiring among and across brain regions \citep{alhourani2015network, Elam2021, Lynn2019}, human diseases \citep{Goh2007,Hammond2007, Royer2022}, as well as disease diffusion across connected regions \citep{alutto2024dynamic,arino2005multi, bansal2007individual,keeling2005networks}.
Interestingly, qualitative behaviours are often preserved even under huge parameter variations, uncertainties and environmental fluctuations \citep{Alon1999,Blanchini2011,Chesi2011b,Khammash2016,Kim2006,Kitano2004b,Shinar2009,Shinar2010,stelling2004robustness,Streif2016,Waldherr2011}, because they rely on the system's interconnection \emph{structure}, related to the network topology. Parameter-free structural approaches can check whether a property is preserved for a whole family of uncertain systems, exclusively due to its \emph{structure} \citep{blanchini2021structural}; when a property fails to hold structurally, robust analysis or probabilistic approaches can be used to understand why this occurs, which system features prevent the property from holding, and which key parameters must be finely tuned to enforce it.
Networks offer thereby a useful framework to study robustness and resilience \citep{Aldana2003,Li2022,Liu2020b,Vitkup2004}. Leveraging the concept of feedback loop and dynamic evolution, the systems-and-control community has developed versatile and powerful methodologies to formulate, investigate and solve relevant problems in the life sciences 
\citep{Atkinson1965,Blanchini2018,blanchini2021structural,samad2006,stelling2002}. In particular, tools to deal with uncertain systems and robustness \citep{Barmish1994,sanchezpena1998,zhou1998essentials} have found successful applications in the realm of systems biology and synthetic biology \citep{Baetica2019,DelVecchio2016}. Further insight emerges when bridging concepts from sibling disciplines, such as network science, statistical physics, mathematical biology: control-theoretic approaches can complement methods from other disciplines, and offer theoretical rigour to computational tools relying on parametric search and to data-driven statistical methods. A control-theoretic framework is also well-suited to embed interdisciplinary research questions and methodologies, and to promote modelling across scales \citep{Doyle2006,Doyle2008,Garabed2019} as well as the design of a system-level integration of interconnected modules that takes into account feedback and retroactivity effects \citep{DelVecchio2008,DelVecchio2009,Jayanthi2011,McBride2019}.

\subsection{Aims and scope}
\label{sec:scope}

Focusing on biological and epidemiological systems, this monograph presents a primer on the concepts of \emph{robustness} and \emph{resilience} of complex dynamical networks. These concepts are often used in the life sciences without a rigorous formal definition; here we review their different meanings in various disciplines and, in Chapters~\ref{ch:struct-an} and \ref{ch:rob-res-sta-mod}, we provide a comprehensive framework to address them with mathematical rigour. 
In biology and epidemiology, full knowledge of the system parameters is often impossible, because they cannot be measured or quantified, they are unobservable or they are highly uncertain and time-varying \citep{Bailey2001,Liu2013,Roda2020}; still, natural systems are able to exhibit extraordinarily robust and resilient behaviours, preserved in spite of parametric uncertainties, intrinsic variability, environmental fluctuations, as well as noise, perturbations and external disturbances.
Classic examples of biological robustness are the circadian clock, which keeps oscillating at the same frequency across orders of magnitude of variations in its parameters  \citep{Bagheri2008,Doyle2006,stelling2004robustness}, and the tumbling frequency in bacterial chemotaxis, which remains unaffected by step perturbations and chemical gradients \citep{Alon1999,Barkai1997}.
How can we rigorously explain the robustness and resilience of a system, or the loss thereof?
And which are the minimal \emph{structural} requirements to predict the emergence of robust and resilient properties in the \emph{absence of quantitative information} about the system, or in the presence of \emph{limited quantitative information} -- often of a probabilistic nature?

A thorough analysis of the robustness and resilience of biological model can be precious to: better understand how natural systems can preserve fundamental properties in different environmental conditions; untangle the interplay of multiple parameters, allowing the systematic analysis of several intertwined factors; identify the crucial parameters in the system dynamics or the crucial motifs \citep{Alon2007,Milo2002a,Stone2019,yeger2004} in the system structure, which can become candidate targets for control actions, therapies and interventions to guarantee the desired system behaviour; inspire the design of robust artificial systems having designated properties and functions; falsify models, by revealing discrepancies between model predictions and experimental observations \citep{Angeli2012}; and inform better modelling approaches for faithful ``digital twins'' -- namely, \emph{in-silico} surrogate models that emulate the behaviour of the original system -- to support analysis, hypothesis testing, predictions of system behaviours and synthesis of control interventions, drugs and treatments. The insight gained from the analysis of models is particularly precious to design control strategies at various levels, ranging from biomolecular controllers to pharmaceutical and non-pharmaceutical interventions to contrast epidemics.

In this monograph, we restrict our scope to dynamical models of ordinary differential equations (ODEs) characterised by a known interconnection structure and by unknown or uncertain parameters (see the examples in Section~\ref{sec:models}). We focus on qualitative structural (\ie parameter-free) methods \citep{Blanchini2012,blanchini2021structural,Jacquezt1993,Kaltenbach2009,Reder1988,Shinar2010}, designed to identify, understand and predict robustness \citep{Chesi2011b,Hara2019,Kim2006,Ma2002,Shinar2007,Shinar2009,stelling2004robustness,Steuer2011,Streif2016,Waldherr2011} and resilience \citep{bhamra2011resilience,carpenter2012general,Dakos2015,fraccascia2018resilience,fisher2015more,Gao2016,gunderson2000ecological,holling1996engineering,Liu2020b,meyer2016mathematical} of biological systems at all scales. Conversely, we do not focus on complementary quantitative and numerical methods to establish parametric robustness and sensitivity, nor on the abundant literature on network robustness in engineered systems (see, \eg \cite{weerakkody2019resilient}). Also, we do not dwell on the vast literature on synthetic biology \citep{Arpino2013,Batt2007,benner2005synthetic,Cosentino2012,Csete2002, DelVecchio2015,DelVecchio2016,DelVecchio2018, Hsiao2018,Koeppl2011,Sootla2016}, which nonetheless can leverage our presented approaches for the design and engineering of robust and resilient systems. Moreover, we do not review here work on network identification \citep{Schön2011,chiuso2012bayesian,Papachristodoulou2007,Porreca2008,Porreca2012,Yu2007}, including dynamical structure reconstruction \citep{gonccalves2008necessary} and reconstruction of the unknown topology in networks of dynamical systems \citep{materassi2010topological,materassi2012problem}, as well as work on system reconstruction or network reconstruction \citep{aalto2020gene,peixoto2019network,sutulovic2025differentiator,runge2018causal} tailored to biological systems, even though the identification of a precise network topology (\emph{structure}) is an indispensable prerequisite for the structural analysis discussed in Chapter~\ref{ch:struct-an}.

In our overview, we aim to establish a connection between control theory and sibling disciplines; hence, several pointers will be included to interdisciplinary studies. However, the literature on mathematical models and approaches for biology and epidemiology is too vast to be completely explored here. The interested reader is referred to the many dedicated special issues \citep{Albert2018,Allgower2011,Arcak2019,Chesi2011, Khammash2008, Yue2010}, surveys \citep{Ashwin2016d,Brockmann2013, Hethcote2000, Iglesias2007,Pastor2015, Sontag2005} and books \citep{J.S.Allen2014, allman2004mathematical,Blackwood2018, Brauer2008, edelstein2005, Iglesias2010, Ingalls2013,Jayaraman2009,Lenhart2007,martcheva2015introduction,Nowzari2016,Queinnec2007,Segel2013, strogatz2018nonlinear,Szallasi2006}, and the references therein. Similarly, we do not survey the whole biological, ecological and epidemiological literature unveiling robust system properties (see, \eg \cite{Alberts2015,Alon1999,Barkai1997,Bertolaso2019,Breindl2011,dethier2015positive,samad2005b,Faulon2004,felix2015,Ghaemi2009,Hernandez2009,Kitano2007,Little1999,Paulino2019,Qian2021,Venturelli2012}), although we provide pointers to some specific publications throughout the text.

We focus on ``classical'' low-dimensional networks \citep{Barabasi2013}, amenable to transparent examples and analysis, with continuous-time dynamics. Large-scale networks \citep{Hori2013,Prill2005}, higher-order networks \citep{Battiston2021,bianconi2021higher,cisneros2021multigroup}, as well as methods for model order reduction \citep{Feliu2012,Rao2013,schaft2015,Snowden2017} and dimension reduction \citep{Laurence2019, Vegue2023} are briefly overviewed in Chapter~\ref{ch:unch-dir}. Maps and discrete models, such as Boolean networks \citep{Chaves2005,Chaves2006,Saadatpour2011} and Petri nets \citep{Angeli2007,Angeli2011, Chaouiya2007,Koch2015}, are also not included in our survey. Moreover, we do not consider agent-based models, which are widely used to simulate systems in the life sciences, but typically not amenable for analytical treatment  \citep{Borshchev2004a,DG2023model,DG2023,DG2024,Fan2022,kerr2021covasim,leonard2024fast,Pagliara2021,PTG2025,Sturniolo2021}.

In addition to introducing the scope of this monograph, the present chapter:
presents the considered classes of models and networks, along with examples from biology and epidemiology;
starts introducing the concept of structure, in relation to graph theory and network theory;
discusses the different concepts of robustness, resilience and stability as understood in various disciplines.

Then, Chapter~\ref{ch:struct-an} briefly surveys the concepts of robustness and of structural analysis for uncertain dynamical systems, aimed at unveiling the mechanisms that enable robustness \emph{qualitatively}, \ie when only the parameter-free description of a \emph{family} of systems is available.

Chapter~\ref{ch:rob-res-sta-mod} compares the notions of robustness and resilience, and motivates the need for a concept of resilience that is distinct from that of robustness in order to assess important properties of natural systems. We introduce novel formal definitions of resilience that enable a more quantitative analysis of dynamical systems in the presence of an exogenous stochastic noise.

Chapter~\ref{ch:res-ind} builds upon the notion of resilience to discuss the development of resilience indicators, their efficacy and their applicability in different domains:
at first low-dimensional systems are considered, then the analysis is expanded to network resilience, and examples from biology and epidemiology are proposed.

Finally, Chapter~\ref{ch:conclusion} offers concluding remarks and outlines promising open challenges in the field, not fully surveyed in other chapters or under recent investigation.

\subsection{Considered classes of models}\label{sec:models}

Natural dynamics range from chemical reaction networks, gene regulatory networks and biomolecular systems, arising from the interplay of biochemical species and molecules \citep{alon2019introduction,Angeli2009c,Cosentino2012,DelVecchio2015,Feinberg2019,Iglesias2010,samad2006,Sontag2005}, to ecological networks, arising from the interactions among the species sharing an ecosystem \citep{Brauer2012,edelstein2005}, and epidemiological systems, describing the spread of infectious diseases \citep{Brauer2012,martcheva2015introduction,Nowzari2016,Pastor2015}. Simple examples of such systems are the Incoherent Feed-Forward Loop in gene regulation \citep{Goentoro2009,Kaplan2008}, the Lotka-Volterra predator-prey system \citep{edelstein2005}, the Susceptible-Infected-Recovered epidemic model \citep{Brauer2012}.
They all can be modelled as nonlinear dynamical networks, namely, interconnections of interacting dynamical systems, whose intrinsic nonlinearities can lead to emergent phenomena \citep{Epstein1998} such as persistent oscillations (circadian rhythms, see \cite{Bagheri2008,stelling2004robustness}; predator-prey dynamics, see \cite{edelstein2005}; epidemic waves, see \cite{Brauer2012}) and multi-stability (cell cycle, see \cite{Angeli2004detection}; pathogenesis, see \cite{Cloutier2012}; multiple ecosystem configurations and community compositions, see \cite{Brauer2012,edelstein2005}; disease-free or endemic equilibria, see \cite{Brauer2012,edelstein2005}).

To study dynamical systems in biology and epidemiology, ODE-based models are widely employed.
In fact, such systems aptly mimic the evolution of natural systems, at least under the assumption of high density of elements, such as a high copy number of expressed genes or proteins in reaction networks, or a sufficiently large number of (infectious) individuals in epidemiological models. Importantly, these models are often amenable to analytical results (and not just numerical simulation campaigns) that enable conceptual insight, verifiable predictions and control.

Let us denote the set of non-negative real numbers as $\mathbb{R}_{\geq 0}$ and of positive real numbers as $\mathbb{R}_{> 0}$. We consider families of nonlinear ODE systems of the form
\begin{equation}\label{eq:modelclass}
\dot x(t) = f(x(t),\vartheta) \, ,
\end{equation}
where $x \in \mathbb{R}_{\geq 0}^n$ is a vector of species concentrations or population densities, while the vector function $f \colon \mathbb{R}_{\geq 0}^n \times \mathbb{R}_{\geq 0}^q \to \mathbb{R}^n$ depends on the parameters $\vartheta \in \mathbb{R}_{\geq 0}^q$ (which may be in turn functions, \eg of time or temperature). Both $f$ and $\vartheta$ may be uncertain, or even unknown; in a more general formulation, they may also be time-varying.

As an example, protein production or degradation may be modelled as a function of one or more of the state variables $x_k$, whose form can be, \eg polynomial (mass action kinetics, $f_i=\kappa \prod_{k=1}^n x_k^{a_k}$, with $a_k \in \mathbb{N}_0$ and $\kappa \in \mathbb{R}_{> 0}$, see \cite{Feinberg2019}) or saturating (Michaelis-Menten or Hill kinetics, $f_i=\kappa \frac{x_k^\ell}{1+\alpha x_k^\ell}$, with $\ell, \alpha, \kappa \in \mathbb{R}_{> 0}$, see \cite{alon2019introduction}), and it may be affected by the temperature. As another example, the force of infection in an epidemiological model may depend on the density $x_h$ of infected individuals either linearly (\eg $f_i=\kappa x_j x_h$, where $x_j$ is the density of susceptible individuals and $ \kappa \in \mathbb{R}_{> 0}$) or nonlinearly (\eg through Holling-type functions, $f_i=\kappa x_j \frac{x_h^p}{1+\beta x_h^q}$, where $x_j$ is the density of susceptible individuals and $p,q,\beta, \kappa  \in \mathbb{R}_{> 0}$, see \cite{Capasso1978}), and it may be affected by environmental factors, seasonality, or non-pharmaceutical interventions.

Often, the functions $f_i$, $i=1,\dots,n$, are monotonic in all their variables: for all pairs $i$ and $j$, $\partial f_i(x,\vartheta) / \partial x_j$ is sign-definite (either non-negative or non-positive, or identically zero) for all $\vartheta$ in the domain of interest. This property expresses either inherent pairwise cooperation (activation, reinforcement) or inherent pairwise competition (inhibition, suppression) among the variables.

System \eqref{eq:modelclass} is endowed with a known \emph{network structure}, which captures the pattern of relations among the state variables.

Here, we provide some examples of ODE models in biology and epidemiology that will be used throughout the monograph to illustrate the discussed definitions and results, and we highlight their network structure through graph representations.

\subsubsection{Dynamical networks in biology}\label{sec:dninb}

Biological dynamics can be described by the ODE system
\begin{equation}
    \dot{x}(t) = Sg(x(t),\vartheta) + g_0(\vartheta) \, ,
    \label{eq:ode-system}
\end{equation}
where $S \in \mathbb{R}^{n \times m}$ is a constant matrix that represents the \emph{interconnection structure} of the system and is independent of both the states $x \in \mathbb{R}_{\geq 0}^n$ (representing species concentrations in (bio)chemical reaction networks, and population densities in epidemiological or ecological systems) and the uncertain parameters $\vartheta \in \mathbb{R}_{\geq 0}^q$; in the case of (bio)chemical reaction networks, $S \in \mathbb{Z}^{n \times m}$ is the stoichiometric matrix associated with the chemical reactions \citep{Angeli2009c,blanchini2021structural,Feinberg2019}. Moreover, $g \colon \mathbb{R}^n_{\geq 0} \times \mathbb{R}^q_{\geq 0} \to \mathbb{R}^m_{\geq 0}$ is a vector of uncertain ``reaction rate'' functions (which represent reaction rates strictly speaking in (bio)chemical reaction networks, and population growth rates in epidemiological or ecological systems). The components $g_i$, $i=1,\dots,m$, are non-negative functions, typically monotonic in the state variables. In the special case in which the law of mass action applies, the components $g_i$ are polynomial functions of the species concentrations; for instance, for a chemical reaction network $p_1 X_1 + p_2 X_2 + \dots + p_n X_n \react{g_i} q_1 X_1 + q_2 X_2 + \dots + q_n X_n$, the rate has the form $g_i = \vartheta_i \prod_{h=1}^n x_h^{p_h}$. Vector $g_0 \colon \mathbb{R}^q_{\geq 0} \to \mathbb{R}^n_{\geq 0}$ represents influxes from the external environment.

Systems of the form \eqref{eq:modelclass} can be cast in the form \eqref{eq:ode-system} by setting $S=I$ (the identity matrix), $g=f$, and $g_0=0$; and vice versa, systems of the form \eqref{eq:ode-system} can be written in the form \eqref{eq:modelclass} by setting $f = Sg+g_0$.

Chemical reaction network systems can be represented as \emph{flow graphs}, where the vertices are associated with chemical species and the edges are associated with reactions turning molecules of some species into molecules of other species \citep{blanchini2021structural}.

Below, we provide three examples of this class of systems.

\begin{figure}[h!]
	\centering
	\includegraphics[width=\textwidth]{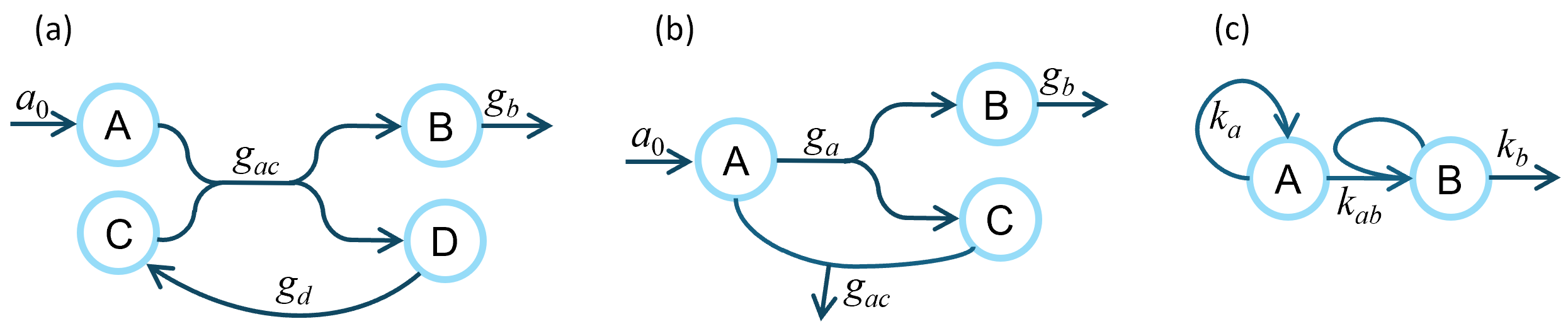}
	\caption{\footnotesize Flow graph representation of the chemical reaction networks associated with (a) the biomolecular network \eqref{eq:biomolec}, (b) the chemical reaction network \eqref{eq:biomolec2} and (c) the Lotka-Volterra system \eqref{eq:LV}.}
	\label{fig:CRNflowgraphs}
\end{figure}

\begin{example}[\textbf{Biomolecular network}]\label{ex:biomol}
The biomolecular network \citep{Chen2010} involving the reactions $\emptyset \react{a_0} A$, $A+C \react{g_{ac}} B+D$, $D \react{g_d} C$ and $B \react{g_b} \emptyset$ corresponds to the dynamical system
\begin{equation}\label{eq:biomolec}
\begin{cases}
\dot a(t) =  a_0 - g_{ac}(a(t),c(t),\vartheta)\, ,\\
\dot b(t) =  g_{ac}(a(t),c(t),\vartheta) - g_b(b(t),\vartheta)\, ,\\
\dot c(t) =  g_d(d(t),\vartheta)-g_{ac}(a(t),c(t),\vartheta)\, ,\\
\dot d(t) = g_{ac}(a(t),c(t),\vartheta)-g_d(d(t),\vartheta) \, ,
\end{cases}    
\end{equation}
where functions $g_k$ are unknown and assumed to be monotonic in each of their arguments; the system can be written in the form \eqref{eq:ode-system} with $x = \begin{bmatrix}
 a \\ b \\ c \\ d
\end{bmatrix}$, $S = \begin{bmatrix}
-1 & 0 & 0\\
1 & -1 & 0\\
-1 & 0 & 1\\
1 & 0 & -1
\end{bmatrix}$, $g(x,\vartheta) = \begin{bmatrix}
g_{ac}(a,c,\vartheta) \\ g_b(b,\vartheta) \\ g_d(d,\vartheta)
\end{bmatrix}$ and $g_0=\begin{bmatrix}
a_0 \\ 0 \\ 0 \\ 0
\end{bmatrix}$.
The flow graph associated with system \eqref{eq:biomolec} is shown in Figure~\ref{fig:CRNflowgraphs}a.
\end{example}

\begin{example}[\textbf{Chemical reaction network}]\label{ex:CRN}
The time evolution of the concentration of the chemical species involved in the reactions $\emptyset \react{a_0} A$, $A \react{g_{a}} B+C$, $A+C \react{g_{ac}} \emptyset$ and $B \react{g_b} \emptyset$ can be described by a system of the form \eqref{eq:ode-system}; in fact, it corresponds to the dynamical system
\begin{equation}\label{eq:biomolec2}
\begin{cases}
\dot a(t) =  a_0 - g_{a}(a(t),\vartheta)-g_{ac}(a(t),c(t),\vartheta)\, ,\\
\dot b(t) =  g_{a}(a(t),\vartheta) - g_b(b(t),\vartheta)\, ,\\
\dot c(t) =  g_{a}(a(t),\vartheta)-g_{ac}(a(t),c(t),\vartheta)\, ,
\end{cases}    
\end{equation}
so $x = \begin{bmatrix}
 a \\ b \\ c
\end{bmatrix}$, $S = \begin{bmatrix}
-1 & 0 & -1\\
1 & -1 & 0\\
1 & 0 & -1
\end{bmatrix}$, $g(x,\vartheta) = \begin{bmatrix}
g_a(a,\vartheta) \\ g_b(b,\vartheta) \\ g_{ac}(a,c,\vartheta)
\end{bmatrix}$, $g_0=\begin{bmatrix}
a_0 \\ 0 \\ 0
\end{bmatrix}$.
The flow graph associated with system \eqref{eq:biomolec2} is shown in Figure~\ref{fig:CRNflowgraphs}b.
\end{example}

\begin{example}[\textbf{Lotka-Volterra system}]\label{ex:LV}
Lotka-Volterra predator-prey dynamics can be associated with a chemical reaction network, where the prey $A$ reproduces through the auto-catalytic reaction $A \react{k_a} 2 A$, interactions between prey $A$ and predator $B$ cause a decrease in prey and a corresponding increase in predators through the conversion reaction $A+B \react{k_{ab}} 2B$, and predator $B$ decreases through the degradation reaction $B \react{k_b} \emptyset$. The resulting model under the law of mass action
\begin{equation}\label{eq:LV}
\begin{cases}
\dot a(t) = k_a a(t) - k_{ab} a(t) b(t) \, ,\\
\dot b(t) = -k_b b(t) + k_{ab} a(t) b(t) \, ,
\end{cases}    
\end{equation}
is a system of the form \eqref{eq:ode-system} with $x = \begin{bmatrix}
 a & b
\end{bmatrix}^\top$, $\vartheta = \begin{bmatrix}
k_a & k_{ab} & k_b
\end{bmatrix}^\top$, $S = \begin{bmatrix}
1 & -1 & 0\\
0 & 1 & -1
\end{bmatrix}$, $g(x,\vartheta) = \begin{bmatrix}
k_a a & k_{ab} a b & k_b b
\end{bmatrix}^\top$, $g_0=0$.
The flow graph associated with system \eqref{eq:LV} is shown in Figure~\ref{fig:CRNflowgraphs}c.
\end{example}

A special case within the class of models \eqref{eq:modelclass}, or equivalently \eqref{eq:ode-system}, is given by activation-inhibition networks, where $N$ species interact pairwise, each by either promoting or repressing the activity of some others. In gene regulatory networks \citep{alon2019introduction} the amount $x_i$ of active species $i$, $i =1,\dots, N$, evolves according to the equation
\begin{equation}
\dot x_i(t) = -\mu_i x_i(t) + \sum_{j \in \mathbb{J}_i} f_{ij}(x_j(t)) + \sum_{\ell \in \mathbb{L}_i} g_{i \ell}(x_\ell(t))+ \kappa_i \, ,
    \label{eq:gene-regulation}
\end{equation}
where $\mu_i$ represents the self-degradation rate of species $i$, $\mathbb{J}_i$ and $\mathbb{L}_i$ are respectively the sets indexing species that activate (promoters) and inhibit (repressors) species $i$, the increasing positive functions $f_{ij}(x_j)$ represent the activating effect of $x_j$ on $x_i$, the decreasing positive functions $g_{i \ell}(x_\ell)$ represent the inhibiting effect of $x_\ell$ on $x_i$, and $\kappa_i$ is an external input corresponding to a basal expression rate.

See Figure~\ref{fig:hill}a for an example of a gene activation circuit, with $N=1$, which can be represented by the equation
\begin{equation}
\dot x(t) = - x(t) + f(x(t)) + k = - x(t) + a \frac{x(t)^h}{1 + x(t)^h} + k \, .
    \label{eq:gene-regulation2}
\end{equation}

The activation and inhibition functions are often chosen as $f_{ij}(x_j) = \frac{A_{ij}(x_j/\beta)^h}{1+(x_j/\beta)^h}$ and $g_{i \ell}(x_\ell) = \frac{B_{i \ell}}{1+(x_\ell/\beta)^h}$, respectively (see Figure~\ref{fig:hill}b). These are Michaelis-Menten functions (if $h=1$) or sigmoidal Hill functions (if $h>1$), where $\beta$ is a dissociation constant, usually normalised to $1$, $A_{ij}$ and $B_{i \ell}$ are respectively the maximum production and degradation rates, and the exponent $h$ represents cooperative binding mechanisms of transcription factors \citep{santillan2008use}. Other regulating functions may be employed, representing logic operators such as \emph{AND} and \emph{OR} \citep{Chen2010}.

\begin{figure}[ht!]
	\centering
	\includegraphics[width=\textwidth]{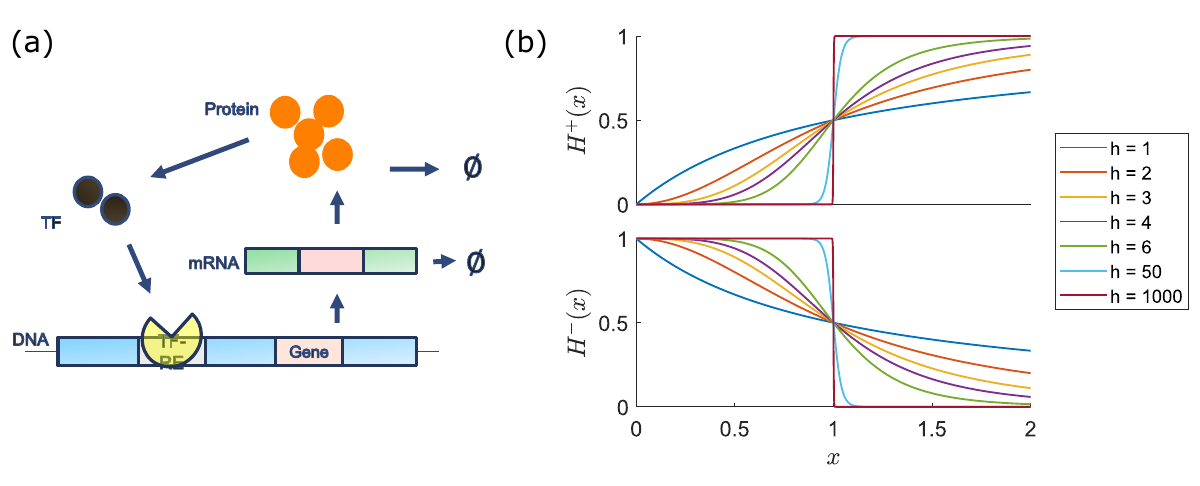}
	\caption{\footnotesize (a) Schematic representation of a gene activation circuit, corresponding to Eq. \eqref{eq:gene-regulation2}. TF stands for transcription factors, namely, proteins that regulate DNA transcription by binding promoters. (b) The family of activating $H^+(x)=\frac{(x/\beta)^h}{1+(x/\beta)^h}$ and inhibiting $H^-(x)=\frac{1}{1+(x/\beta)^h}$ Hill functions, for various values of $h$; $h = 1$ corresponds to Michelis-Menten functions, while the so-called ``logic approximation'' is $\lim_{h \to \infty} H^*(x) = \Theta^*(x - \beta)$, where $^*$ denotes either $^+$ or $^-$ and $\Theta^*$ is the (increasing or decreasing) Heaviside step function.}
	\label{fig:hill}
\end{figure} 

Activation-inhibition systems can be associated with \emph{signal graphs}, where the vertices represent chemical species and the signed edges correspond to inhibiting or activating interactions; the sign of the edges matches the sign pattern of the system Jacobian matrix \citep{blanchini2021structural}.

\begin{example}[\textbf{IFFL}]
The Incoherent Feed-Forward Loop (IFFL)
\begin{equation}\label{eq:IFFL}
\begin{cases}
\dot x_1(t) = -\mu_1 x_1(t) + \kappa_1\\
\dot x_2(t) = -\mu_2 x_2(t) + f_{23}(x_3(t)) +f_{21}(x_1(t))\\
\dot x_3(t) = -\mu_3 x_3(t) + g_{31}(x_1(t))
\end{cases}    
\end{equation}
captures the regulatory interactions among three genes: $X_1$ activates $X_2$ and inhibits $X_3$, which in turn activates $X_2$. The parameters $\mu_i$, $i=1,2,3$, are positive.
The signal graph corresponding to system \eqref{eq:IFFL} is shown in Figure~\ref{fig:IFFLsignal}.
\end{example}

\begin{figure}[ht!]
	\centering
	\includegraphics[width=0.4\textwidth]{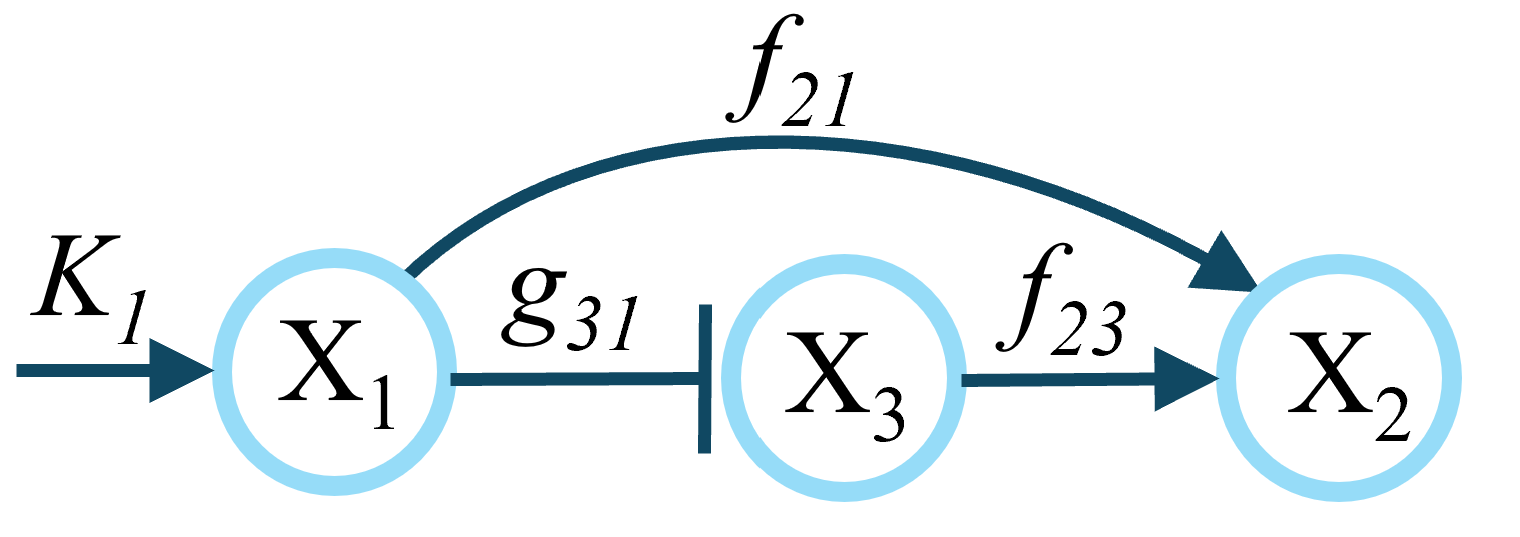}
	\caption{\footnotesize Signal graph representation of the Incoherent Feed-Forward Loop activation-inhibition network associated with system \eqref{eq:IFFL}. The negative self-loops for each node are not reported for simplicity.}
	\label{fig:IFFLsignal}
\end{figure}

Differently from reaction network systems, which are mechanistic models grounded on chemical laws, activation-inhibition networks are phenomenological models that do not correspond to a list of chemical reactions: they aim at reproducing empirical observations by considering interactions that only implicitly account for the true physical mechanisms. For instance, in the above IFFL model, the variables $x_i$ can be seen as the concentration of the proteins expressed by the corresponding genes $X_i$ and model \eqref{eq:IFFL} can be seen as the reduction of a more complex system involving gene-RNA-protein dynamics via timescale separation arguments \citep{Chen2010}. Notably, any system with a sign-definite Jacobian can be seen as an activation-inhibition network and thus associated with a signal graph, which is often referred to as \emph{Jacobian graph}.

\subsubsection{Dynamical networks in epidemiology}
\label{sec:epi_model}

Compartmental models in epidemiology are aggregate mean-field models, which partition the considered population into different compartments associated with relevant disease stages. Such models effectively describe contagion phenomena in a large population of well-mixed individuals (whose interaction network can be described by a complete graph), and can be visualised by flow graphs, where the vertices represent compartments and the edges illustrate the transition of individuals from one compartment to another \citep{Brauer2012,edelstein2005}; see, \eg Figure~\ref{fig:epi-model}.

Compartmental epidemiological models have a long tradition (\cite{kermack1927contribution}, see also the survey by \cite{Breda2012}). The simple Susceptible-Infectious-Susceptible (SIS) and Susceptible-Infectious-Removed (SIR) models (see Figure~\ref{fig:epi-model}) have been expanded over the span of several decades into more complex models with several compartments \citep{Arino2007,CalaCampana2024}, aimed at capturing specific disease dynamics \citep{Brauer2008, Giordano2020,proverbio2021dynamical}, as well as, for instance, age structure \citep{Martcheva2020}, latent infectivity \citep{Safi2013}, hospitalisation procedures \citep{herdImmunity}, vaccination \citep{brauer2011backward,Giordano2021,Krueger2022,PanPox2022}, viral shedding into the environment \citep{Proverbio2022}, seasonality \citep{aron1984seasonality}.

\begin{figure}[ht!]
	\centering
	\includegraphics[width=0.8\textwidth]{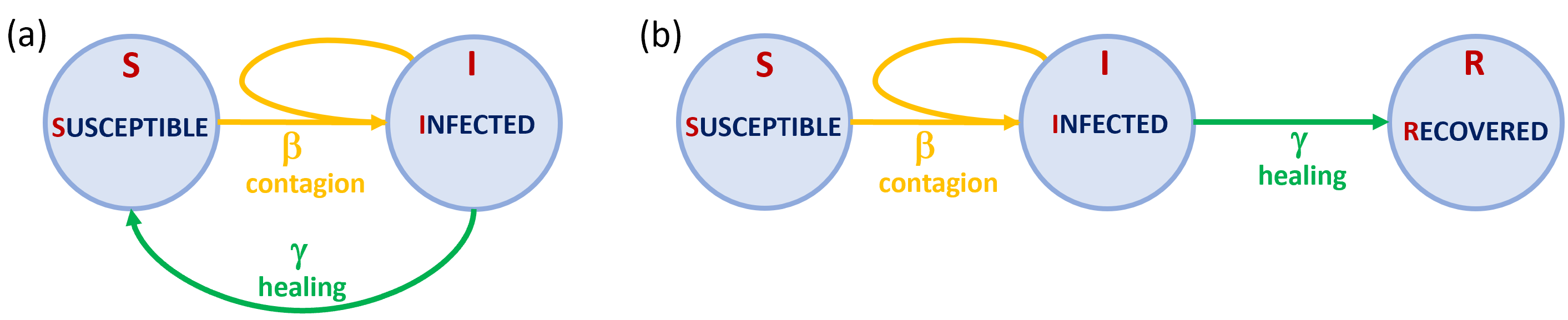}
	\caption{\footnotesize Graph representation of (a) the SIS model, Eq. \eqref{eq:sis}, and (b) the SIR model, Eq. \eqref{eq:sir}, akin to the flow graph representation of chemical reaction networks. Vertices represent compartments and edges represent flows of individuals governed by rate constants $\beta$ and $\gamma$.}
	\label{fig:epi-model}
\end{figure}

A general epidemiological model with an arbitrary number of infected and non-infected compartments \citep{CalaCampana2024} can be written as
\begin{equation}\label{eq:gencompmod}
\begin{cases}
\dot w(t) = G w(t) - \mbox{diag}(Cx(t))w(t) +Dx(t) + a\\
\dot x(t) = F x(t) + b\,w(t)^\top Cx(t) 
\end{cases}
\end{equation}
where $w(t), \, a\in \mathbb{R}_{\geq 0}^{n_2}$ and $x(t), \, b \in \mathbb{R}_{\geq 0}^{n_1}$; variables $w_i(t)$, $i=1,\dots,n_2$ represent population fractions in non-infected (\eg susceptible, recovered, vaccinated) compartments, while variables $x_i(t)$, $i=1, \dots n_1$, represent population fractions in different infected compartments. Matrix $G \in \mathbb{R}^{n_2 \times n_2}$ contains parameters related to vaccination and waning of immunity that drive the flows among non-infected compartments, while matrix $F \in \mathbb{R}^{n_1\times n_1}$ describes the flows among infected compartments, along with recovery.  Vector $a$ represents natality in the population, while $F$ and $G$ also include mortality rates. Matrices $F$ and $G$ are typically Metzler (\ie their off-diagonal entries are non-negative), to account for cooperativity; also, when the natality term $a$ is non-zero, $F$ and $G$ are typically Hurwitz (\ie all their eigenvalues have negative real part), to account for the dying out of the population in the absence of other terms, as well as the boundedness of the non-infected population (which is increased by $a$).
The infection matrix $C \in \mathbb{R}_{\geq 0}^{n_2 \times n_1}$ includes the contagion parameters,
while $D \in \mathbb{R}_{\geq 0}^{n_2 \times n_1}$ includes the recovery parameters.
The total population is often assumed constant over a finite time horizon $T>0$ (meaning that natality and background mortality compensate for each other, which is reasonable from an epidemiological perspective, provided $T$ is short enough and the population is large enough, so that small fluctuations in the population do not affect the spread of contagion): for all $t \in [0,T]$,
$\sum_{i=1}^{n_1} x_i(t) + \sum_{j=1}^{n_2} w_j(t) = 1$. Also, the system is positive: if $x(0), w(0) \geq 0$, then $x(t), w(t)\geq 0$ for all $t \geq 0$.

The SIS and SIR models, as well as more complex SIR-like models, can be seen as special cases of the general model \eqref{eq:gencompmod}.
In all the examples below, we assume that the population densities are normalised with respect to the total population size (\ie they amount to fractions of the population).

\begin{example}[\textbf{SIS model}]
The SIS model (Figure~\ref{fig:epi-model}a) can be written as
\begin{equation}\label{eq:sis}
\begin{cases}
    \dot{S}(t) = -\beta S(t) I(t) + \gamma I(t) \, ,\\
    \dot{I}(t) = \beta S(t) I(t) - \gamma I(t) \, ,
\end{cases}
\end{equation}
where $S$ is the fraction of susceptible individuals and $I$ of infected (and infectious) individuals (corresponding to \eqref{eq:gencompmod} with $w=S$, $x=I$, $G=0$, $C=\beta$, $D=\gamma$, $a=0$, $F=-\gamma$, $b=1$). The infection rate is modelled by $\beta$, while infected (and infectious) individuals recover, and become susceptible again, at rate $\gamma$. The total population is conserved, since $\dot S + \dot I \equiv 0$, and in particular $S(t) + I(t) = 1$ for all $t$.
\end{example}

\begin{example}[\textbf{SIR model}]
The SIR model (Figure~\ref{fig:epi-model}b) is given by
\begin{equation}\label{eq:sir}
\begin{cases}
    \dot{S}(t) = -\beta S(t) I(t) \, ,\\
    \dot{I}(t) = \beta S(t) I(t) - \gamma I(t) \, ,\\
    \dot{R}(t) = \gamma I(t) \, .
\end{cases}
\end{equation}
and includes the additional variable $R$ representing the fraction of removed individuals, no longer involved in the contagion process because either recovered (and permanently immunised) or dead; $\beta$ is the transmission rate, $\gamma$ the recovery rate. The system corresponds to \eqref{eq:gencompmod} with $w=\begin{bmatrix} S\\ R \end{bmatrix}$, $x=I$, $G=\begin{bmatrix} 0 & 0\\ 0 & 0 \end{bmatrix}$, $C=\begin{bmatrix} \beta\\ 0 \end{bmatrix}$, $D=\begin{bmatrix} 0\\ \gamma \end{bmatrix}$, $a=\begin{bmatrix} 0\\ 0 \end{bmatrix}$, $F=-\gamma$, $b=1$. Again, the total population is conserved, since $\dot S + \dot I + \dot R \equiv 0$, and in particular $S(t) + I(t) + R(t) = 1$ for all $t$.
\end{example}

\begin{example}[\textbf{Extended SIRV model}]\label{ex:SIRV}
The model
\begin{small}
\begin{equation}\label{eq_SIRV}
\begin{cases}
\dot{S}(t) =  \psi - u S(t) + \chi_1 R(t) - \beta S(t)I(t) + \chi_2V(t)- \mu_S S(t)  \\
\dot{I}(t) = \beta S(t) I(t) -\gamma I(t)  + \pi_1\beta R(t)I(t) + \pi_2 \beta V(t)I(t) -\mu_I I(t) \\
\dot{R}(t) = \gamma I(t) -\chi_1R(t)  - \pi_1\beta R(t)I(t) -\mu_R R(t)\\
\dot{V}(t) = u S(t) -\chi_2V(t) - \pi_2\beta V(t)I(t) - \mu_V V(t)
\end{cases}
\end{equation}
\end{small}
considers susceptible ($S$), infected ($I$), recovered ($R$) and vaccinated ($V$) individuals; $\psi$ is the birth rate, $\mu_i$ are the mortality rates within each compartment, $\beta$ is the transmission rate, $\pi_1,\pi_2 \in [0,1]$ are the (re)infection probabilities for recovered and vaccinated, $u$ is the vaccination rate, $\gamma$ is the recovery rate and $\chi_1,\chi_2$ are the rates at which immunity wanes. All the parameters are positive constants.
System \eqref{eq_SIRV} can be cast in the form \eqref{eq:gencompmod}, with $x = I$ and $w = [S~R~V]^{\top}$, by taking $F = -\gamma-\mu_I$, $b=1$, $C = [\beta~\pi_1 \beta~\pi_2 \beta]^{\top}$, $G =\begin{bsmallmatrix}
-u -\mu_S &\chi_1&\chi_2\\ 0&-\chi_1 -\mu_R&0\\ u&0&-\chi_2 -\mu_V
\end{bsmallmatrix} $, $D = [0~\gamma~0]^{\top}$, $a=[\psi~0~0]^\top$ \citep{CalaCampana2024}.
\end{example}

\begin{example}[\textbf{Extended SEIRV model}]
The model
\begin{small}
\begin{equation}\label{eq_SEIRV}
\hspace{-2mm}\begin{cases}
\dot{S}(t) =  \psi - u S(t) + \chi_1 R(t) - \beta S(t)I(t) + \chi_2V(t) -\mu_S S(t) \\
\dot{E}(t) = \beta S(t)I(t) + \pi_1 \beta R(t)I(t) +  \pi_2\beta V(t)I(t) - \rho E(t) -\mu_E E(t)  \\
\dot{I}(t) = \rho E(t) -\gamma I(t) -\mu_I I(t) \\
\dot{R}(t) = \gamma I(t) -\chi_1R(t)  - \pi_1\beta R(t)I(t)- \mu_R R(t) \\
\dot{V}(t) = u S(t) -\chi_2V(t)  - \pi_2\beta V(t)I(t) - \mu_V V(t)
\end{cases}
\end{equation}
\end{small}
considers susceptible ($S$), recovered ($R$) and vaccinated ($V$) individuals along with two infected compartments, exposed ($E$), who are not yet infectious (\ie contagious), and infectious ($I$).
Exposed individuals become infectious at rate $\rho>0$. The other coefficients have the same meaning as in the extended SIRV model in Example~\ref{ex:SIRV}.
System \eqref{eq_SEIRV} can be written as in \eqref{eq:gencompmod} with $x = [E~I]^{\top}$, $w = [S~R~V]^{\top}$,
$F = \begin{bsmallmatrix}
-\rho- \mu_E &0 \\ \rho & -\gamma -\mu_I
\end{bsmallmatrix}$, $G =\begin{bsmallmatrix}
-u -\mu_S &\chi_1&\chi_2\\ 0&-\chi_1-\mu_R&0\\ u&0&-\chi_2-\mu_V
\end{bsmallmatrix} $, $C = \begin{bsmallmatrix}
 0 &\beta \\0& \pi_1 \beta   \\ 0&\pi_2 \beta 
\end{bsmallmatrix}$, $D = \begin{bsmallmatrix}
0 &0 \\ 0&\gamma  \\ 0&0
\end{bsmallmatrix}$, $b=[1~0]^\top$ and $a=[\psi~0~0]^\top$ \citep{CalaCampana2024}.
\end{example}

\begin{example}[\textbf{Extended SIDARTHE-V model}]
The SIDARTHE model proposed by \cite{Giordano2020} is tailored to the specificities of the COVID-19 pandemic and discriminates between infected individuals depending on whether they have been diagnosed and on the severity of their symptoms, resulting in five infected compartments: asymptomatic undetected infected ($I$), asymptomatic diagnosed infected ($D$), symptomatic undetected infected ($A$), symptomatic diagnosed infected ($R$), diagnosed infected with life-threatening symptoms ($T$).
The model has been expanded to consider vaccination by \cite{Giordano2021}, so that the non-infected categories are, in addition to susceptibles ($S$), recovered ($H$) and deceased ($E$), also vaccinated ($V$). Here, we further introduce waning immunity with rates $\chi_i$ and (re)infection probabilities $\pi_i \in [0,1]$ for recovered and vaccinated, as well as a birth rate $\psi$ and non-COVID-related mortality rates $\mu_i$, and we denote by $u$ the vaccination rate. The resulting model \citep{CalaCampana2024} is

\begin{scriptsize}
     \begin{equation*}\label{eq:SIDARTHEV}
         \begin{cases}
             \dot{S} &= \psi -S(\alpha I +\beta D +\gamma A +\delta R) - uS +\chi_1 H + \chi_2V -\mu_S S\\
             \dot{I} &= S(\alpha I +\beta D +\gamma A +\delta R) -(\varepsilon +\zeta +\lambda)I +\pi_1\alpha HI +\pi_2\alpha VI + \pi_3\beta HD + \pi_4\beta VD \\
             &+ \pi_5 \gamma HA + \pi_6 \gamma VA + \pi_7 \delta HR + \pi_8\delta VR  -\mu_I I\\
             \dot{D} &= \varepsilon I -(\eta +\rho) D -\mu_D D\\
             \dot{A} &= \zeta I - (\theta +\mu+\kappa)A -\mu_A A\\
             \dot{R} &= \eta D +\theta A -(\nu +\xi +\tau_1)R -\mu_R R\\
             \dot{T} &= \mu A +\nu R -(\sigma +\tau_2)T -\mu_T T\\
             \dot{H} &= \lambda I +\rho D +\kappa A +\xi R +\sigma T -\chi_1H-\pi_1\alpha HI -\pi_3\beta HD -\pi_5\gamma HA -\pi_7 \delta HR -\mu_H H\\
             \dot{E} &= \tau_1 R +\tau_2 T\\
             \dot{V} &= u S -\chi_2V -\pi_2 \alpha VI -\pi_4\beta VD -\pi_6\gamma VA -\pi_8\delta VR -\mu_V V
         \end{cases}
     \end{equation*}
\end{scriptsize}

The parameters $\alpha$, $\beta$, $\gamma$ and $\delta$ denote transmission rates; $\varepsilon$ and $\theta$ diagnosis rates; $\zeta$ and $\eta$ symptom onset rates; $\mu$ and $\nu$ aggravation rates; $\tau_1$ and $\tau_2$ COVID-related mortality rates; $\lambda$, $\kappa$, $\xi$, $\rho$ and $\sigma$ recovery rates. Also this system can be written as in \eqref{eq:gencompmod}.

The state is $x = [I~D~A~R~T]^\top$ and $w=[S~H~E~V]^\top$, while
$a = [\psi~0~0~0~0]^\top$, $b = [1~0~0~0~0]^\top$, and the matrices are
$$F=
\begin{bsmallmatrix}
        -(\varepsilon+\zeta+\lambda +\mu_I) &0&0&0&0 \\
        \varepsilon & -(\eta+\rho+\mu_D) &0&0&0\\
        \zeta & 0& -(\theta+\mu+\kappa +\mu_A) & 0&0\\
        0&\eta & \theta & -(\nu+\xi +\tau_1 +\mu_R) & 0\\
        0&0&\mu&\nu& -(\sigma +\tau_2 +\mu_T)
    \end{bsmallmatrix},$$
$$C = 
    \begin{bsmallmatrix}
        \alpha & \beta & \gamma & \delta & 0 \\
        \pi_1\alpha & \pi_3 \beta& \pi_5\gamma & \pi_7\delta&0 \\
        0 & 0& 0&0&0 \\
        \pi_2\alpha & \pi_4\beta& \pi_6 \gamma& \pi_8\delta&0 
    \end{bsmallmatrix}, \qquad\qquad
    D = \begin{bsmallmatrix}
        0 & 0& 0&0&0 \\
        \lambda & \rho & \kappa & \xi & \sigma \\
        0&0&0&\tau_1&\tau_2\\
        0 & 0& 0&0&0 \\
    \end{bsmallmatrix},$$
 $$G = \begin{bsmallmatrix}
        -u-\mu_S &\chi_1&0&\chi_2 \\
        0 & -\chi_1-\mu_H& 0&0\\
        0 & 0& 0&0\\
        u & 0 & 0& -\chi_2-\mu_V\\
    \end{bsmallmatrix}.$$
\end{example}

\begin{remark}[\textbf{Epidemic models as chemical reaction systems}]\label{rem:epidemicCRN}
As is well known (see, \eg \cite{blanchini2021structural}), compartmental epidemic models can be interpreted in terms of chemical reactions (and visualised through the same type of flow graphs, as in Figure~\ref{fig:epi-model}) with mass action kinetics. For instance, the SIS model corresponds to chemical reactions $S + I \react{\beta} 2I$ and $I \react{\gamma} S$, while the SIR model corresponds to chemical reactions $S + I \react{\beta} 2I$ and $I \react{\gamma} R$. Hence, they can always be cast in the form \eqref{eq:ode-system}.
\end{remark}

Heterogeneous disease transmission in distinct geographical areas (\eg a network of cities connected by individual commuting and mobility), or in different age classes, can be captured by \emph{meta-population} or \emph{multi-patch} compartmental models, which describe the spreading process through SI-like models defined on a graph: each node (associated with a patch) represents the local dynamic evolution of the epidemic phenomenon through an SI-like compartmental model, while the interconnections represent population mobility and cross-interactions that couple contagion dynamics at a larger scale \citep{bichara2018multi,DellaRossa2020,Grenfell1997,HernandezVargas2022}; see, \eg Figure~\ref{fig:multipatch-model}. An example of meta-population SIS model \citep{boguna2013nature} is
\begin{equation}
\dot{I}_k(t) = -\gamma_k I_k(t) + \sum_{j=1}^N \beta_{kj} [1-I_k(t)] I_j(t), \quad k=1,\dots, N \, ,
\end{equation}
where $I_k$ represents the fraction of infected and $S_k=1-I_k$ the fraction of susceptibles in the $k$th patch, $\gamma_k$ is the recovery rate in the $k$th patch, and $\beta_{kj}$ is the transmission rate for contacts between an individual in the $k$th patch and an individual in the $j$th patch (and is possibly equal to zero if the two patches are non-communicating).

Multi-patch and meta-population models \citep{Aalto2025,Ball2014,Bertuzzo2020,DellaRossa2020,Gatto2020,Grenfell1997,Rowthorn2009} that account for the heterogeneity of the epidemic evolution in distinct geographical areas and for the mobility of individuals, thus capturing the hierarchical levels of the epidemic phenomenon at various spatial scales (city, province, region, country, continent), are fundamental to support the planning of coordinated interventions with a national or continental perspective \citep{Priesemann2020,Priesemann2021,Valdez2022}.

\begin{figure}[ht!]
	\centering
	\includegraphics[width=\textwidth]{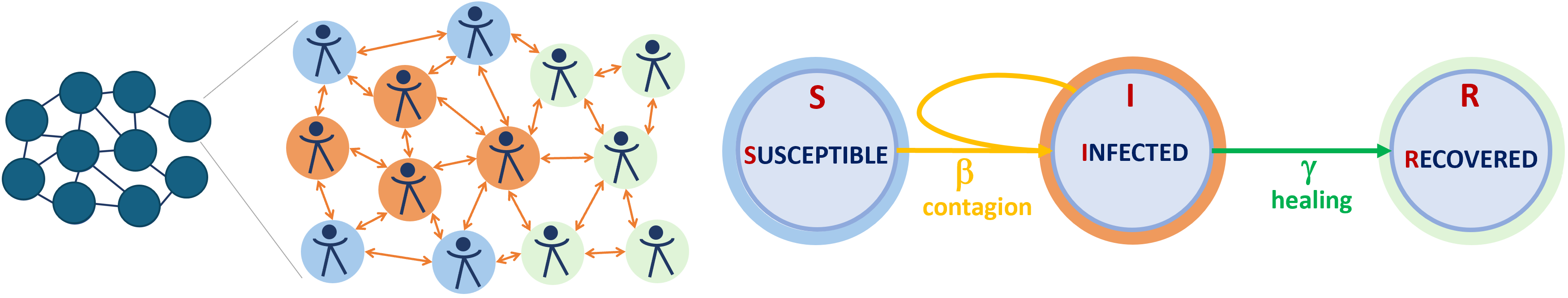}
	\caption{\footnotesize A multi-patch, or meta-population, networked model for epidemics capturing the heterogeneous disease transmission in different regions, or age classes (patches, associated with the blue nodes). The dynamics of each patch is described by a different SIR model with specific parameters, and patches are affecting one another through cross-coupling.}
	\label{fig:multipatch-model}
\end{figure}

In addition to between-host models that describe the evolution of contagion in a population, in-host models can be developed to describe the evolution of the infection within an individual \citep{Almocera2018,Almocera2019,Abuin2020,Hernandez-Vargas2020,HernandezVargas2022} and to support optimal treatment planning for bacterial \citep{deJong2023,AKCG2025} and viral \citep{Hernandez-Vargas2010,Hernandez2014} infections, also building upon results on the optimal control of compartmental models \citep{BBCDG2021,Blanchini2023}. Analogous approaches can be adopted, \eg for non-infectious diseases like cancer \citep{DG2019,Giordano2016convex}.
Multi-scale models can be developed that bridge the microscopic in-host scale and the macroscopic between-host scale \citep{Almocera2018,Almocera2019,Bellomo2020,Feng2011,Feng2015,Gandolfi2014,Garabed2019,Hart2020,Murillo2013}, and study how the former affects the latter \citep{Cai2017} by considering the interplay between the immunological mechanisms in the host and the epidemiological mechanisms of contagion between hosts.

\subsection{Structure, networks and graphs}

We call \emph{structure} the qualitative (parameter-free) description of a whole family of systems, defined without resorting to any arbitrary parameter bound \citep{blanchini2021structural}; each element of the family can be obtained via a specific choice of the involved functions and numerical parameter values (or bounds). We thus discriminate between qualitative and quantitative models depending on whether precise numerical values are specified, even though we always consider mathematical models (which are still quantitative in comparison to descriptive models used in biology and epidemiology, such as western blot analysis or verbal descriptions according to which a concentration may ``be increasing'', ``be lower than...'', see, \eg \cite{nemudryi2020temporal}). 

A structure can be effectively visualised using graphs, which are versatile and general tools \citep{Boccaletti2006,gross2005graph} to describe the complex interplay of interactions giving rise to biological and epidemiological systems. Here, we briefly recall some graph-theoretic terminology, which is sometimes employed with different meaning even in sibling disciplines. 

A \emph{graph} is a triple $(V,E,W)$, where $V=\{v_i\}_{i=1}^{n_V}$ and $E=\{e_i\}_{i=1}^{n_E}$ are finite sets of, respectively, $n_V$ vertices (or nodes) and $n_E$ edges (or arcs, or links), where edges correspond to ordered pairs of vertices, and $W \in \mathbb{R}^{n_V \times n_V}$ is the adjacency matrix, whose entry $W_{ij}$ corresponds to the weight of the edge $e_k=(i,j)$ with origin vertex $i \in V$ and destination vertex $j \in V$; if $(i,j) \not\in E$, then $W_{ij}=0$. This framework can represent both undirected and directed graphs, possibly including self-loops. A \emph{hyper-graph} $(V,E,M)$ also contains edges that involve more than two vertices (\emph{hyper-edges}); its topology is captured by the incidence matrix $M \in \mathbb{Z}^{n_V \times n_E}$, where $M_{ij}=-1$ if $i$ is an origin vertex for edge $j$, $M_{ij}=1$ if $i$ is a destination vertex for edge $j$, and $0$ otherwise; furthermore, the edges can be associated with weights. The concepts of hyper-graph and hyper-edge are fundamental for biological networks: for instance, flow graphs representing chemical reactions are hyper-graphs, unless all the reactions are unimolecular (namely, either of the form $X_i \react{} X_j$ or of the form $X_i \react{} \emptyset$); and, \eg the hyper-edge associated with reaction $X_1 + X_2 \react{} X_3 + X_4$ connects two origin vertices ($X_1$ and $X_2$) to two destination vertices ($X_3$ and $X_4$).

A \emph{labelled graph} is a graph with labelling functions $L_V \colon V \to D_V$ and $L_E \colon E \to D_E$, where $D_V$ and $D_E$ are arbitrary sets, called the vertex labels and the edge labels, respectively.
\emph{Signed graphs} are labelled graphs where each edge is associated with a positive or negative sign, depending on the nature of the interaction.

A \emph{network} is the applied counterpart of a mathematical graph, representing a model of a real entity with its inherent dynamics (associated with the nodes or with the links) and interaction functions.
A graph thus corresponds to the topology of a network, which also includes dynamical information (\eg parametrised evolution functions).
For a dynamical system endowed with a network topology, the \emph{structure} corresponds to the underlying graph, potentially augmented with qualitative information about the nature of the interactions (such as sign patterns, or monotonicity of the involved functions). Then, a \emph{family of systems} is characterised by a \emph{structure} if all the elements of the family are represented by the \emph{same graph} and their dynamics (\eg the involved functions) satisfy \emph{common qualitative assumptions}, such as positivity or monotonicity.

Network scientists often use a different terminology \citep{Barabasi2013}, as summarised in Table~\ref{table_netw_graph}, where the term \emph{network} covers both static representations and ``dynamical graphs''. The structure would still correspond to the topological, static wiring of nodes and links, \ie their qualitative pattern. We will adopt the definitions we previously proposed, and distinguish between structure (or graph structure) and network.

\begin{table}[ht!]
	\centering
	\begin{tabular}{c|c }
		\textbf{Network Science} & \textbf{Graph Theory}   \\ \cline{1-2}
		Network                  & Graph                   \\ 
		Node                     & Vertex                  \\
		Link                     & Edge (or arc)         \\
	\end{tabular}
\caption{\footnotesize Different terminology used in network science and in graph theory.}
\label{table_netw_graph}
\end{table}

As testified by the examples in this chapter, the meaning of graphs can vary depending on the context. Vertices are usually associated either with individual state variables (representing, \eg the concentration of a chemical species or the density of a population category) or with subsystems (\ie subsets of state variables, for instance in multi-patch systems). In the case of \emph{flow graphs}, edges represent the flow of species moving from some vertices to others (corresponding, \eg to chemical reactions in chemical reaction networks); flow graphs are the most common representation of compartmental models of epidemiological systems \citep{herdImmunity,kermack1927contribution}, and are used in systems biology to model, \eg chemical reactions \citep{Anderson2011} or molecular fluxes \citep{Sharma2016}. In the case of \emph{signal graphs}, which are signed graphs, the signed edges are associated with either inhibiting or activating interactions, so as to model gene regulatory networks \citep{Angeli2004detection, Edwards2015} or biological control loops \citep{alon2019introduction}.

\subsection{Stability, robustness and resilience across disciplines}
\label{sec:informal_defs}

Here we aim at exploring and comparing the concepts of stability, robustness and resilience for dynamical systems subject to uncertainties, stochastic noise and fluctuations, and at discussing their use in different research communities. Since perturbations are ubiquitous in any realistic application and the system's characteristic parameters are never known exactly, determining whether a system  can reliably maintain a certain performance (for example, reach a prescribed set in the state space) under the influence of perturbations or noise, or uncertain parameters,  is a long-standing research question in numerous disciplines, including physics, engineering and biology. A multitude of (either competing or complementary) concepts offer a lens to study these general questions. We now provide an informal overview of how these terms are understood by scholars working in different research fields: a literature review reveals that stability, robustness and resilience have conceptually different definitions and meanings across disciplines. Hence, an effort to unify these concepts within a single rigorous control-theoretic framework is of the utmost importance. 

\subsubsection{Stability}

The notion of stability introduced by Lyapunov in 1892 \citep{lyapunov1992general} is usually associated with the equilibrium points (\ie steady states) of a dynamical system, and captures the idea that trajectories emanating from points that are close enough to the equilibrium remain in its vicinity for all future times. Subsequent research in systems and control has been extremely successful in expanding the toolbox of stability concepts by incorporating additional ones, including asymptotic stability, finite-time stability, Bounded-Input-Bounded-Output stability \citep{carlson1998signal}, Input-to-State and Input-to-Output stability \citep{sontag2013mathematical}, and robust stability \citep{sanchezpena1998,zhou1998essentials}. However, these concepts have only partially permeated into other disciplines. For example, in ecology, the concept of stability is often used informally to cover other (eco)system properties, such as the ability to maintain function or performance, returning to a prescribed equilibrium state, despite disturbances \citep{arnoldi2016resilience,grimm1997babel, pimm1984complexity}, \ie notions similar to disturbance rejection in control theory. Similarly, biological studies such as that by \cite{Dai2015} associate stability with the \emph{rate} of return to the rest state, a concept that overlaps with the notion of resilience in engineering applications, as we discuss later. These are just two examples demonstrating the need for a unified framework of concepts across the span of multiple research communities.

\subsubsection{Robustness}

In systems and control theory, robustness roughly refers to the ability of a system to maintain a desired performance (or, more generally, the ability of a system property to persist), despite uncertainties in the model and/or in the system parameters (as well as despite the introduction of disturbances), provided an uncertainty (or disturbance) bounding set is assumed \emph{a priori} \citep{Barmish1994, ZHANG201041}. The underlying motivation is to guarantee that the system, or closed-loop system, in question is reliable in achieving the necessary performance. Robustness can be further extended to \emph{structural} robustness \citep{Blanchini2011}, leading to the concept of structural property \citep{Blanchini2012,blanchini2021structural}, which adopts a different perspective: it does not consider uncertainty bounding sets and focuses on the sole knowledge of the structure of a dynamical system so as to investigate its properties. A prototypical example of structural robustness can be found in networked systems, where the graph topology plays a crucial role in determining the dynamic behaviour and performance \citep{blanchini2021structural}; the concepts of robust and structural properties will be further discussed in Chapter~\ref{ch:struct-an}.

Robust control and estimation methods have been successfully used for control and surveillance of epidemic spreading \citep{Alamo2021,Alamo2022,HernandezVargas2022,Lee1995,Leitmann1998,Morris2021}, and biological systems are often studied within a robustness framework \citep{Khammash2016,Kim2006, Kitano2004,Kitano2007towards}. For instance, robustness of protein and genetic regulation is hypothesised to promote evolvability of living systems \citep{bloom2006protein, Kwon2008,masel2010robustness}. In other sub-fields of biology, however, the notion of robustness has been used to encompass primarily \emph{structural stability} \citep{Nikolov2007, shil2001methods}, \ie the guarantee of a unique stable steady state despite parametric uncertainty. In a wider (but often mathematically indefinite) sense, it is also interpreted as the ability of biological systems to maintain several properties, withstanding the influence of a great number of external perturbations \citep{Kartal2009, Kitano2004b, Lesne2008, Whitacre2012}. 

In network science, robustness characterises the ability of a system to withstand failures and perturbations, such as link removal or dismantling \citep{artime2024robustness}, without loss of function. Here, the notion is usually tied to the topology of the associated network, by quantifying its ability to maintain connectivity when a fraction of the nodes, or links, are damaged or removed \citep{liu2015vulnerability, ren2019generalized, zdeborova2016fast}. Robustness is thus defined in terms of the critical number of node (or link) removals that make the graph disconnected, and is associated with the property of connectivity as measured by, \eg the diameter or size of the largest connected component \citep{jeong2000large}. 
This perspective strongly overlaps with the control-theoretic notion of robustness. Indeed, networks can be seen as weighted directed graphs, where the edge weights represent the \emph{strength} of the connection between the corresponding nodes; then, a family of realisations of a given network can be seen as an uncertain family of graphs, with the uncertainty bounding set representing all possible weights that can be assigned to the different edges. For a rich description of recent works on network robustness from a network science perspective, we refer the interested reader to \cite{artime2024robustness} and \cite[Section 6]{Liu2020b}. 

\subsubsection{Resilience}
\label{sec:res_informal}

The notion of resilience is heuristically defined in several disciplines \citep{bhamra2011resilience,Crespi2021, fraccascia2018resilience,Krakovska2024,Vitousek2025} and is often used as an umbrella term \citep{baggio2015boundary,standish2014resilience}, to the point of covering up to 70 definitions \citep{fisher2015more}. Overall, its definitions span various nuances that fall somewhere between the following three key interpretations: 
\begin{itemize}
    \item \emph{System resilience} (often called ``\emph{ecological resilience}'', \cite{Holling1973}, see 	also \cite{dakos2022ecological,gunderson2000ecological}) focuses on the magnitude of the change or perturbation that a system can withstand while still returning to its original state, or maintaining its original function, without transitioning to another stable regime (\ie attractor). From a control-theoretic perspective, this interpretation resembles the notions of a stability radius under structured/unstructured perturbations and $\mu$-analysis (see, \eg \cite{doyle1982analysis}).
    \item \emph{Engineering resilience} corresponds to fault tolerance or robustness to extrinsic disturbances, usually measured in terms of recovery rate, \ie how quickly the disturbance is damped and the system goes back to its original equilibrium state  \citep{holling1996engineering}; see also the overview by \cite{Clement2021}.
    \item In ecology and evolutionary theory, the notion of \emph{adaptive resilience} corresponds to the capacity of a system to adapt or evolve (over possibly long time spans, as in Darwinian selection or evolution over generations) in response to unfamiliar events or changes in operating conditions. For example, \cite{carpenter2012general} study how ecological systems such as coral reefs survive and recover after extreme shocks like hurricanes.
\end{itemize}
The majority of other definitions lie in between these three, and sometimes overlap with the concept of robustness \citep{carpenter2001metaphor,Liu2020b}. For instance, ecological resilience \citep{scheffer2015generic} tends to overlap with biological robustness: in the case of systems with multiple stable equilibria and subject to exogenous disturbances, seen as externally-controlled parameters, the resilience of a noise-free (nominal) system may be quantified through the distance from a critical set of parameters for which the system experiences a transition from one equilibrium (or attractor) to another. 

Both in biology and epidemiology, resilience is often associated with the magnitude of an exogenous signal that allows a system to spring back to a prescribed (original) stable steady state, after the system has tipped onto another stable steady state \citep{lemmon2020achieving}; resilient systems are those requiring minimal effort, by external governors, to return to the original state. Conversely, the concept of engineering resilience does not usually consider the case of alternative stable states.
Finally, in network science, resilience is defined as the ability to avoid dysfunction despite node attacks \citep{Loppini2019, mugisha2016identifying, zhang2020resilience} and to maintain function despite noise and out of equilibrium dynamics \citep{Bhandary2021}. Through this lens, resilience relates to the dynamics associated with specific links and nodes, rather than with topological disruptions to the network.

The concept of resilience, along with the introduction of novel formal definitions from a control-theoretic perspective, will be at the heart of Chapter~\ref{ch:rob-res-sta-mod}.

\subsubsection{A comparative overview}

Our comparative discussion of the concepts of stability, robustness and resilience, including their differences and potential overlaps, is summarised in Table~\ref{tab:summary-concepts}, while Figure~\ref{fig:venn} provides a visualisation of how these concepts may intersect across disciplines. Despite the plethora of definitions and the partial overlaps between the concepts, some overarching themes can be identified. In systems biology, ecology and network science, stability generally refers to asymptotic behaviours and existence of equilibria, or attractors; robustness is associated with ``static'' properties, such as network topology; while resilience encompasses non-equilibrium phenomena, including noisy fluctuations and the existence of alternative regimes. Moreover, as also noted by \cite{zhu2024disentangling}, resilience and robustness address two distinct but intertwined goals in control: resilience covers unforeseen events and recovery, while robustness is intended to manage anticipated uncertainties. Consequently, the different concepts are not mutually exclusive, but are complementary to understand system properties. Moreover, having established one property among stability, robustness and resilience (\eg stability) may promote the inference of another of these three properties (\eg resilience), or enable the derivation of indicators to measure it, as we discuss in Chapter~\ref{ch:res-ind}. Although we made an attempt to associate existing notions of resilience with suitable concepts from systems and control theory, a formal and precise mapping between various definitions of resilience across different disciplines does not hitherto exist, as can be seen by the multitude of available definitions and interpretations. Indeed, from our survey, the need emerges for a unified set of (formal) definitions, which will foster multidisciplinary inquiries into novel research directions.

\begin{table}[ht!]
	\centering \footnotesize
\renewcommand{\arraystretch}{1.4}
 \resizebox{\textwidth}{!}{
	\begin{tabular}{p{0.12\textwidth} | p{0.21\textwidth} | p{0.21\textwidth}| p{0.21\textwidth}|p{0.21\textwidth}| }
\cline{2-5}
		                             & Popular definition     & Ecology, biology, epidemiology
                               & Network science  & Control theory  \\ \cline{1-5}
		Stability                 & Persistence of function. &  Rate of return to rest state.   & Persistence of an equilibrium. & Trajectories converging to (asymptotic stability), or remaining within a given distance from (Lyapunov stability), an equilibrium or set.
                 \\ \cline{1-5}
		Robustness                & Insensitivity to uncertainty or perturbations.   &   Ability to absorb perturbations and persist.  & Network ability to maintain structural functions despite perturbations such as link removal. &   Preservation of a property within a family of systems (in spite of uncertainty or perturbations).
         \\ \cline{1-5}
		Resilience                & Prompt recovery after shock.   & Distance from tipping point, or ability to spring back after dysfunction.    &  Ability to maintain function despite perturbation of dynamical properties. & Fault tolerance and readiness to recover after an external perturbation or disturbance. \\ \cline{1-5}
	\end{tabular}
 }
 \caption{\footnotesize Summary of concepts and their definitions in various disciplines. }
 \label{tab:summary-concepts}
\end{table}

\begin{figure}[ht!]
	\centering
	\includegraphics[width=0.45\textwidth]{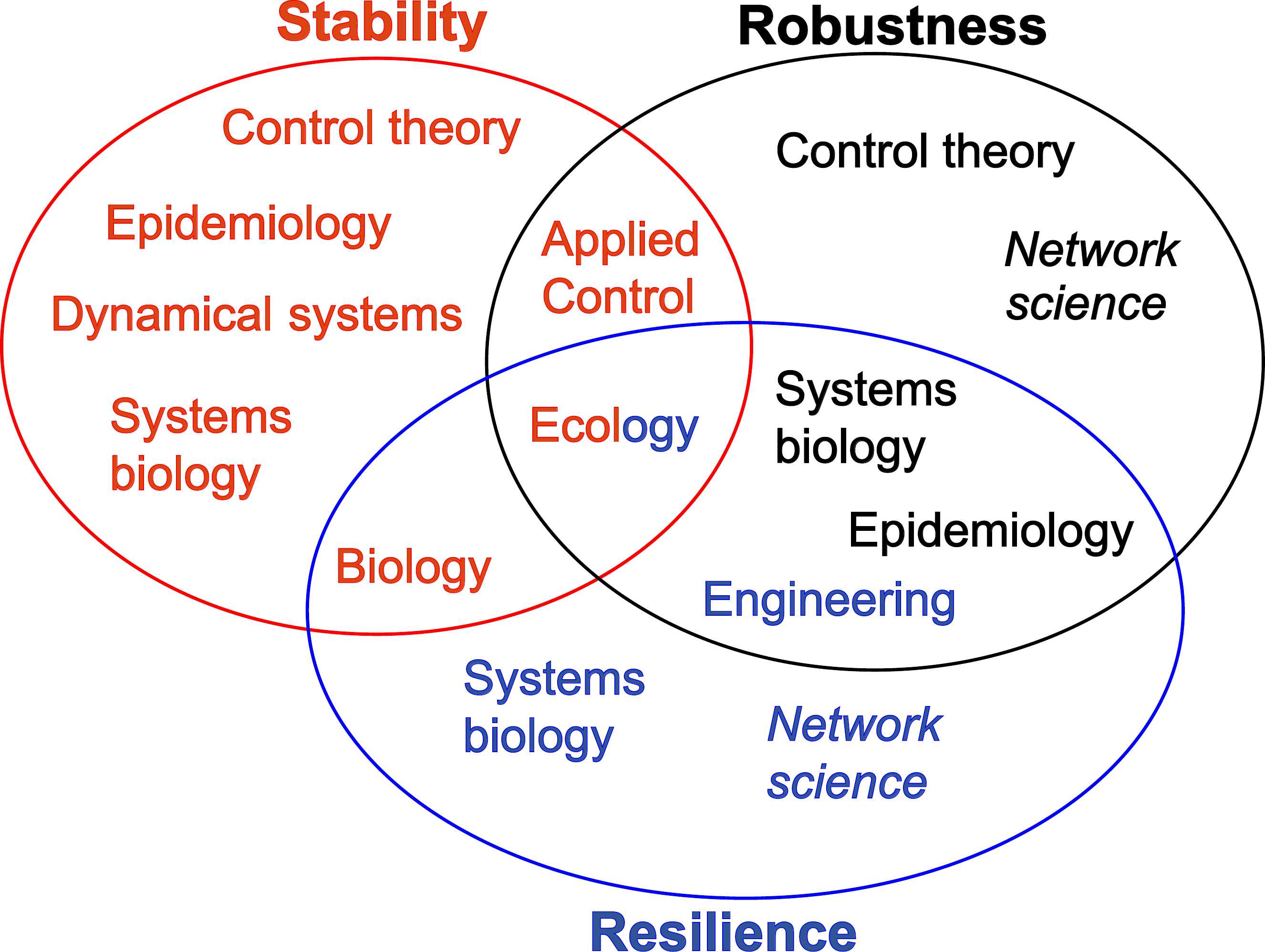}
	\caption{\footnotesize Venn diagram that visualises the overlap of concepts related to stability, robustness and resilience (according to our definitions) across disciplines. Certain disciplines may denote some concept with a nomenclature (denoted by the colour) that is overlapping with other meanings (denoted by the position in the diagram). In ecology, the three concepts are used, although the terms stability and resilience are often used interchangeably. In biology, the term stability is often used to denote resilience. In systems biology, the terms stability and resilience are mostly used properly, but the term robustness is often also used as a synonym for resilience. Also in epidemiology, the term robustness is often used to mean resilience. In control theory, robustness coherently refers to robust/structural properties, while formal notions of resilience are missing. When control strategies are designed outside the formal domains of control theory, e.g., in biology or applied physics (which we dub “applied control” for brevity), the term stability is often used to mean robustness. Branches of engineering outside control engineering may use the term resilience to mean robustness. Network theory uses the terms robustness and resilience without ambiguity, but with two clear distinct meanings that differ from the definitions considered in this monograph (hence, the italic font; see also Table~\ref{tab:summary-concepts}). The figure summarises mainstream trends in the considered disciplines, more deeply discussed in Table~\ref{tab:summary-concepts} and in the main text; individual authors and articles may use different nuances.}
	\label{fig:venn}
\end{figure}

We finally comment on the use of the above terms in management, which often percolates into public discourse and social healthcare \citep{Blanchet2017}. In this context, stability is often described as  persistence of performance, as opposed to growth; resilience corresponds to performance degradation followed by recovery; robustness involves insensitivity to uncertainty, and the additional concept of anti-fragility relates to the ability to prosper and benefit from adversity, \ie performance gain when exposed to adversity \citep{munoz2022resilience}. Intuitively, these definitions can also be associated with corresponding definitions in systems and control, through the introduction of a suitable performance index and a mathematical formalisation. However, we will not proceed with this inquiry, which is out of the scope of this monograph, and we leave it as a direction for future exploration.

It is important to stress that, when mathematical models are available, the concepts of stability, resilience and robustness can be rigorously formalised. The confusion between the terms arises when such models are \emph{not} available, and researchers must resort to verbal descriptions, where the distinctions between the different properties become blurred and difficult to assess without ambiguity. Although this ambiguity may be unavoidable in empirical research, in the next chapters we will show that, in the presence of mathematical models, the concepts of stability, robustness and resilience can be formally defined, and the resulting definitions can be leveraged to gain deeper insight into the system properties. In turn, these definitions will foster the formulation of testable hypotheses and predictions, and promote more precise modelling approaches in both qualitative and quantitative disciplines.

\newpage
\section{Structural analysis and robust analysis}
\label{ch:struct-an}

Systems in nature, at all scales, are characterised by a remarkable complexity and, at the same time, by an astounding robustness, in spite of huge uncertainties, variability, and/or environmental fluctuations. Their study is of fundamental importance to unveil the mechanisms of life and understand, \eg cellular dynamics, the onset of diseases, the spread of infections and epidemics, but it is also extremely challenging due to their inherent complexity, their intrinsically nonlinear behaviours and the many parameters at play, whose intertwined effects are often hard to disentangle.

In the life sciences, parameter-dependent numerical simulations are typically used to predict the behaviour of natural systems case by case. However, the exact models and parameter values are often poorly known, and extensive simulation campaigns only offer quantitative snapshots of the system behaviour within the parameter landscape, which cannot provide any theoretical guarantees or insight into the system functioning. On the other hand, qualitative behaviours are often preserved even with
huge parameter variations, because they rely on the system's \emph{interconnection structure}. 

Parameter-free \emph{structural approaches} are powerful to check whether a property is preserved for a whole \emph{family of uncertain systems}, exclusively due to its structure.
Rather than testing specific systems with given parameters, structural methods provide criteria that are independent of parameter values and specific modelling choices.
These methods help reveal cases in which crucial qualitative behaviours are maintained, despite time-varying, uncertain or unknown parameters, thanks to peculiar motifs \citep{Alon2007,Milo2002a,Stone2019,yeger2004} in the overall system structure: the interconnection topology alone, along with qualitative features of the individual dynamic units, guarantees that essential life-preserving properties hold regardless of parameter values. Structural analysis \citep{blanchini2021structural} has been successfully employed to provide a parameter-free assessment of properties, to better understand biological and physiological systems and to support the synthesis of biomolecular circuits.
Many properties of biological interest have been structurally assessed, including stability \citep{Al-Radhawi2016,Blanchini2014d,Blanchini2017f,Clarke1980}, bistability and multistability \citep{Angeli2004detection,Blanchini2014structural,Cosentino2012b,Kaufman2007,Mincheva2008,Snoussi1998,soule2003}, onset of persistent oscillations \citep{Blanchini2014structural,Culos2016,Katz2025osc,Mincheva2008,Snoussi1998}, and input-output influences among the dynamic components of the network \citep{Feliu2019,Giordano2016,Mochizuki2015,Sontag2014,Vassena2023},
as well as emergence and evolution of patterns in the space-time dynamics \citep{Arcak2013,Hori2015}, possibly leading to synchronous behaviours \citep{Bagheri2008,Scardovi2010}. 
Several methods have been developed for structural analysis, based for instance on algebraic geometry and graph theory \citep{Brualdi1997,gross2005graph,Mincheva2008,Radde2009,Radde2010,schaft2013,schaft2015},
including graphical methods such as signed loop analysis \citep{Angeli2009b,Gouze1998,Radde2010,Snoussi1998,soule2003},
qualitative methods \citep{Jacquezt1993,Maybee1969},
sparsity and sign patterns \citep{Domijan2012,Culos2016,Katz2025osc}, topological degree theory \citep{Lloyd1978},
as well as parametric robustness \citep{Barmish1994,sanchezpena1998,zhou1998essentials},
set-theoretic and Lyapunov methods \citep{Al-Radhawi2016,Blanchini2014d,Blanchini1999,BM2015},
dynamic immersion in differential inclusions \citep{AubinCellina},
positive \citep{BermanPlemmons1994} and monotone systems \citep{Hirsch2006b,Leenheer2007,Smith1995}.
The theory of chemical reaction networks has provided many structural results \citep{Angeli2009c,Feinberg2019,Horn1972,Horn1973}, including the celebrated Deficiency Zero Theorem by \cite{Feinberg1987}.

Structural analysis aims at checking whether a property of interest holds \emph{structurally}, for all the elements of a family of systems characterised by a common structure. The answer is either ``yes'' or ``no''. A ``yes'' answer offers a powerful and rigorous proof of the extraordinary parameter-free robustness of a whole family of systems, along with very strong theoretical guarantees on the behaviour of all systems within the family, independent of parameter values.
However, the answer is often ``no'', and then some (but not all) elements of the family may still enjoy the property. How many? With which probability? In this case, structural analysis leaves us clueless and robustness analysis -- based on either worst-case \citep{Barmish1994} or probabilistic \citep{Metropolis1953,tempo2013randomized,Sagar1998,Sagar2011} approaches -- is needed to assess whether the property is satisfied \emph{robustly} (\eg within a given subset of the parameter space, characterised by lower and upper bounds for the parameters) and quantify the degree of robustness (related, \eg to the size of the parameter subset in which the property holds, or, given some probability distributions for the parameters, by the probability with which the property holds). Also, in some cases a property consistently holds when repeatedly testing the system with random parameter values, but it cannot be proven to hold structurally, and it is then crucial to understand why; this situation suggests the existence of a zero measure set in the parameter space for which the property does not hold.

Here, we survey methodologies for the structural analysis of systems in nature and we highlight their ability to: offer a deeper insight into the fundamental properties of biological phenomena, which explains their robustness and the connection between dynamic behaviours and network topology \citep{Palumbo2005,stelling2002}; suggest the identification of therapeutic targets for disease treatment \citep{Saadatpour2011}; unravel the design principles and motifs \citep{alon2019introduction} selected by evolution to ensure the robustness and resilience of living organisms; distil structure-based design guidelines for biomolecular circuits in synthetic biology \citep{Baetica2019,DelVecchio2016,Hsiao2018,Koeppl2011} so that they exhibit the desired behaviour in view of their structure, regardless of parameter values and environmental fluctuations.
Also, since structural approaches reveal which dynamic behaviours are necessarily stemming out of the system structure, or incompatible with such a structure, they can enable model falsification, \ie model invalidation \citep{Anderson2009,Angeli2012,Bates2011,Porreca2012}. In fact, whenever experimental results yield a behaviour that is incompatible with the model structure, then the model is structurally unsuitable to represent the considered phenomenon: no matter how the model parameters are chosen, the behaviour observed in experiments cannot be reproduced. Analogous techniques can be used to compare alternative models proposed for the same phenomenon \citep{Anderson2009,Giordano2018,Hamadeh2011,Polynikis2009}.

Finally, this chapter discusses the limitations of structural approaches and outlines future research directions to integrate structural, robust and probabilistic methods, in order to gain a deeper understanding of systems for which no structural guarantees can be provided for the given property of interest. In these cases, it is fundamental to understand why the (desired or expected) property does not hold structurally, and to identify the system features that enable or prevent the property and the key parameters that must be finely tuned for the property to hold.

\subsection{Capture the structure: sign patterns and decompositions}

To assess properties that are independent of the specific parameter values and functional expressions, and exclusively rely on the \emph{system structure} (graph topology along with qualitative properties), we first need to formally capture the system structure.
As mentioned in Section~\ref{sec:dninb}, graphs -- both signal graphs and flow graphs -- are precious to visualise the structure of a system, and graphical methods are useful for its analysis; see, \eg \cite{Radde2009,Radde2010,Thomaseth2017}. Here, we briefly illustrate two possible mathematical characterisations of the system structure, in association with these two types of graphs.

\begin{remark}[\textbf{Embracing biology and epidemiology}]
Although the focus is often on chemical reaction networks and biological models, the methodologies discussed throughout this chapter embrace systems in a wide range of domains, including both biochemical and epidemiological models, as well as biological and ecological models. In fact, as recalled in Remark~\ref{rem:epidemicCRN}, any compartmental epidemic model can be seen as a chemical reaction network system of the form \eqref{eq:ode-system}, and hence all the mathematical results of chemical reaction network theory also apply to compartmental models in epidemiology. The same holds for ecological models, which can typically be modelled either as chemical reaction networks (see, \eg the prey-predator model discussed in Example~\ref{ex:LV}) or as activation-inhibition systems of the form \eqref{eq:modelclass}, to capture mutualism as well as competition among species.
\end{remark}

\subsubsection{System structure: sign pattern}\label{sec:signpattern}

For systems of the form \eqref{eq:modelclass}, namely $\dot x(t) = f(x(t),\vartheta)$, where $\vartheta$ is a vector of uncertain parameters, it can often be reasonably assumed that $f_i$, $i=1,\dots,n$, are monotonic functions of the state variables; an example is provided by activation-inhibition systems discussed in Section~\ref{sec:dninb}. Then, the system Jacobian $J(x,\vartheta)$ is a sign-definite matrix and can be associated with a \emph{sign pattern} matrix $\Sigma$, having entries $\Sigma_{ij} =\mbox{sign}(J_{ij}(x,\vartheta)) = \mbox{sign}(\partial f_i(x,\vartheta) / \partial x_j)$. Matrix $\Sigma$ captures the \emph{system structure} and is an equivalent representation of the system's signal graph: its positive (respectively, negative) entries are associated with positive (respectively, negative) edges in the signal graph, corresponding to activating (respectively, inhibiting) interactions.

\begin{example}[\textbf{Structure of the IFFL}]
The structure of the IFFL activation-inhibition network system \eqref{eq:IFFL}, whose signal graph is shown in Figure~\ref{fig:IFFLsignal}, can be described by the sign pattern matrix
\begin{equation*}
\Sigma_{\mbox{IFFL}} = \begin{bmatrix}
- & 0 & 0 \\
+ & - & + \\
- & 0 & -
\end{bmatrix}.
\end{equation*}
\end{example}

\begin{example}[\textbf{Repressilator and promotilator}]\label{ex:represspromot}
The repressilator (respectively, promotilator), described, \eg by \cite{alon2019introduction,DelVecchio2015,Elowitz2000,samad2005}, is an activation-inhibition network formed by a loop in which each node inhibits (respectively, activates) the next. In particular, the three-node repressilator system is
\begin{equation*}
\begin{cases}
\dot x_1(t) = -\mu_1 x_1(t) + g_{13}(x_3)\\
\dot x_2(t) = -\mu_2 x_2(t) + g_{21}(x_1)\\
\dot x_3(t) = -\mu_3 x_3(t) + g_{32}(x_2)
\end{cases}
\end{equation*}
and has sign pattern matrix
\begin{equation*}
\Sigma_{R} = \begin{bmatrix}
- & 0 & - \\
- & - & 0 \\
0 & - & -
\end{bmatrix},
\end{equation*}
while the three-node promotilator system is
\begin{equation*}
\begin{cases}
\dot x_1(t) = -\mu_1 x_1(t) + f_{13}(x_3)\\
\dot x_2(t) = -\mu_2 x_2(t) + f_{21}(x_1)\\
\dot x_3(t) = -\mu_3 x_3(t) + f_{32}(x_2)
\end{cases}
\end{equation*}
with sign pattern matrix
\begin{equation*}
\Sigma_{P} = \begin{bmatrix}
- & 0 & + \\
+ & - & 0 \\
0 & + & -
\end{bmatrix}.
\end{equation*}
\end{example}

\subsubsection{System structure: $BDC$-decomposition}

Consider a system of the form \eqref{eq:ode-system}, namely $\dot{x}(t) = Sg(x(t),\vartheta) + g_0(\vartheta)$, associated with a flow graph, and assume that $g_i$, $i=1,\dots,m$, are continuously differentiable functions with \emph{sign-definite partial derivatives}.
Then, the system structure can be effectively captured by its $BDC$-decomposition \citep{Blanchini2014d,Blanchini2017f,BG2019,blanchini2021structural,Giordano2016,GiuliaPhD2016}.

\begin{definition}[\textbf{Local $BDC$-decomposition}]\label{def_adm_BDC}
System \eqref{eq:ode-system} admits a \emph{local $BDC$-decomposition} if, for any $x \in \mathbb{R}^n_{\geq 0}$, the system Jacobian $J(x,\vartheta) = \frac{\partial Sg(x,\vartheta)}{\partial x}$ can be written as the positive linear combination of constant rank-one matrices:
\begin{equation}\label{eq:BDC_BDC_rank1}
J(x,\vartheta)= \sum_{h=1}^q B_h \Delta_h(x,\vartheta) C_h^\top= \sum_{h=1}^q R_h \Delta_h(x,\vartheta),
\end{equation}
where $B_h$ and $C_h^\top$ are constant column and row vectors, respectively, so that $R_h = B_h C_h^\top$ are constant rank-one matrices, while $\Delta_h(x,\vartheta)$, $h=1,\dots,q$, are positive scalar functions.
\end{definition}
We can compactly write $J(x,\vartheta) = B\Delta(x,\vartheta)C$, where $\Delta(x,\vartheta)$ is a diagonal matrix with positive diagonal entries $\Delta_h(x,\vartheta)$, $B$ is the matrix formed by the columns $B_h$ and $C$ is the matrix formed by the rows $C_h^\top$.

The definition holds for any non-negative $x$, hence, in particular, for any equilibrium point $\bar x(\vartheta)$ such that $S g(\bar x(\vartheta), \vartheta) + g_0(\vartheta) = 0$, which allows for a $BDC$-decomposition-based analysis around the equilibrium. 

As shown by \cite{Giordano2016}, a system \eqref{eq:ode-system} always admits a $BDC$-decomposition: matrices $B$ and $C$ can be built systematically, based on matrix $S$ and on qualitative information about $g(\cdot)$.

\begin{proposition}
Any system \eqref{eq:ode-system} admits a $BDC$-decomposition, $J(x,\vartheta) = B\Delta(x,\vartheta)C$, according to Definition~\ref{def_adm_BDC}.
\end{proposition}
\begin{proof}
The statement is proved constructively by \cite{Giordano2016}: equation~\eqref{eq:ode-system} is rewritten as
$\dot x =  \sum_{j=1}^s S_j~g_j(x,\vartheta) + g_0$,
where $S_j$ is the $j$th column of matrix $S$. The corresponding Jacobian is
\[J(x,\vartheta)= \sum_{j=1}^s S_j~ 
\left[ \frac{\partial g_j(x,\vartheta)}{\partial x_1}
~~\frac{\partial g_j(x,\vartheta)}{\partial x_2} ~~ \dots ~~
\frac{\partial g_j(x,\vartheta)}{\partial x_n} \right ] \ .\]
Denoting the absolute values of all the non-zero partial derivatives by $\Delta_1(x,\vartheta), \Delta_2(x,\vartheta),\dots, \Delta_q(x,\vartheta)$, we can write
$J(x,\vartheta)=  \sum_{h=1}^q B_h \Delta_h(x,\vartheta) C_h^\top$,
where
\begin{itemize}
\item $\Delta_h(x,\vartheta) = \left |\frac{\partial g_j(x,\vartheta)}{\partial x_i}\right|$ for some $i$ and $j$;
\item $B_h = S_j$, the column of $S$ associated with $f_j$;
\item $C_h^\top$ has a single non-zero entry in the $i$th position, equal to the sign of $\frac{\partial g_j(x,\vartheta)}{\partial x_i}$.
\end{itemize}
\end{proof}

\begin{example}\label{ex:BDC_CRN}
Consider the chemical reaction network in Example~\ref{ex:CRN}, assume that all the partial derivatives are positive and denote them as
$\alpha = \frac{\partial g_{a}(a,\vartheta)}{\partial a}$,
$\beta = \frac{\partial g_{b}(b,\vartheta)}{\partial b}$,
$\gamma = \frac{\partial g_{ac}(a,c,\vartheta)}{\partial a}$,
$\zeta = \frac{\partial g_{ac}(a,c,\vartheta)}{\partial c}$.
The Jacobian matrix and its $BDC$-decomposition are
\begin{equation*}
J =
\begin{bsmallmatrix}
-(\alpha+\gamma) & 0 & -\zeta \\
\alpha & -\beta & 0 \\
\alpha-\gamma & 0 & -\zeta
\end{bsmallmatrix} =
\underbrace{\begin{bsmallmatrix}
-1 & ~~0 & ~-1 & -1\\
~~1 & -1 & ~~0 & ~~0\\
~~1 & ~~0 & -1 & -1 &
\end{bsmallmatrix}}_{=\mbox{\normalsize $B$}} \underbrace{\small \mbox{diag}\{ \begin{smallmatrix}\alpha, \beta, \gamma, \zeta\end{smallmatrix}\}}_{=\mbox{\normalsize $\Delta$}} \underbrace{\begin{bsmallmatrix}
1 & ~~0 & ~~0\\
0 & ~~1 & ~~0\\
1 & ~~0 & ~~0\\
0 & ~~0 & ~~1
\end{bsmallmatrix}}_{=\mbox{\normalsize $C$}}.
\end{equation*}
The $BDC$-decomposition is unique up to permutations of the state variables: an order must be chosen for the positive partial derivatives. Then, for instance, since the first derivative is $\alpha = \frac{\partial g_{a}(a,\vartheta)}{\partial a}$, the first column of $B$ corresponds to $S_1$, associated with the reaction rate function $g_{a}$, and the first row of $C$ has a $1$ entry in the first position, corresponding to variable $a$. Column $S_3$ is repeated twice in matrix $B$ because $g_{ac}$ depends on two state variables.
\end{example}

The constant matrices $B$ and $C$ can be immediately built based on the network topology, \ie the graph structure, and the diagonal matrix $\Delta$ carries on the diagonal the absolute values of all the non-zero partial derivatives in the system. Essentially, any nonlinear system whose Jacobian matrix is the \emph{positive} linear combination of constant rank-one matrices (which are independent of the parameters and of the state) admits a $BDC$-decomposition; hence, not only chemical reaction networks, but also most flow networks in general, not just in the life sciences, but also in the physical sciences and in engineering, admit a $BDC$-decomposition.

The local $BDC$-decomposition in Definition~\ref{def_adm_BDC} is related to the system Jacobian. However, admitting a $BDC$-decomposition is not just a local property associated with the linearised system, but also a global property of the original nonlinear system \citep{GiuliaPhD2016}, as can be shown based on the integral variational equation.

\begin{proposition}\label{integralformula}
Given a continuously differentiable function $g \colon {\cal D} \subset \mathbb{R}^n \rightarrow \mathbb{R}^n$, where ${\cal D}$ is a convex domain with non-empty interior, it holds that $g(x_2) - g(x_1) =\int_0^1 \frac {\partial g}{\partial x}(\xi (x_2-x_1)+x_1) d \xi  (x_2-x_1)$ for all $x_1,x_2 \in {\cal D}$.
\end{proposition}
\begin{proof}
Denoting by $\varphi(\xi) \doteq g(\xi (x_2-x_1)+x_1)$ for $0 \leq \xi \leq 1$, since
$\frac{d\varphi}{d\xi}(\xi) =\frac{\partial g}{\partial x}(\xi (x_2-x_1)+x_1) \frac{d(\xi (x_2-x_1)+x_1)}{d\xi}=\frac{\partial g}{\partial x}(\xi (x_2-x_1)+x_1) (x_2-x_1)$,
we can write the difference
$g(x_2) - g(x_1) = \varphi(1) - \varphi(0) = \int_0^1 \varphi'(\xi) d \xi = \int_0^1 \frac {\partial g}{\partial x}(\xi (x_2-x_1)+x_1) d \xi (x_2-x_1)$.
\end{proof}

Then, we can define a global $BDC$-decomposition as follows.
Given system \eqref{eq:ode-system} (namely, $\dot x(t) = S g(x(t),\vartheta) + g_0$) along with the equilibrium condition $0 =  S g(\bar x(\vartheta), \vartheta) + g_0$, denoting $z \doteq x - \bar x(\vartheta)$ and subtracting the two equations yields the shifted system
 \begin{equation}\label{eq:shifted_z}
\dot z = S [g(z + \bar x(\vartheta)) - g(\bar x(\vartheta))].
  \end{equation}
Since the system admits a \emph{local} $BDC$-decomposition, for any fixed equilibrium $\bar x(\vartheta)$ we can consider $z = x - \bar x(\vartheta)$ and the system can be \emph{equivalently rewritten} through the \emph{global} $BDC$-representation
 \begin{equation}\label{eq:shiftedBDCrepres}
\dot z = [B D(z) C] z.
  \end{equation}
In fact, applying the integral formula in Proposition~\ref{integralformula} to the right-hand side of system \eqref{eq:shifted_z} provides
$\dot z = \left[ \int_0^1 J(\xi z + \bar x) d\xi \right] z$.
In view of the local $BDC$-decomposition, equivalently
$\dot z = \left[ B \left( \int_0^1 \Delta(\xi z + \bar x) d \xi \right) C \right] z = \left[ B \left( \int_0^1 \mbox{diag}\left\{ \frac{\partial g_i (\xi z + \bar x)}{\partial x_j} \right\} d\xi \right) C \right] z$, where we explicitly use the fact that matrices $B$ and $C$ are parameter independent, and hence are the same for all $x$.

Therefore, matrices $D$ (in the global $BDC$-decomposition) and $\Delta$ (in the local $BDC$-decomposition) are connected by the integral relationship
$D(z) = \int_0^1 \Delta(\xi z + \bar x) d \xi = \int_0^1 \mbox{diag}\left\{ \frac{\partial g_i (\xi z + \bar x)}{\partial x_j} \right\} d\xi = \mbox{diag}\left\{ \int_0^1  \frac{\partial g_i (\xi z + \bar x)}{\partial x_j} d\xi \right\} = \mbox{diag}\left\{ \Gamma_{ij}(z) \right\}$,
where $\Gamma_{ij}(z) := \int_0^1 \frac{\partial g_i (\xi z + \bar x)}{\partial x_j} d\xi.$
Due to the monotonicity of the functions $g_i(\cdot)$, whose integral is computed on a non-zero interval, $\Gamma_{ij}(z)$ is strictly positive; it also admits a maximum and a minimum in any closed and bounded domain:
$0 < \nu \leq \nu_{ij} \leq \Gamma_{ij}(z) \leq \mu_{ij} \leq \mu$.

Hence, a system admits a global $BDC$-representation \eqref{eq:shiftedBDCrepres} \emph{if and only if} it admits a local $BDC$-decomposition \eqref{eq:BDC_BDC_rank1}, as shown by \cite{GiuliaPhD2016}.
\begin{proposition}[\textbf{Global $BDC$-decomposition}]
System \eqref{eq:ode-system} admitting an equilibrium $\bar x(\vartheta)$ can be equivalently written in the form $\dot z = B D(z) C z$,
where $z = x-\bar x(\vartheta)$, if and only if it admits a local $BDC$-decomposition according to Definition~\ref{def_adm_BDC}.
\end{proposition}
\begin{proof}
If $J(x) = B\Delta(x)C$ for all $x$ in the domain, then $J(z) = B\Delta(z)C$ for all $z=x-\bar x(\vartheta)$, and integration exploiting the integral formula in Proposition~\ref{integralformula} entails the result.
Conversely, if system \eqref{eq:ode-system} is equivalent to $\dot z = BD(z)C z$, then its linearisation yields $J(z) = B\Delta(z)C$, hence $J(x) = B\Delta(x)C$.
\end{proof}

The constant matrices $B$ and $C$ are the very same matrices in both the local and the global $BDC$-decomposition of a given system, and they capture the system structure. The unknown parameter values and functional expressions are isolated in the positive definite diagonal matrix ($\Delta$ in the local case, $D$ in the global case), whose entries are unknown, or uncertain; when assessing structural properties, we do not require any knowledge about the diagonal matrix $\Delta \succ 0$, or $D \succ 0$, because we seek properties that hold for the whole family of systems characterised by a given pair $(B,C)$.

\begin{remark}[\textbf{Generality of the $BDC$-decomposition}]
Also any system of the form \eqref{eq:modelclass}, with $f_i$ ($i=1,\dots,n$) monotonic functions of the state variables, admits a $BDC$-decomposition. In fact, as observed in Section~\ref{sec:dninb}, a system of the form \eqref{eq:ode-system} can always be cast in the form \eqref{eq:modelclass}, and vice versa.
As an example, here is the $BDC$-decomposition of the Jacobian of the IFFL system in \eqref{eq:IFFL}:
\begin{equation*}
J = \begin{bsmallmatrix}
-\mu_1 & 0 & 0\\
\alpha & -\mu_2 & \beta\\
-\gamma & 0 & -\mu_3
\end{bsmallmatrix}=
\underbrace{\begin{bsmallmatrix}
0 & 0 & 0 & -1 & 0 & 0\\
1 & 1 & 0 & 0 & -1 & 0\\
0 & 0 & -1 & 0 & 0 & -1
\end{bsmallmatrix}}_{=\mbox{\normalsize $B$}} \underbrace{\small \mbox{diag}\{ \begin{smallmatrix} \alpha, \beta, \gamma, \mu_1, \mu_2, \mu_3 \end{smallmatrix}\}}_{=\mbox{\normalsize $\Delta$}} \underbrace{\begin{bsmallmatrix}
1 & 0 & 0 \\
0 & 0 & 1 \\
1 & 0 & 0 \\
1 & 0 & 0 \\
0 & 1 & 0 \\
0 & 0 & 1\end{bsmallmatrix}}_{=\mbox{\normalsize $C$}},
\end{equation*}
where $\alpha = \frac{d f_{21}(x_1)}{d x_1}$, $\beta = \frac{d f_{23}(x_3)}{d x_3}$ and $\gamma = - \frac{d g_{31}(x_1)}{d x_1}$ are positive scalars.
The $BDC$-decomposition can capture a much broader class of structures than the sign pattern matrix $\Sigma$: every sign-definite Jacobian admits a $BDC$-decomposition, but $BDC$-decomposable Jacobian matrices are not necessarily sign definite (as can be seen, \eg in Example~\ref{ex:BDC_CRN}, where the entry $\alpha-\gamma$ is not sign definite). The $BDC$-decomposition can capture much more general system structures by taking into account cross-constraints and couplings among the matrix coefficients.
\end{remark}

\begin{remark}[\textbf{System in reaction rate coordinates}]\label{rem:EDF}
The same type of decomposition can be adopted for chemical reaction network systems where the state variables are, instead of the species concentrations, the reaction rates, thus yielding the so-called $EDF$-decomposition \citep{Blanchini2016,Blanchini2017f,blanchini2021structural}. Given the state vector $r(t) = g(x(t))-g(\bar x)$, where $\bar x$ is the equilibrium such that $0 = S g(\bar x) + g_0$, the resulting system is
\begin{equation*}
\dot r(t) = \frac{\partial g}{\partial x} \dot x(t) = \frac{\partial g}{\partial x} [Sg(x(t))+g_0] = \left[ \frac{\partial g}{\partial x} S \right] r(t).
\end{equation*}
Then, the system Jacobian can be written as $J_r = E \Delta(x) F$ and the overall nonlinear system can be equivalently rewritten as
\begin{equation*}
\dot r(t) = E D(x(t)) F r(t),
\end{equation*}
where $\Delta$ and $D$ are positive definite diagonal matrices, while matrix $F$ is formed by possibly repeated rows of the stoichiometric matrix $S$ and the $i$th column of $E$ has a single non-zero entry, in position $j$, that is equal to the sign of the derivative $\Delta_i = \frac{\partial g_h}{\partial x_k}$ with $r_j = g_h(x)-g_h(\bar x)$. The $BDC$-decomposition and the $EDF$-decomposition of the same system are intimately connected; \eg the equality $CB = FE$ always holds \citep{Blanchini2016}.
For instance, consider the chemical reaction network in Example~\ref{ex:CRN}, whose $BDC$-decomposition is shown in Example~\ref{ex:BDC_CRN}. By ordering the reaction rates as $g_a$, $g_b$ and $g_{ac}$, the system in reaction rate coordinates can be written as
\begin{equation*}
\dot r(t) =
\underbrace{\begin{bsmallmatrix}
\alpha & 0 & 0 \\
0 & \beta & 0 \\
\gamma & 0 & \zeta
\end{bsmallmatrix}}_{=\mbox{\normalsize $\frac{\partial g}{\partial x}$}}
\underbrace{\begin{bsmallmatrix}
-1 & ~~0 & ~-1 \\
~~1 & -1 & ~~0 \\
~~1 & ~~0 & -1 
\end{bsmallmatrix}}_{=\mbox{\normalsize $S$}} r(t) =
\begin{bsmallmatrix}
-\alpha & 0 & -\alpha \\
\beta & -\beta & 0 \\
-\gamma+\zeta & 0 & -(\gamma+\zeta)
\end{bsmallmatrix} r(t),
\end{equation*}
and its Jacobian admits the following $EDF$-decomposition:
\begin{equation*}
\begin{bsmallmatrix}
-\alpha & 0 & -\alpha \\
\beta & -\beta & 0 \\
-\gamma+\zeta & 0 & -(\gamma+\zeta)
\end{bsmallmatrix} =
\underbrace{\begin{bsmallmatrix}
~~1 & ~~0 & ~~0 & ~~0\\
~~0 & ~~1 & ~~0 & ~~0\\
~~0 & ~~0 & ~~1 & ~~1
\end{bsmallmatrix}}_{=\mbox{\normalsize $E$}} \underbrace{\small \mbox{diag}\{ \begin{smallmatrix}\alpha, \beta, \gamma, \zeta\end{smallmatrix}\}}_{=\mbox{\normalsize $\Delta$}} \underbrace{\begin{bsmallmatrix}
-1 & ~~0 & -1\\
~~1 & -1 & ~~0\\
-1 & ~~0 & -1\\
~~1 & ~~0 & -1
\end{bsmallmatrix}}_{=\mbox{\normalsize $F$}}.
\end{equation*}

\end{remark}

\subsection{Robust and structural properties}

The sign pattern matrix $\Sigma$, as well as the matrices $B$ and $C$ of the $BDC$-decomposition, are examples of \emph{structures} of a system, which are preserved because of the underlying (signal or flow) graph topology and the qualitative features of the system, and are not affected by quantitative aspects and by uncertainty or variability in the parameters.
If a property holds exclusively due to the peculiar structure of the system, for all possible (physically meaningful) parameter values and functional expressions, then it is a \emph{structural property}. Conversely, if, given the system structure, the property holds for (large) prescribed parameter variations within an assigned uncertainty bounding set, then it is a \emph{robust property}. A whole family of systems can be characterised by a common structure, if all the elements of the family can be represented by the same (\eg flow or signal) graph and the involved functions satisfy common qualitative assumptions (\eg monotonicity). We consider the following definition \citep{blanchini2021structural}. 

\begin{definition}[\textbf{Robust property and Structural property}]\label{def:robustness}
Given a family of systems $\mathcal{F} = \{G_{\lambda}\}_{\lambda \in \mathcal{I}}$ and a property $\mathcal{P}$, the property $\mathcal{P}$ is \emph{robustly satisfied} (\emph{robust}) if  $G_{\lambda}$ enjoys the property $\mathcal{P}$ for every $\lambda\in \mathcal{I}$. If, in addition, the family $\mathcal{F}$ is specified qualitatively by a structure, without resorting to numerical bounds, the property $\mathcal{P}$ is \emph{structurally satisfied} (\emph{structural}). 
\end{definition}

Requiring some quantities or parameters to be non-negative in view of their physical nature, as it often happens when modelling phenomena in the life sciences (\eg concentrations, number of individuals, mass, rates cannot be negative), is not considered to be a numerical bound (because it is not an arbitrary choice, but a necessity for the consistency of the mathematical model with the natural phenomenon that it aims at representing) and can still characterise a structure.

Given a system structure, our analysis aims at assessing whether a property of interest $\mathcal{P}$ holds either structurally (for all possible, physically meaningful, parameter values) or robustly (in a subset of the parameter set, and in this case we seek the largest possible subset of the parameter set where the property holds). If the property $\mathcal{P}$ is incompatible with the given structure, meaning that there exists no possible parameter choice for which the property holds, then the property $\neg \mathcal{P}$ holds structurally.

Following in part the recent survey by \cite{blanchini2021structural}, we provide in this section an overview of several properties of interest and possible approaches to assess them structurally, or robustly.

It is worth noting that many properties are relevant across different domains. An example is the case of \emph{spiking systems}, characterised by a specific qualitative behaviour: the system output first increases, up to a maximum, and then decreases \citep{BG2024}. Such a behaviour is recurrent in epidemics, in relation to the peak of infections. It is also widespread in biology \citep{Levine2013,Nakamura2024}. In neurosciences, an analogous concept is that of excitable systems, characterised by spike trains \citep{izhikevich2000neural,Izhikevich2007,Meisel2015a}. For example, consider the trajectories of the output variable of the Incoherent Feed-Forward Loop (IFFL) motif \eqref{eq:IFFL} in systems biology, and of the infected variable in a compartmental Susceptible-Infected-Recovered (SIR) model \eqref{eq:sir} in epidemiology; as shown in Figure~\ref{fig:IFFL-SIR}, both exhibit a qualitatively spiking behaviour for suitable choices of the parameters.

\begin{figure}[ht!]
	\centering
		\includegraphics[width=\textwidth]{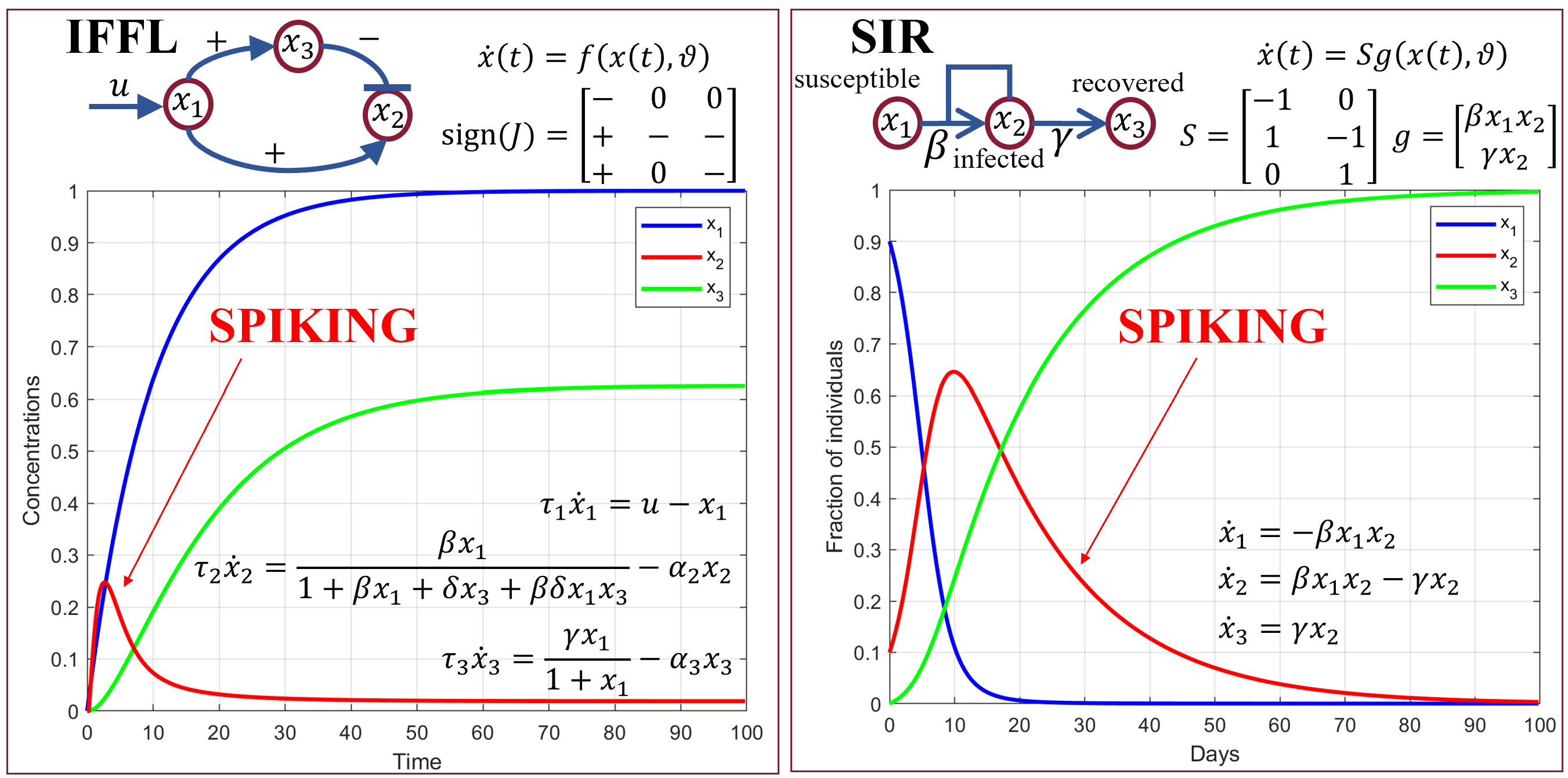}
	\caption{\footnotesize \small The IFFL model and the SIR model both exhibit a spiking dynamic behaviour (qualitatively meant as the system output first increasing up to a maximum, and then decreasing) for appropriate choices of the parameter values. }
	\label{fig:IFFL-SIR}
\end{figure}

\subsubsection{Positivity and boundedness}

Models in the life sciences typically need to be intrinsically positive and bounded. For instance, biochemical species concentrations and population densities cannot be negative. At the same time, they cannot grow unbounded \citep{Anderson2011,Angeli2011,August2010}, due to the presence of degradation or secretion phenomena in living cells, for biochemical species, and of death rate and finite carrying capacity, for populations.

A system $\dot x(t) = f(x(t),\vartheta)$ is \emph{positive} if, given the initial condition $x(t_0) \geq 0$, the system evolution $x(t) \geq 0$ for all $t \geq t_0$, where the inequalities are interpreted componentwise; equivalently, the non-negative orthant is positively invariant. In view of Nagumo's theorem (\cite{Nagumo1942}, see also \cite{Blanchini1999,BM2015}), positivity of the system is equivalent to the fact that, when $x_i=0$, then $f_i(x,\vartheta)|_{x_i=0}\geq 0$; in the linear case $\dot x(t)= M(\vartheta) x(t)$, positivity amounts to requiring that $M(\vartheta)$ is a Metzler matrix, namely, $M_{ij}(\vartheta) \geq 0$ for $i \neq j$. If the positivity condition holds for all possible values of $\vartheta$, then the system is \emph{structurally positive}. For instance, a system of the form \eqref{eq:ode-system} is structurally positive provided that $g_0(\vartheta) \geq 0$ for all $\vartheta$ and that, if $S_{ij}<0$, then the reaction rate function $g_j$ depends on $x_i$ and is zero when $x_i=0$ (meaning that species $X_i$ is one of the reagents in the associated chemical reaction).

For a given $\vartheta$, system $\dot x(t) = f(x(t),\vartheta)$ is uniformly bounded if, given any $H>0$, there exists $K>0$ such that $\| x(0) \| \leq H$ implies $\| x(t) \| \leq K$ for all $t\geq0$; it is bounded if, for any initial condition, there exists $K>0$ such that $\| x(t) \| \leq K$ for all $t\geq0$. A positive system is bounded if, for any initial condition $x(0) \geq 0$, there exists $K>0$ such that $0 \leq x_i(t) \leq K$ for all $t \geq 0$ and for all $i$.
The system is \emph{structurally} bounded (or uniformly bounded) if the condition holds irrespective of the choice of $\vartheta$. Even though the value $K$ depends on the parameters, the existence of such a bound for all $\vartheta$ is a structural property.

Boundedness is typically ensured also by the presence of conservation laws. For instance, in (bio)chemical reaction network systems of the form \eqref{eq:ode-system}, the state evolves in the stoichiometric compatibility class $\mathcal{S}(x_0) = \{ x \geq 0 \colon x = x_0 + \mbox{col}([S \, g_0])\}$, where $\mbox{col}$ denotes the column space of a matrix; $\mathcal{S}(x_0)$ is the intersection of an affine manifold depending on the initial conditions with the positive orthant, and is often a bounded set for all $x_0 \geq 0$ \citep{Angeli2009c}. As a consequence, given the initial condition $x_0$, convergence to an equilibrium $\bar x$ is only possible if $\bar x \in \mathcal{S}(x_0)$.
Also, in SIR-like epidemiological models, the total number of individuals (\ie the sum of the state variables) is often constant (because all the derivatives in the system add up to zero), which guarantees boundedness.
More in general, the boundedness of biological systems \citep{Anderson2011} has been structurally assessed through algorithmic approaches \citep{Angeli2011}, positively invariant sets and Lyapunov (or Lyapunov-like) functions. In fact,    if a bounded subset $\mathcal{S}$ of the state space is positively invariant (namely, $x(0) \in \mathcal{S}$ implies that $x(t) \in \mathcal{S}$ for all $t \geq 0$), then all the system trajectories originating within $\mathcal{S}$ are bounded; under suitable regularity assumptions, positive invariance can be verified through the equivalent conditions provided by Nagumo's theorem. Also Lyapunov-based algorithms to structurally assess boundedness of the system solutions have been proposed \citep{Blanchini2014d,Blanchini2017f}.

\subsubsection{Equilibria: existence, number and stability}\label{sec:equilibria}

When studying the dynamics of nonlinear systems $\dot x(t) = f(x(t),\vartheta)$, also in the life sciences, it is fundamental to assess the existence of equilibria, as well as their number and their stability properties.

\subsubsection*{Existence and uniqueness of equilibria}

Structurally proving ultimate boundedness of the system solutions within a compact and convex set $\mathcal{S}$ (meaning that, for each initial condition $x(0)=x_0$, there exists a time $T(x_0)$ such that the corresponding system solution $ x(t) \in \mathcal{S}$ for all $t\geq T(x_0)$) is particularly important, because it implies the existence of at least one equilibrium point in $\mathcal{S}$ (\cite{Srzednicki1985}, see also \cite{Richeson2002,Richeson2004}). Also positive invariance of a compact and convex set $\mathcal{S}$ implies the existence of at least one equilibrium point in $\mathcal{S}$ \citep{BM2015}.
If an equilibrium exists for $\dot x(t) = f(x(t),\vartheta)$, with $f$ smooth, its uniqueness can be guaranteed if (i) the convex and compact set $\mathcal{S}$, having a non-empty interior, is positively invariant for the system, (ii) there are no equilibria on the boundary $\partial \mathcal{S}$, and (iii) the system Jacobian is structurally non-singular within $\mathcal{S}$, namely $\det(J(x,\vartheta)) \neq 0$ for all $x$ and for all $\vartheta$ \citep{Kellogg1976,Zampieri1992}.

Interesting approaches to structurally assess the existence, uniqueness and possibly stability of equilibria are offered by topological degree theory \citep{Lloyd1978}. An important result is the following: if the convex and compact set $\mathcal{S}$, having a non-empty interior, is positively invariant for the generic system $\dot x(t) = f(x(t))$, with $f$ smooth, and if the system admits $N \geq 1$ equilibria $\bar x_h$ that are in the interior of $\mathcal{S}$ and are non-degenerate (namely, $\det(J(\bar x_h)) \neq 0$), then the system \emph{index} is
\begin{equation}\label{eq:degreetheory}
\sum_{h=1}^N \mbox{sign}[\det(-J(\bar x_h))]=1.
\end{equation}
As a consequence, it is impossible to have an even number of non-degenerate equilibria in the interior of $\mathcal{S}$.

For an odd number of non-degenerate equilibria, the index relation must be satisfied. For instance, if $\det(-J(x))>0$ for all possible values of $x \in \mathcal{S}$, then the equilibrium must be unique. Conversely, if $\det(-J(\bar x_1))<0$, at least other two equilibria with $\det(-J(\bar x_k))>0$, $k=2,3$, must exist.
In the presence of exactly three equilibria, precisely one must have $\det(-J(\bar x_k))<0$ and hence it must be unstable; in fact, $\det(-J(\bar x_k)$ is the constant term of the associated characteristic polynomial, and all the coefficients of a monic Hurwitz stable polynomial must be positive \citep{Gantmacher}.

\begin{example}[\textbf{Gene regulation system: equilibria}]
We can apply the result from degree theory to the equilibrium analysis of the gene regulation system $\dot x(t) = - x(t) + a \frac{x(t)^h}{1 + x(t)^h} + k$, namely \eqref{eq:gene-regulation2}, with $h=2$ and $k=0.1$, also with the help of Figure~\ref{fig:gene-reg-phase-portrait}. For varying $a$, the system admits a different number of equilibria, depending on the number of intersections between the sigmoid $a \frac{x(t)^h}{1 + x(t)^h}$ and the line $x-k$. Either the equilibrium is unique, non-degenerate and stable (for $a<a_{c,1} \approx 1.77$ and $a > a_{c,2} \approx 2.63$), or there are three equilibria that are non-degenerate, two stable and one unstable (for $a_{c,1}<a < a_{c,2}$), or there is a degenerate equilibrium (for $a=a_{c,1}$ and $a=a_{c,2}$, namely when the line is tangent to the sigmoid) and thus the result from degree theory does not apply.
This gene regulation example, associated with system \eqref{eq:hill-equation}, will be examined more thoroughly in Section~\ref{ex:bio}.
\end{example}

\begin{figure}[ht!]
	\centering
	\includegraphics[width=0.5\textwidth]{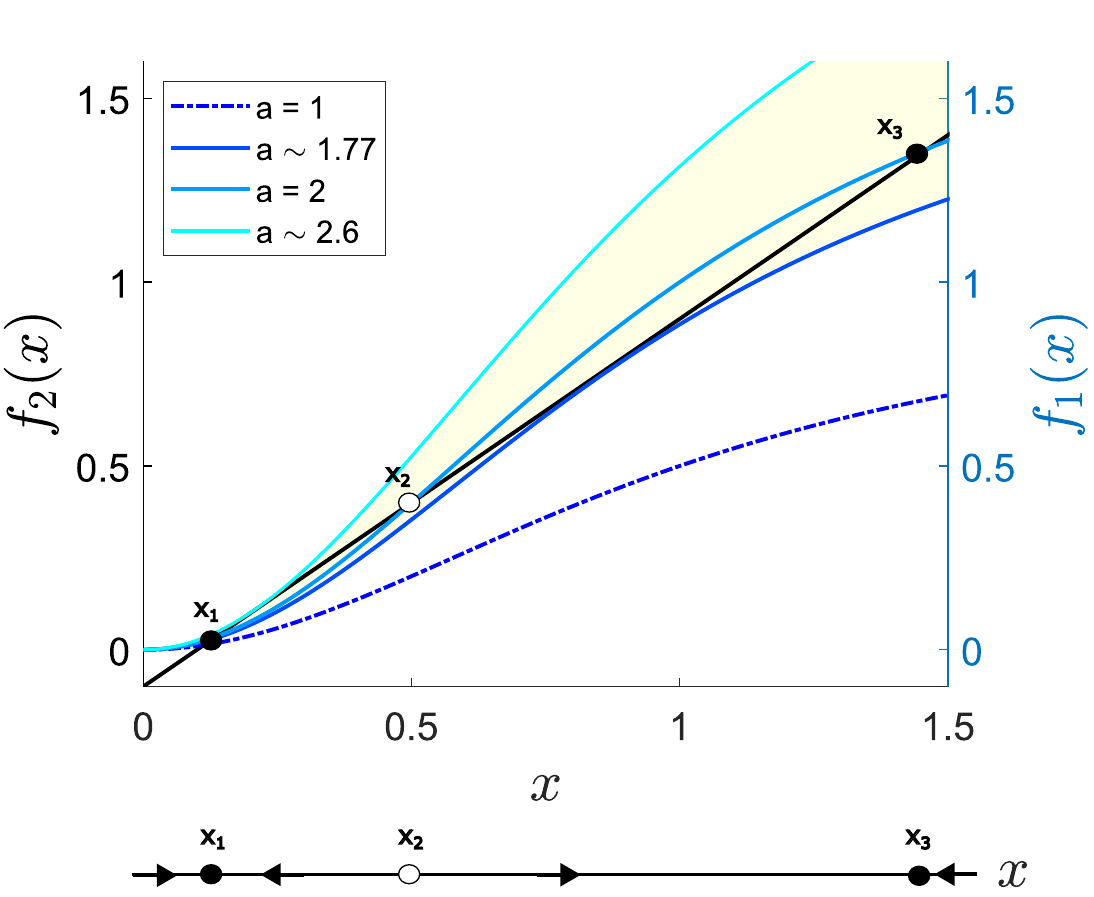}
	\caption{\footnotesize Top: the intersections of the two functions $f_1(x) = a x^h / (1+x^h)$ and $f_2(x) = x - k$ correspond to setting $f_G(x)=0$ and identify the equilibria of system \eqref{eq:hill-equation}; we set $h=2$ and $k=0.1$, and vary $a$. At the critical values $a = a_{c,1} \approx 1.77$ and $a = a_{c,2} \approx 2.63$, $f_1$ is tangent to $f_2$; the bistable region $a\in (a_{c,1},a_{c,2})$ is shaded. Bottom: the 1D phase space for system \eqref{eq:hill-equation}, with stable (full circle) and unstable (empty circle) equilibria and vector flows (arrows), for $a = 2$, $h=2$ and $k=0.1$. }
	\label{fig:gene-reg-phase-portrait}
\end{figure}

The theory of chemical reaction networks has provided several structural conditions for the equilibrium analysis of chemical reaction networks, typically with mass action kinetics \citep{Feinberg1987,Feinberg1995,Feinberg1995multiple,Feinberg2019}.
An important example is the deficiency-one theorem by \cite{Feinberg1987}, which relies on the assessment of a quantity, the network \emph{deficiency}, that exclusively depends on the interconnection structure.

To define the deficiency of a general chemical reaction network system of the form \eqref{eq:ode-system}, we need to introduce some concepts. A \emph{complex} is the integer linear combination of chemical species corresponding to either the reagents or the products of a reaction; for instance, in Example~\ref{ex:biomol} the complexes are $A$, $B$, $C$, $D$, $A+C$, $B+D$, while in Example~\ref{ex:CRN} the complexes are $A$, $B$, $B+C$, $A+C$, and in Example~\ref{ex:LV} the complexes are $A$, $B$, $2A$, $2B$, $A+B$.
A \emph{linkage class} is a (sub)set of complexes that are connected by reactions through \emph{undirected} paths (\ie the reaction arrows can be followed in either direction). We denote the number of complexes by $c$ and the number of linkage classes by $\ell$. As an example, the chemical reaction network with reactions $A \revreact{}{} 2B$, $A+C \revreact{}{} D$, $D \react{} B+E$, $B+ E \react{} A+C$ has $c=5$ complexes ($A$, $2B$, $A+C$, $D$ and $B+E$) and $\ell=2$ linkage classes ($\{A,2B\}$ and $\{A+C,D,B+E\}$); see \cite{Feinberg2019}. The \emph{reaction rank} is $r=\mbox{rank}(S)$, where $S$ is the stoichiometric matrix in \eqref{eq:ode-system}.
Given the number $c$ of complexes, the stoichiometric matrix can be decomposed as
\begin{equation*}
S=N M,
\end{equation*}
where entry $(i,\ell)$ of matrix $N \in \mathbb{Z}^{n \times c}$ is the stoichiometric coefficient of species $i$ in complex $\ell$, while $M \in \{-1,0,1\}^{c \times m}$ is the complex-reaction incidence matrix, whose $(\ell,j)$ entry is $-1$ if complex $\ell$ is the reagent in reaction $j$, $1$ if complex $\ell$ is the product of reaction $j$, and $0$ otherwise. Then, the deficiency of the chemical reaction network is the integer
\begin{equation*}
\delta = \mbox{dim}\{\mbox{ker}(N) \bigcap \mbox{col}(M)\}.
\end{equation*}
For networks that are \emph{weakly reversible} (\ie for which all the complexes in each linkage class are connected through \emph{directed} paths, namely by following the reaction arrows only in the direction in which they point), the deficiency can be equivalently computed as
\begin{equation*}
\delta = c - \ell - r.
\end{equation*}
As a general property, it holds that $\delta \geq 0$ \citep{Angeli2009c,Feinberg2019}.

For networks with deficiency $\delta=1$, conditions for the existence and uniqueness of equilibria are indeed given by the deficiency-one theorem \citep{Feinberg1987} and the existence of multiple equilibria has been investigated, \eg by \cite{Feinberg1995multiple}.

Structural conditions to assess or rule out the presence of multiple equilibria have also been proposed by \cite{Banaji2007,Banaji2009,Craciun2005,Craciun2006}, based on injectivity properties and on graph-theoretic tools.
A survey of methods to assess the multistationarity of chemical reaction networks is provided by \cite{Joshi2015}.

Also, the robustness of equilibria has been investigated by \cite{Breindl2011}, who quantified how well a given model structure can reproduce a desired equilibrium pattern when the kinetic mechanisms and parameters are poorly known.

Existence and uniqueness of the equilibrium can be shown for systems of the form $\dot x(t) = S g(x(t)) + g_0$ based on the $BDC$-decomposition \citep{Blanchini2017f}. In particular, if $J=B \Delta C$ is structurally non-singular (\ie $\det(B \Delta C) \neq 0$ for all diagonal $\Delta \succ 0$) and the entries of matrix $\Delta$ are bounded as $0 < \varepsilon \leq \Delta_h(x) \leq \mu$, for all $h$ and for all $x$ (which implies that the vector function $g$ is radially unbounded: $\|g(x_k)\| \to \infty$ for any sequence $x_k$ such that $\|x_k\| \to \infty$), then an equilibrium exists for every $g_0$. Moreover, if $J=B \Delta C$ is structurally non-singular, then the equilibrium must be unique.

\subsubsection*{Stability of equilibria}

The local stability of equilibria is often analysed through Lyapunov's first method, \ie by computing the eigenvalues of the linearised system, or through the Routh-Hurwitz criterion. However, in general, applying this methodology requires the knowledge of the parameter values.

\paragraph{Qualitative stability and sign stability.}
In special cases, methods from qualitative stability \citep{Clarke1975,edelstein2005,Jeffries1974,Levins1974,Levins1977,May1973,May1974,Maybee1969,Quirk1965} or sign stability \citep{Jeffries1977,Logofet1982} can guarantee the structural stability of a whole family of systems characterised by a Jacobian matrix with a given sign pattern. Approaches for qualitative modelling and (stability) analysis have a long tradition in mathematical biology, and in particular in ecology (see also \cite{Dambacher2002,Dambacher2003stability,Dambacher2003,Dambacher2005,Dambacher2007,Dambacher2009,Marzloff2011}). As an example, any matrix with the sign pattern
\begin{equation*}
\Sigma = \begin{bmatrix}- & -\\ + & -\end{bmatrix},
\end{equation*}
corresponding to the negative feedback interconnection (activation-inhibition) of two chemical species, is structurally Hurwitz stable; also, all matrices with the sign pattern
\begin{equation*}
\Sigma = \begin{bmatrix}0 & + & +\\ - & 0 & 0\\ - & 0 & -\end{bmatrix},
\end{equation*}
corresponding to an ecosystem with one predator and two prey species, are structurally Hurwitz stable \citep[Example 2, page 239]{edelstein2005}.

\paragraph{Structural stability.}
For special classes of systems, such as compartmental systems \citep{Jacquezt1993}, which are akin to chemical reaction networks exclusively composed of unimolecular reactions of the form $X_i \react{} X_j$ or $X_i \react{} \emptyset$ (and whose flow graph is thus a graph strictly speaking, instead of being a general hyper-graph), strong structural stability results are available \citep{Maeda1978} based on the $1$-norm as a Lyapunov function.

Structural stability has been long studied in the context of chemical reaction network theory \citep{Clarke1980,Feinberg1987,Feinberg2019,Horn1972,Horn1973,Horn1973b,Reder1988,Sontag2001}.
A pivotal result on the structural stability of chemical reaction networks with mass action kinetics is the celebrated deficiency-zero theorem by \cite{Feinberg1987}: a weakly reversible chemical reaction network with mass action kinetics and deficiency $\delta=0$ admits a unique equilibrium point in each stoichiometric compatibility class, and such an equilibrium is (locally) asymptotically stable (within the stoichiometric compatibility class). The result holds independent of the parameter values (\ie the positive reaction rate constants) and its proof relies on employing, as a Lyapunov function, the system entropy $V(x)=\sum_{i=1}^n x_i \log \left( \frac{x_i}{\bar x_i}\right) - x_i - \bar x_i$.
Several subsequent developments have been provided \citep{Hangos2010,Ke2019,schaft2013}, also proving global stability within the stoichiometric compatibility class for significant classes of systems \citep{Anderson2008,Rao2017} and extending the deficiency-zero theorem to generalised mass action kinetics where $x_i^{p_i}$ is replaced by $h_i(x_i)^{p_i}$ \citep{Sontag2001}.

As a side note, also different concepts of stability, such as Input-to-State stability, have been investigated for biological systems \citep{Chaves2006b,Chaves2008,nilssongiordano2025}.

\paragraph{Lyapunov methods.}
Lyapunov approaches have been frequently adopted for the stability analysis of biological and biochemical systems \citep{Al-Radhawi2016,Blanchini2011,Blanchini2014d,Chesi2008,Clarke1980,Franco2013,Grognard2005,Pasquini2020,Rao2017}.
In particular, local stability of chemical reaction networks has been assessed through parameter-dependent quadratic Lyapunov functions \citep{Clarke1980}. For generic $BDC$-decomposable systems, $BDC$ is Hurwitz for all positive definite diagonal matrices $D$ if and only if the Lyapunov equation $[BDC]^\top P(D) + P(D) [BDC] = -I$ admits a positive definite solution $P(D)$ for all $D$, which yields a parameter-dependent Lyapunov function $V(x) = x^\top P(D) x$. The methods by \cite{Chesi2008,Chesi2011b} allow to solve this problem efficiently through sum-of-squares polynomials and LMI-based convex optimisation \citep{BCCG2020}.

Quadratic functions are known to be conservative \citep{Molchanov1986,Molchanov1989}. Non-quadratic Lyapunov functions have also been proposed, such as polynomial Lyapunov functions \citep{Chesi2011c}, as well as polyhedral (\ie piecewise-linear) Lyapunov functions \citep{Al-Radhawi2016,Blanchini2011,Blanchini2014d,Grognard2005} that generalise those introduced by \cite{Maeda1978} for compartmental systems \citep{Jacquezt1993}. In fact, even in simple cases, quadratic Lyapunov functions can only certify the stability of individual systems of the form $\dot x = S g(x)+g_0$ \emph{with known parameters}, because they depend on the reaction rates. In these cases, the only possible structural (parameter-independent) Lyapunov function is polyhedral \citep{Blanchini2015c} and there are no other rate-independent Lyapunov functions \citep{Blanchini2017f}.
Piecewise quadratic Lyapunov functions have been considered by \cite{Pasquini2018,Pasquini2020}. Lyapunov approaches have also been adopted, more in general, to rule out oscillatory behaviours \citep{Angeli2022}.

\paragraph{Structural Lyapunov functions for $BDC$-decomposable systems.}
Based on the $BDC$- decomposition, we can \emph{systematically} build \emph{structural} piecewise-linear Lyapunov functions. A class of polyhedral Lyapunov functions has been introduced by \cite{Blanchini2014d} to address the structural stability analysis in the case of \emph{general} monotonic reaction rate functions $g$ (thus going beyond mass-action-kinetics networks).
Polyhedral Lyapunov functions \citep{BM2015} can be expressed either in terms of the vertices $\pm P_i$ of a symmetric convex polytope including the origin in its interior, so that the Lyapunov function is
\begin{equation}
V_P(x) = \inf \{ \| w \|_1 \colon P w = x, \, w \in \mathbb{R}^s \},
\end{equation}
whose unit ball has vertices $\pm P_i$, where $P_i$ is a column of the full-row-rank matrix $P \in \mathbb{R}^{n \times s}$,
or in terms of the set of inequalities $\pm F_i^\top x \leq 1$ that describe the polytope, so that the Lyapunov function is
\begin{equation}
V^H(x) = \| H x \|_\infty,
\end{equation}
whose unit ball has facets on the planes $\pm H_i^\top x = 1$, where $H_i^\top$ is a row of the full-column-rank matrix $H \in \mathbb{R}^{s \times n}$.
Given the graph structure, \ie matrices $B$ and $C$ of the $BDC$-decomposition, a systematic algorithm can generate the unit ball of the polyhedral Lyapunov function -- if any -- that guarantees the structural stability of the whole family of systems, regardless of the parameters (\ie for all possible choices of the positive diagonal entries of matrix $\Delta$, or $D$).
In particular, \cite{Blanchini2014d} absorb the nonlinear system $\dot z(t) = B D(z(t)) C z(t)$ in a linear differential inclusion $\dot z(t) = B D(t) C z(t)$, where the state dependent matrix $D$ becomes time dependent (see also \cite{Angeli2009} for linear differential inclusion approaches to chemical reaction networks); any trajectory of the original system is also a trajectory of the linear differential inclusion. A polyhedral Lyapunov function is sought by associating the linear differential inclusion with a suitable discrete difference inclusion having equivalent stability properties; then, an iterative numerical algorithm (starting from the diamond $\|x\|_1 \leq 1$ and iteratively expanding the set according to the structural directions of growth for the family of systems, contained in matrices $B$ and $C$) converges to the unit ball of the polyhedral Lyapunov function that guarantees the structural stability of the system, provided that such a function exists. The test for structural stability can fail, due to the lack of convergence, even when some individual elements of the family of systems are stable; however, if the test fails, we cannot claim that all the elements of the family are stable.
If the iterative algorithm does provide a polyhedral Lyapunov function, then \emph{local} Lyapunov stability is structurally certified, for any choice of the monotonic reaction rate functions, \ie for all the systems in the family.

In addition, if the system admits a polyhedral Lyapunov function, then local \emph{asymptotic} stability is ensured if and only if $J=B\Delta C $ is structurally non-singular (\ie $\det(-B\Delta C) >0$ for all positive definite diagonal matrices $\Delta$). Also \emph{global} asymptotic stability and convergence are ensured if and only if $BDC$ is structurally non-singular \citep{Blanchini2017f}.

Piecewise-linear-in-rate Lyapunov functions proposed by \cite{Al-Radhawi2016,Al-Radhawi2020} can be seen as the \emph{dual} class of functions: a piecewise-linear-in-rate Lyapunov function for a system in concentration coordinates, $\dot x(t) = S g(x(t)) + g_0$, can be seen as a piecewise-linear Lyapunov function for the system in reaction coordinates $\dot r(t) = \left[ \frac{\partial g}{\partial x} S \right] r(t)$ \citep{Blanchini2016}, discussed in Remark~\ref{rem:EDF}.
An interesting insight comes from duality, as shown by \cite{BG2022}. Given a chemical reaction network
\begin{equation*}
\dot x(t)=Sg(x(t)),
\end{equation*}
with an equilibrium at zero ($g(0)=0$), we define the \emph{dual} system as
\begin{equation*}
\dot y(t) = S^\top h(y(t)),
\end{equation*}
where the dual stoichiometric matrix is $S^\top$ (so the two systems have different size in general, and their Jacobian matrices are not the transpose of each other), while $h$ is a vector of new reaction rate functions (associated with the species of the \emph{primal} system), and zero is still an equilibrium ($h(0)=0$).

Given a primal chemical reaction network with $n$ species and $m$ reactions, the dual network has $m$ species and $n$ reactions. The stoichiometric matrix of the dual network has size $m \times n$ and is the transpose of the stoichiometric matrix of the primal network (which has size $n \times m$); the reaction rate vector of the dual network has dimension $n$ and the $k$th reaction rate function depends on all the variables $i$ for which entry $(i,k)$ of the stoichiometric matrix is negative.

\begin{example}[\textbf{Dual chemical reaction network}]
The dual of the network $2X_1 + X_2 \react{y_1} X_3$, $X_3 + X_4 \react{y_2} X_1 + X_2$, $X_3 \react{y_3} X_4$, with
$S = \begin{bsmallmatrix}
-2 & 1 & 0\\
-1 & 1 & 0\\
1 & -1 & -1\\
0 & -1 & 1
\end{bsmallmatrix}$ and $g(x)=
\begin{bsmallmatrix}
y_1(x_1,x_2)\\
y_2(x_3)\\
y_3(x_3)
\end{bsmallmatrix}$, is the network $2 Y_1 \react{h_1} Y_2$, $Y_1 \react{h_2} Y_2$, $Y_2 + Y_3 \react{h_3} Y_1$, $Y_3 \react{h_4} Y_4$, which has stoichiometric matrix
$S^\top = \begin{bsmallmatrix}
-2 & -1 & 1 & 0\\
1 & 1 & -1 & -1\\
0 & 0 & -1 & 1
\end{bsmallmatrix}$ and reaction rate vector $h(y)=
\begin{bsmallmatrix}
h_1(y_1)\\
h_2(y_1)\\
h_3(y_2,y_3)\\
h_4(y_3)
\end{bsmallmatrix}$.
\end{example}

Then, within the stoichiometric compatibility class possibly induced by the presence of conservation laws, the primal system is structurally locally stable if and only if the dual system is structurally locally stable \citep{BG2022}.
Moreover, consider the corresponding systems in reaction rate coordinates, where the state variables are no longer the concentrations of the chemical species, $x$ and $y$, but are now the reaction rates $r=g(x)$ and $w=h(y)$:
\begin{equation*}
\dot r(t) = \left[ \frac{\partial g}{\partial x} S \right] r(t)
\end{equation*}
and
\begin{equation*}
\dot w(t) = \left[ \frac{\partial h}{\partial y} S^\top \right] w(t).
\end{equation*}
These two systems are also the dual of each other.
The four systems (primal and dual, in concentration coordinates and in reaction rate coordinates) can be equivalently rewritten according to their global $BDC$-decomposition (or $EDF$-decomposition); it turns out that the dual chemical reaction network in concentration coordinates, $\dot y(t) = F^\top D(\cdot) E^\top y(t)$, has the same structure as the transpose of the primal chemical reaction network in reaction rate coordinates, $\dot r(t) = E D(\cdot) F r(t)$, while the dual chemical reaction network in reaction rate coordinates, $\dot w(t) = C^\top D(\cdot) B^\top w(t)$, has the same structure as the transpose of the primal chemical reaction network in concentration coordinates, $\dot x(t) = B D(\cdot) C x(t)$ \citep{BG2022}.
Then, we can adopt the iterative numerical procedure by \cite{Blanchini2014d} to compute a structural polyhedral Lyapunov function, if any, for each of these systems.
An interesting symmetry emerges in the architecture of these systems also in relation to global stability assessed through Lyapunov norms: the primal chemical reaction network in concentration coordinates admits a polyhedral Lyapunov function if and only if the dual chemical reaction network in reaction rate coordinates admits a polyhedral Lyapunov function, which is the dual polyhedral norm; the same holds for the primal chemical reaction network in reaction rate coordinates and the dual chemical reaction network in concentration coordinates \citep{BG2022}.
Hence, for any network admitting a piecewise-linear Lyapunov function in concentrations \citep{Blanchini2014d,Blanchini2017f} there exists a dual network admitting a piecewise-linear Lyapunov function in rates \citep{Al-Radhawi2016,Al-Radhawi2020,Blanchini2016}, and vice versa.
Duality can thus be useful to assess the structural stability of particularly challenging network topologies through the analysis of their dual structure.

\paragraph{Parametric robustness for $BDC$-decomposable systems.}
The $BDC$-decomposition can also be used for robust stability analysis, when the unknown diagonal entries of matrix $\Delta =\mbox{diag} \{\Delta_1, $ $\dots,\Delta_q\}$ are bounded in known intervals:
\begin{equation}\label{eq:bounds}
0 < \Delta_i^- \leq \Delta_i \leq \Delta_i^+.
\end{equation}
In this case, the robust Hurwitz stability of $B\Delta C$ can be assessed thanks to fundamental results from the literature on parametric robustness, such as the zero exclusion theorem and the mapping theorem \citep{Barmish1994}.
The zero exclusion theorem guarantees that matrix $B\Delta C$, with the diagonal entries of $\Delta$ as in \eqref{eq:bounds}, is robustly Hurwitz stable if and only if $B \Delta^* C$ is Hurwitz for some $\Delta^*$ and, for all $\omega \geq 0$, the value set $\mathcal{V}_\omega = \{z \in \mathbb{C} \colon z = \det[\jmath \omega I - B\Delta C], \, \Delta_i \mbox{ as in \eqref{eq:bounds}}\}$ does not include the origin. Computing the precise shape of the value set is hard. However, since $\det[\jmath \omega I - B \Delta C]$ is a multi-affine function of the parameters $\Delta_i$, the mapping theorem ensures that the convex hull of the value set at the frequency $\omega$ is the convex hull of the points $\det[\jmath \omega I - B \hat \Delta_k C]$, where $\hat \Delta_k$ are the vertices of the hyper-rectangle defined by \eqref{eq:bounds}. Then, the convex hull of the value set $\mathcal{V}_\omega$ can be easily drawn and a computable sufficient condition for robust Hurwitz stability is that such a convex hull never includes the origin.
Such methodologies, relying on the $BDC$-decomposition and on parametric robustness, can be employed in general for the robustness analysis of biological systems \citep{BCGZ2020,BCGZ2022}.

\paragraph{Structural topology-independent stability.}
Strong structural stability results can also be obtained for special classes of systems, such as Michaelis-Menten networks. Consider an activation-inhibition biological network of the form \eqref{eq:gene-regulation}, where an arbitrary number of species interact pairwise through either activating or inhibiting interactions (this happens, \eg in gene regulatory networks, where gene expression levels are regulated through DNA, RNA and proteins, often by means of kinetic mechanisms for activation and inhibition based on enzymatic reactions). Activation-inhibition regulatory networks typically include myriads of interactions and their topology is often not fully known; it would therefore be precious to obtain results that are not only structural (\ie parameter-independent) but also topology-independent (\ie independent of the precise network topology; see also \cite{Blanchini2017,DG2020,DG2021,DG2021ecc}). We call Michaelis-Menten network a system of the form \eqref{eq:gene-regulation} where the interaction functions are of the Michaelis-Menten type, $f_{ij}(x_j) = \frac{A_{ij}(x_j/\beta)}{1+(x_j/\beta)}$ and $g_{i \ell}(x_\ell) = \frac{B_{i \ell}}{1+(x_\ell/\beta)}$. As shown by \cite{BLANCHINI2023110683}, under suitable assumptions, Michaelis-Menten networks structurally admit a single positive equilibrium, which is locally asymptotically stable, not only regardless of parameter values (\ie structurally), but also regardless of the network topology (and of the size of the network), and even in the presence of arbitrary constant delays in the interaction functions. Hence, the stability result is not only structural and topology-independent, but also delay-independent.

\subsubsection{Sensitivity and input-output influences}\label{sec:influences}

When an external stimulus is applied to a biological system resting at an equilibrium state, how does the system react? How is the equilibrium perturbed? Which new equilibrium is reached when a persistent perturbation is applied? Assessing the signs of sensitivities of steady states to perturbations is fundamental and several approaches have been proposed (see \cite{Sontag2014}, as well as \cite{Mochizuki2015,Fiedler2015,Brehm2018} and \cite{Feliu2019,Vassena2023}). In particular, the $BDC$-decomposition can be leveraged to structurally predict the sign of input-output influences at steady state through a vertex algorithm (\cite{Giordano2016}, see also \cite{BG2019}).
Consider a general nonlinear system
\begin{equation}\label{eq:influences}
\begin{cases}
\dot x(t) = f(x(t),u(t))\\
y(t)=h(x(t))
\end{cases}
\end{equation}
with state $x \in \mathbb{R}^n$, input $u \in \mathbb{R}$ and output $y \in \mathbb{R}$, resting at an asymptotically stable equilibrium point $(\bar x, \bar u)$ such that $f(\bar x, \bar u)=0$ and corresponding to the steady-state output $\bar y = h(\bar x)$. Assume that a positive constant perturbation $\delta \bar u>0$ is added to the input, which becomes $\bar u + \delta \bar u$, and is sufficiently small to still guarantee existence and stability of an equilibrium. Then, after a transient, the new steady-state value of the output is $\bar y + \delta \bar y$. The \emph{steady-state input-output influence} of $u$ over $y$ is then $\mbox{sign}(\delta \bar y)$, and it is \emph{structurally} positive `$+$' (or negative `$-$', or zero `$0$') if, regardless of the system parameters, $\delta \bar y>0$ (or $\delta \bar y<0$, or $\delta \bar y=0$), while it is indeterminate `?' if its sign depends on the values of the system parameters \citep{Giordano2016}. Precisely, by defining the implicit input-output (SISO) function
\begin{equation}
y = \Phi(u) \colon f(x,u)=0, \, y=h(x),
\end{equation}
the input-output influence can be written as
\begin{equation}
\frac{\delta \bar y}{\delta \bar u} = \frac{d}{du} \Phi(u) = - \frac{\partial h(x)}{\partial x} \left[ \frac{\partial f(x,u)}{\partial x} \right]^{-1} \frac{\partial f(x,u)}{\partial u}
\end{equation}
and its structural assessment yields $\mbox{sign}\left[\frac{\delta \bar y}{\delta \bar u}\right] \in \{+,-,0,?\}$.

When the steady-state input-output influence is structurally zero, the system structurally exhibits \emph{perfect adaptation}, which is a widely observed phenomenon in natural systems. Perfect adaptation is the ability of a system, resting at a stable equilibrium and subject to a persistent (step) perturbation (of sufficiently moderate magnitude), to asymptotically recover (after a transient) the very same steady-state value and output value that it admitted before the perturbation was applied; if the pre-stimulus configuration is recovered not exactly, but just approximately, the property is denoted as \emph{adaptation} \citep{alon2019introduction,Araujo2018,Blanchini2012,Cosentino2012,Drengstig2008,Ma2009,Sontag2003,Waldherr2012}.
Perfect adaptation has been observed, \eg in yeast osmoregulation \citep{Muzzey2009}, in calcium homeostasis \citep{samad2002} and in bacterial chemotaxis \citep{Alon1999,Barkai1997,Clausznitzer2010,Levchenko2002,Spiro1997,Yi2000}, and suitable structures, such as the antithetic integral feedback controller, have been proposed to ensure robust perfect adaptation in biomolecular networks \citep{Aoki2019,Briat2016,Gupta2022,Steel2018}.

The analysis of steady-state input-output influences has a long history in ecology, to robustly predict the effect of perturbations in communities of interacting species \citep{Dambacher2002,Dambacher2003,Dambacher2007}; \emph{press perturbation} experiments are carried out by applying a persistent perturbation to a species in the community, and assessing the ensuring variation of the density of the various species at the new equilibrium \citep{Bender1984,Giordano2017interaction,Giordano2017qualitative,Montoya2009,Novak2011,Novak2016,Schmitz1997,Yodzis1988}.

The $BDC$-decomposition allows us to assess influences structurally, as discussed by \cite{Giordano2016}. Consider the linearisation of system \eqref{eq:influences}, where $J = \frac{\partial f(x,u)}{\partial x} = B \Delta C$, $E = \frac{\partial f(x,u)}{\partial u}$, $H = \frac{\partial h(x)}{\partial x}$, and $z = x - \bar x$:
\begin{equation}\label{eq:influences_local}
\begin{cases}
\dot z(t) = B \Delta C z(t) + E \delta \bar u\\
\delta y(t) = H z(t)
\end{cases}
\end{equation}
Then, the sign of
$$\frac{\delta \bar y}{\delta \bar u} = \frac{d}{du} \Phi(u) = H (-B \Delta C)^{-1} E = \frac{\det \begin{bmatrix} -B \Delta C & -E\\ H & 0 \end{bmatrix}}{\det(-B \Delta C)}$$ is the same as the sign of
\begin{equation}
\mathcal{I}(\Delta) = \det\begin{bmatrix} -B \Delta C & -E\\ H & 0 \end{bmatrix},
\end{equation}
since $\det(-B \Delta C)>0$ structurally (for all $\Delta \succ 0$) in view of the asymptotic stability assumption.
$\mathcal{I}(\Delta)$ is a multi-affine function of the uncertain parameters $\Delta_i$. Without loss of generality, let us rescale the parameters in the unit hypercube, so that $\Delta \in \mathcal{D} = \{\Delta \colon 0 < \Delta_i \leq 1 \,\, \forall i\}$, and define the set $\hat{\mathcal{D}} = \{\Delta \colon \Delta_i \in \{0,1\} \,\, \forall i\}$ of the vertices of $\mathcal{D}$. Denote by $I$ the identity matrix. Then, since a multi-affine function of the parameters admits its minima and maxima at the vertices of the parameter space, we have that: $\mbox{sign}\left[\frac{\delta \bar y}{\delta \bar u}\right]=+$ if and only if $\mathcal{I}(I)>0$ and $\mathcal{I}(\Delta) \geq 0$ for all $\Delta \in \hat{\mathcal{D}}$; $\mbox{sign}\left[\frac{\delta \bar y}{\delta \bar u}\right]=-$ if and only if $\mathcal{I}(I)<0$ and $\mathcal{I}(\Delta) \leq 0$ for all $\Delta \in \hat{\mathcal{D}}$; $\mbox{sign}\left[\frac{\delta \bar y}{\delta \bar u}\right]=0$ if and only if $\mathcal{I}(\Delta) = 0$ for all $\Delta \in \hat{\mathcal{D}}$; otherwise, $\mbox{sign}\left[\frac{\delta \bar y}{\delta \bar u}\right]= \, ?$ (\cite{Giordano2016}, see also \cite{BG2019}).
Therefore, the structural steady-state influence can be computed through a vertex algorithm.
Note that the same algorithm can be adopted to assess the sensitivity of an output of interest (which can also be an individual state component) with respect to variations in the parameter values; it suffices to consider $\bar u$ as the nominal parameter value and $\delta \bar u$ as its variation.

The approach can also be used to support the design of synthetic architectures guaranteeing that relevant steady-state influences are structurally signed \citep{Giordano2016negative}, \eg in rate-regulatory biomolecular circuits \citep{Franco2014,Giordano2013}.

\begin{remark}[\textbf{Leveraging multi-affinity}]\label{rem:multiaff}
By exploiting the multi-affinity of a suitable function associated with the property to be investigated, the $BDC$-decomposition allows us to derive vertex results to assess whether a property holds in a whole hyper-rectangle in the parameter space just by checking it on the vertices.
For instance, in the presence of a constant input perturbation affecting the system,
\begin{equation}
\begin{cases}
\dot x(t) = Sg(x(t)) + E u(t), \qquad J(\Delta(x))=B \Delta(x) C\\
y(t)=Hx(t)
\end{cases}
\end{equation}
we can perform a robust sensitivity analysis: we can leverage vertex results to determine robust lower and upper bounds for the input-output sensitivity $\mathcal{F}=-HJ(\Delta)^{-1}E$, by exploiting the fact that it is a multi-affine function of the uncertain parameters in $\Delta$. Similarly, when the system is affected by a \emph{periodic input perturbation}, we can perform a robust frequency analysis and obtain robust lower and upper bounds for the magnitude and phase of the transfer function $W(s,\Delta)=H(sI-J(\Delta))^{-1}E$.
Assuming that the uncertain parameters lie in known intervals \eqref{eq:bounds}, methodologies relying on the $BDC$-decomposition, and on multi-affinity of the function of interest in the uncertain parameters, can be employed for the robustness analysis of biological systems \citep{BG2019,BCGZ2020,BCGZ2022}. For instance, they allow us to derive robust conditions for microphase separation in the presence of diffusion and chemical reactions \citep{BFGO2023}. In view of the generality of this approach, this type of analysis can be easily extended to other types of natural and engineered systems.
\end{remark}

The \emph{structural steady-state influence matrix} (SSIM) can be built by taking vectors $E$ and $H$ in system \eqref{eq:influences_local} with a single non-zero entry, say the $j$th entry for $E$ and the $i$th entry for $H$, which we can set to $1$ without loss of generality. This choice of $E$ and $H$ corresponds to assessing the steady-state influence on variable $i$ due to a step input affecting the equation of variable $j$. Then, considering all possible $(i,j)$ pairs, with $i,j \in \{1,\dots, n\}$, reveals the effect of all possible input perturbations, acting on one of the state equations, on the steady state of each of the state variables taken as an output \citep{Giordano2016}. Hence, entry $(i,j)$ of the SSIM does not only include the direct influence associated with entry $(i,j)$ in the system Jacobian matrix: each entry of the SSIM expresses the overall net influence that takes into account both the direct effects and the indirect effects propagating through the network interconnection.
In ecology, the steady-state influence matrix has been widely studied as the adjoint of the negative community matrix (\cite{Levins1974,May1974}, see also \cite{Dambacher2002,Dambacher2003,Dambacher2003stability,Dambacher2005,Dambacher2007}), so that its $(i,j)$ entry quantifies the effect on species $i$ of a press perturbation on species $j$, or equivalently as the negated inverse of the community matrix \citep{Bender1984,Schmitz1997,Yodzis1988}, which has the same sign pattern under stability assumptions.

\begin{figure}[ht]
	\centering
	\includegraphics[width=0.4\textwidth]{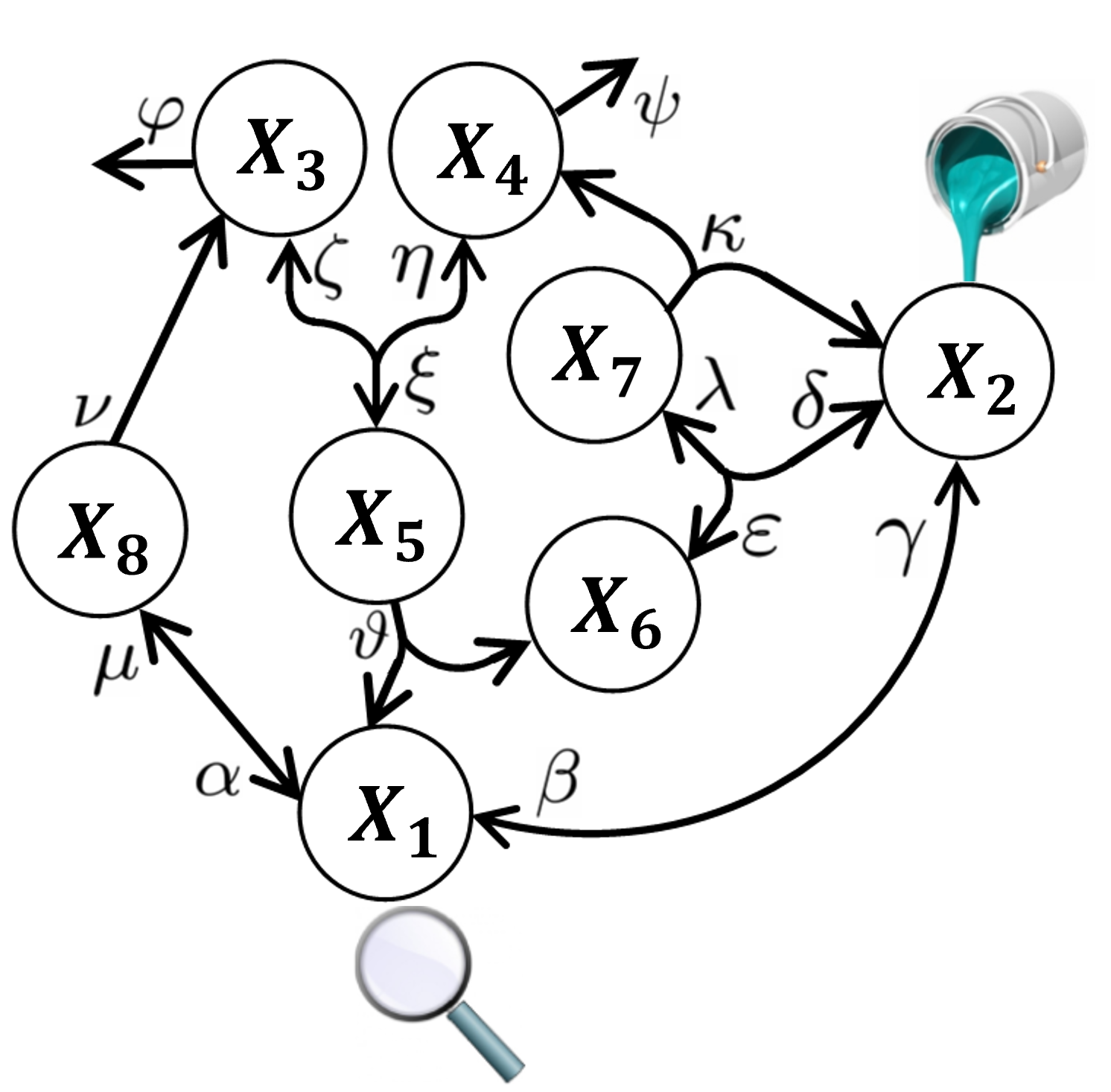}
	\caption{\footnotesize Flow graph corresponding to the EnvZ-OmpR osmoregulation system in \emph{E. coli} \citep{Shinar2010}; the positive entry $(1,2)$ in the structural steady-state influence matrix in \eqref{eq:SSIMosmo} means that, when a positive persistent perturbation acts as an input on species $X_2$, then the new steady-state concentration of species $X_1$ is always larger than its pre-perturbation value, regardless of the system parameters; figure adapted from \cite{Giordano2016}. }
	\label{fig:osmo}
\end{figure}

Computing the structural SSIM for real biological systems can help reveal their extraordinary robustness: the fraction of indeterminate entries is often surprisingly small. As an example, the EnvZ-OmpR osmoregulation system in \emph{E. coli} \citep{Shinar2010}, shown in Figure~\ref{fig:osmo}, has the influence matrix \citep{Giordano2016}
\begin{equation}\label{eq:SSIMosmo}
\mbox{SSIM}_{\mbox{osmo}} = \begin{bsmallmatrix}
+ & + & + & + & + & + & + & + \\
+ & + & + & + & + & + & + & + \\
+ & + & + & 0 & + & 0 & + & + \\
0 & 0 & 0 & + & + & + & + & 0 \\
+ & + & + & + & + & + & + & + \\
? & ? & ? & ? & ? & ? & ? & ? \\
+ & + & + & + & + & + & + & + \\
+ & + & + & + & + & + & + & +
\end{bsmallmatrix}.
\end{equation}
Comparing the structural SSIM with experimental results is a precious tool for model falsification. For instance, entry $(1,2)$ of the structural SSIM shows the structural sign of the steady-state variation in the concentration of the first species ensuing a positive step input perturbation acting on the second species: after a transient, the new equilibrium concentration of the first species will always be larger than it was before the perturbation. Then, if we happen to observe experimentally that the new concentration of the first species is smaller than before (even in a single experiment), we can conclude that the model is unsuitable to represent the phenomenon, because it will never be able to reproduce the behaviour we have observed, no matter how the functional expressions and parameter values are chosen.

\begin{example}[\textbf{Protein-mediated ceramide transfer}]\label{ex:Golgi}
Another interesting case study is offered by the mechanism of protein-based intracellular ceramide transfer from the endoplasmic reticulum to the trans-Golgi network \citep{Giordano2018}, for which two alternative models have been proposed \citep{Hanada2010,Weber2015}: the short distance shuttle (SDS) model, where the protein CERT moves back and forth between the two membranes so as to deliver ceramide, like a ferry or a shuttle, and the neck swinging (NS) model, where CERT needs to be attached to both membranes at the same time in order to deliver ceramide by bending its neck back and forth. A graphical visualisation of the two models is provided in Figure~\ref{fig:Golgi}. Can the two models be structurally compared?
Both models can be seen mathematically as flow-inducing networks (see \cite{Giordano2017}), \ie compartmental dynamical systems where some flows between compartments can be driven by the state of other species in the system, which tune the flow without being affected by it.
Then, using the $BDC$-decomposition and the vertex algorithm (see the methodology described by \cite{Giordano2016} and adopted by \cite{Giordano2018} to analyse this case study), the structural SSIM can be computed for the two systems:
\begin{equation}
\hspace{-4mm} \mbox{SSIM}_{\mbox{SDS}} = \begin{bsmallmatrix}
? & ? & ? & ? & - & - & - \\
+ & + & ? & ? & + & + & + \\
- & - & ? & ? & - & - & - \\
+ & + & ? & ? & + & + & + \\
? & ? & ? & ? & ? & ? & ? \\
+ & + & ? & ? & ? & ? & ? \\
? & ? & ? & ? & ? & ? & ?
\end{bsmallmatrix} \mbox{ and } 
\mbox{SSIM}_{\mbox{NS}} = \begin{bsmallmatrix}
? & ? & - & - & - & - & - \\
+ & + & + & + & + & + & + \\
- & - & ? & ? & - & - & - \\
+ & + & + & + & + & + & + \\
? & ? & - & - & ? & ? & ? \\
? & ? & - & - & ? & ? & ? \\
? & ? & + & + & + & + & + 
\end{bsmallmatrix}
\end{equation}
The two models are not structurally incompatible: influences associated with the same $(i,j)$ entry either have the same sign, or are both undetermined, or one is signed and the other is undetermined. In the latter case, we gain insight that would allow us to design experiments in order to discriminate between the two models. For instance, if an experiment revealed a positive influence corresponding to entry $(1,3)$, where the SDS model has an undetermined entry while the NS model has a structurally negative entry, then the NS model would be falsified, while the SDS model would still be valid.
\end{example}

\begin{figure}[ht!]
	\centering
	\includegraphics[width=0.8\textwidth]{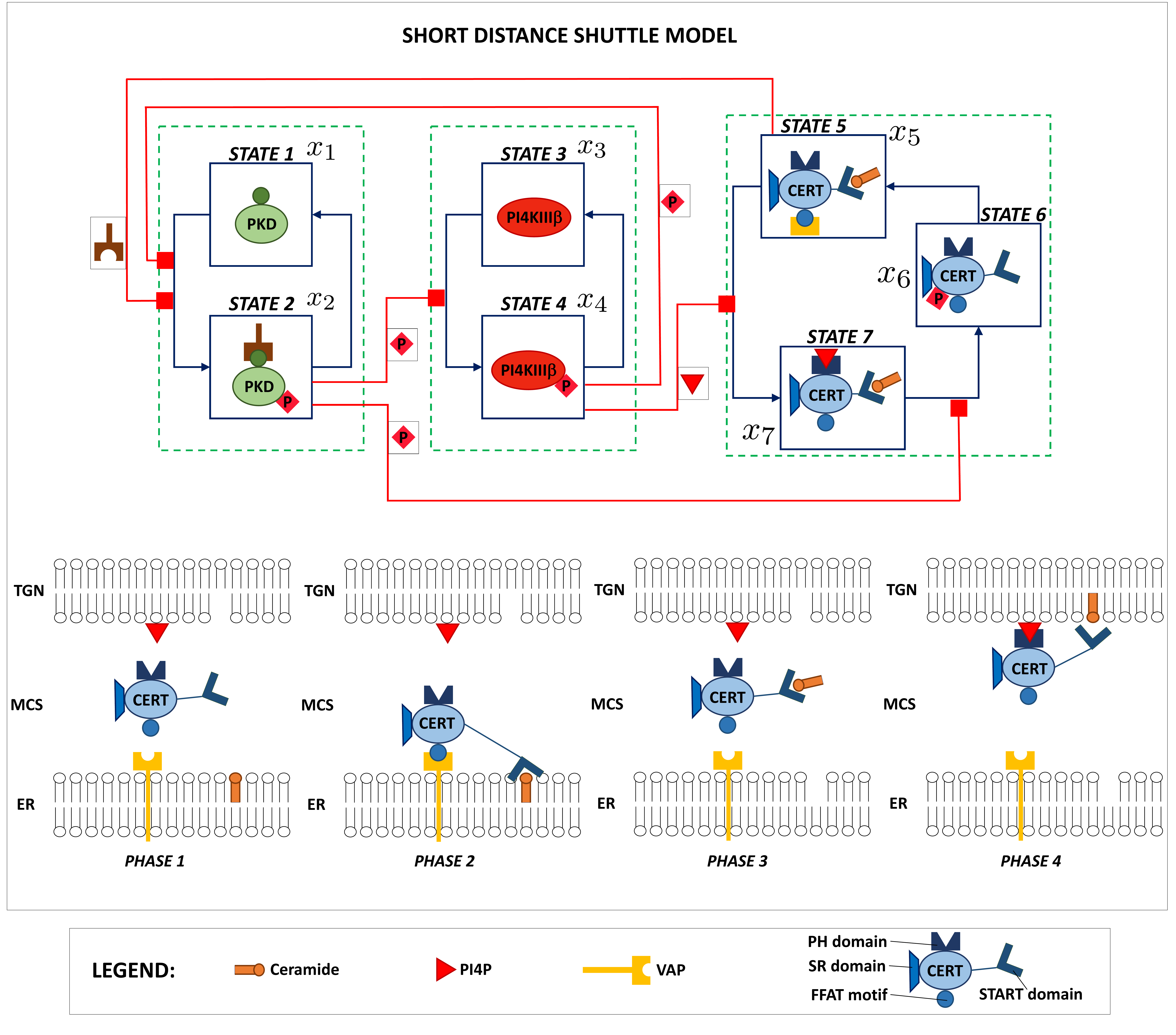} \\[-2pt]
	\includegraphics[width=0.8\textwidth]{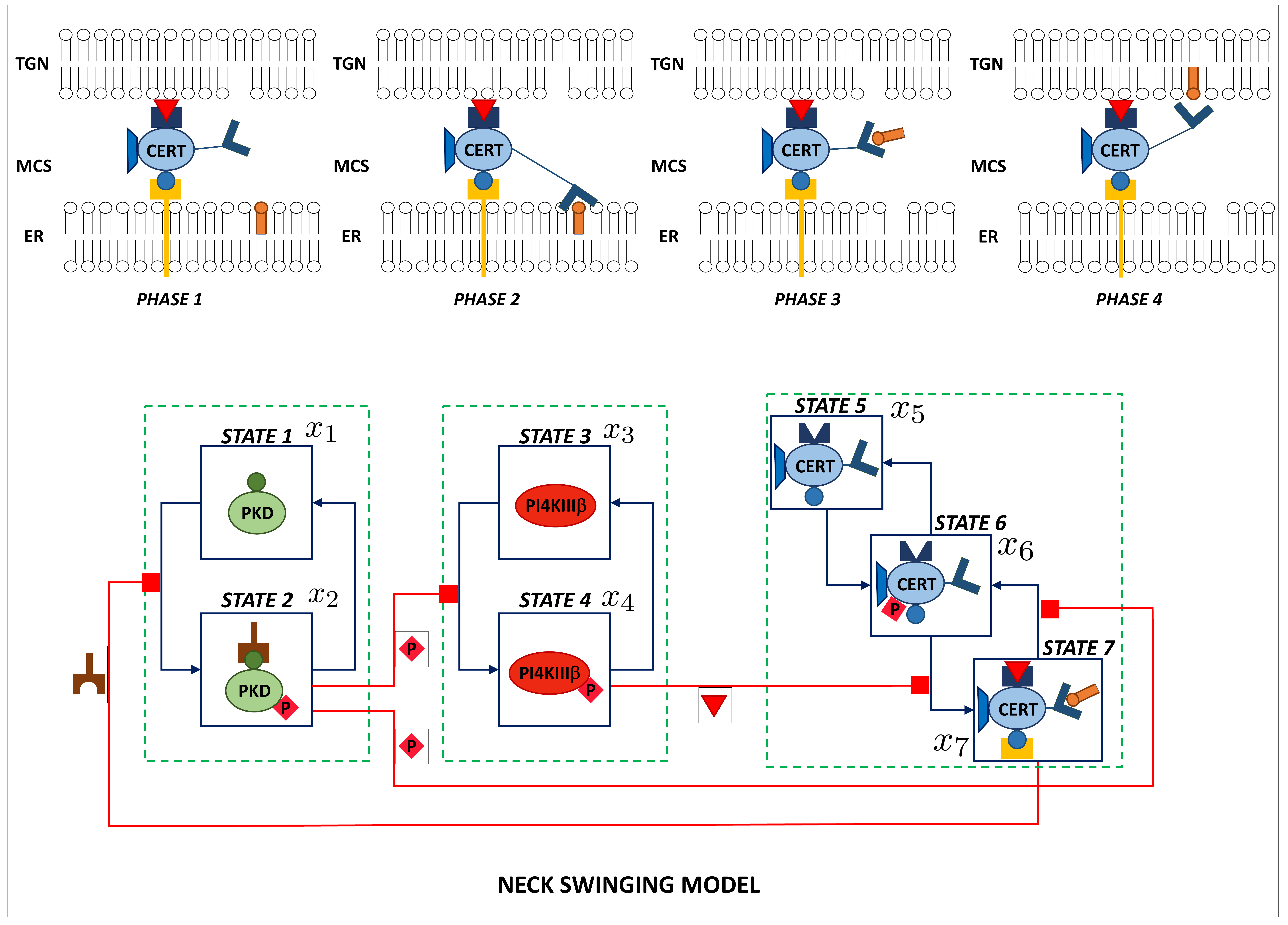}
	\caption{\footnotesize The short distance shuttle (SDS) model, top, and the neck swinging (NS) model, bottom, for CERT-mediated ceramide transfer \citep{Hanada2010,Weber2015} visualised as flow-inducing networks \citep{Giordano2017}; figure adapted from \cite{Giordano2018}. }
	\label{fig:Golgi}
\end{figure}

So far, we have considered the effect of step (constant) inputs; however, the analysis can be extended to systems subject to periodic inputs \citep{Cournac2009,Fiore2018,Ingalls2004,Rahi2017,Russo2010}. Robust results based on the $BDC$-decomposition can be provided also in the case of periodic inputs \citep{BCGZ2022}, as mentioned in Remark~\ref{rem:multiaff}. Robustness analysis in the frequency domain has also been performed by \cite{Hara2019}.

As a side note, the system sensitivity to parameter variations can also be investigated through bifurcation analysis, and structural approaches to this aim have been proposed by \cite{Okada2018}.

\subsubsection{Oscillations and multistability}

Oscillatory and multistable behaviours are ubiquitous in biological systems at all scales.
Oscillations are fundamental within cells, spanning from periods of seconds in calcium signalling to 24 hours in the circadian clock -- fundamental examples are calcium oscillations, glycolytic oscillations, cAMP oscillations in pancreatic cells, oscillations in NF-$\kappa$B, p53 and Wnt signalling, oscillations in the kinases driving the cell cycle, the segmentation clock that drives pattern generation in vertebrate embryonic development \citep{Doyle2006,Epstein1998,Goldbeter2012,MENGEL2010,Yu2007,Zak2005} -- and determining their robustness is crucial \citep{Ghaemi2009}. Oscillations are also observed in whole organisms: some examples are heartbeat, sleep phases, circadian rhythms, menstrual cycle, seasonal changes in the coat or plumage of some animal species that allow camouflage into the environment \citep{Ebenhoh2013}, and even adaptation to wave motion \citep{Burlando2022,Mucci2020}. Synchronised oscillations and entrainment often occur as well \citep{Ashwin2016d,Bagheri2008,Bagheri2008,balazsi2003increased,Gunawan2006,Hammond2007,Katz2024,Kuramoto1984,Li2006,Russo2010,Scardovi2010,Uriu2012}.
On the other hand, multistability is crucial, \eg in cell differentiation and in ruling cell fate, including life cycle, cell division, and apoptosis \citep{Craciun2006b,Eissing2004,Ferrell2012,Xiong2003,Zorzan2019}. It is also linked to pattern formation \citep{Arcak2013,Epstein1998,Hori2015} and pathogenesis \citep{Burlando2020,Cloutier2012,Demori2022,Demori2024}.

Building robust and tunable biomolecular oscillators \citep{Atkinson2003,Blanchini2014c,Samaniego2016,Samaniego2017b,CGF2020,Elowitz2000,Jayanthi2012,Landau2023,Montagne2011,novak2008,Purcell2010,samad2005,Tsai2008,Wang2006,Wilhelm1995,Wilhelm1996} and switches \citep{Atkinson2003,Samaniego2016,Gardner2000,sabouri2008,samad2005,Wilhelm2009} is a sought-after objective in synthetic biology. In highly uncertain environments, design principles based on structural insight are precious to guarantee that artificial biological circuits will keep exhibiting the desired qualitative behaviour thanks to their structure, and regardless of environmental fluctuations, perturbations and variability.

Structural explanations for the onset of oscillatory and multistationary behaviours have been widely studied in systems with a sign-definite Jacobian matrix, whose structure can be captured by the sign pattern matrix $\Sigma$ defined in Section~\ref{sec:signpattern} or, equivalently, by the associated signal graph where each node represents a state variable and each edge represents one of the non-zero entries of the Jacobian matrix, along with its positive or negative sign. Then, it is of interest to look at the sign of the cycles in the signal graph, which is defined as the sign of the product of the Jacobian entries associated with the edges forming the cycle; the sign of the cycles can be assessed without knowing the exact parameter values associated with the strengths of the interactions, as it only requires knowing their sign, and is therefore a structural characteristic of the system that describes a whole family.

A pioneering conjecture by \cite{Thomas1981} is that the presence of (at least one) positive cycle is a necessary condition for multistationarity, while the presence of (at least one) negative cycle is a necessary condition for sustained oscillations. The conjecture was subsequently proved by \cite{Angeli2009b,Gouze1998,Snoussi1998}. It is worth pointing out that these conditions are known not to be sufficient. For instance, the presence of positive cycles (associated, \eg with autocatalysis) can destabilise the system equilibrium and favour the onset of sustained oscillations \citep{Tsai2008}.

Graph-theoretic conditions for oscillations and multistability (and in some cases even pattern formation) have been studied \citep{Domijan2012,Mincheva2008}, also in connection to bifurcations \citep{Domijan2009,Leite2010}.

Graph-based conditions for multistationarity have been also provided by \cite{Craciun2006b,Kaufman2007,soule2003},
while structural conditions based on injectivity properties and on graph-theoretic tools to assess or rule out the presence of multiple equilibria have been proposed by \cite{Banaji2007,Banaji2009,Craciun2005,Craciun2006}. 
On the other hand, graph-based conditions for oscillations have been proposed by \cite{Culos2016,Mincheva2011,Richard2011}, also in connection to Hopf bifurcations \citep{Angeli2014}.

A structural classification of  systems with a sign-definite Jacobian, based on the sign of the cycles, has been proposed by \cite{Blanchini2014structural}.  A strong (respectively, weak) candidate oscillator can exclusively (respectively, possibly) transition to instability due to the presence of a complex pair of eigenvalues, while a strong (respectively, weak) candidate multistationary system can exclusively (possibly) transition to instability due to the presence of a real eigenvalue. Then, a system is a strong candidate oscillator if and only if all the cycles in the associated Jacobian graph are negative, and a strong candidate multistationary system if and only if all the cycles are positive. Strong candidate oscillators (respectively, multistationary systems) do not necessarily exhibit sustained oscillations (respectively, multiple equilibria); however, any transition to instability possibly induced by parameter variations is necessarily associated with the local appearance of sustained oscillations, which can be possibly associated with a Hopf bifurcation and with a limit cycle (respectively, of more equilibria, which can be possibly associated with a pitchfork bifurcation). The result is exemplified next.

\begin{example}[\textbf{Repressilator and promotilator: Oscillations and multistability}]
The repressilator and promotilator circuits in Example~\ref{ex:represspromot} are characterised by a single negative cycle and a single positive cycle, respectively. For both systems, it can be seen that a compact and convex set $\mathcal{S} = \{x \colon x_i \geq 0, \, i=1,2,3, \mbox{ and } \sum_{i=1}^3 x_i \leq K \}$, where $K$ is a positive constant, is positively invariant and therefore, as discussed in Section~\ref{sec:equilibria}, an equilibrium exists; moreover, there are no equilibria on the boundary of $\mathcal{S}$. The repressilator admits a unique equilibrium, in view of degree theory and in particular of the equality in \eqref{eq:degreetheory}, because its Jacobian $J_R$ is structurally non-singular and $\det(-J_R)>0$ structurally. This is however not the case for the promotilator. Consider the ``transition to instability'' in which the system parameters are perturbed so that the Jacobian is marginally stable and has either (a) a real eigenvalue equal to zero, or (b) a pair of purely imaginary eigenvalues, while all other eigenvalues have a negative real part. The structural classification by \cite{Blanchini2014structural} implies that only the scenario (b) is possible for the repressilator, while only the scenario (a) is possible for the promotilator. In fact, for the repressilator, $J_R$ is structurally non-singular, which rules out scenario (a), and can only admit imaginary eigenvalues (all the coefficients of the characteristic polynomial are positive); in this case, the linearised system exhibits local sustained oscillations. Conversely, for the promotilator, the Jacobian $J_P$ is a Metzler matrix (\ie it has non-negative off-diagonal entries) and therefore, in view of Perron-Frobenius theory (see, \eg \cite{BermanPlemmons1994,Bullo2024}), its dominant eigenvalue (\ie the one with the largest real part) must be real, which rules out scenario (b); in this case, at the considered equilibrium $\bar x(\vartheta^*)$, $\det[J_P(\bar x(\vartheta^*))]=0$. If the parameters vary so that $\det[J_P(\bar x(\vartheta))]<0$, then the equality in \eqref{eq:degreetheory} implies the appearance of (at least) two stable equilibria, and hence multistability.
\end{example}

The classification by \cite{Blanchini2014structural} (suitably extended to sign-definite interconnections of structurally stable monotone subsystems, see \cite{Blanchini2015structuralclass}, or of structurally stable positive-impulse-response systems, see \cite{Blanchini2017e}) has been applied to gain insight into physiological mechanisms, such as regulation of blood pressure and volume (where two concurrent candidate oscillator structures are present, one acting through blood volume and one through vasoconstriction, see \cite{Burlando2019}), into the onset of diseases such as Mal de Debarquement syndrome \citep{Burlando2022,Mucci2020}, fibromyalgia \citep{Demori2022,Demori2024} and amyotrophic lateral sclerosis \citep{Burlando2020}, and into ecological systems related to honey bee health where bistability leads to two possible extreme outcomes in terms of survival or extinction of the bee colony \citep{breda2022deeper}.

\paragraph{Timescale separation and oscillations.} In the case of a negative loop of several subsystems, each associated with its own time constant, the insertion of a single element that is slower than the others has a stabilising effect (``strong separation of timescales counteracts the tendency to oscillate'', see \cite{alon2019introduction}, page 100) and conversely similar timescales promote oscillations \citep{alon2019introduction,Blanchini2018homogeneous,novak2008}. Analysing the system's timescales is therefore important to predict the onset of oscillations.

\paragraph{Cooperative systems.} A monotone system, or cooperative system, $\dot x(t) = f(x(t))$ generates an order-preserving flow \citep{Smith1988}; namely, if $x_1(0) \geq x_2(0)$, then $x_1(t) \geq x_2(t)$ for all $t \geq 0$, where inequalities hold componentwise \citep{Hirsch2006b,Smith1995}. Monotone (and near-monotone, \cite{Sontag2007}) systems are commonly encountered in nature: many models of biological phenomena and of chemical reaction networks are monotone \citep{Angeli2006,Angeli2010,Leenheer2007} and monotone system theory can be used to invalidate biological models \citep{Angeli2012}. A nonlinear system is monotone if its Jacobian matrix $J(x)$ is Metzler for all $x$. In view of their order-preserving properties, monotone systems cannot have attractive periodic orbits \citep{Smith1995} (nor chaotic behaviours) and hence cannot yield persistent oscillations. The same dynamic properties are enjoyed by systems that become monotone after a state transformation that only alters the signs of the state variables; note that such a transformation would not alter the sign of the cycles in the Jacobian graph. This fact is consistent with the classification discussed above, since, for systems associated with a strongly connected Jacobian graph, monotonicity is equivalent to the absence of negative cycles \citep{Sontag2007}; a monotone system (or a system that can be cast into the monotone system framework) is thus a candidate multistationary system. If the graph is weakly connected, monotonicity is equivalent to the fact that all oriented paths connecting the same pair of nodes have the same sign.
The concept of monotone system has been extended to systems with inputs and outputs \citep{Angeli2003b}: a system $\dot x(t)=f(x(t),u(t))$, with output $y(t)=g(x(t))$, is input/output monotone if, given ordered initial conditions $x_1(0) \geq x_2(0)$ and inputs $u_1(t) \geq u_2(t)$,  both the states and the outputs preserve the order, $x_1(t) \geq x_2(t)$ and $y_1(t) \geq y_2(t)$ for all $t \geq 0$ (all inequalities hold componentwise). This concept is precious to investigate interconnections of monotone systems; in such a context, conditions for multistability have been investigated for monotone systems \citep{Enciso2005b} and monotone input/output systems \citep{Angeli2004}, and also conditions for oscillations (as opposed to convergence to a stable equilibrium) have been provided for input/output monotone systems under negative feedback \citep{Angeli2008}. As mentioned earlier, also the classification of candidate oscillatory and multistationary systems has been extended to sign-definite interconnections of structurally stable monotone subsystems \citep{Blanchini2015structuralclass}.

\paragraph{2-cooperative systems.} Structural oscillations can be observed in strongly 2-cooperative systems, which are characterised by a Jacobian matrix with a specific sign structure. A system of size $n\geq 3$ is strongly 2-cooperative if its Jacobian is an irreducible matrix with sign pattern
\begin{equation*}
\Sigma_{S2C} = \begin{bsmallmatrix}
    * & +& 0&\dots  &0&-  \\
    +& * & +&\dots&0&0\\
    0& + & *& \dots&0&0\\
    \vdots & \vdots & \vdots & \ddots & \vdots &\vdots\\
     0& 0 &   0& \dots&*&+\\
    -& 0 &0&    \dots&+&*
\end{bsmallmatrix},
\end{equation*}
where $*$ entries are allowed to have any sign, while $+$ entries (respectively, $-$ entries) must be non-negative (respectively, non-positive).
For a strongly 2-cooperative system, under suitable technical assumptions, we can explicitly characterise a set of initial conditions from which the system solutions asymptotically converge to a non-trivial periodic orbit \citep{Katz2024osc,Katz2025osc}. Biological models that exemplify this result are the Goodwin model (\cite{Goodwin1965}, see also \cite{gonze2021,Griffith1968,sanchez2009,Tyson1975}) and the Field-Noyes model for the Belusov-Zhabotinskii chemical reaction (\cite{FieldNoyes1974}, see also \cite{hastings1975existence,murray1974}), as well as biomolecular oscillator models (\eg \cite{Blanchini2014c,Samaniego2016,Samaniego2017b}), which indeed yield sustained oscillations for appropriate initial conditions, regardless of the values of the system parameters (\ie structurally), thanks to the sign pattern of their Jacobian matrix.

\subsection{Divide et impera: managing large scale and complexity}\label{sec:divideetimpera}

Most biological and epidemiological systems are large-scale and complex networks, with a tight interplay between network structure and system dynamics. They are typically very difficult to investigate through closed-form analytical approaches (which are more suitable for small-size systems) and are frequently studied by resorting to vast numerical simulation campaigns. How can one manage such a large scale and complexity, and still carry out a rigorous mathematical structural analysis? To achieve this goal, divide-and-conquer approaches are precious.

Analysing large-scale systems efficiently is also fundamental to untangle the interplay of multiple parameters in highly complex and uncertain settings, and hence enable a systematic multi-factorial analysis \citep{Simmons2021,Catford2022} to single out the individual role of each parameter in enabling or preventing a property of interest. Identifying crucial parameters and key structural motifs can guide the design of tailored experiments to study the effect of particular factors or stressors associated with a parameter, or to falsify or contrast competing models by comparing experimental observations with structural model predictions.

\paragraph{Network motifs.} Often, system behaviours and dynamic properties are highly correlated with the relative abundance of small subnetworks known as \emph{network motifs} \citep{Alon2007,Blokhuis2020,Drengstig2008,Milo2002a,Prill2005,Sneppen2010,Stone2019,Wong2017,yeger2004}. Such small subnetworks are easy to analyse and their presence within a large-scale network can shape the resulting emergent behaviour. In biological networks, some characteristic motifs -- such as negative autoregulation, incoherent and coherent feedforward loops, double inhibition, double-activation and activation-inhibition, and positive and negative feedback loops in general -- happen to be much more frequent than they would be in randomly generated network \citep{alon2019introduction}, and hence appear to have evolutionary advantages thanks to the key dynamic behaviours that they induce. Thus, identifying a well-known motif within a complex network can be enough to predict its overall dynamics, or to differentiate between the characteristic behaviours of alternative models (see, \eg \cite{Cloutier2012,Giordano2018,Goentoro2009,Kaplan2008,Rahi2017}).

\paragraph{System simplification.} To simplify the system, one can reduce its size through model order reduction techniques; see, \eg \cite{Feliu2012,Rao2013,schaft2015,Snowden2017}, as well as the brief discussions in Section~\ref{sec:ind-on-net-dimred} and Section~\ref{sec:dimred}). Alternatively, singular perturbation arguments \citep{KhalilBook2002} can be leveraged, which often rely on timescale separation; this topic extends beyond the scope of this monograph and the interested reader is referred, \eg to \cite{Chaves2006,Chaves2019,franci2016three,Jayanthi2011,meyer2003}; see also the books by \cite{alon2019introduction,DelVecchio2015}.

\paragraph{System decomposition and aggregation.} Also decomposing the system into subsystems that enjoy special structural properties, and can thus be collapsed into a single \emph{aggregated node}, enables the structural analysis of large-scale complex network systems. By lumping whole subsystems with a peculiar structural behaviour into a single aggregated node, a huge network with myriads of nodes can be partitioned into the interconnection of a smaller number of aggregated nodes, each of which structurally enjoys a relevant property (see Figure~\ref{fig:AggregationScaling}, left). A classic property that ensures that a subsystem is ``well-behaved'' is monotonicity: due to its order-preserving property, a monotone system approximately behaves as a first-order system (and can thus be replaced by a single aggregated node). The decomposition of biological systems into monotone subsystems has been investigated, \eg by \cite{DasGupta2007}.
For instance, as noted earlier, the structural classification of candidate oscillatory and multistationary systems by \cite{Blanchini2014structural} can be generalised to more complex systems resulting from the sign-definite interconnection of subsystems that are structurally stable and monotone \citep{Blanchini2015structuralclass} or that are structurally stable and have a positive impulse response \citep{Blanchini2017e}. Also the methods for the computation of the structural influence matrix \citep{Giordano2016} can be analogously generalised to this case \citep{Blanchini2018aggregates}.

\paragraph{System scaling.} An interesting research question is whether -- or under which conditions on the interconnection topology -- a system interconnecting several subsystems that structurally enjoy a certain property $\mathcal{P}$ exhibits the same property $\mathcal{P}$. How does the topology of the \emph{aggregated} graph affect the global behaviour? Is the structural property of the individual subsystems preserved regardless of the interconnection topology, or under appropriate conditions? For instance, the sign-definite interconnection of monotone subsystems is in turn a monotone system if all the interactions among subsystems are activating, or more in general if there are no negative cycles in the \emph{aggregated graph} (whose \emph{aggregated nodes} correspond to the subsystems). In this way, one can analyse the properties of \emph{networks of networks}; an interesting case is that of networks of networks whose components share the same structure (\ie belong to the same family, possibly enjoying a structural property), but are arbitrarily heterogeneous in the parameters, as shown in Figure~\ref{fig:AggregationScaling}, right. Sometimes it is possible to ``scale up'' properties that hold for the interconnected subsystems and project them onto the whole network of networks, even regardless of the network topology; interesting results have been obtained for the topology-independent stability of networks of networks \citep{Blanchini2017,DG2020,DG2021,DG2021ecc}, and networks interconnecting first order dynamics with Michaelis-Menten interaction functions, under suitable assumptions, have been shown to be stable regardless of the interconnection topology \citep{BLANCHINI2023110683}.
Scaling is particularly important when considering systems of akin-but-heterogeneous systems that often arise
in biology: cell populations, interacting individuals, metapopulation models in ecology and in epidemiology.

\begin{figure}[ht!]
	\centering
		\includegraphics[width=\textwidth]{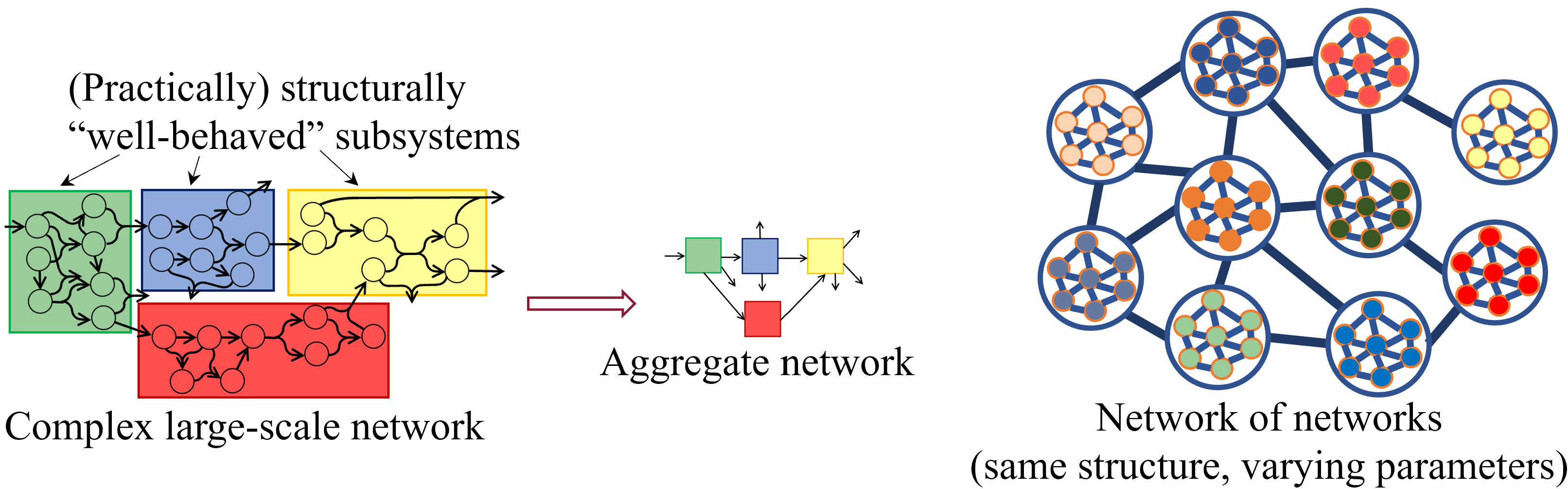}
	\caption{\footnotesize \small Left: aggregate network obtained by collapsing ``well-behaved'' subsystems into a single aggregated node. Right: network of networks, all with the same structure but with arbitrarily different parameter values.}
	\label{fig:AggregationScaling}
\end{figure}

\subsection{Beyond structural properties}

Consider again the case study in Example~\ref{ex:Golgi}. To explain the same phenomenon, two different models have been proposed: which is the right one? Structural analysis is powerful and captures the essential (qualitative) differences between the two models \citep{Giordano2018}. However, to fully distinguish among them, we also need to consider parameter-dependent properties and quantitative aspects.

In addition, although surprisingly often steady-state input-output influences are structurally signed, as discussed in Section~\ref{sec:influences}, undetermined influences (`?') can also be observed, and structural approaches cannot cast any further light on undetermined influences, or quantify how far the influence is from being structurally signed. A quantitative assessment can be obtained, for instance, through the integration of structural and probabilistic methodologies. Computing a probability triplet $\{p_+, p_-, p_0\}$, with $p_+ + p_- + p_0=1$ (where $p_*$ denotes the probability of the influence having sign $*$), by suitably sampling the parameter space -- \eg via rigorous uncertainty randomisation and Monte Carlo methods, see \cite{BarmishLagoa1997,Metropolis1953,tempo2013randomized,Sagar1998,Sagar2011} -- can shed light on undetermined influences, and for instance reveal that the influence is practically positive ($p_+=0.98$), or that all signs are equally likely ($p_+ = p_- = p_0=1/3$).

Another critical case is that of a family of systems characterised by a structure for which a property $\mathcal{P}$ of interest
cannot be proven to be structural, even though it consistently holds for randomly picked elements of the family.
For instance, some chemical reaction network structures considered by \cite{Blanchini2014d}, for which structural stability cannot be guaranteed through a structural polyhedral Lyapunov function, still have all their eigenvalues in the left half plane when the system functions and parameters are randomly generated. Rigorous probabilistic and randomised approaches can appropriately sample the family of systems, so as to unveil the stability probability and the key features and parameters that can prevent the system from being stable.

\subsubsection{Structural analysis: strengths and limitations}
Structural analysis offers strong insight into parameter-free properties of systems, but cannot provide any sharper information when a property fails to hold structurally because it is parameter-dependent.
In particular, a structural assessment provides a binary ``yes'' versus ``no'' answer to the question: does the property $\mathcal{P}$ of interest hold structurally?
When the answer is ``yes'', discovering that a property is structural is an exciting achievement: it offers very strong guarantees on the behaviour of a whole family of systems with the same interconnection structure, independent of parameter values, and provides powerful insight into the mechanisms underlying its remarkable robustness.
However, often a property does not hold structurally (and, even when it does, this is typically very hard to prove). For some models, one may conjecture that a property holds structurally, because it consistently holds for random parameter values, yet no proof is available; and the conjecture may be false, because the property may not hold in a zero measure set. When the answer to the question is ``no'', structural approaches cannot provide any insight.

When a property -- especially an expected one, based on repeated experimental observations -- does not hold structurally, it is crucial to understand why. Paraphrasing Tolstoy’s quote from \emph{Anna Karenina} (\emph{Happy families are all alike; every unhappy family is unhappy in its own way}): given a property, system families enjoying it structurally are all alike as far as the property of interest is concerned; every system family that does not enjoy it structurally lacks it in its own way.
Perhaps the property is not structural because it just holds in a (small) subset of the parameter space; or it
always holds when some parameters satisfy a specific relation; or it holds not always, but with a sufficiently high probability;
or it would be structural if a more appropriate model were adopted to describe the real phenomenon. Or a property may be practically structural in probability, in the sense that it holds in the whole parameter space apart from a zero measure set.
Hence, given a family of systems characterised by a structure (specified without resorting to numerical bounds), understanding why an expected property fails to hold structurally is a challenge for which new methodologies need to be developed, going beyond structural analysis.

While structural methods can assess whether a property of interest holds for all systems in a given family, computational approaches can test whether the property holds for one specific system in the family; however, case-by-case numerical simulations focusing on single snapshots of the parameter space cannot reveal the essence of mechanisms, nor provide rigorous theoretical guarantees. To bridge the gap between purely qualitative (structural) and purely quantitative (numerical) analysis, an integration of structural \citep{blanchini2021structural} and probabilistic \citep{tempo2013randomized} methods, possibly tailored to exploit the peculiar features of systems in biology and epidemiology, is required.
The aforementioned gap, as we have seen, is partially filled by approaches for robustness analysis, well studied in the realm of control theory \citep{Barmish1994,sanchezpena1998,zhou1998essentials}.
Still, the classic robustness paradigm checks whether a property holds for all the systems within a (possibly quantitatively) specified sub-family, and therefore again provides a binary ``yes'' versus ``no'' answer, without offering any insight as to why the (desired or expected) property is, or is not, verified.
A paradigm shift is needed to enable a whole \emph{continuum} of answers, checking whether the property holds structurally, or robustly, or with a guaranteed probability.  In particular, it is fundamental to understand why the property does not hold for some parameter choices, by identifying the system features or motifs that are crucial to enable or prevent the property, and the key parameters that must be finely tuned to enforce the property. For instance, the property may hold only in specific portions of the parameter space, or only for specific values of certain parameters, or only if specific relations among the parameter values are satisfied.

\subsubsection{Integrating structural and probabilistic approaches}

Cutting-edge probabilistic and randomised algorithms for the analysis and the control of uncertain systems are illustrated, \eg by \cite{tempo2013randomized}. These methodologies have solid theoretic foundations and provide instructions on how (and how much) to sample randomly, when running extensive numerical simulations, so as to achieve rigorous guarantees in probability. The uncertainty is confined within a set, as in the worst-case approach of robust analysis, and, in addition, it is a random variable with a given probability distribution. Then, rigorous uncertainty randomisation allows to explore the parameter space. Probabilistic techniques such as Monte Carlo methods \citep{Metropolis1953} can identify the good (or bad) set, where the desired property is (or is not) verified. The confidence of the prediction can be quantified rigorously through probability inequalities, and the sample complexity required to draw conclusions with a desired confidence can be derived, \eg via the Chernoff bound \citep{Chernoff1952} or based on statistical learning theory \citep{Sagar1998}.

Probabilistic robustness has been fruitfully introduced and investigated in engineering \citep{tempo2013randomized}, to characterise the likelihood of a system property to hold, given the probability distribution of model parameters, and to synthesise controllers that achieve a desired objective with arbitrarily high probability. Here we focus however on the \emph{analysis} of natural systems, and not on control synthesis. The randomised approaches described by \cite{BarmishLagoa1997,tempo2013randomized} have not been specifically tailored -- so far -- to the analysis of uncertain biological dynamics. Studies in this direction exist in genomics and proteomics (see, \eg \cite{Komurov2010,Sagar2011,Zaki2012}). Probabilistic studies of biological systems with uncertain parameters have been proposed (see, \eg \cite{Liebermeister2005,Weisse2010}), but they often lack formal guarantees and mostly rely on extensive simulation campaigns; recently, numerical methods have been proposed, for the probabilistic analysis of complex biological systems with parametric uncertainty, that significantly reduce the required computational burden in exploring the parameter space \citep{sutulovic2024efficient}.

Here, we sketch a possible framework to integrate structural and probabilistic approaches.

\paragraph{An integrated structural and probabilistic (ISP) framework.} Given a qualitative family $\mathcal{F}$ of systems, associated with a structure, a property $\mathcal{P}$ of interest can hold with varying prevalence, as visualised in Figure~\ref{fig:StructuralProbabilistic} (top row): for all systems in $\mathcal{F}$, and hence with probability $p=1$ (\emph{structural}); for all systems in a quantitatively specified sub-family $\mathcal{F}_r \subset \mathcal{F}$ (\emph{robust}); with large enough probability $p>p^*$ for a system in $\mathcal{F}_r$ (\emph{practically robust}) or $\mathcal{F}$ (\emph{practically structural}); for at least one system in $\mathcal{F}$, which is guaranteed if it holds with probability $p>0$ (\emph{compatible}); for no system in $\mathcal{F}$, and hence with probability $p=0$ (\emph{structurally incompatible}).
Importantly, in a probabilistic framework, the above statements apply up to a zero-measure set: $p=1$ guarantees that the property holds for \emph{all} systems \emph{possibly excluding a set of measure zero}; $p=0$ entails that the property holds \emph{at most on a set of measure zero}.

\begin{figure}[ht!]
	\centering
		\includegraphics[width=\textwidth]{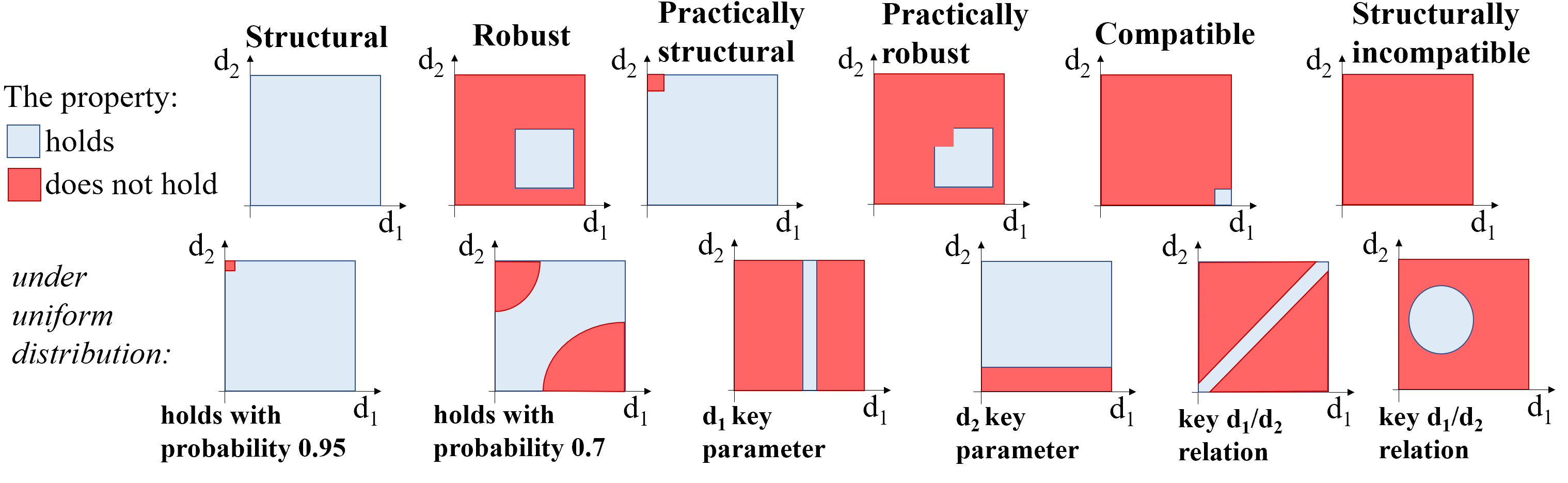}
	\caption{\footnotesize \small Illustration of varying prevalence with which a property can hold in a parameter space. The top row shows a progression of scenarios where the property of interest holds with various prevalence. Left to right, the property holds: always, for all values in the parameter space (structural); for all values in a subset of the parameter space (robust); almost always, for almost all values in the parameter space (practically structural); for almost all values in a subset of the parameter space (practically robust); for some values in the parameter space (compatible); for no value in the parameter space (structurally incompatible). The bottom row illustrates examples of probabilistic assessment (the property holds with probability $0.95$ or with probability $0.7$ under uniform distribution), as well as examples of cases when there is a key parameter that needs to be fixed to a precise value or restricted to a precise interval for the property to hold, and examples of cases when the parameters must be in a key relation for the property to hold (\eg $d_1=d_2$, or the point $(d_1,d_2)$ belongs to a circle in the parameter space).}
	\label{fig:StructuralProbabilistic}
\end{figure}

To bridge the gap between numeric and structural approaches (which is only partially filled by
worst-case robust approaches that also yield ``yes'' versus ``no'' answers), we can look for a continuum of possible answers.
For example, consider the generic family of nonlinear systems $\dot x(t)=f(x(t),d)$, endowed with an underlying network structure (which we assume is known) and affected by the uncertainty $d$ (a vector of parameters or functions). The uncertain parameters can be assumed to be bounded within a known support set (which, for structural assessments, corresponds to the entire parameter space) and characterised by a known probability distribution, or probability density function. For randomised analysis, considering different distributions (\eg Gaussian, Poisson, Exponential) can deeply affect the results \citep{tempo2013randomized}; when \emph{a-priori} knowledge about the probability distribution is not available, uncertainty can be assumed to have the worst-case distribution, which is often the uniform distribution \citep{BarmishLagoa1997}.
Numerical approaches can assess whether a property $\mathcal{P}$ holds for one element of $\mathcal{F}$, corresponding to a known $d=d^*$. However, in general, $d \in \{d^* + \rho \mathcal{D}\}$, where $\mathcal{D}$ is an uncertainty shape set (\eg the unit ball of some norm) and $\rho$ is the \emph{uncertainty radius}, quantifying the magnitude of the uncertainty.
Then, as visualised in Figure~\ref{fig:Probabilistic}, the property $\mathcal{P}$ is probabilistically (\ie up to a zero-measure set) \emph{structural} if it is verified with probability $p=1$ for all $\rho>0$, \emph{robust} if it is verified with probability $p=1$ for $0 < \rho < \rho_T$ and with lower probability for $\rho \geq \rho_T$, \emph{practically structural} (respectively, \emph{practically robust}) if it holds with probability $p > p^*$, with $1 > p^* > 0$ large enough according to the considered application, for all $\rho > 0$ (respectively, for $0 < \rho < \rho_T$).
Identifying the threshold $\rho_T$ can quantify the system's fragility (or resilience) and further investigations can reveal which are the crucial
components of $d$, or relations between components of $d$ (see Figure~\ref{fig:StructuralProbabilistic}, bottom row), or the crucial motifs in the interconnection structure, that can either enable $\mathcal{P}$ or prevent it from holding. This insight can be achieved by exploring the uncertainty space with randomised sampling techniques, to disentangle the effect of the various parameters and of important structural motifs; this type of analysis can help explain why the property does (or does not) hold with a certain probability, and uncover the crucial parameters, or
structural motifs, that enable or prevent $\mathcal{P}$, which can thus be targeted when planning interventions aimed at enforcing $\mathcal{P}$.

\begin{figure}[ht!]
	\centering
		\includegraphics[width=0.7\textwidth]{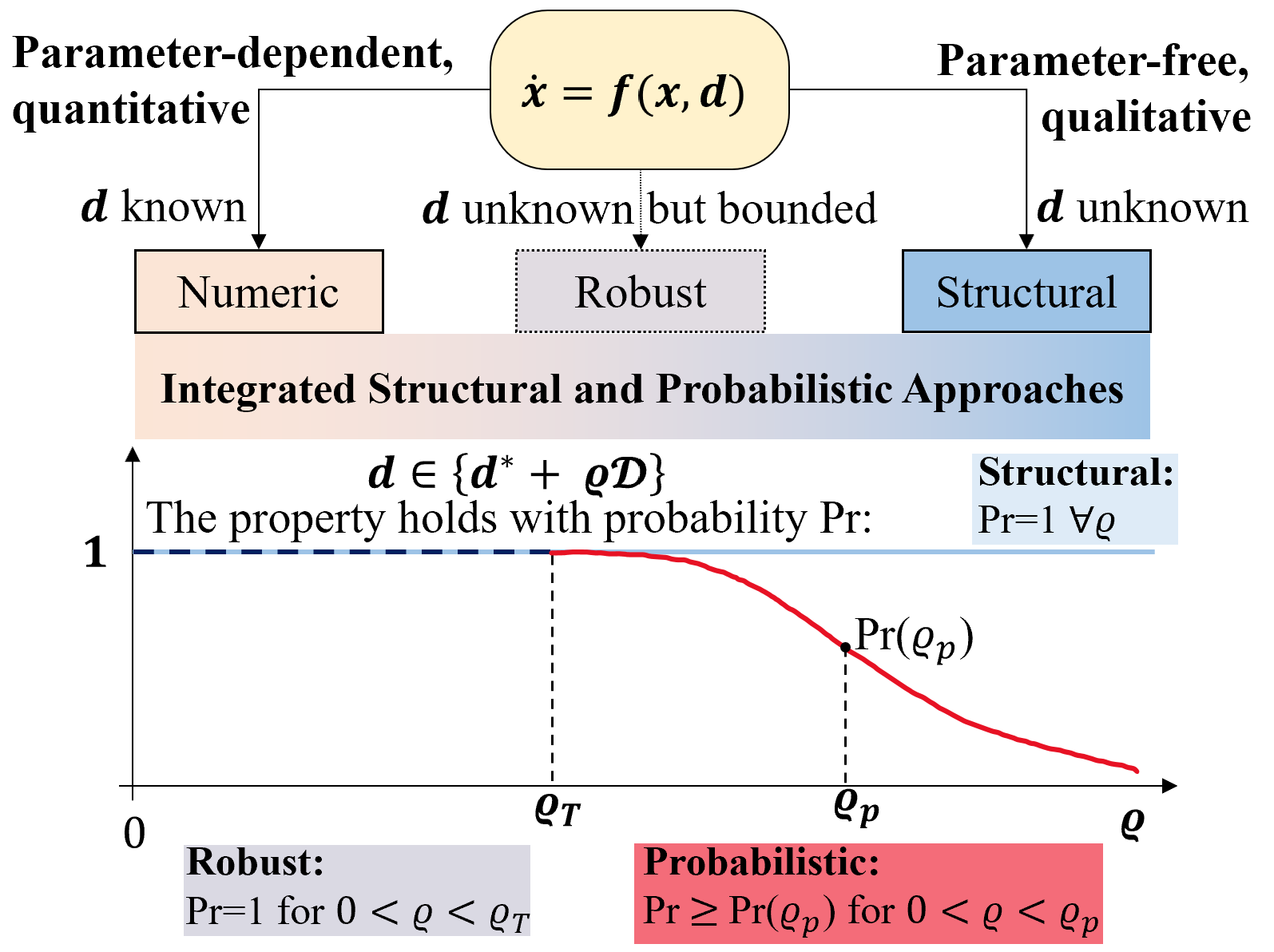}
	\caption{\footnotesize \small Integrated structural, robust and probabilistic assessment of properties of systems.}
	\label{fig:Probabilistic}
\end{figure}

\paragraph{ISP hypercube analysis.} Realistic parameter values in biological and epidemiological systems can vary across several orders of magnitude. Still, they can be assumed to lie within known intervals: the support set $\mathcal{C}$ is thus the Cartesian product of all the intervals where the (positive) uncertain parameters are confined,
\begin{equation*}
\mathcal{C} = \{d \colon d_i^- \leq d_i \leq d_i^+, \, i=1,\dots,q\}.
\end{equation*}
Then, a property is robust if it is satisfied for all $d \in \mathcal{C}$, and it is structural if $d_i^-=0$ and $d_i^+=\infty$ (assuming that only non-negative parameters are of interest).
Note that for $d_i^+ < \infty$, starting from the parameter hyperrectangle $\mathcal{C}$, the parameters can always be normalised as $d_i = d_i^- + \delta_i (d_i^+ - d_i^-)$, with $\delta$ in the unit hypercube, $0 \leq \delta_i \leq 1$ for all $i=1,\dots,q$.
An analysis of the parameter space that integrates probabilistic methods with parametric robustness can assess whether a property $\mathcal{P}$ is (practically) structural, or (practically) robust.
When $\mathcal{P}$ is not structural, by suitably slicing and sampling the parameter hypercube through integrated structural and probabilistic (ISP) approaches, one can identify crucial parameters that enable/prevent the property (\eg $\mathcal{P}$ holds provided that $d_k$ is fixed to a particular value, $d_k=d_k^*$) and subsets of the parameter set (characterised by special parameter relations) where the property holds practically structurally, robustly, or with a guaranteed probability.

Exploring the parameter space through randomised algorithms \citep{tempo2013randomized} can help identify crucial parameters, or parameter relations, for the property $\mathcal{P}$ to hold robustly. Fixing one parameter at a time to a given value, and checking whether the property becomes structural for the remaining parameters, can reveal crucial parameters. Probabilistic approaches can also help find correlations between parameters, associated with special relations that must hold to enable the property. Some parameters may need to satisfy specific constraints (for instance, for the Goodwin oscillator with $n$ stages and Hill coefficient $N$, there exists at least one choice of the parameter values for which the linearised system admits complex eigenvalues with non-negative real part if and only if $\cos(\pi/n)\sqrt[n]{N}>1$, otherwise the equilibrium is structurally stable; see \cite{Blanchini2018homogeneous}), or one parameter may need to be larger than another (\eg $\beta > \gamma$ to enable spiking in the SIR epidemic model \eqref{eq:sir}, otherwise the infected variable is monotonically decreasing).
Integrating probabilistic approaches allows us to quantify the indeterminacy when a property $\mathcal{P}$ is \emph{not} structural (or robust), and yield a probability assessment with a given confidence, depending on the number of samples. Randomised sampling can guide the identification of hyperplanes, or generic hypersurfaces, that separate the subsets where the property $\mathcal{P}$ holds (good set) and where it does not (bad set), as shown in Figure~\ref{fig:Hypercube}; it may help reveal a variety where the property is robustly verified, when the points of the good set accumulate in its neighbourhood (\eg if a property only holds for small/high enough values of a certain parameter, the randomisation will generate points of the good set only close to the corresponding face in the hypercube).

\begin{figure}[ht!]
	\centering
		\includegraphics[width=0.4\textwidth]{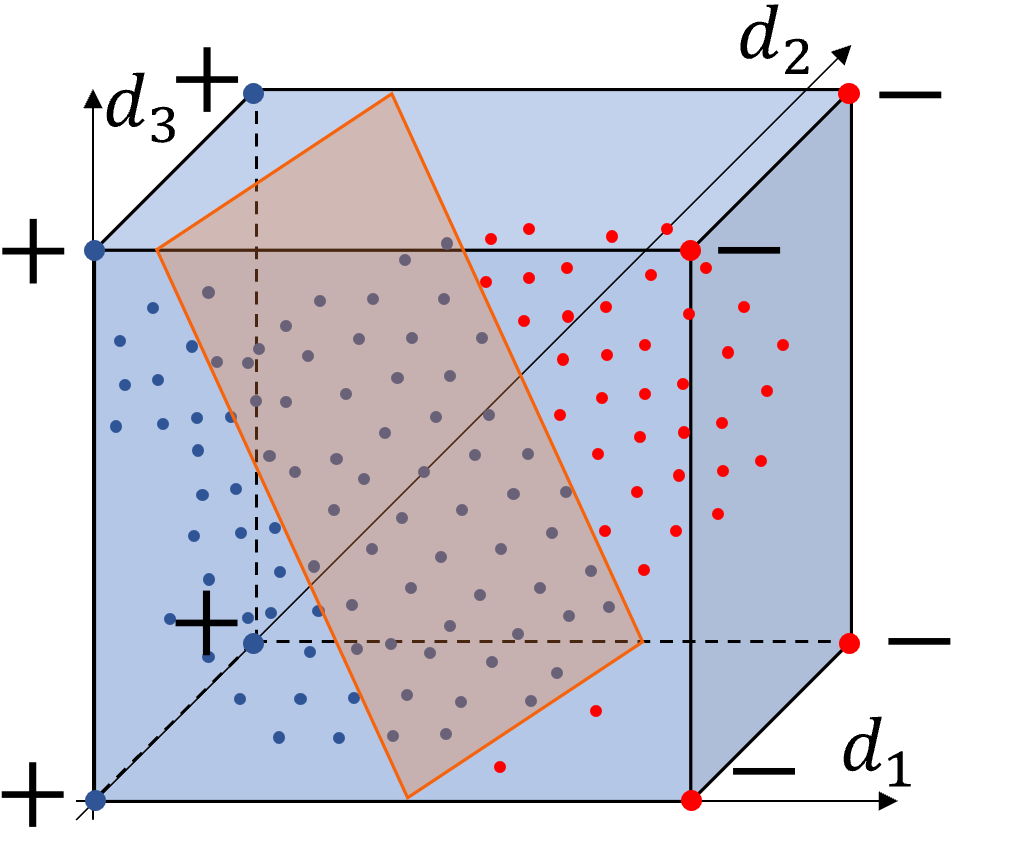}
	\caption{\footnotesize \small Exploration of the parameter space (hypercube) through randomised sampling.}
	\label{fig:Hypercube}
\end{figure}

As discussed earlier (see, \eg Remark~\ref{rem:multiaff}), multiaffine functions defined on a hypercube (without loss of generality) have their minima and maxima at the vertices, and then parametric robustness approaches yield powerful vertex results: a property holds in the whole parameter hypercube if it holds at all its vertices (as is the case when studying structural steady-state input-output influences, as discussed in Section~\ref{sec:influences}, see also \cite{Giordano2016}). Such vertex approaches can help disentangle the effect of multiple uncertainty factors in complex large-scale systems with myriads of interconnections, where detecting and explaining the effect of multiple parameters is challenging and a multifactorial analysis that singles out the role of each parameter is crucial.
Hence, efficient approaches are possible when the property of interest can be expressed by an index $\varphi(d)$ that is a multiaffine function of the uncertain parameters, to be evaluated for $d \in \mathcal{C}$ (\eg for Hurwitz stability, $\varphi(d)$ is the largest real part of the Jacobian eigenvalues, and hence it must be $\varphi(d)<0$ for all $d \in \mathcal{C}$).
Structural approaches lead to algorithms with exponential complexity: $2^q$ vertices must be checked in general, where $q$ is the number of parameters. Performing the vertex exploration efficiently is
crucial. Iterative approaches that integrate structural and probabilistic methodologies could allow to first perform fast randomised sampling of the whole hypercube, and then exhaustively explore the vertices for a progressively reduced number of parameters; we propose this type of approaches as a promising avenue for future research.

Integrated approaches for structure randomisation are also possible, aimed at identifying, in a complex network, the subnetworks and motifs whose sole presence can guarantee that a property of interest is structurally verified, or whose introduction can compromise the structural property. Starting from a given system structure and a property $\mathcal{P}$, the structure can be randomly varied, \eg by introducing
``unitary'' changes and modifications (\ie changing or removing one interconnection at a time), and then reassessing the property each time. This operation can identify key interconnections (\eg chemical reactions or pathways in cell biology, or between-host interactions in epidemiology) whose presence enables $\mathcal{P}$, or which need to be removed for $\mathcal{P}$ to hold.

Probabilistic approaches have been preliminarily explored for the qualitative analysis of large-scale agent-based opinion formation models, whose complexity does not allow for closed-form analysis \citep{DG2023model}, in order to assess the probability of emergent qualitative outcomes in the opinion distribution (such as perfect consensus, consensus, polarisation, clustering and dissensus) when the network topology is unknown and the model parameters and the initial agent opinions are also unknown, but have a known distribution \citep{DG2022,DG2023,DG2024}. This type of approach could also be applied to complex large-scale networked systems in biology and epidemiology (\eg \cite{kerr2021covasim}).

In general, the combination of structural and probabilistic/randomised approaches looks very promising to obtain a complete qualitative and quantitative picture of the system behaviour, highlighting important relations among the system parameters and offering rigorous guarantees.
This is an exciting research direction, which is further explored in the next chapter.
The novel integration of structural, robust and probabilistic methods (possibly tailored to the peculiarities of natural systems) can enable not only analysis, but also control paradigms \citep{Azuma2022,zanudo2017} for complex uncertain dynamical networks in nature, which leverage the obtained insight to guarantee a desired global property through targeted local interventions.

\subsubsection{Uncertainty quantification and generalised polynomial chaos}\label{sec:uq-gpc}

As mentioned earlier, Monte Carlo methods, which require running many simulations with sampled random variables and corresponding random process realisations, are the most common tool for probabilistic analysis. However, they suffer from poor scalability and require extensive computational power and time to gain information about even relatively simple models. They are thus impractical when simulating large models with uncertain parameters, since their sample complexity scales poorly with model dimensionality and grows rapidly with increase in the desired accuracy, leading to the so-called ``curse of dimensionality'' \citep{tempo2013randomized}. Developing flexible and efficient alternative methods to quantify probabilistic robustness is crucial to scale up the investigation of parametric uncertainties in complex models, allow the extensive exploration of large parameter spaces, and enable wider adoption of robustness analysis even when low computational power is a requirement.
An initial study in this direction is offered by \cite{sutulovic2024efficient}, who first propose, to quantify probabilistic robustness in complex models affected by parametric uncertainties,  the use of generalised polynomial chaos (gPC) for uncertainty quantification, as an alternative to computationally expensive Monte Carlo methods, so as to make the analysis more efficient and scalable. They then assess the performance of gPC methods in providing surrogate models for the probabilistic robustness analysis of widely used complex systems in neuroscience.

Generalised polynomial chaos is a spectral-decomposition-based method for dimensionality reduction, which approximates the solution of differential equations with parametric uncertainties having a known probability distribution that belongs to a prescribed class; the uncertainty is meant to capture the degree of confidence of the modeller about the actual parameter values. First developed in the realm of uncertainty quantification \citep{pepper2019multiscale, xiu2002wiener}, gPC also found applications in other fields, such as model predictive control for stochastic systems \citep{kim2013generalised}, to replace or accelerate Monte Carlo methods. Analysing uncertain systems by means of gPC is computationally efficient \citep{fisher2011optimal} and the investigation of large parameter spaces can be drastically accelerated.
In particular, \cite{sutulovic2024efficient} consider autonomous ordinary-differential-equation systems of the form
\begin{equation}
    \dot x = f(x,Z),
    \label{eq:general nonlinear uncertain system}
\end{equation}
where $x \in \mathbb{R}^n$ is the state of the system and $Z(\omega) \in \mathbb{R}^q$ is a parametric uncertainty.
Differently from stochastic differential equations (such as stochastic models of chemical reaction networks, \cite{Anderson2015, Anderson2017, Anderson2015b}), where the dynamics are driven by a stochastic process, the source of stochasticity in \eqref{eq:general nonlinear uncertain system} is induced by parametric uncertainties; in fact, $Z \colon \Omega \to I_Z \subset \mathbb{R}^q$ is assumed to be a vector of mutually-independent (without loss of generality) random variables defined on some probability space $(\Omega,\mathcal{F},\mathbb{P})$, where $\Omega$ denotes a sample space containing samples $\omega \in \Omega$, $\mathcal{F}\subseteq 2^\Omega$ is a $\sigma$-algebra and $\mathbb{P}$ is a probability measure. Also, $Z$ is assumed to be absolutely continuous with respect to the Lebesgue measure, and to have a probability density function $\rho_Z(z)$.  
Due to the parametric uncertainty, the state in \eqref{eq:general nonlinear uncertain system} is a stochastic process $x(t;Z)$, as each stochastic realisation of the parameters yields a different system trajectory.
For this class of deterministic ODE systems with the inclusion of a stochastic noise in the form of parametric uncertainty, \cite{sutulovic2024efficient} propose the use of various gPC methods to estimate the first moments (summary statistics) of $x(t;Z)$ without resorting to extensive numerical simulations (from which the summary statistics could be subsequently computed). They consider both intrusive (\eg Galerkin) and non-intrusive (\eg collocation) gPC approaches, which turn out to be scalable and allow for a fast and comprehensive exploration of parameter spaces. They also explore the trade-off between efficiency and accuracy of gPC methods; in particular, the collocation method emerges as an overall better approach for the study of neural dynamics, which conjugates efficiency and accuracy.
Using surrogate models based on gPC, or more in general on uncertainty quantification methods, is thus a promising avenue to scale-up studies of probabilistic robustness, and is expected to bear fruitful insight in future research.

\subsection{Impact on biology and epidemiology}

Natural systems in biology, ecology and epidemiology are being investigated with various methodologies, to unravel and predict their properties, behaviours and underlying functioning mechanisms.
Understanding natural phenomena despite their intrinsic uncertainty is crucial to identify the key elements that ensure their robustness, or lead to fragility and indeterminacy. Disentangling the effect of the numerous uncertain parameters and components in a system is also paramount to guide targeted interventions, by identifying the most appropriate species or interactions to be controlled so as to reliably enforce a desired property.
Coping with uncertainty in the most efficient way is fundamental, in particular for natural systems at all scales, from molecules to cells and whole organisms, up to population dynamics in ecosystems and to the evolution of spreading infectious diseases. In systems biology, extremely complex uncertain phenomena need to be untangled to explain the functioning of pathways and identify the mechanisms that can originate diseases, which is essential to enable targeted therapies. In synthetic biology, aimed at constructing biomolecular circuits with robust and tuneable behaviours, uncertain controllers need to be built for uncertain processes. In the control of epidemics, despite the highly uncertain conditions, fast and crucial decisions need to be made, by conjugating sound insight, grounded on proven theoretical guarantees, with the need for speed and efficiency.

The ISP (integrated structural and probabilistic) methods discussed in this chapter -- which include structural and robust approaches, along with their integration with probabilistic methods -- are tailored to help in this endeavour.
In systems and synthetic biology, ISP approaches can offer a better understanding of biological mechanisms in networks of biomolecules and of cells, by allowing us to investigate the robustness of biological systems and identify essential parameters and key structures that enable them to behave consistently in diverse environmental conditions.
These approaches support both the analysis and the systematic design of synthetic biomolecular feedback systems with a desired behaviour that is guaranteed to manifest due to the circuit's structure, in spite of uncertainty, noise and variability; examples are provided by biomolecular oscillators and switches \citep{Atkinson2003,Samaniego2016,Samaniego2017b,CGF2020,Elowitz2000,Gardner2000,Landau2023}, feedback controllers \citep{Aoki2019,Briat2016,BrittoBisso2025,Chevalier2019,Lugagne2017}, pulsing circuits \citep{Levine2013,Nakamura2024}, and biological neural networks \citep{Pandi2019,Okumura2022}, including classifiers for pattern recognition \citep{Cuba2024}.

ISP approaches to studying physiological models and pathways in organs and organisms, both in their healthy functioning and in their pathological state, can help identify the crucial mechanisms that lead to the onset of diseases, and thereby detect therapeutic targets to suggest new treatments and pharmaceutical interventions (see, \eg \cite{Burlando2020,Mucci2020}).

In neuroscience, initial studies have applied ISP methods to gain a deeper insight into the robustness of neural dynamics and to obtain its probabilistic quantification \citep{sutulovic2024efficient}. Systems of various complexity have been considered, ranging from mechanistic models such as the Jansen-Rit model for multiple connected neurons \citep{jansen1995electroencephalogram} to semi-phenomenological models of single-neuron dynamics such as the Hindmarsh-Rose model \citep{hindmarsh1984model,innocenti2007dynamical,Montanari2022}, 
up to phenomenological models of epileptic activity in brain regions such as the Epileptor model \citep{jirsa2014nature}.
Models from neuroscience provide a challenging case study, because they display multiple regimes (\ie different long-term behaviours of solutions) depending on the chosen parameters; this complexity is effectively tackled by \cite{sutulovic2024efficient}, who assess the persistence of neuronal signalling regimes subject to parametric uncertainties through approaches based on generalised polynomial chaos, which display remarkable scalability and computational efficiency.

On broader scales, untangling the structure and parameter uncertainties in ecological networks \citep{Dambacher2003stability,Dambacher2009,Jeffries1974,Levins1977,May1973,Marzloff2011} can support strategies for sustainable agriculture (see, \eg \cite{deJong2024}) and ecosystem preservation.

The prediction and control of epidemic phenomena can benefit from ISP methodologies, which can be tailored to analyse epidemic models of various complexity, ranging from SIR-like compartmental models, such as SIDARTHE and its extensions \citep{Giordano2020,Giordano2021,CalaCampana2024}, to network-based models describing the spread of
infectious diseases in an interacting population, where nodes are associated with individuals and arcs with their
interactions \citep{Aalto2025,alutto2024dynamic,Nowzari2014,Pastor2015}, possibly including both between-host contagion dynamics and in-host infection dynamics that model interactions between pathogen and immune system \citep{Hernandez-Vargas2019}.
For instance, when describing large-scale pandemics across multiple regions through network-based metapopulation approaches, aggregation methods (briefly discussed in Section~\ref{sec:divideetimpera}) are fundamental.
An ISP characterisation of possible contagion dynamics and patterns can help outline what-if scenarios and identify key factors to predict and to contain the disease spread. ISP methods can suggest how to prevent undesired behaviours, such as recurrent epidemic waves, and promote desired outcomes, such as quasi-steady-states with low infection numbers. The achieved insight can facilitate the holistic design of multi-pronged strategies to contain the contagion via optimal combinations of vaccination plans and non-pharmaceutical interventions, also taking into account resource constraints, thus supporting decision-making with rigorous guarantees.

Overall, integrated structural and probabilistic approaches can be embedded in the analysis of natural systems at all scales, and contribute to their study, design and control. Many of the methods described in this chapter can be readily applied, and future developments that successfully bridge the gaps we have outlined will foster further contributions from the control community to the life sciences.

\newpage
\section{Resilience and robustness}
\label{ch:rob-res-sta-mod}

Structural analysis \citep{blanchini2021structural} and robustness analysis \citep{Barmish1994}, as discussed in Chapter~\ref{ch:struct-an},
aim at guaranteeing that a property is preserved by a whole family of uncertain systems independent of parameter values, or for all parameter values within a specified set, respectively. However, some properties of interest hold neither structurally nor robustly, but with a certain probability, and in such cases structural or robust approaches just provide a negative qualitative outcome and cannot quantify to which degree the property holds.
Also, systems in nature are subject not only to parametric uncertainties, but also to stochastic perturbations, which may yield regime shifts \citep{scheffer2012anticipating}, mostly understood as shifts among basins of attraction \citep{ashwin2012tipping}. To address these phenomena, multiple concepts of resilience for networks and dynamical systems have been introduced \citep{Liu2020b}, but formal definitions are so far lacking in the literature, with the exception of very recent contributions \citep{Proverbio2024,PROVERBIO2024445}. Resilience indicators 
\citep{Dakos2015, kuehn2011mathematical,Krakovska2024} have also been proposed to quantify the ability of a system to withstand perturbations, while maintaining properties of a given attractor. Their design is still relatively at its infancy and different methods have been suggested recently \citep{Liu2020b, proverbio2023systematic}. However, most indicators rely on simplistic surrogate models and lack rigorous and widely applicable mathematical definitions; see Chapter~\ref{ch:res-ind} for more details. Also, appropriate resilience indicators cannot themselves be equated to a definition of resilience; otherwise, different resilience indicators that yied opposite results may lead to lack of consensus as to whether a given system is resilient or not. Resilience needs to be defined rigorously prior to the design of indicators.

To address these open questions, this chapter discusses new rigorous definitions of resilience aimed at quantifying the probability that the behaviour of a nominal deterministic system with respect to a prescribed attractor is preserved under stochastic perturbations.
First, we analyse an example from systems biology to motivate the conceptual difference between robustness and resilience in (deterministic) dynamical models, and we show that the control-theoretic notion of robustness may be too restrictive to capture some of the properties that biological systems exhibit. Then, we complement the concept of robustness with new formal notions of resilience. In particular, for a family of ODE systems consisting of stochastic perturbations of a nominal deterministic system, we characterise the concept of resilience through the properties of an attractor-basin pair $(A,B(A))$ induced by the nominal deterministic system, where $A$ is an attractor of interest and $B(A)$ is the corresponding basin of attraction, and require that the system trajectories remain sufficiently close to the deterministic attractor $A$ even in the presence of stochastic perturbations, which perturb the nominal, deterministic, dynamics.
Also, a suitably defined attraction time is presented as a formal resilience indicator.

\begin{remark}[\textbf{Attractors and regimes}]
Stable invariant sets (\ie attractors) map to the notion of ``basic sets'' discussed by \cite{lemmon2020achieving} and generalise the concept of equilibrium point to the case of non-equilibrium dynamics \citep{ashwin2012tipping}. They are also powerful tools to formalise the concept of ``regime'' (used in ecology and biology), ``property'' (used in control theory) or ``function'' (used in systems biology). These notions often refer to state properties/behaviours that are preserved over certain time frames even in the presence of varying conditions. In fact, a regime can be characterised by the state of a system remaining within a prescribed basin of attraction \citep{lemmon2020achieving}. As we will show in this chapter, this association allows us to establish a connection between qualitative properties (including those considered in other research fields) and quantitative dynamics.
\end{remark}

Our definitions of resilience generalise the notion of probabilistic robustness, in the spirit of \cite{tempo2013randomized} as described in the previous chapter, by considering dynamical disturbances in addition to (static) parametric uncertainty, and enable a quantitative assessment of the persistence of system properties, beyond qualitative answers offered by robustness analysis.
The validity and efficacy of these definitions are showcased in applications to models of biological systems, including a bistable gene regulation model, and a reaction-diffusion system for plankton dynamics that exhibits Turing patterns.

We start by discussing resilience in relation to stability properties, which are of paramount importance in the life sciences as well as in control theory; extensions to other properties will be discussed as an avenue for future work.
We believe that the proposed rigorous definitions can be useful for scientists across a wide array of disciplines, such as network theory, systems biology, ecology, epidemiology, and can be easily modified to capture other desired phenomena aside from stability. We hope that such a set of fundamental definitions can be the starting point to a fruitful research effort that will build new bridges across scientific communities. 

\subsection{The need for resilience}

To illustrate the difference between the concepts of robustness and resilience, and to motivate the need for a notion of resilience that differs from that of robustness, we introduce the potential landscape method, and then proceed with a motivating example from systems biology. 

\subsubsection{The potential landscape}
\label{sec:potential_land}

A useful tool to visualise the difference between robustness and resilience is the potential landscape, adopted in the literature in addition to the state space representation of a dynamical system \citep{Berglund2006, freidlin1998random}. In models of mechanical and electrical systems, the potential can be immediately associated with notions of energy, as they are defined in physics; also, recently, the potential has been successfully associated with biological concepts such as \emph{creodes} and stability landscapes \citep{Ferrell2012, Wang2010}, quasi-potentials \citep{ Nolting2016, Zhou2012} and epigenetic landscapes in developmental biology  \citep{Eugenio2014,Moris2016}. In general, given the autonomous ODE system 
\begin{equation*}
		\begin{aligned}
			& \dot{x}(t) = f(x(t)),\ \ x(t_0) = x_0\, ,
		\end{aligned}
	\end{equation*}
with $x(t)\in \mathbb{R}^n$, we assume that its solutions exist and are extensible for all $t\in \mathbb{R}$. The function $f$ is \emph{integrable} if there exists a scalar function $V(x)$ such that $f(x)=-\nabla V(x)$; then, $V(x)$ is the \emph{potential function} of the ODE system and the system describes a motion of the state $x$ from high to low potential values. The potential is useful to visualise the behaviour of the trajectories, in a Lyapunov-like fashion, by studying $V(x(t))$ along the system solutions, with $V(x(t,x_0))$ being the potential along the system solution emanating from $x(t_0)=x_0$.

\begin{remark}[\textbf{Langevin equations}]
The notion of adjoint stochastic potential is also well-defined for stochastic Langevin equations \eqref{eq:langevin_eq} through the associated Fokker-Planck equation \eqref{eq:fokker_plank}; see Appendix~\ref{sec:ODESDE}. Hence, a similar analysis to the one employed here is also possible for some SDEs. We do not pursue this topic here; for additional details, see \cite{proverbio2022buffering, Sharma2016, Su2019} and the references therein. 
\end{remark}

The study of the behaviour of $V(x(t))$ is not only relevant to show convergence of the system trajectories to equilibrium points, but it can also be extended to general attractors and applications to non-equilibrium dynamics, such as persistent oscillations observed, \eg in circadian rhythms and brain activities \citep{Doyle2006,khona2022attractor,Montanari2022}, seasonal epidemic patterns \citep{Grassly2006}, or transients in multivariate data \citep{heino2022attractor}.

Thinking in terms of the potential landscape is useful in the context of robustness and resilience, in particular for stability analysis; it offers tools to complement robust stability analysis \citep{Iglesias2010} by including noise and quantifying resilience. Unfortunately, an analysis based on the potential landscape is not always possible when complex networks are involved and detailed mechanistic models are not available \citep{meyer2016mathematical, schultz2017potentials}. Nonetheless, under suitable assumptions, certain measures extracted from the potential landscape representation can be used as indicators for resilience changes (see, \eg \cite[Section 5.2]{Krakovska2024}); resilience indicators described in the literature, along with their range of applicability and their performance, are discussed in Chapter~\ref{ch:res-ind}.

We now present a motivating example from systems biology, aimed at pointing out the difference between the notions of robustness and resilience, where the potential landscape is used to analyse how the system attractor changes due to parameter variations and/or to the addition of exogenous noise to the right-hand side of the ODE system.

\subsubsection{An example from systems biology: gene regulation}\label{ex:bio}
A well-known model from systems biology helps us highlight the difference between the two notions of robustness and resilience, by showing various changes in attractors that can be experienced by a family of \emph{deterministic} systems \citep{Proverbio2024,Proverbio2025cdc}.
In the rest of the chapter, we will consider stochastic perturbations (\ie noise) suitably added to deterministic ODEs, and our ensuing definitions of resilience will employ a probabilistic framework (see Section~\ref{Sec:ResDef}).
For now, we consider the nonlinear system 
\begin{equation}
    \dot{x}(t) = f_G(x(t)) = - x(t) + a \frac{x(t)^h}{1+x(t)^h} + k,\ x(0)=x_0\,,
    \label{eq:hill-equation}
\end{equation}
where $x(t)\in \mathbb{R}_{\geq0}$, $t\in [0,\infty)$, represents the concentration of a biochemical species.
System \eqref{eq:hill-equation} is a particular case of gene regulatory network system, belonging to the class of systems presented in Chapter~\ref{ch:intro}, Eq. \eqref{eq:gene-regulation}; see \eqref{eq:gene-regulation2}. Indeed, the one-dimensional system \eqref{eq:hill-equation} can be obtained from \eqref{eq:gene-regulation} by assuming that only activating Hill-type interactions are present in the network and applying, to the full Hill-type dynamics on the whole network, suitable dimension reduction techniques such as weighted dimension reduction (\emph{cf.} Section~\ref{sec:ind-on-net-dimred} and \cite{Gao2016}).

The potential function $V(x)$ for system \eqref{eq:hill-equation}, such that $f_G(x)=-\nabla V(x)$, is shown in Figure~\ref{fig:pot_land}a for different choices of the system parameter $a$.

\begin{figure}[ht!]
	\centering
	\includegraphics[width=\textwidth]{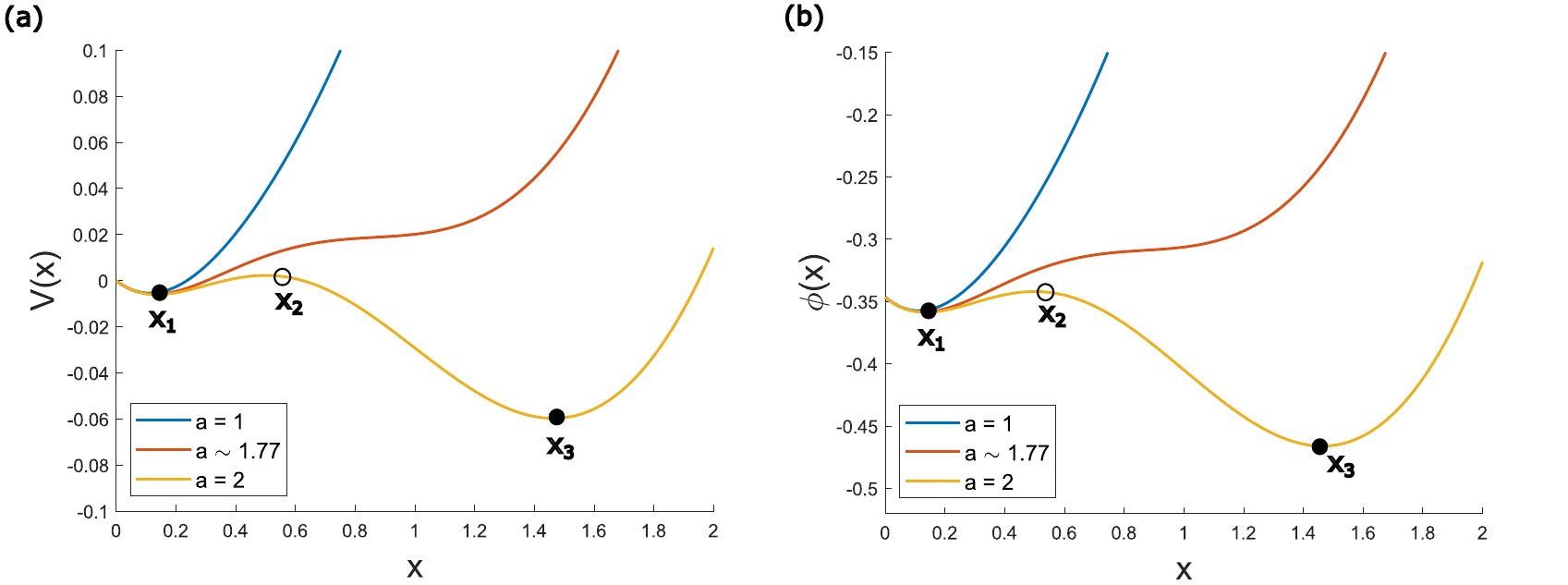}
	\caption{\footnotesize (a) Deterministic potential function for Eq. \eqref{eq:hill-equation}. (b) Stochastic potential function for Eq. \eqref{eq:hill-equation-stoch}. Different values of $a$ are considered, while the other parameters are chosen as $\sigma = 0.5$, $k=0.1$ and $h=2$. The equilibria (stable, full circle; unstable, empty circle) refer to the case $a=2$ and relate to those in Figure~\ref{fig:gene-reg-phase-portrait}. }
	\label{fig:pot_land}
\end{figure}

We now consider different combinations of the parameters $a$, $h$ and $k$ in equation \eqref{eq:hill-equation}.
If $h = 1$, then the Hill function is actually a Michaelis-Menten function, and the system is \emph{robustly} stable according to Definition~\ref{def:robustness}; in fact, regardless of how the values of the parameters $a$ and $k$ are chosen in their biologically meaningful range (provided that, however, the third parameter is constrained as $h=1$), there is a unique globally asymptotically stable equilibrium:
\begin{equation}
    \mathbf{x_1} = \frac{k+a-1 + \sqrt{(k+a-1)^2+4k}}{2} \, , \; \ \forall a \geq 0, \, k \geq 0 \, ,
\end{equation}
such that $f_G(\mathbf{x_1})=0$.
This result is valid for \emph{any} gene regulatory network that exclusively involves Michaelis-Menten interactions (namely, with $h=1$), regardless of the parameter values and of the system structure \citep{BLANCHINI2023110683}, and it is therefore a structural property for such a class of Michaelis-Menten networks.

Conversely, if $h \geq 2$, robust stability no longer holds. In fact, the system may be bistable, and thus have two locally asymptotically stable equilibria. In particular, when $h = 2$ and $0 < k < \frac{1}{3\sqrt{3}}$, the system may admit three equilibria, of which two ($\mathbf{x_{1}}$ and $\mathbf{x_{3}}$) are stable and one ($\mathbf{x_2}$) is unstable \citep{weber2013stochastic}. For different values of $h > 2$, different combinations of the parameters $a$ and $k$ can be found for which the system admits multiple equilibria \citep{Proverbio2025cdc}. 

Analysing the vector field and phase portrait (see Figure~\ref{fig:gene-reg-phase-portrait}, bottom) immediately highlights the system equilibria, whose number and whose values depend on the system parameter $a$ for given values of $h$ and $k$. In particular, given $h = 2$ and a fixed value of $k$ with $0 < k < \frac{1}{3\sqrt{3}}$, there exist $0<a_{c,1} < a_{c,2}$ such that for $ a \in (a_{c,1}, a_{c,2})$ the system is bistable, yielding two contrasting regimes associated with two different stable attractors (equilibria).
Hence, when $a\in (0,a_{c,1})$ or $a \in (a_{c,2}, +\infty)$, the system is robustly stable, for the given $k$ and $h$, while it is \emph{not} robustly stable (it does \emph{not} have a unique globally attractive equilibrium) when $a\in (a_{c,1},a_{c_2})$. This is shown in Figure~\ref{fig:cusp-gen}. Identifying these parameter-bounding sets provides \emph{a-priori} conditions guaranteeing that the system has one versus two attractors.

For a given $h \geq 2$, which combinations of values of $k$ and $a$ yield a single equilibrium and which yield three equilibria? With the help of Figure~\ref{fig:gene-reg-phase-portrait}, we can identify the borderline critical case of two equilibria, which is achieved when the line $f_2(x) = x - k$ is tangent to the sigmoid $f_1(x)= a x^h / (1+x^h)$, where $f_G(x)=f_1(x)-f_2(x)$, by solving
\begin{equation}
    \begin{cases}
        f_G(\bar{x}) = 0 \, , \\
         \frac{df_G(x)}{dx} \Big\rvert_{\bar{x}} = 0 \, .
        \label{eq:to-cusp}
    \end{cases}
\end{equation}
The first equation implies that $\bar{x}$ is an equilibrium point (\ie an intersection of $f_1(x)$ and $f_2(x)$), while the second equation requires $f_2$ to be tangent to $f_1$ at the point, \ie $\frac{df_1(x)}{dx}\Big\rvert_{\bar{x}}=\frac{df_2(x)}{dx}\Big\rvert_{\bar{x}}$. 

In general, system \eqref{eq:to-cusp} provides an explicit relationship between the two free parameters $a$ and $k$. For $h=2$, system \eqref{eq:to-cusp} can be solved explicitly, yielding two curves
\begin{footnotesize}
\begin{multline*}
	a_1(k) = -k + \tfrac{3^{1/3} + -9 k + \sqrt{-3 + 81 k^2})^{2/3}}{2 \left(-3 k + \sqrt{-(1/3) + 9 k^2}\right)^{1/3}} \\
	 -\tfrac{1}{36 k}\left( -3 + \tfrac{3^{4/3}}{(-9 k + \sqrt{-3 + 81 k^2})^{2/3}} + 3^{2/3} \left(-9 k + \sqrt{-3 + 81 k^2} \right)^{2/3}  \right),
\end{multline*}
\begin{multline*}
	a_2(k) = \tfrac{1}{12 k} \left(1-12k^2 + \left(1 + 24 \sqrt{3} \sqrt{k^2 (-1 + 27 k^2)^3} - 108 k^2 (5 + 54 k^2)\right)^{1/3} \right. \\
	\left. + \tfrac{1 + 216 k^2}{(1 + 24 \sqrt{3} \sqrt{k^2 (-1 + 27 k^2)^3} - 
	108 k^2 (5 + 54 k^2))^{1/3}}   \right),
\end{multline*}
\end{footnotesize}
that identify the parameter combinations leading to the critical case (yellow and blue curves, respectively, in Figure~\ref{fig:cusp-gen}); for intermediate values of $a(k)$ such that $a_2(k) \geq a(k) \geq a_1(k)$ (in the shaded region in Figure~\ref{fig:cusp-gen}), bistability occurs; for external values of $a(k)$ (in the white regions in Figure~\ref{fig:cusp-gen}), monostability occurs.
It is worth noting that, when $h=2$ and $k>\frac{1}{3\sqrt{3}}$, the system has a single, globally asymptotically stable, equilibrium for all possible values of $a \geq 0$. Otherwise, multiple regimes (in terms of number of equilibria) are possible, depending on the values of $a$ and $k$.

\begin{figure}[ht!]
	\centering
	\includegraphics[width=0.5\textwidth]{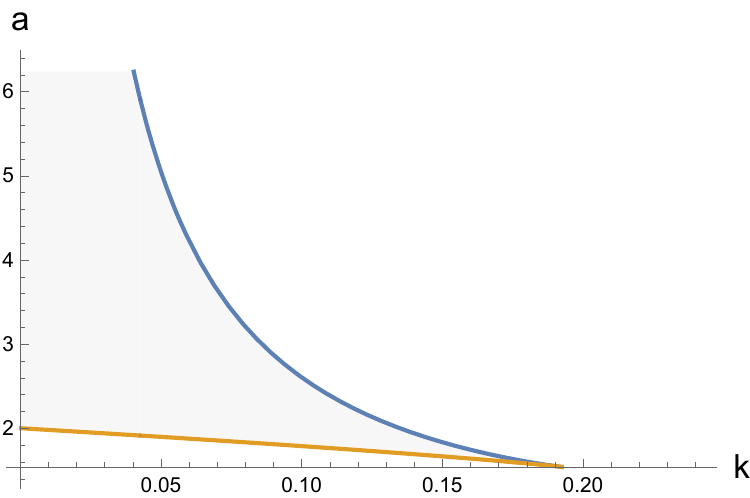}
	\caption{\footnotesize The blue and yellow lines $a = a(k)$ represent parameter combinations yielding bifurcation points for Eq. \eqref{eq:hill-equation} with $h=2$, and identify a region (shaded) where bistability occurs. Outside that region, monostability occurs.}
	\label{fig:cusp-gen}
\end{figure}

So far, this example has demonstrated the connection between parameter values, and parameter-bounding sets, and the corresponding regimes, in terms of number of attractors that the system may have. Considering the question through the prism of robust stability, only two answers are possible: either the system is monostable for all possible parameter values within the given parameter-bounding set, or it is not.
In our particular case, for $a\in (a_{c,1},a_{c,2})$, and $k \in \left(0, \frac{1}{3\sqrt{3}}\right)$, the system does \emph{not} have a unique globally attractive equilibrium when $h = 2$. When $h>2$, similar bounds for $a$ and $k$ can be found numerically; for the same values of $a$, variable $k$ may take larger values while still allowing for bistability if $h>2$, and hence we will consider the interval $0 < k < \frac{1}{3\sqrt{3}}$ throughout the following comparative analysis. Although the qualitative \emph{robustness} analysis, for $a>0$ and $k>0$, always yields a negative answer whenever $h\geq 2$, the behaviour of the system \emph{when subject to disturbances} changes significantly for different values of $h\geq 2$.
To capture this phenomenon, a more nuanced characterisation of changes in system properties is required, which calls for the notion of resilience.

\begin{example}[\textbf{Noise}]
With additive white noise, Eq. \eqref{eq:hill-equation} becomes
\begin{equation}
    \dot{x}(t) = f_G(x(t)) + \eta(t) = - x(t) + a \frac{x(t)^h}{1+x(t)^h} + k + \eta(t) \, ,
    \label{eq:hill-equation-stoch1}
\end{equation}
where $\eta(t)$ is uncorrelated white noise with (implicit) intensity $\sigma$. In the literature, this situation is often analysed resorting to the \emph{stochastic} potential of the system \citep{Gardiner1985}, given by
\begin{equation}
     \phi(x)  = \frac{1}{2}\ln \sigma - \frac{1}{\sigma}\int_0^x f(x')dx' \, ,
\end{equation}
which is shown in Figure~\ref{fig:pot_land}b. The overall shape of $\phi(x)$ is similar to the deterministic $V(x)$, but the numerical values are different and, in particular, $\Delta \phi = |\phi(\mathbf{x_3}) - \phi(\mathbf{x_2})|$ is larger than $\Delta V = |V(\mathbf{x_3}) - V(\mathbf{x_2})|$, which means that more energy is required to steer the system from the equilibrium $\mathbf{x_3}$ to the equilibrium $\mathbf{x_2}$: according to this criterion, noise improves the \emph{resilience} of the equilibrium $\mathbf{x_3}$ (this observation complements that by \cite{weber2013stochastic}, who conduct a bifurcation analysis pointing out that noise preserves the existence of the equilibrium $\mathbf{x_1}$ for a larger interval of parameter values, and enlarges the bistable region). However, it is important to stress that there is an intrinsic difference between $\phi$ and $V$: while $V$ is always decreasing along the deterministic trajectories and is therefore a Lyapunov-like function for the deterministic system, function $\phi$ does not have analogous properties in relation to the stochastic system. For stochastic differential equations, the stochastic potential is related to the stationary probability distribution and determines state-space regions where the state has the highest probability to be found (potential wells), while random noise allows the state to fluctuate and escape such wells. For details, see \cite{Proverbio2025cdc}.
\end{example}

As we mentioned, considering different values of $h\geq 2$ does not affect the stability and robustness properties of the system with $a>0$ and $k>0$, but it does affect the resilience of the system. In fact, varying $h\in (2,\infty)$ alters the potential function associated with the system: precisely, it alters the geometric properties (shape) of the basins of attraction around the equilibria, thereby modifying the \emph{resilience} property of the locally asymptotically stable equilibria, understood as the system's ability to reject disturbances while remaining close to equilibrium states, which is necessary for prompt responses to external stimuli and survival of biological systems \citep{proverbio2022buffering}.
In particular, for $0 < k < \frac{1}{3\sqrt{3}}$ and $a_{c,1} < a < a_{c,2}$, we can plot the potential function obtained from Eq. \eqref{eq:hill-equation} for different values of $h$ (see Figure~\ref{fig:hypergeometric}b). The analytical solution involves the Gauss hypergeometric function $_2 F _1(a,b;c;z)$, expressed in terms of $x$ and $h$ for any fixed $a$ and $k$; see Figure~\ref{fig:hypergeometric}a. The changes in the potential function that can be observed in Figure~\ref{fig:hypergeometric}b for different values of $h$ yield different values of $\Delta V(x)$ and thus different levels of system resilience, associated with a different configuration and depth of the basin of attraction of the equilibria \citep{lundstrom2007dynamic}.

\begin{figure}[ht!]
	\centering
	\includegraphics[width=\textwidth]{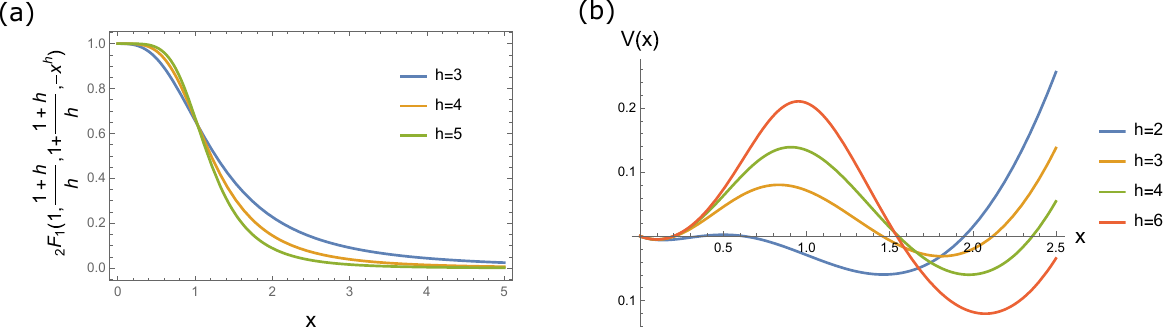}
	\caption{\footnotesize (a) Gauss hypergeometric function $_2 F _1(1, (1+h)/h;1+(1+h)/h;z^{-h})$, part of the integral for Eq. \eqref{eq:hill-equation}. (b) Potential function $V(x)$ for system \eqref{eq:hill-equation}, with $k=0.1\in \left(0, \frac{1}{3\sqrt{3}}\right)$ and $a=2\in (a_{c,1},a_{c,2})$, leading to bistability, for different choices of $h \geq 2$.}
	\label{fig:hypergeometric}
\end{figure}

Overall, our analysis has shown that the system \eqref{eq:hill-equation} is structurally stable when $h=1$ (for all $a \geq 0$ and all $k \geq 0$), while when $h \geq 2$ it is robustly stable only for specific intervals of $a$ and $k$, depending on the precise value of $h$. On the other hand, the resilience of the system is determined by combinations of $a$, $h$ and $k$, which shape the basins of attraction of the equilibria.
A purely qualitative assessment of the robustness of system \eqref{eq:hill-equation} with respect to the existence of a unique globally asymptotically stable equilibrium does not capture all possible attractor configurations for the system, nor their geometric properties, and a quantitative assessment of the persistence of system properties in the face of disturbances is required. To complement robustness, the notion of resilience is thus necessary; yet, to date, the numerous definitions of resilience introduced in the literature, which we surveyed in Section~\ref{sec:informal_defs}, are mostly heuristic and do not rely on rigorous mathematical formulations.

In this chapter, following \cite{Proverbio2024,PROVERBIO2024445}, we propose formal definitions of resilience for a class of autonomous ODE systems subject to stochastic noise (see Section~\ref{Sec:RorR}), and demonstrate their applicability to the analysis of systems in the life sciences.

To this aim, we will need some background from control theory and dynamical systems, and in particular fundamental definitions and results related to ordinary and stochastic differential equations, which are summarised in Appendices~\ref{App:A0} and \ref{sec:ODESDE}.

\subsection{Randomness and noise in biology and epidemiology}

Random phenomena can be modelled by introducing uncertainty into the dynamics in the form of a stochastic process. The definition of stochastic processes \citep{freidlin1998random,Gardiner1985} builds on probability theory; see Appendix~\ref{sec:ODESDE} for a brief summary.
Many modelling approaches embed stochastic processes \citep{Ullah2011} using different mathematical tools that include, \eg Boolean networks \citep{Chaves2005,Saadatpour2011} or discrete processes \citep{J.S.Allen2014}. Since this monograph focuses on continuous dynamics, we do not review these additional approaches, but we restrict ourselves to including noise terms into differential equations so as to introduce and discuss the concept of resilience.

In biological systems and models, noise is ubiquitous \citep{su2019phenotypic,Tsimring2015a}. Systems biology mostly copes with uncertainty in intrinsic biochemical processes (``intrinsic'' noise, \cite{Hasty2000,norman2015stochastic}) or cell-cell variation (``extrinsic'' noise, \cite{Zhang2012}), as well as stochasticity due to small copy numbers of molecular or biological species \citep{Gillespie2000}. In case of timescale separation, uncertain fast variables are often modelled as noise \citep{o2018stochasticity,Tian2004}. Additional case studies are surveyed by \cite{J.S.Allen2014,Ullah2011}. 

In epidemiology, noise also plays a relevant role in modelling uncertain dynamics \citep{Aleta2020, Becker1977}, \eg due to low numbers of infectious individuals \citep{brett2017anticipating}. Noise can model the uncertainty in the epidemic evolution \citep{Diekmann2000, Hellewell2020} driven by human behaviours or intervention strategies \citep{Buonomo2018}, or associated with fast and unknown unmodelled processes \citep{Proverbio2022}.

\subsection{Robustness or resilience?}\label{Sec:RorR}

To distinguish between the notions of resilience and robustness, we need to provide adequate definitions for both concepts. 

Definition~\ref{def:robustness} of robustness considers \emph{families of systems}, whence it is especially useful for capturing structured perturbations (\ie perturbations that preserve the underlying system structure), including perturbations that can be attributed to noise and/or randomness. An example of structural analysis is the investigation of structural stability, mentioned in Chapter~\ref{ch:struct-an} (see also \cite[Section~6]{blanchini2021structural}). In the biological case study of Section~\ref{ex:bio}, the property $\mathcal{P} =$ `the system admits a unique globally asymptotically stable equilibrium' is \emph{robust} for the family of systems $\mathcal{F}=\{ G_{\lambda}\}_{\lambda \in \mathcal{I}}$ if, for each $\lambda \in \mathcal{I}$, $G_{\lambda}$ has a \emph{singleton global attractor} $\left\{\bar x(\lambda) \right\}$.

Here, we propose a possible definition of resilience relying on attractors. To this aim, we first characterise the considered class of systems.

\begin{definition}[\textbf{Family of systems}]\label{def:family}
We consider a family of systems $\mathcal{F} = \left\{G_{\lambda} \right\}_{\lambda \in \mathcal{I}}$ subject to the following assumptions:
\begin{enumerate}
    \item There exists $\lambda_0 \in \mathcal{I}$ such that $G_{\lambda_0}$ is a \emph{deterministic system} corresponding to an autonomous ODE system
\begin{equation}\label{eq:ODESys}
 \dot{x} = f(x),\ \ x(0) = x_0, \quad t\geq 0, 
\end{equation}
where $x\in U \subseteq \mathbb{R}^n$, $U$ is an open set and $f \colon U \to \mathbb{R}^{n}$ is a smooth and globally Lipschitz function. We consider system $G_{\lambda_0}$ as the \emph{nominal} (deterministic) system.
    \item For $\lambda \in \mathcal{I} \setminus \{\lambda_0\}$, the system $G_{\lambda}$ is obtained from the nominal system $G_{\lambda_0}$ by adding the \emph{stochastic} term $g(x)\eta_{\lambda}$ to the right hand side of the ODE system in \eqref{eq:ODESys}, where $\eta_{\lambda}$ is a stochastic stationary noise term and $g \colon U\to \mathbb{R}^n$ is a smooth and globally Lipschitz function.      
Hence, the state of $G_{\lambda}$ at time $t>0$ is a random variable, defined on a suitable probability space with probability $\mathbb{P}_{\lambda}$ that is induced by the stochastic noise $\eta_{\lambda}$.
\end{enumerate}
With a slight abuse of notation, we consider $\eta_{\lambda}$ to be of the form $\eta_{\lambda} = \lambda \frac{\mathrm{d}\mathrm{B}_t}{\mathrm{d}t}$, where $\mathrm{B}_t$ is a one-dimensional Brownian motion and $\lambda\geq 0$ is a noise intensity parameter.
We assume that the solutions to all the systems in the family $\mathcal{F}$ satisfy existence, uniqueness and extensibility to all times, in the appropriate sense (e.g., \cite[Theorem 5.5]{oksendal2013stochastic} guarantees the existence of strong solutions in the Itô sense).
\end{definition}

Let $A$ be a closed attractor of the nominal system $G_{\lambda_0}$ and consider the associated basin of attraction
\begin{equation*}
B(A)=\{\chi \in U \colon \lim_{t \to \infty} \operatorname{dist}(x(t;\chi,0),A) = 0 \},
\end{equation*}
where $x(t; \chi,0)$ denotes the trajectory of the nominal system $G_{\lambda_0}$, with $\eta_{\lambda_0} \equiv 0$ ($\lambda_0=0$), starting from initial condition $\chi$.
Given $\lambda \in \mathcal{I}\setminus \left\{\lambda_0 \right\}$,
a natural question to ask is whether the attractor-basin pair $(A,B(A))$ preserves its properties for some/all the systems $G_{\lambda}$ in $\mathcal{F}$.
Since the trajectories of $G_{\lambda}$ are stochastic processes, the answer to this question will have a probabilistic nature. To ensure the question is well posed (see Chapter 5 in \cite{khasminskii2012stochastic} for a discussion on \emph{stability in probability} of SDEs), we assume that function $g$ satisfies $g\vert_{A}\equiv 0$, namely, $g(x)=0$ for all $x \in A$.

If for all $\lambda \in \mathcal{I}$ it is true that, for any initial condition $x_0\in B(A)$, we have 
\[\lim_{t\to \infty }\operatorname{dist}(x(t;x_0,\eta_{\lambda}),A)=0\] 
almost surely, then $(A,B(A))$ is a robust attractor-basin pair for the family $\mathcal{F}$. 
Through the prism of robustness analysis as in Chapter~\ref{ch:struct-an},
since we are interested in an attractor-basin pair $(A,B(A))$ induced by the nominal deterministic system $G_{\lambda_0}$, the property $\mathcal{P}$ is given by
\begin{equation}\label{eq:pattractor1}
\hspace{-2.4mm}\mathcal{P}: (A,B(A)) \text{ is an attractor-basin pair almost surely,}
\end{equation}
and $\mathcal{P}$ is robust for $\mathcal{F}$ if, for any initial condition $x_0\in B(A)$ and for all $\lambda \in \mathcal{I}\setminus \left\{\lambda_0 \right\}$, we have 
\begin{equation}\label{eq:pattractor2}
\mathbb{P}_{\lambda} \left( \lim_{t\to \infty }\operatorname{dist}(x(t;x_0,\eta_{\lambda}),A)=0 \right)= 1, 
\end{equation}
where $x(t;x_0,\eta_{\lambda})$ is the solution of system $G_{\lambda}$ that emanates from $x_0$ and $\mathbb{P}_{\lambda}$ is the probability measure that is associated with the system $G_{\lambda}$, with dynamics driven by $\eta_{\lambda}$.
When $A=\left\{e \right\}$, where $e$ is an equilibrium point of the nominal system $G_{\lambda_0}$, \eqref{eq:pattractor2} can be rewritten as 
$\mathbb{P}_{\lambda}\left(\lim_{t\to \infty}x(t;x_0,\eta_{\lambda}) = e \right)=1$, and the robustness of property $\mathcal{P}$ in \eqref{eq:pattractor1} for $\mathcal{F}$ is related to \emph{asymptotic stability in probability} of $e$ for all $x_0\in B\left(\left\{e \right\} \right)$ and all $\lambda\in \mathcal{I}\setminus \left\{\lambda_0 \right\}$, as defined in Section 5.4 of \cite{khasminskii2012stochastic}. Therefore, the definitions of resilience proposed in this chapter can be viewed as an extension of the notions of stability in probability.
For the class of systems in Definition~\ref{def:family}, in this chapter we define rigorously the concept of resilience, which (differently from robustness, capturing the system's ability to withstand parameter variations) is aimed at capturing the noise rejection property of a system.

In fact, the viewpoint of robust or structural analysis is insufficient to address cases in which a property of interest holds neither structurally nor robustly, but with (high) probability.
In these cases, structural and robust approaches simply provide a negative qualitative outcome and cannot quantify to which degree (\ie with which probability) the property holds.
Also, as discussed earlier, dynamical systems originating in the life sciences are often extraordinarily robust with respect to huge uncertainties and environmental fluctuations that induce parameter variations, but do not exhibit an analogous ``robustness'' with respect to noise. It is thus necessary to go beyond the concept of robust, or structural, property. In particular, the noise-rejection and function-preserving property of a dynamical system can be captured by the alternative concept of resilience, of which we now aim at proposing rigorous and formal definitions.

Before proceeding, however, we discuss two examples from the biological literature that help us clarify the concept of attractor and basin of attraction.

\begin{example}[\textbf{Gene regulation - continued}]
For the biological example \eqref{eq:hill-equation-stoch1}, the family of systems $\mathcal{F} = \left\{ G_{\lambda}\right\}$, where $\lambda \in \mathcal{I} = [0,\lambda_*),\ \lambda_*>0$, is such that $G_{\lambda}$, $\lambda \in \mathcal{I}$, is given by the stochastic ODE
\begin{equation}
    \dot{x}  = f_G(x)+ \eta_\lambda(t)g(x)= - x + a \frac{x^h}{1+x^h} + k + \lambda \tfrac{dB_t}{dt} g(x) \, ,
    \label{eq:hill-equation-stoch}
\end{equation}
where $h \geq 2$, $k \in (0,\tfrac{1}{3 \sqrt{3}})$ and $a\in (a_{c,1},a_{c,2})$, while $ \lambda \tfrac{dB_t}{dt}$ is an uncorrelated white noise with intensity $\lambda \in \mathcal{I}$. We set $\lambda_0 = 0$, while the attractor $A$ is chosen as a singleton containing one of the two locally stable equilibrium points of $G_{0}$, with $B(A)$ being the corresponding basin of attraction. Moreover, $g(x)$ is 
a smooth transition function satisfying $g(x)=0$ for $x \in A + \varphi B(0,1)$ and $g(x)= 1$ for $x\notin A+2\varphi B(0,1)$, where $\varphi>0$ is small and $B(0,1)$ is the unit ball around zero.
\end{example}

\begin{example}[\textbf{Turing patterns}]
\cite{bashkirtseva2021stochastic} consider pattern formation and stochastic sensitivity of the pattern-attractors of the following deterministic partial-differential-equation (PDE) Levin-Segel population model with diffusion \citep{levin1976hypothesis}:
\begin{equation}\label{eq:PopMNodel}
\begin{cases}
u_t = au+eu^2-buv +D_u u_{xx}\,,\\
v_t = cuv - dv^2 +D_vv_{xx}\,,
\end{cases}
\end{equation}
having zero-flux boundary conditions
\begin{equation*}
\begin{array}{lll}
&u_x(0,0) = u_x(0,L) = v_x(0,0)=v_x(0,L) = 0.
\end{array}
\end{equation*} 
System \eqref{eq:PopMNodel} is composed of two coupled one-dimensional nonlinear diffusion PDEs, where the state variables $u(t,x)$ and $v(t,x)$ represent, respectively, the density of phytoplankton (activator) and herbivore (inhibitor). The space variable is $x\in [0,L]$, $D_u$ and $D_v$ are the diffusion coefficients, the parameters $a$, $b$, $c$, $d$ and $e$ are positive.
Given initial conditions for the deterministic model \eqref{eq:PopMNodel}, its solutions exhibit pattern formation (\ie convergence to periodic wave-like behaviours of the graphs of $u(t,\cdot)$ and $v(t,\cdot)$ as $t\to \infty$), provided that the following inequalities hold:
\begin{equation}\label{eq:PattParam}
\frac{D_u}{D_v}< \left( \sqrt{\frac{b}{d}} - \sqrt{\frac{b}{d} - \frac{e}{c}}\right)^2,\quad bc>ed\,.
\end{equation}

To fit \eqref{eq:PopMNodel} into our framework, we perform semi-discretization of \eqref{eq:PopMNodel} on a finite set of uniformly spaced grid points $0=x_1<x_2<\dots<x_N=L,\ x_{j+1}-x_j\equiv h>0$, thereby replacing \eqref{eq:PopMNodel} with the following ODE system of dimension $2N$:
\begin{equation}\label{eq:TurPatSemiDisc}
\dot{w}(t) = Mw(t)+H(w(t))=: f(w(t)),\quad t\geq 0,
\end{equation}
with
$w(t) = \begin{bsmallmatrix}
w_1(t)\\ \vdots\\ w_N(t)
\end{bsmallmatrix}$, $w_j(t) \approx \begin{bsmallmatrix}
    u(t,x_j)\\ v(t,x_j)
\end{bsmallmatrix}$,
$H(w) = \begin{bsmallmatrix}
    h(w_1) \\ \vdots \\ h(w_N)
\end{bsmallmatrix}$,\\ $h(\begin{bsmallmatrix}\xi\\ \rho \end{bsmallmatrix}) = \begin{bsmallmatrix}
    e\xi^2 - b \rho \xi \\
    c \rho \xi -d\rho^2
\end{bsmallmatrix}$, $N_ 0 = \begin{bsmallmatrix}
    a -\frac{2D_u}{h^2} & 0 \\
    0 & -\frac{2D_v}{h^2}
\end{bsmallmatrix}$, $N_1 = \frac{2}{h^2}\begin{bsmallmatrix}
    D_u & 0 \\ 
    0 & D_v
\end{bsmallmatrix}$,\\ $M  = \begin{bsmallmatrix}
    N_0 & N_1 & 0 & 0 &\dots & 0 & 0 & 0\\
    \frac{1}{2}N_1 & N_0 & \frac{1}{2}N_1 & 0 & \dots & 0 & 0 & 0\\
    \vdots & \vdots & \vdots & \vdots & \ddots & \vdots & \vdots & \vdots \\
    0 & 0 & 0 & 0 & \dots & \frac{1}{2}N_1 & N_0 & \frac{1}{2}N_1\\
    0 & 0 & 0 & 0 & \dots & 0 & N_1 & N_0
\end{bsmallmatrix}$.

We assume that the parameters satisfy \eqref{eq:PattParam}.
The Turing pattern of \eqref{eq:PopMNodel} is approximated by computing an equilibrium point $\mathbf{w} \in (0,\infty)^{2N}$ of \eqref{eq:TurPatSemiDisc} with small step-size $h$ and separately interpolating the values that correspond to each of the functions $u$ and $v$. Let the system \eqref{eq:TurPatSemiDisc} be the nominal deterministic system with the attractor $A = \left\{ \mathbf{w} \right\}$, corresponding to the Turing pattern generated by the nominal deterministic system; we also consider the associated basin of attraction $B(A)$. Then, we consider the family $\mathcal{F} = \left\{ G_{\lambda} \right\}_{\lambda\in [0,\lambda_*)}$ with $\lambda_*>0$, such that
\begin{equation}\label{eq:turingadditive}
\dot{w}(t) = f(w(t)) + \lambda \tfrac{\mathrm{d}B_t^{2N}}{\mathrm{d}t}g(w(t)),
\end{equation}
where $g(w)$ is a smooth transition function satisfying $g(w)=0$ for $w \in A + \varphi B(0,1)$ and $g(w)= 1$ for $w\notin A+2\varphi B(0,1)$, where $\varphi>0$ is small and $B(0,1)$ is the unit ball around zero.
Thus, $\lambda \frac{\mathrm{d}B_t^{2N}}{\mathrm{d}t}g(w(t))$ is a vector of stationary noises that drives the dynamics of $G_{\lambda}$, provided $w(t)$ is outside of a small ball around the attractor; $\frac{\mathrm{d}B_t^{2N}}{\mathrm{d}t}$ denotes a $2N$-dimensional vector of uncorrelated white Gaussian noise terms; $\lambda$ is the noise intensity, and the choice $\lambda = \lambda_0 = 0$ recovers the nominal deterministic system \eqref{eq:TurPatSemiDisc}.
For $\lambda>0$, the stochastic solutions of the system venture away from the deterministic pattern-attractor and fall around it with an appropriate probability distribution.
\end{example}

\subsection{Resilience definitions}\label{Sec:ResDef}

Given a family of systems $\mathcal{F} = \left\{ G_{\lambda}\right\}_{\lambda \in \mathcal{I}}$, as specified in Definition~\ref{def:family}, we consider $\lambda_0\in \mathcal{I} $ and recall that $G_{\lambda_0}$ is a deterministic nominal system, whose dynamics is described by an ODE, with attractor $A$ and a corresponding basin of attraction $B(A) \supseteq A$. Let $\lambda \in \mathcal{I}\setminus \left\{\lambda_0 \right\}$, along with the corresponding stochastic term $\eta_{\lambda}$ and probability measure $\mathbb{P}_{\lambda}$.

For such a family of systems, we present novel definitions of resilience \citep{Proverbio2024,PROVERBIO2024445}.

\begin{definition}[\textbf{Practical resilience}]\label{def:PrRes}
Consider a family of systems $\mathcal{F} = \left\{G_{\lambda} \right\}_{\lambda \in \mathcal{I}}$ and let $(A,B(A))$ be an attractor-basin pair corresponding to $G_{\lambda_0}$.
Let the time horizon $\tau\in (0,\infty]$, the distance $\delta \in [0,\infty)$ and the confidence level $\gamma \in (0,1]$ be fixed. Consider the set $A_\varepsilon = \{ \chi \colon \operatorname{dist}\left(\chi, A \right) \leq \varepsilon \}$, with $0 \leq \varepsilon \leq \delta$. The system $G_{\lambda}$ is $(\tau,\gamma,\delta, \varepsilon)$-\emph{practically resilient} if, for all $x_0\in A_\varepsilon \cap B(A)$,
\begin{equation*}
\mathbb{P}_{\lambda}\left( \sup_{ t\in [0,\tau)} \operatorname{dist}\left({x}(t;x_0,\eta_{\lambda}), A \right) \leq \delta \right)\geq \gamma.
\end{equation*}
The family $\mathcal{F}$ is $(\tau,\gamma,\delta,\varepsilon)$-\emph{practically resilient} if, for all $x_0\in A_\varepsilon \cap B(A)$,
\begin{equation*}
\inf_{\lambda\in \mathcal{I}\setminus \left\{\lambda_0 \right\}}\mathbb{P}_{\lambda}\left( \sup_{ t\in [0,\tau)} \operatorname{dist}\left({x}(t;x_0,\eta_{\lambda}), A \right) \leq \delta \right)\geq \gamma.
\end{equation*}
\end{definition}

Intuitively, the system $G_{\lambda}$, $\lambda\in \mathcal{I}\setminus \left\{\lambda_0 \right\}$, is $(\tau,\gamma,\delta, \varepsilon)$-practically resilient if, subject to the stochastic noise $\eta_{\lambda}$, the trajectory $x(t;x_0,\eta_{\lambda})$, emanating from an arbitrary point $x_0$ that lies within an $\varepsilon$-distance from the attractor $A$, and within its basin of attraction $B(A)$, remains within a $\delta$-distance from $A$ with probability at least $\gamma$ over the interval $t\in [0,\tau)$; see the visual representation in Figure~\ref{fig:pract_res}.
The noise $\eta_{\lambda}$ may prevent the state $x(t;x_0,\eta_{\lambda})$ of the \emph{perturbed} system from converging to the set $A$ (attractor of the \emph{unperturbed} nominal system), but at least the nominal dynamics of the system keeps the state $\delta$-close to the attractor $A$ of the nominal system, with probability at least $\gamma$. In our considered framework, almost sure ISS (Input-to-State Stability) would correspond to practical resilience with $\gamma=1$. Also, practical resilience corresponds to \emph{stability in probability} as defined in Section 5.3 of \cite{khasminskii2012stochastic}, provided that $\gamma < 1$, with $A=\{0\}$ and an appropriately chosen $\delta$.

\begin{figure}[ht]
    \centering
    \includegraphics[width=0.4\textwidth]{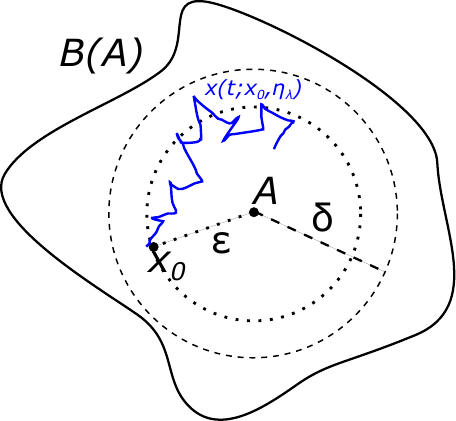}
    \caption{\footnotesize Intuitive schematic representation for the  concept of $(\tau,\gamma,\delta, \varepsilon)$-\emph{practical resilience} in Definition~\ref{def:PrRes}. In the considered realisation, subject to the stochastic noise $\eta_{\lambda}$, the trajectory $x(t;x_0,\eta_{\lambda})$ in blue, emanating from a point $x_0$ that lies within an $\varepsilon$-distance from the attractor $A$, and within its basin of attraction $B(A)$, remains within a $\delta$-distance from $A$ over the interval $t\in [0,\tau)$. The system $G_\lambda$ is $(\tau,\gamma,\delta, \varepsilon)$-\emph{practically resilient} if this happens with probability at least $\gamma$; the family $\mathcal{F}$ is $(\tau,\gamma,\delta, \varepsilon)$-\emph{practically resilient} if this happens with probability at least $\gamma$ for \emph{all} systems in the family.}
    \label{fig:pract_res}
\end{figure}

In particular, the family $\mathcal{F}$ is $(\infty,1,\delta,\delta)$-practically resilient if, for all $x_0\in A_\delta \cap B(A)$,
\begin{equation}
\inf_{\lambda\in \mathcal{I}\setminus \left\{\lambda_0 \right\}}\mathbb{P}_{\lambda}\left( \sup_{ t>0} \text{dist}\left({x}(t;x_0,\eta_{\lambda}), A \right) \leq \delta \right)= 1.
\end{equation}  

If we assume that $\eta_{\lambda}$ is deterministic (\eg an exogenous disturbance), we can omit $\mathbb{P}_{\lambda}$ and reinterpret $(\infty,1,\delta, \delta)$-practical resilience as equivalent to $x(t;x_0,\eta_{\lambda})\in A_{\delta}$ for all $t>0$ and $\lambda \in \mathcal{I}\setminus \left\{\lambda_0 \right\}$. In comparison to the definition of $A$ being a \emph{robust} attractor for the trajectories $x(t;x_0,\eta_{\lambda})$ of the family $\mathcal{F}$, as per Definition~\ref{def:robustness}, $(\infty,1,\delta, \delta)$-practical resilience of the family means that for any \emph{deterministic} perturbation in the class $\{ \eta_\lambda\}_{\lambda \in \mathcal{I}}$, the state   $x(t;x_0,\eta_{\lambda})$ remains in $A_\delta$ for all $t>0$, \ie no deterministic perturbation in $\{ \eta_\lambda\}_{\lambda \in \mathcal{I}}$ can drive the state away from a $\delta$-neighbourhood of $A$. 

We can associate the notion of resilience as intended in systems ecology (\emph{cf.} Table~\ref{tab:summary-concepts}) with $(\infty,1,\delta_*, \delta_*)$-practical resilience, where $\delta_* = \text{dist}(A, \partial B(A))$ represents the distance between $A$ and the boundary of $B(A)$. In this case, resilience offers sufficient probabilistic guarantees for the trajectories to remain within the basin of attraction.

To allow the system state $x(t;x_0,\eta_{\lambda})$ to exhibit damped oscillations and arbitrary large detours away from the attractor $A$ (this is what happens, \eg in the case of excitable systems, such as those discussed by \cite{wieczorek2011excitability,izhikevich2000neural,Izhikevich2007}, or spiking systems, such as those discussed by \cite{BG2024}), provided that the system state returns to a neighbourhood of the attractor $A$ with high enough probability, we introduce the additional definition of \emph{asymptotic practical resilience}.
\begin{definition}[\textbf{Asymptotic practical resilience}]\label{def:AsympPrRes}
Consider a family of systems $\mathcal{F} = \left\{G_{\lambda} \right\}_{\lambda \in \mathcal{I}}$ and let $(A,B(A))$ be an attractor-basin pair corresponding to $G_{\lambda_0}$.
Let the distance $\delta \in [0,\infty)$ and the confidence level $\gamma \in (0,1]$ be fixed. The system $G_{\lambda}$ is $(\gamma,\delta)$-\emph{asymptotically practically resilient} if, for all $x_0\in B(A)$,
\begin{equation*}
\mathbb{P}_{\lambda}\left( \limsup_{ t\to \infty} \operatorname{dist}\left({x}(t;x_0,\eta_{\lambda}), A \right) \leq \delta \right)\geq \gamma.
\label{eq:practical-res}
\end{equation*}
The family $\mathcal{F}$ is $(\gamma,\delta)$-\emph{asymptotically practically resilient} if, for all $x_0\in B(A)$,
\begin{equation*}
\inf_{\lambda\in \mathcal{I}\setminus \left\{\lambda_0 \right\}}\mathbb{P}_{\lambda}\left( \limsup_{ t\to \infty} \operatorname{dist}\left({x}(t;x_0,\eta_{\lambda}), A \right) \leq \delta \right)\geq \gamma.
\end{equation*}   
When $\delta=0$, the system (respectively, the family) is $\gamma$-\emph{asymptotically resilient}.
\end{definition}

Also, when $\delta=0$ and $\gamma=1$, Definition~\ref{def:AsympPrRes} of asymptotic resilience reduces to robustness of the property $\mathcal{P}$ in \eqref{eq:pattractor1} according to Eq. \eqref{eq:pattractor2}.

Differently from Definition~\ref{def:PrRes}, which deals with the transient behaviour of the system trajectories, Definition~\ref{def:AsympPrRes} deals with their asymptotic properties and requires the trajectories emanating from $x_0 \in B(A)$ to asymptotically converge to a $\delta$-neighbourhood (with $\delta>0$) of the attractor $A$ with probability at least $\gamma>0$. 

Definition~\ref{def:AsympPrRes} has a strong analogy with the condition in Eq. (5.15) in Section 5.4 of \cite{khasminskii2012stochastic}, in relation to \emph{asymptotic stability in probability}: such a condition would guarantee asymptotic resilience, provided that $\gamma < 1$, with $A=\{0\}$ and $\delta=0$, but \emph{for all initial conditions in a sufficiently small neighbourhood of $A$}, and not anywhere in $B(A)$. Moreover, the asymptotic condition in Definition 5.1 of \cite{khasminskii2012stochastic}, concerning \emph{asymptotic stability in the large}, would imply asymptotic resilience, provided that $\gamma < 1$, with $A=\{0\}$ and $\delta=0$.

\begin{remark}[\textbf{Alternative definitions of resilience}]
The proposed definitions of resilience heavily rely on the structure of the family $\mathcal{F}$. In particular, the attractor-basin pair $(A,B(A))$ is determined (and fixed) by reference to the nominal deterministic system $G_{\lambda_0}$. The resilience of $G_{\lambda}$ for $\lambda\neq \lambda_0$ is then determined by relating the behaviour of the trajectories $x(t;x_0,\eta_{\lambda})$ to a neighbourhood of $A$. Consider, for example, Definition~\ref{def:AsympPrRes} for $\lambda\neq \lambda_0$. Even when $\eta_{\lambda}$ is deterministic, its addition to the right-hand side of the ODEs may alter the attractor $A$ to a new attractor $A(\lambda)$, which will be close to $A$, subject to appropriate assumptions on $\eta_{\lambda}$. Hence, allowing for $\delta>0$ in Definition~\ref{def:AsympPrRes} is essential to obtain a meaningful definition.
An alternative approach to defining resilience would be to require the \emph{property of existence of an attractor-basin pair} (which may be different from the pair $(A,B(A))$, induced by the nominal system $G_{\lambda_0}$) to be preserved, with high enough probability.
\end{remark} 

Other concepts surveyed in Table~\ref{tab:summary-concepts} can be interpreted through the lenses of Definitions~\ref{def:PrRes} and \ref{def:AsympPrRes}. For example, the recovery or approximate recovery after (random) shocks, with high probability, can be interpreted via Definition~\ref{def:AsympPrRes}. To some extent, the entire idea of a dynamical system preserving its qualitative behaviour despite (random) perturbations is associated with the concept of residency within an attractor: network resilience as introduced by \cite{Gao2016} is a special case of system resilience assessed through bifurcation analysis (see Chapter~\ref{ch:res-ind}).

Building upon the provided resilience definitions, we consider, as a specific resilience indicator, the \emph{attraction time} of a system, which probabilistically quantifies the time it takes for a perturbed or noisy system to reach a neighbourhood of a prescribed equilibrium state. Here, we show how such an indicator can be formally defined, for the family of systems in Definition~\ref{def:family}, relying on asymptotic practical resilience in Definition~\ref{def:AsympPrRes}.

\begin{definition}[\textbf{Attraction time}]\label{Def:RetTime}
Consider the family $\mathcal{F}= \left\{ G_{\lambda}\right\}_{\lambda \in \mathcal{I}}$ and let $(A,B(A))$ be an attractor-basin pair of the nominal system $G_{\lambda_0}$. Let $\lambda \neq \lambda_0$ and assume that $G_{\lambda}$ is $(\gamma,\delta)$-asymptotically practically resilient for some $\delta\in [0,\infty), \gamma\in (0,1]$, as per Definition~\ref{def:AsympPrRes}. Let $x_0\in B(A)$, $\nu \in [0,\infty)$, $\tau \in (0,\infty)$ and $\mu \in (0,1]$. Then, the system $G_{\lambda}$ has a \emph{$(\tau,\mu,\nu)$-attraction time} with respect to $x_0$ if 
\begin{equation*}
\mathbb{P}_{\lambda}\left( \sup_{t\in [\tau,\infty)} \operatorname{dist} \left(x(t;x_0,\eta_{\lambda}),A \right) \leq \nu\right)\geq \mu. \label{eq:attr_time}
\end{equation*}
Similarly, given $x_0\in B(A)$, the family $\mathcal{F}$ has a \emph{$(\tau,\mu,\nu)$-attraction time} with respect to $x_0$ if 
\begin{equation*}
\inf_{\lambda\in \mathcal{I}\setminus \left\{\lambda_0 \right\}} \mathbb{P}_{\lambda}\left( \sup_{t\in [\tau,\infty)} \operatorname{dist} \left(x(t;x_0,\eta_{\lambda}),A \right) \leq \nu\right)\geq \mu. 
\end{equation*}
\end{definition}

Intuitively, $G_{\lambda}$ has a \emph{$(\tau,\mu,\nu)$-attraction time} if the probability of the state to return to a $\nu$-neighbourhood of the attractor $A$ after time $\tau$ is at least $\mu$. 

Note that, given $\nu$, smaller $\tau$ and/or larger $\mu$ for which $G_{\lambda}$ has $(\tau,\mu,\nu)$-attraction time imply higher resilience of the attractor $A$ of $G_{\lambda_0}$ when the dynamics is perturbed by $\eta_{\lambda}$. 
In fact, consider the case of $\mu$ fixed. Then reducing $\tau$ in Definition~\ref{Def:RetTime} to $\tau'<\tau$ means that the state of $G_{\lambda}$ returns to the $\nu$ neighborhood of the attractor $A$ faster (with probability at least $\mu$), which indicates higher resilience of the \emph{nominal} system $G_{\lambda_0}$ to the addition of the stochastic perturbation $\eta_{\lambda}$ to its right-hand side. 

Importantly, the \emph{$(\tau,\mu,\epsilon)$-attraction time} is well-defined, in the sense that a $(\gamma,\delta)$-asymptotically practically resilient system also has $(\tau,\mu,\nu)$-attraction time for some $\tau$, $\mu$, and $\nu$, as the following proposition shows.
\begin{proposition}
    Let $G_{\lambda}$ be $(\gamma,\delta)$-asymptotically practically resilient  for $x_0\in B(A)$ and some $\delta\in [0,\infty), \gamma\in (0,1]$. Then, there exist some $\nu,\tau \in (0,\infty)$ and $\mu \in (0,1)$ such that $G_{\lambda}$ has a \emph{$(\tau,\mu,\nu)$-attraction time}. 
\end{proposition}
\begin{proof}
Choose $\nu > \delta$ and define the sets
\begin{equation*}
E_0  = \left\{\limsup_{ t\to \infty} \operatorname{dist}\left({x}(t;x_0,\eta_{\lambda}), A \right) \leq \delta \right\}
\end{equation*}
and
\begin{equation*}
E_{n} = \left\{\sup_{t\geq n} \operatorname{dist}\left(x(t;x_0,\eta_{\lambda}),A \right)\leq \nu \right\}.
\end{equation*}
By definition of the limit superior, $E_0 \subset \bigcup_{n\in \mathbb{N}}E_n$. Also, $E_n \subseteq E_{n+1}$ for all $n\in \mathbb{N}$. Hence, by monotone convergence
\begin{equation*}
\gamma \leq \mathbb{P}_{\lambda}(E_0) = \lim_{n\to \infty}\mathbb{P}_{\lambda}\left(E_0\cap E_n \right)\leq \lim_{n\to \infty}\mathbb{P}_{\lambda}(E_n).    
\end{equation*}
Therefore, for any $0<\zeta<1$, there exists a sufficiently large $n$ such that $G_{\lambda}$ has a $\left(n,\zeta \gamma,\nu \right)$-attraction time.
\end{proof}

We will further discuss resilience indicators proposed in the literature in Chapter~\ref{ch:res-ind}, where we will illustrate other types of indicators and outline how they can be refined or adapted in light of the novel definitions presented above.

\subsection{Examples of resilience analysis}

We now showcase how the definitions of (asymptotic) practical resilience and attraction time can be used to gain quantitative insight into the dynamics of case studies in systems biology, thus complementing robustness analysis.
The resilience analysis of a dynamical system allows us to quantitatively study the noise-rejection properties of the system and better understand how biological regulation mechanisms maintain their function (\eg protein production) despite external perturbations (in addition to parametric uncertainties).

\subsubsection{Gene regulation}\label{subsec:genereg}
Consider the gene regulation model in Eq. \eqref{eq:hill-equation-stoch}, where $A = \{\mathbf{x}_3\}$ and the smooth transition function is
\begin{equation}
g(x) := g(x; \mathbf{x}_3)= \ell \left(\frac{x - \mathbf{x}_3 - \varphi}{\varphi}\right)+\ell \left(\frac{-x + \mathbf{x}_3 - \varphi}{\varphi}\right),
\end{equation}
where
\begin{equation}
\ell(x)=\frac{\vartheta(x)}{\vartheta(x)+\vartheta(1-x)}
\end{equation}
and
\begin{equation}
\vartheta(x)= \begin{cases}e^{-\tfrac{1}{x}} \mbox{ when } x>0\\ 0 \mbox{ when } x \leq 0 \end{cases}
\end{equation}
with $x \in \mathbb{R}$, and $\varphi = 0.0001$, as shown in Figure~\ref{fig:function}.

\begin{figure}[ht]
    \centering
    \includegraphics[width=0.4\textwidth]{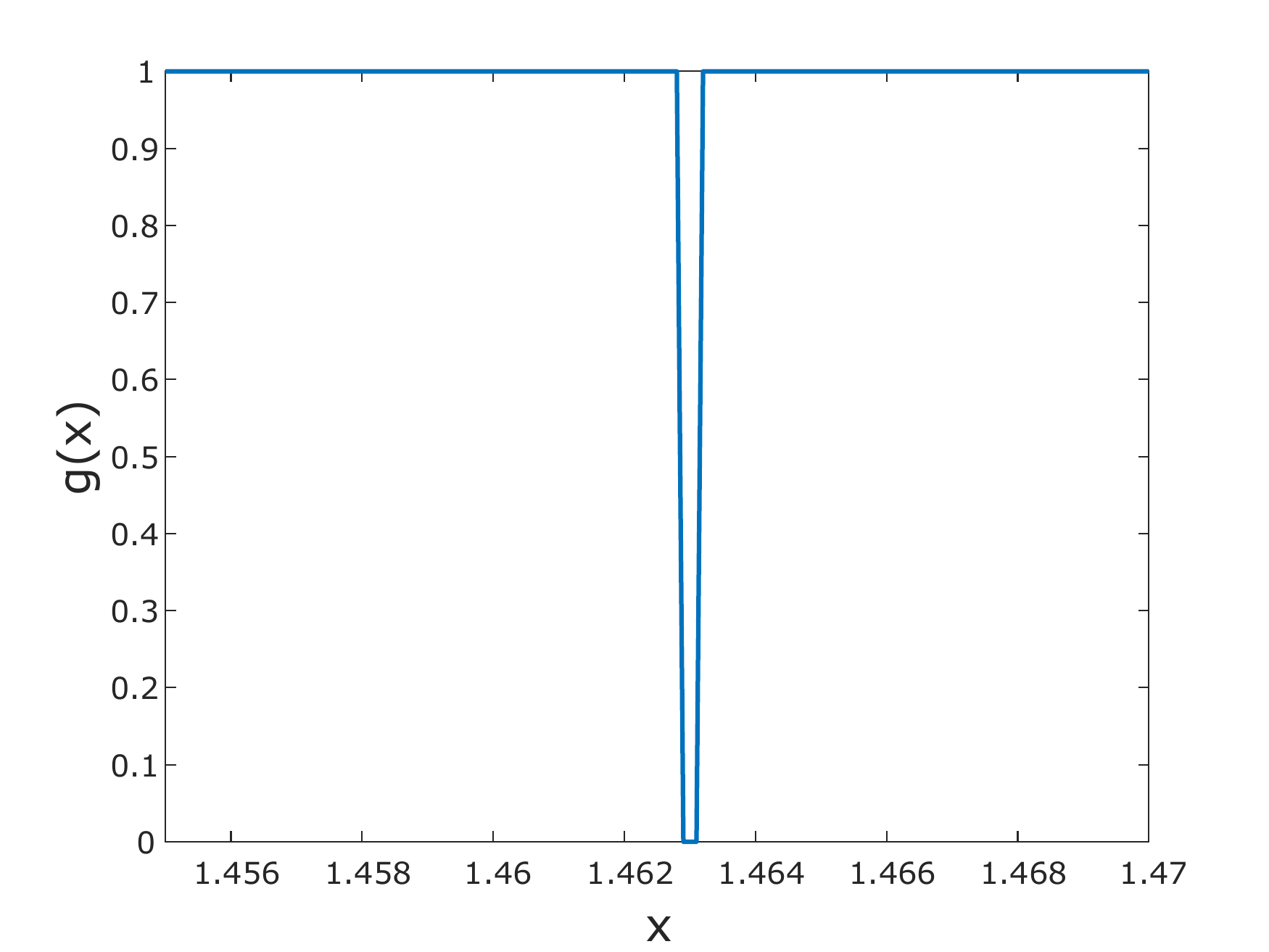}
    \caption{\footnotesize The function $g(x):=g(x; \mathbf{x}_3)$ that weighs the noise, chosen as described in Section~\ref{subsec:genereg}, with $\varphi = 0.0001$; $\mathbf{x}_3 \approx 1.463$.}
    \label{fig:function}
\end{figure}

To assess $(\tau,\gamma,\delta, \delta)$-practical resilience, we perform numerical simulations of the system \eqref{eq:hill-equation-stoch} around the stable equilibrium $\mathbf{x}_3$, for uniformly spaced initial conditions $x_0 \in (\mathbf{x}_3 - \delta; \mathbf{x}_3+\delta)$, and assess whether the trajectories lie within the interval $(\mathbf{x}_3 -  \delta; \mathbf{x}_3 + \delta)$ over a finite horizon $t_{fin}$, for a range of noise intensities $\lambda$.
Let us consider the model parameters $k=0.1$, $h=2$ and $a=2$ (within the bistable region in Figure~\ref{fig:gene-reg-phase-portrait}), for which the lower stable equilibrium is $\mathbf{x}_1 \approx 0.137$, the intermediate unstable equilibrium is $\mathbf{x}_2 = 0.5$ and the upper stable equilibrium is $\mathbf{x}_3 \approx 1.463$.
Figure~\ref{fig:attr_time} shows sample trajectories of the nominal system $G_{\lambda_0}$ with $\lambda_0=0$, corresponding to the deterministic system in \eqref{eq:hill-equation}, converging to the stable equilibrium $\mathbf{x}_3$, for different initial conditions $x_0 \in (\mathbf{x}_3 - \delta; \mathbf{x}_3+\delta)$, with $\delta = 0.8$. 

\begin{figure}[ht]
    \centering
    \includegraphics[width=0.5\textwidth]{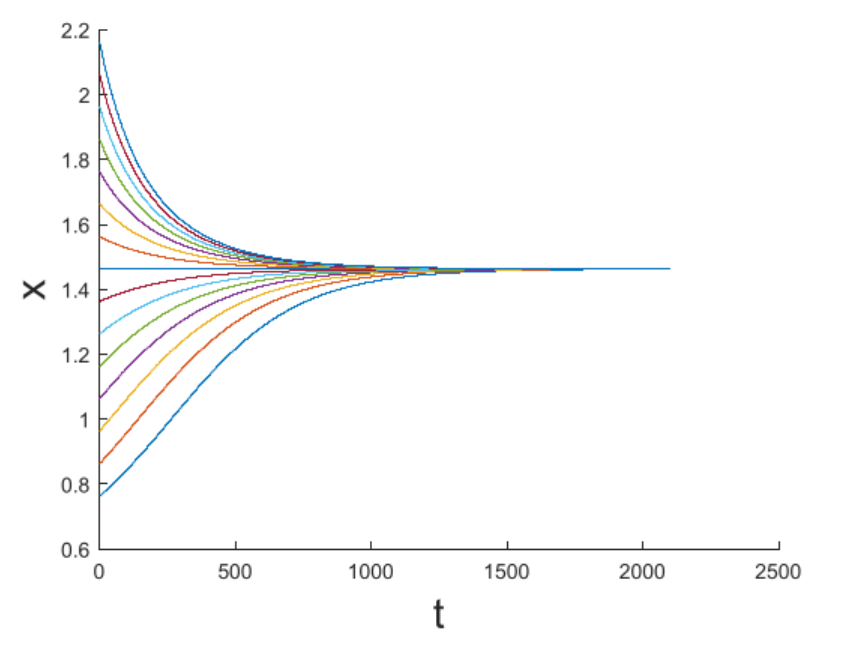}
    \caption{\footnotesize Convergence of several trajectories of the nominal deterministic system in Eq. \eqref{eq:hill-equation} to the stable equilibrium $\mathbf{x}_3$, for a range of initial conditions $x_0$ in a $\delta$-neighbourhood of the attractor, $x_0 \in (\mathbf{x}_3 - \delta; \mathbf{x}_3+\delta)$, $\delta = 0.8$. }
    \label{fig:attr_time}
\end{figure}

For the stochastic gene regulation model \eqref{eq:hill-equation-stoch}, we can numerically assess the dependency of $\mathbb{P}_\lambda$ on both $\delta$ and $\lambda$.

We consider different values of the distance $\delta \in [\delta_{min},\delta_{max}]$, with $\delta_{min} = 0.05 \mathbf{x}_3$ and $\delta_{max}=|\mathbf{x}_3 - \mathbf{x}_2|$, which is the distance from the unstable equilibrium $\mathbf{x}_2$. 
In fact, all trajectories leave the basin of attraction of the stable equilibrium $\mathbf{x}_3$ after hitting the unstable equilibrium $\mathbf{x}_2$, regardless of the noise intensity $\lambda$.
Figure~\ref{fig:res-time}, left, shows the values of $\mathbb{P}_\lambda$ depending on $\delta$ and $\lambda$: as expected, the probability of remaining within a $\delta$-neighbourhood of the attractor is larger if $\delta$ is larger and if the noise intensity $\lambda$ is smaller. Setting the desired threshold $\gamma$ for $\mathbb{P}_\lambda$ identifies specific resilience levels.

The worst-case attraction time $\tau$ over all simulated initial conditions can be used as resilience metric, yielding the values shown in Figure~\ref{fig:res-time}, right. For system \eqref{eq:hill-equation-stoch}, we compute the worst-case attraction time $\tau$, over all simulated initial conditions, by setting $\nu=\delta$ and $\mu = 1$, for each of the values of $\delta$ and $\lambda$ considered above. 
As shown by comparing Figures~\ref{fig:res-time} right and left, the region with high values of $\mathbb{P}_\lambda$ (close to 1) is correlated with smaller values of the worst-case attraction time. The upper bound on $\tau$ in Figure~\ref{fig:res-time}, right, designates the simulation horizon and identifies cases where asymptotic practical resilience was not observed within the pre-specified time. 

\begin{figure}[ht]
\centering
\begin{minipage}{.45\textwidth}
  \centering
  \includegraphics[width=\textwidth]{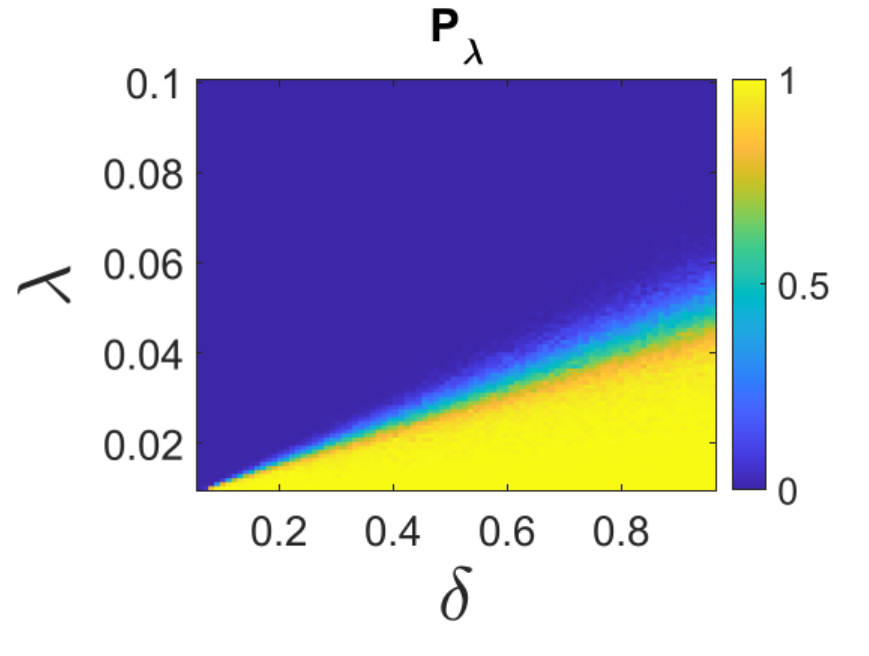}
\end{minipage}
\begin{minipage}{.45\textwidth}
  \centering
  \includegraphics[width=\textwidth]{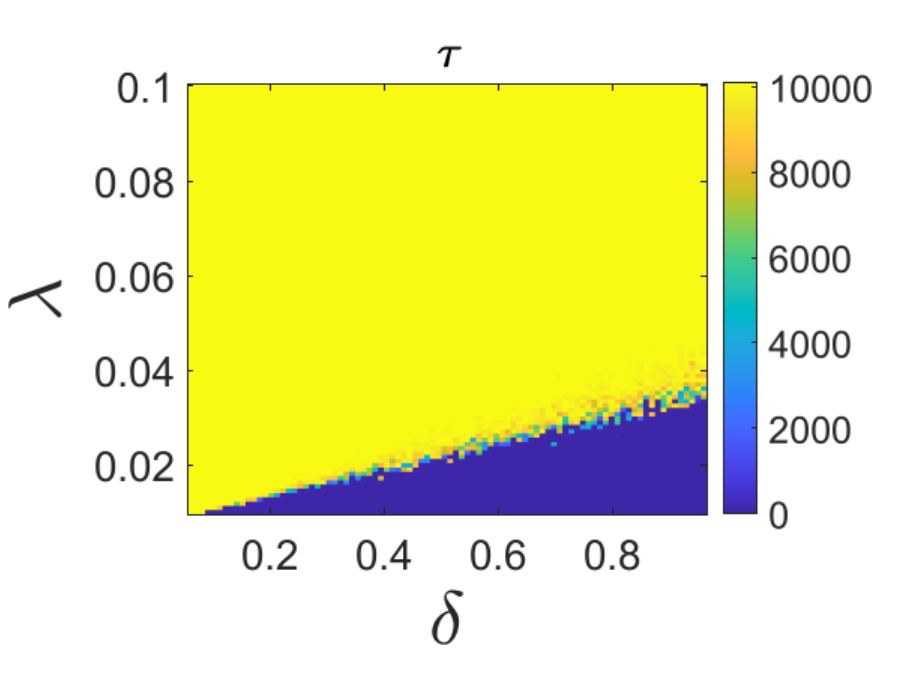}
\end{minipage}
\caption{\footnotesize For the stochastic gene regulation model in Eq. \eqref{eq:hill-equation-stoch}, with parameters $k=0.1$, $h=2$ and $a=2$ (within the bistable region in Figure~\ref{fig:gene-reg-phase-portrait}), the plots show: (left) the dependence of $\mathbb{P}_\lambda$ on $\delta$ and $\lambda$, and (right) the worst-case attraction time $\tau$ for different $\delta$ and $\lambda$, considering $\nu=\delta$ and $\mu = 1$. The region with values of $\mathbb{P}_\lambda$ close to 1 is correlated with smaller values of the worst-case attraction time.}
\label{fig:res-time}
\end{figure}

Simulations so far demonstrated the definitions for prescribed values of the system parameters, including $a$. How do resilience levels change as parameter $a$ is varied? The deterministic model \eqref{eq:hill-equation} with fixed $k=0.1$ and $h=2$ exhibits a bifurcation when $a=a_{c,1}$ (see Figure~\ref{fig:gene-reg-phase-portrait}). It is of interest to understand how proximity of $a$ to its bifurcation value $a_{c,1}$ affects the resilience of the perturbed stochastic system $G_\lambda$ in \eqref{eq:hill-equation-stoch}.

To this end, we study $(\infty, \cdot, \delta(a), \delta(a))$-practical resilience, where $\delta(a)=|\mathbf{x}_3(a) - \mathbf{x}_2(a)|$ and $a \in (a_{c,1}; a_{max})$. Here, $a_{c,1} \approx 1.77$ corresponds to the lower bifurcation value (see Figure~\ref{fig:gene-reg-phase-portrait}) and $a_{max} \approx 1.89 < a_{c_2} \approx 2.63$ is within the bistable parameter region, while $\mathbf{x}_3(a)$ is a stable equilibrium and $\mathbf{x}_2(a)$ is the unstable equilibrium (see Figure~\ref{fig:pot_land}). The resulting values of the probability $\mathbb{P}_\lambda$ are given in Figure~\ref{fig:p_critical_a}. The perturbed system trajectories reside within a $\delta(a)$-neighbourhood of $\mathbf{x_3}(a)$ when the bifurcation parameter $a$ is sufficiently far from $a_{c,1}$, for all simulated noise intensities. However, as $a$ tends to $a_{c,1}$, the probability associated with practical resilience decreases sharply. Specifying a desired $\gamma$ then allows one to identify regions of interest where practical resilience holds with the required probability.

\begin{figure}[ht]
    \centering
    \includegraphics[width=0.5\textwidth]{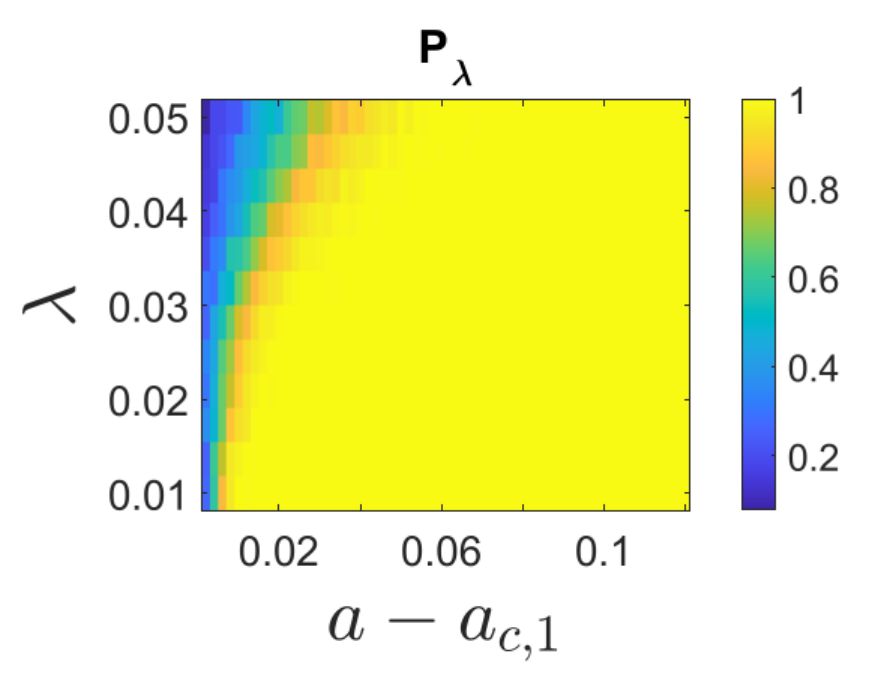}
    \caption{\footnotesize Dependence of $\mathbb{P}_\lambda$ on $\lambda$ and $a-a_{c,1}$, where $a_{c,1}$ is the lower bifurcation value (see Figure~\ref{fig:gene-reg-phase-portrait}), for system \eqref{eq:hill-equation-stoch} with $k=0.1$ and $h=2$.}
    \label{fig:p_critical_a}
\end{figure}

\subsubsection{Turing patterns}

Also the population model in Eq. \eqref{eq:PopMNodel} can help us illustrate how resilience definitions can quantify the sensitivity of a model to noise and initial conditions, over its dynamical evolution. As discussed by \cite{bashkirtseva2021stochastic}, the deterministic Levin-Segel model \eqref{eq:PopMNodel} exhibits regions of multistability, where patterns may temporarily form but ultimately dissipate and other stable structures emerge, with the limit case of spatially homogeneous equilibrium values $(\tilde{u},\tilde{v})=\left( \frac{ad}{bc-de}, \frac{ac}{bc-de} \right)$ when $D_u = D_v = 0$. Noise may further alter the picture, yielding alternative patterns during the time evolution. Practical resilience and attraction time provide summary statistics to quantitatively compare different scenarios.

Consider the system parameters $a=d=e=0.5$, $b=c=1$, $D_v = 0.005$, $D_u=1.4 \cdot 10^{-4}$.
Figure~\ref{fig:turing} shows the stable deterministic pattern obtained by discretising the deterministic equation \eqref{eq:PopMNodel} through the Crank-Nicolson numerical scheme with Neumann boundary conditions, spatial step $dx = 0.01$ and time step $dt = 0.01$ over a rectangle with spatial length $L=1$ and time horizon $t_{fin}=250$ s.

\begin{figure}[ht]
\centering
  \includegraphics[width=\textwidth]{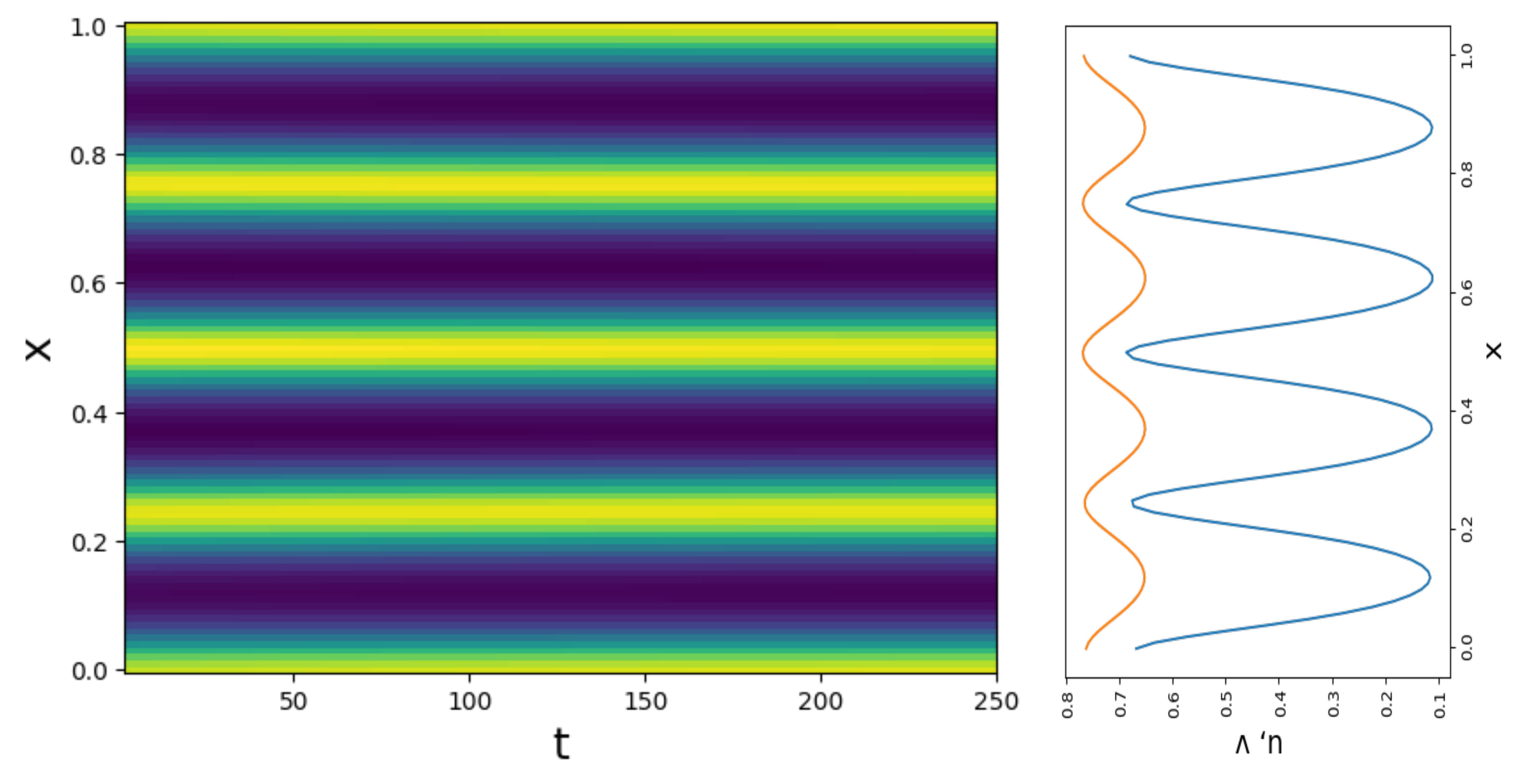}
\caption{\footnotesize Left: Turing-pattern attractor corresponding to $u$ from the deterministic system in Eq. \eqref{eq:PopMNodel} with $\lambda = \lambda_0 = 0$, visualised over space $x$ and time $t$, with the parameter values described in the main text. Right: Turing-pattern attractor for $u$, blue, and $v$, orange.}
\label{fig:turing}
\end{figure}

In Figure~\ref{fig:turing_noise}, the stochastic system is simulated, after discretization over space as in \eqref{eq:TurPatSemiDisc}, using the Milstein method \citep{bayram2018numerical}, and with the weighting function $g _{2N}\colon \mathbb{R}^{2N} \to \mathbb{R}$ that multiplies the noise given by  $g_{2N}(w; \mathbf{w})=\prod_{i=1}^{2N} g(w^{(i)};\mathbf{w}^{(i)})$, where $w^{(i)}$ ($\mathbf{w}^{(i)}$) denotes the $i$-th element of vector $w$ ($\mathbf{w}$) and $g \colon \mathbb{R} \to \mathbb{R}$ is defined as in the previous numerical case study, with $\varphi=0.0001$. System solutions emanating from initial conditions within $\pm \delta$ distance from the Turing-pattern attractor, and affected by noise with intensity $\lambda$, yield system realisations that are different from the nominal Turing pattern at each time point.

\begin{figure}[ht]
\centering
  \includegraphics[width=0.5\textwidth]{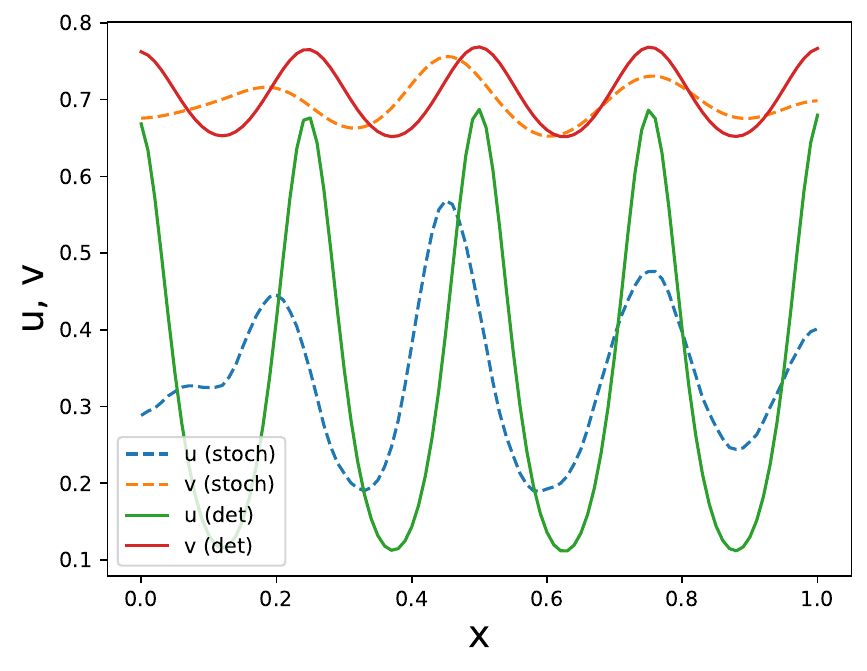}
\caption{\footnotesize Realisations of states $u(t,\cdot)$ and $v(t,\cdot)$ of system \eqref{eq:PopMNodel}, with $t=250$ s, in the deterministic case $\lambda=\lambda_0=0$ (solid) and in the stochastic case $\lambda = 40 \cdot 10^{-4}$ (dashed), with the other parameter values as described in the main text.}
\label{fig:turing_noise}
\end{figure}

We consider simulations of the stochastic system starting from sampled initial conditions close the steady-state ones, on a simulation horizon $t_{fin}=250$ s.
Practical resilience (Definition~\ref{def:PrRes}) with $\varepsilon=\delta$ and attraction time (Definition~\ref{Def:RetTime}) reveal up to which distance $\delta$ (uniform over the spatial grid) the Turing patterns can be considered resilient, depending on the noise intensity $\lambda$. Table~\ref{tab:turing} reports selected case studies. For low noise intensity $\lambda$, the system state may remain within distance $\delta$ from the deterministic Turing-pattern attractor, with a probability $\mathbb{P}_\lambda$ that increases for larger $\delta$, as expected. However, increasing $\lambda$ disrupts the patterns and drives the states away from the original attractor. Using the attraction time as a resilience indicator offers a consistent insight. We compute the worst-case attraction time over all the simulated initial conditions. As shown in Table~\ref{tab:turing}, $\tau$ is smaller when the corresponding probability is larger, and \emph{vice versa}; the indication `$-$' for $\tau$ denotes settings for which asymptotic practical resilience was not achieved within the simulation horizon $t_{fin}=250$ s. Overall, this analysis allows us to identify \emph{probabilistic} performance guarantees, beyond robustness analysis, according to a desired confidence level $\gamma$.

\begin{table}[htb]
	\centering
	\begin{tabular}{c|ccc|ccc}
          & \multicolumn{3}{c|}{$\delta$ (for $\mathbb{P}_\lambda$)} & \multicolumn{3}{c}{$\delta$ (for $\tau [s]$)}  \\
		$\mathbf{\lambda \, (\cdot10^{-4})}$    &  0.01 &   0.05 & 0.1  &0.01 & 0.05 & 0.1    \\
		\hline
1 & 0.12  & 0.855 & 0.945 & 200 & 88 & 67   \\
        9 &  0 & 0  &  0.47  & - &  - & 170  \\
        25 &  0 & 0  &  0  & - &  - & -  \\
		40 &  0 & 0  &  0  & - &  - & -  \\
		\hline
	\end{tabular}
	\caption{\footnotesize \small Values of $\mathbb{P}_\lambda$ (center) and of the worst-case attraction time $\tau$ (right) for given choices of the distance $\delta$ and of the noise intensity $\lambda$, for system \eqref{eq:PopMNodel} with parameter values as in the main text.}
	\label{tab:turing}
\end{table}

\subsection{Novelty and future directions}

Robustness and resilience are concepts that address different aspects of dynamical systems, including their ability to withstand parameter variations and external perturbations, and have been used with varying degrees of rigour in numerous disciplines. While robustness analysis has a long tradition in control theory, rigorous and formal control-theoretic definitions of resilience could not be found in the literature until very recently, thus hampering the development of a cohesive framework. Motivated by applications in systems biology, this chapter introduces definitions of resilience for a class of stochastic dynamical systems, to complement control-theoretic notions of robustness. These definitions allow for a quantitative probabilistic assessment of the impact of noise on desired system properties, related to the preservation of an attractor. Furthermore, the proposed concepts allow us to formally define the attraction time as a rigorous resilience indicator. As shown by the examples, choosing the appropriate definition among the ones we proposed, as well as suitably tuning the parameters involved in the definition, allows flexibility and applicability in various concrete problems in different contexts, and enables us to formulate and answer different types of questions, related to diverse nuances of the concept of resilience \citep{tamberg2022modeler}.

We have shown how the considered definitions can be applied to gain insight into the behaviour of widely used biological systems, motivating the inclusion of resilience in the study of systems properties. The case studies from systems biology exemplify the information that can be obtained through our newly proposed resilience concepts. Future work will include the application of the proposed framework to the numerical and analytical study of complex systems subject to parametric uncertainties and stochastic disturbances. Our new resilience definitions also pave the way to the design of novel resilience indicators, in addition to the attraction time and to the indicators that will be discussed in Chapter~\ref{ch:res-ind}.

This chapter introduces some novelties with respect to recent contributions on the topic, and complements their perspectives. \cite{artime2024robustness} surveyed the notion of robustness and resilience in network science, associating them with persistence of function against adversarial attacks targeting the network structure or against node disturbances (such as node removal); their work mostly relies on heuristic definitions and focuses on \emph{static} properties, while our distinction between robustness and resilience can embrace dynamic properties and differentiates between parametric uncertainty and stochastic perturbations. \cite{Liu2020b} offers a survey of working definitions and heuristic methods to assess network resilience; our definitions discussed above allow us to derive most of them as special cases, by appropriately tuning the parameters in the definition. \cite{Liu2020b} also surveys indicators for resilience loss -- as also done, with additional mathematical rigor, by \cite{Krakovska2024}. 

Even though resilience indicators are the main subject of the next chapter, it is useful to immediately discuss the novelties presented in this chapter with respect to existing resilience indicators. First of all, \cite{Krakovska2024} do not introduce any formal definition of resilience itself, but only implicitly define resilience through the definition of several different resilience indicators that have been proposed in the literature over time. Such indicators are tailored to two main alternative scenarios. For purely deterministic parameter-dependent systems, resilience is quantified through concepts of return time after a perturbation in the initial conditions, or through the geometry of the basin of attraction, or through bifurcation analysis in the presence of parameter variations. Conversely, for stochastic differential equations obtained from a deterministic system by adding white noise, typically having the form \eqref{eq:langevin_eq}, resilience is either quantified in probability through expected escape times (akin to the Mean First Passage Time that we will discuss in Section~\ref{subsec:escape}) or analysed by considering the deterministic dynamical system satisfied by the covariance matrix of the process and quantifying the maximum norm of its solution \citep[Section 3.4]{Krakovska2024}.
Our approach is significantly different and complementary, as it considers a purely stochastic formulation and provides rigorous \emph{definitions of resilience} quantified in terms of probability. Our proposed resilience indicator, the \emph{attraction time}, is complementary to the definitions of \emph{return time} \citep[Definitions 3.1 and 5.1]{Krakovska2024}, which are only suitable for deterministic systems, and to the definition of expected escape times, which are associated with the time needed on average to escape from the attractor -- and not with the probability of returning to a neighbourhood of the attractor after a given time.

As a final note, it would be interesting to explore the connection between the proposed definitions and a stochastic version of the stability radius \citep{HinrichsenPritchard2005}, which may allow for further quantification of resilience. We hope that bridging robustness and resilience will propel new analysis and insights, both in the control field and in multidisciplinary endeavours, to improve the modelling and the understanding of the behaviour of complex systems subject to perturbations, disturbances and uncertainties.

\newpage
\section{Data-driven detection of resilience loss}
\label{ch:res-ind}

In the previous chapter, we have provided quantitative definitions of resilience, and we have introduced the attraction time as a novel resilience indicator. Our framework effectively links formal definitions and existing indicators of resilience, such as those surveyed by \cite{Krakovska2024}. However, additional challenges arise when scarce information is available about the system and reliable models are not available. In these cases, semi-quantitative and data-driven methods can be employed to extract useful knowledge. In this chapter, we review heuristic and semi-quantitative methods discussed in the literature to infer resilience loss in poorly-known systems. We then present a methodology that addresses the discussion on robustness and resilience, by introducing the notion of \emph{generic resilience indicators} building on bifurcation analysis, which were initially defined in mathematical ecology \citep{Dakos2015, kuehn2011mathematical} and are increasingly employed in systems biology and biomedicine \citep{Trefois2015a}. Our goal is to provide a new perspective on this vibrant topic and bridge it across disciplines, by highlighting its strong connection with robustness analysis, presenting recent advances and outlining directions for future research that can leverage our quantitative definitions proposed in Chapter~\ref{ch:rob-res-sta-mod}.
Since this research area is still evolving, the available results are mostly fragmented and consisting of case-by-case studies: strong generalisable results are still missing, and several problem-specific research directions have been initiated. Hence, due to the current state of development of the field, this chapter mainly aims at connecting different disciplines and highlighting open problems.

In the case studies throughout the chapter, when we refer to attractors of dynamical systems, we mostly consider attractors defined by fixed points (equilibria). Limit cycles and pseudo-regimes \citep{lemmon2020achieving}, strange attractors \citep{strogatz2018nonlinear} or global bifurcations \citep{izhikevich2000neural} are not explicitly considered (although they could be included).

\subsection{Looking for resilience indicators: why?}

In nature, many systems are known to experience regime shifts \citep{scheffer2012anticipating}, which are mostly understood as shifts among basins of attraction \citep{ashwin2012tipping, kuehn2011mathematical}. Such shifts can pose significant threats to the system functioning and survival: examples include transitions to cancerous cell states \citep{Chen2012c}, genomic switching \citep{Aguirre2015c}, onset of epileptic seizures \citep{Maturana2020} or epidemic re-emergence \citep{ORegan2013}. Whenever the system of interest can be aptly captured by a quantitative model, the resilience definitions from Chapter~\ref{ch:rob-res-sta-mod} immediately apply.

Often, however, complete and validated models enabling quantitative predictions are missing or poorly identifiable; for example, only time-series data may be available. In addition, regime shifts are often abrupt and unexpected \citep{Boettiger2013, brock2008regime}. Consequently, monitoring, anticipation and management are particularly challenging \citep{lemmon2020achieving}, especially when only little information is available, \ie when the exact functions governing the dynamics are not known, or no viable mechanistic models are available. In these cases, a quantitative investigation of system resilience becomes harder. On the other hand, purely data-driven statistical methods (see, \eg \cite{brett2020detecting, Bury2021,kostoulas2021epidemic,macintyre2023artificial}) may fall short: they lack causal explanations, are often opaque to interpretation and are heavily based on past data, therefore being less suitable for events that are rare or for which big data are hard to gather. Looking for generic indicators to alert for resilience loss is thus an active area of research, often going under the umbrella of the Critical Transitions framework \citep{scheffer2009early}. 

The field is still relatively at its infancy and different methods have been recently suggested and tested. To the best of our knowledge, little systematic studies have been conducted to link the various, albeit problem-specific, research directions that are being developed. Preliminary reviews have been provided by \cite{Krakovska2024,Liu2020b, proverbio2022classification}. In this chapter, we survey this promising area, highlight its recent developments and its connections to the results presented in the previous chapters, and discuss open problems and research directions to which the systems-and-control community could greatly contribute.

\subsection{Regime shifts and analysis methods}
\label{sec:preliminar_res_ind}

Regime shifts from an attractor to an alternative one are known to occur, usually following three main mechanisms \citep{ashwin2012tipping, proverbio2023systematic, thompson2011predicting}: a bifurcation (qualitative dynamical change produced by parameter variations, \cite{Guckenheimer2009a}) altering the system and making the original equilibrium vanish; random (noisy) deviations, driving the system state away from the original attractor due to stochasticity \citep{ritchie2017probability}; and non-autonomous ramping parameters, yielding so-called rate-induced transitions \citep{Alkhayuon2018,slyman2023rate}. 

The first (deterministic) case is mainly investigated in bifurcation theory; we refer to \cite{Guckenheimer2009a,kuznetsov2013elements} for a formal introduction. As an example, consider the gene regulation model from Eq. \eqref{eq:hill-equation}, whose equilibria have been discussed along with their stability properties in Chapter~\ref{ch:rob-res-sta-mod}. Its bifurcation diagram, reporting the dependency of its equilibria $\hat{x}$ on the parameter $a$, is shown in Figure~\ref{fig:res_func}.  Consider initially $a > 2.6$, so that the system rests on the upper stable equilibrium branch. If the parameter $a$ is lowered in a quasi-steady-state fashion (see Appendix~\ref{sec:slowfast}) until it crosses the critical value $a \approx 1.77$, the system trajectories converge to the lower equilibrium. 

\begin{figure}[ht!]
	\centering
	\includegraphics[width=0.5\textwidth]{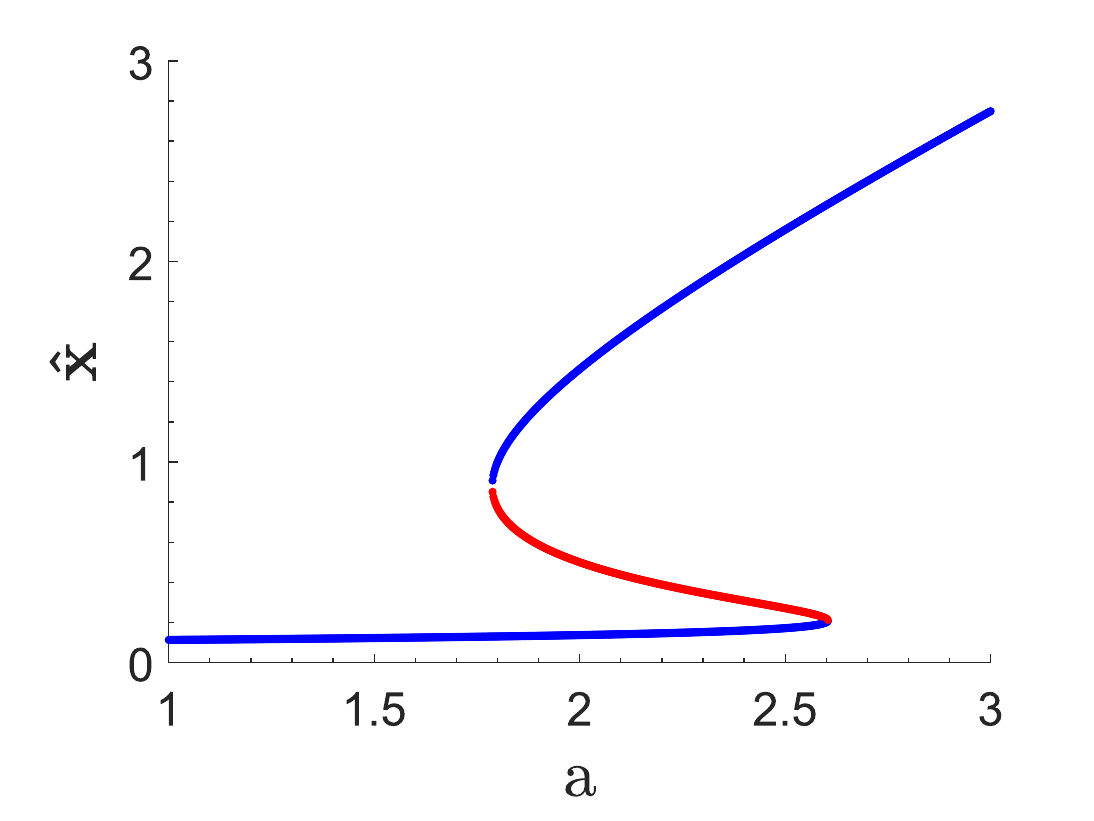}
	\caption{\footnotesize The bifurcation diagram ($a$, $\hat{x}$) of Eq. \eqref{eq:hill-equation}, where $\hat{x}$ identifies the equilibrium points. The bifurcation diagram is often referred to as a \emph{resilience function}, as it depicts the critical parameter values for regime shift. Blue denotes stable equilibria, red unstable ones.}
	\label{fig:res_func}
\end{figure}

Stochastic nonlinear systems can have behaviours that could never be observed in the deterministic case; when some characteristics of the stochastic external disturbance (\eg its variance) cross certain thresholds, transition phenomena can take place. Noise-induced transitions can occur because switching between alternative states can be triggered, with a certain probability, due to random ``jumps'' driven by the stochastic dynamics.
An example is the fluctuation of enzymes across potential barriers \citep{min2005fluctuating} in cell biology; in epidemiology, individual variations or random super-spreading events can trigger local outbreaks \citep{Lloyd-Smith2005a, Small2005}, making disease-free equilibria shift into epidemic states.

Rate-induced shifts are typical of non-autonomous systems \citep{Ashwin2017a}. They were experimented in engineered thermo-acoustics systems \citep{Bonciolini2018} and may occur in living organisms too, \eg in the case of speed-dependency concurring in regulating cellular decision making \citep{Nene2012}; however, non-autonomous systems are not the focus of the monograph and we do not delve further into this topic.

Another interesting case, recently getting much attention from the modelling and applied mathematics community \citep{kuehn2011mathematical,proverbio2022classification,zou2023uncertainty}, is that of systems primarily governed by deterministic functions, with some added stochasticity, for which partial or qualitative information is available. In the terminology of the Critical Transitions framework, this is the case of noisy bifurcation-induced shifts, where stochastic noise is added on top of slowly varying parameters, in the slow-fast system interpretation of Eq. \eqref{eq:slowfast}. 

A biological example, supported by experimental evidence, is offered by a colony of budding yeast, a unicellular organism. The bifurcation diagram reconstructed from experiments \citep{dai2012generic}, reporting how the steady-state cell density in the colony depends on a suitable chemical dilution factor (DL), is shown in Figure~\ref{fig:Dai2012}. DL acts as a control parameter; its exact effect is not entirely known yet, but evidence suggests that it alters the mutual interconnections among cells, thereby altering the normal state of the colony.
Experimentally, it is possible to change the dilution factor. At very low concentrations, the colony prospers (it survives at high cell density). Despite random loss of individual cells (see the standard deviation bars in the figure), the colony survives under moderate to high levels of dilution in the environment, up to a threshold concentration (DL $\approx 1600$) that yields sudden collapse. In this case, the colony experiences a slow approach to a bifurcation point (driven by increasing dilution factor); this scenario is characterised by stochasticity, as individual variations, cell-cell heterogeneity and random cell death contribute to colony randomness.

\begin{figure}[ht!]
	\centering
	\includegraphics[width=0.5\textwidth]{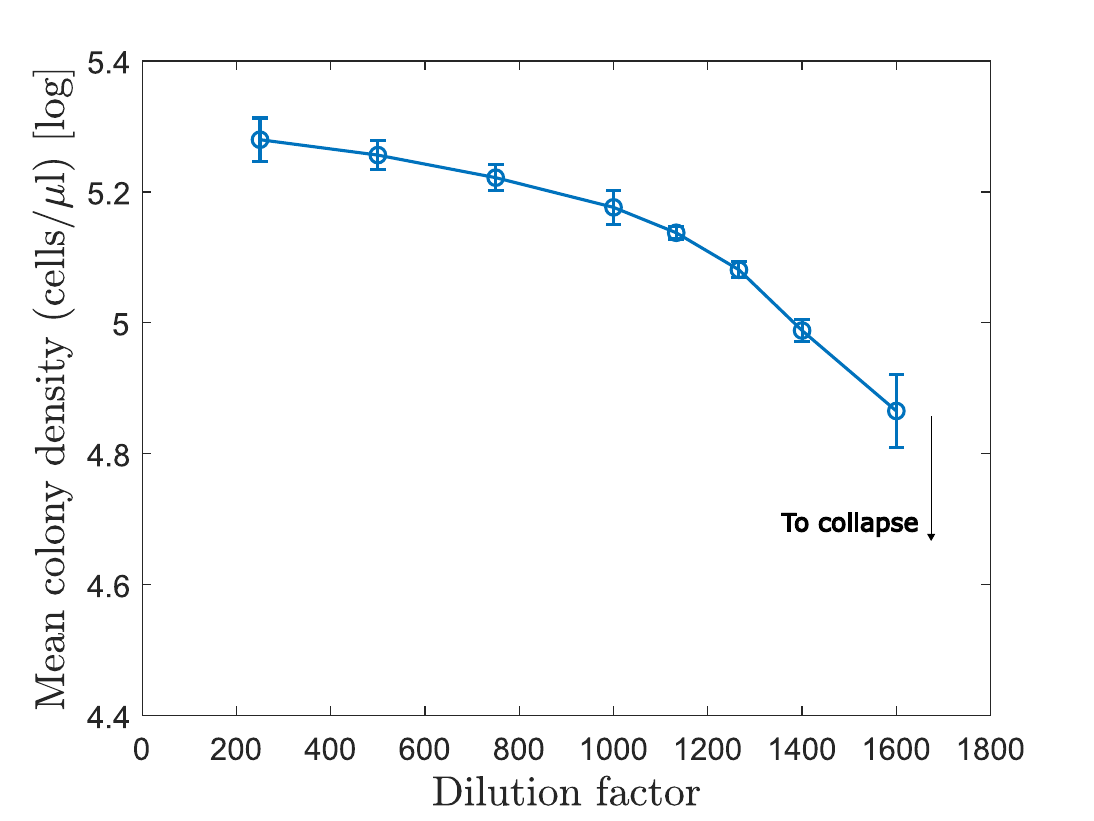}
	\caption{\footnotesize Cell density in yeast colonies versus dilution factor (biological stressor). Increasing dilution values slightly modifies the cell density, up to a critical value (last point) after which the population is driven to collapse (density of zero $\text{cells}/\mu l$, not shown in the graph). The population undergoes a bifurcation that drastically changes its regime. Elaboration on public data from \cite{dai2012generic}; a full experimental phase portrait is in the original publication.}
	\label{fig:Dai2012}
\end{figure}

The quantitative drivers of colony collapse, including dynamic intrinsic and extrinsic factors such as cell-cell variability or varying concentrations of dilution factors, can be understood within the conceptual framework of system resilience. However, as complete models are not available, hybrid methods should be employed to quantify resilience, combining data-driven techniques informed by minimal surrogate models. Can one identify generic indicators of resilience, or of loss thereof, that work for broad classes of systems and are easy to extract from empirical data, which are often sparse or noisy?

Here, we review the recent literature proposing multiple resilience indicators, demonstrate their application in examples, highlight open questions and topics that the control community could contribute to, and discuss perspectives to improve our understanding, prediction and management of resilience and regime shifts.

\subsubsection{Bifurcation normal forms as minimal models}
\label{sec:bif-normal-forms}

Bifurcations \citep{Golubitsky2003a, kuznetsov1994bifurcation} are powerful tools to analyse qualitative changes in system dynamics, due to parameter variations. For many (high-dimensional) systems, the transition from an attractor to another can be understood \emph{locally} through an approximation by a low-dimensional bifurcations, which plays an important role for dimension reduction and analysis. Bifurcations are widely discussed in systems biology \citep{Angeli2004detection,Angeli2014,Hui2011,Leite2010,Mojtahedi2016b,Moris2016} and can be connected to robustness studies using the generalised Nyquist stability criterion \citep{Iglesias2010}, thereby being useful for models from complex system theory. Deducing the existence of bifurcations from scarce data, or formally proving it in complex models, is a particularly challenging problem. Here, we just briefly discuss how to design effective inference methods.

An important notion is that of \emph{normal forms}. Sufficiently close to the critical value of a local bifurcation, a dynamical system may be mapped upon a nonlinear canonical form of a given bifurcation (mostly associated with an ODE whose right-hand side is a polynomial of the state, see Appendix~\ref{App:Normal}) via a coordinate transformation, which defines the normal form of the dynamical system.\footnote{Demonstrating the soundness of such an approximation requires the \emph{center manifold theorem} \citep{Crawford1991a} and other notions of applied bifurcation theory, which are beyond the scope of this monograph. The interested reader is referred, \eg to \cite{Golubitsky2003a, Haragus2010, kuznetsov2013elements}.} Normal forms can be used as surrogate models to reduce complex high-dimensional dynamics to local low-dimension dynamics, along a direction identified by the leading eigenvalue of the system. Studying properties of normal forms enables simplified stability analysis, whose local results can be extended to broad classes of systems characterised by the same type of bifurcation \citep{kuehn2021universal}. For this reason, bifurcation normal forms can be seen as ``universal patterns'' for critical and explosive phenomena \citep{kuehn2021universal, Thom2554}.

To visualise how normal forms can effectively approximate systems in the neighbourhood of a bifurcation point, consider the gene regulation model \eqref{eq:hill-equation} close to a regime shift (abrupt transition between different values of the protein concentration, $x$) and compare it to the normal form for the fold (\ie saddle-node) bifurcation (see Appendix~\ref{App:Normal}): Figure~\ref{fig:example_normal_forms} compares their vector fields, and Figure~\ref{fig:example_normal_forms2} their bifurcation diagrams. The fold bifurcation is also involved in the transition to epidemic outbreaks in SIR-like models with vaccination \citep{brauer2011backward}.

\begin{figure}[ht!]
	\centering
	\includegraphics[width=0.8\textwidth]{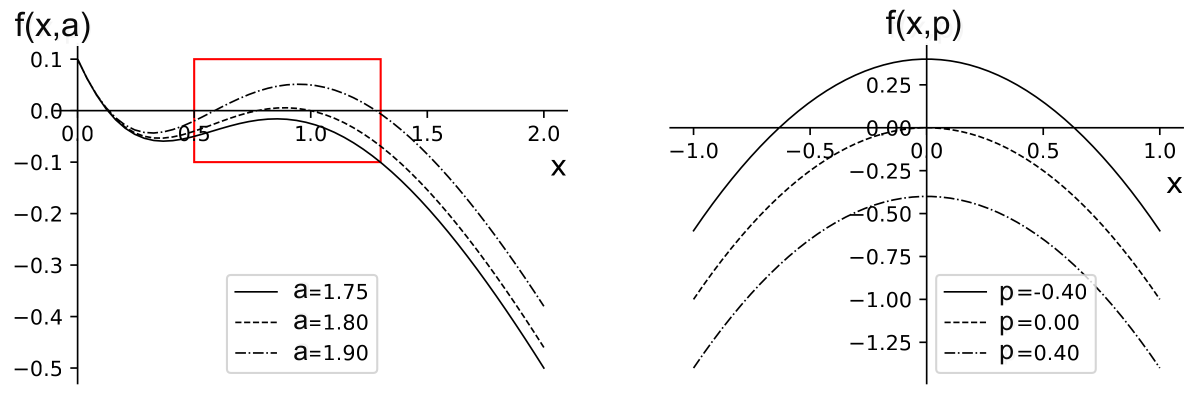}
	\caption{\footnotesize Local bifurcation behaviour and normal forms. Left: plot of $f(x,a)=f_G(x,a)$, the right-hand side of the ODE for the gene regulation model \eqref{eq:hill-equation}. Right: plot of $f(x,p) = -p-x^2$, corresponding to the fold bifurcation normal form (see Appendix~\ref{App:Normal}). The red rectangle highlights the local region where the gene regulation model is approximated by the normal form: both undergo a bifurcation as the function crosses the $x$-axis due to changes in the bifurcation parameter.}
	\label{fig:example_normal_forms}
\end{figure}

\begin{figure}[ht!]
	\centering
	\includegraphics[width=0.5\textwidth]{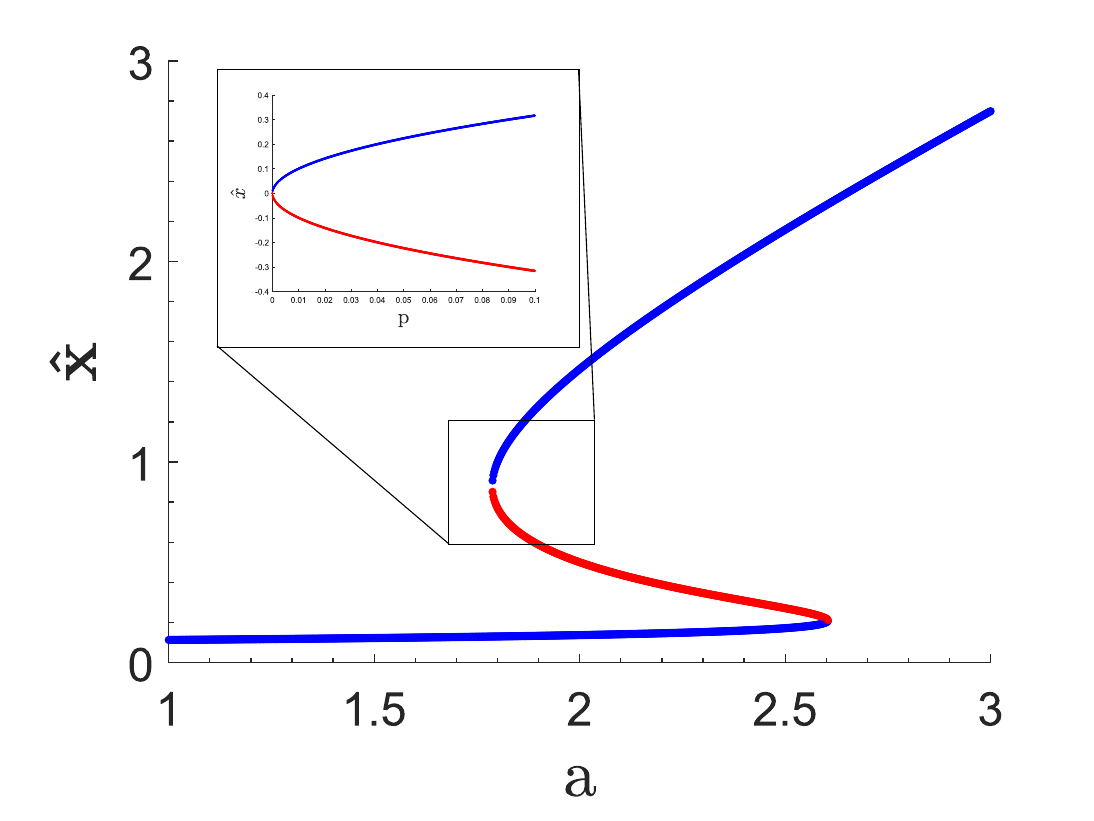}
	\caption{\footnotesize The local system behaviour can be approximated by the associated normal form also in terms of bifurcation diagram close to the bifurcation point. Big: bifurcation diagram for the gene regulation system \eqref{eq:hill-equation}. Inset: bifurcation diagram  for the fold normal form $\dot{x} = -p-x^2$. A complete analysis, including the case of $h > 2$ in \eqref{eq:hill-equation}, is provided by \cite{proverbio2022buffering}. Stable equilibria are in blue and unstable equilibria in red.}
	\label{fig:example_normal_forms2}
\end{figure}

Other abrupt regime shifts, \eg in epidemiological models \citep{ORegan2013}, can be locally captured by bifurcation normal forms. For instance, the SIR model \eqref{eq:sir} extended to include a birth or influx rate $\pi$ and a natural death or outflux rate $\mu$,
\begin{equation}
    \begin{cases}
        \dot{S}(t) =  \pi - \beta S(t) I(t) - \mu S(t)\,,\\
        \dot{I}(t) =  \beta S(t) I(t) - (\mu + \gamma)I(t)\,, \\
        \dot{R}(t) = \gamma I(t) - \mu R(t) \, ,
        \label{eq:sir_demography}
    \end{cases}
\end{equation}
is characterised by a transcritical bifurcation (see Appendix~\ref{App:Normal}) in the transition from the disease-free equilibrium to an epidemic outbreak \citep{balamuralitharan2018bifurcation}. The same bifurcation also occurs in models accounting for non-pharmaceutical interventions and quarantine; see, \eg \cite{proverbio2021dynamical} and Figure~\ref{fig:sir_trans}. 

\begin{figure}[ht!]
	\centering
	\includegraphics[width=0.5\textwidth]{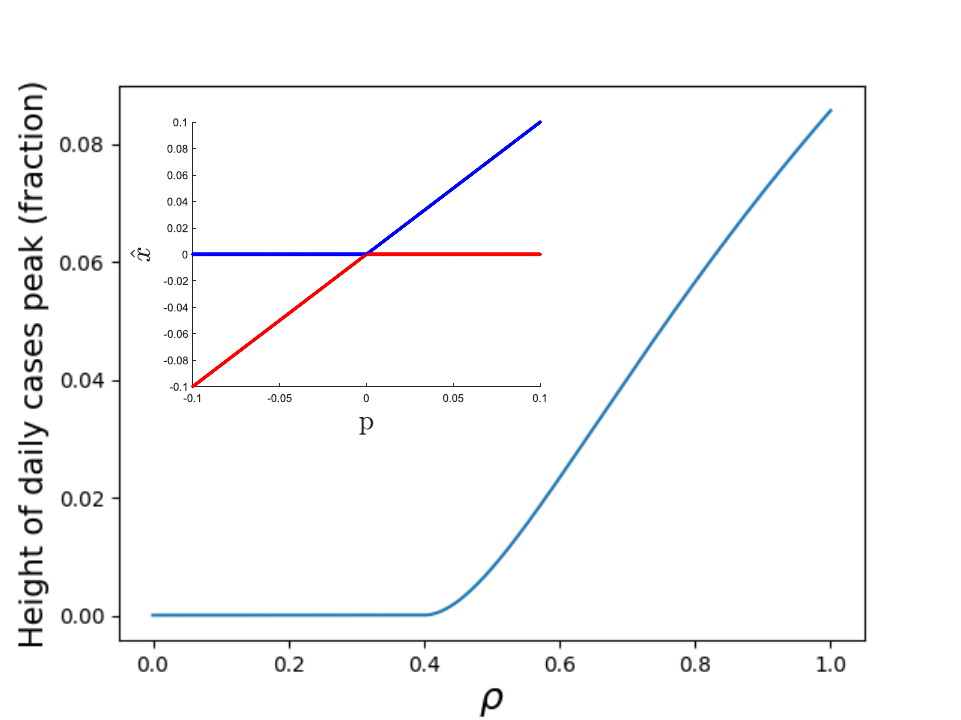}
	\caption{\footnotesize Local bifurcation behaviour and normal form. The local behaviour of an extended SIR-like model akin to Eq. \eqref{eq:sir_demography}, with $\beta \to \beta \rho$ and $\max(I)$ shown as a function of the bifurcation parameter $\rho$ \citep{proverbio2021dynamical}, can be approximated by a transcritical normal form $\dot{x} = px - x^2$ (see Appendix~\ref{App:Normal}), in the inset. Stable equilibria are in blue and unstable equilibria in red.}
	\label{fig:sir_trans}
\end{figure}

Bifurcations and their associated normal forms may fall short when multiple eigenvalues of commensurable amplitude exist; moreover, their application to natural systems heavily relies on slow-fast approximations of relative timescales (see, \eg \cite{Berglund2006, proverbio2023systematic}, and Appendix~\ref{sec:slowfast}). Still, they provide powerful minimal models that encompass numerous applications and allow for resilience analysis; they have been proposed as candidates for bridging scales and for the design of models, together with data-driven inference \citep{Tegner2016}; and they provide minimal benchmark models in case of unknown dynamics, subject to some (minimal) modelling assumptions. 

In the next section, we use bifurcation normal forms as minimal surrogate models $G_{\lambda_0}$ and augment (when necessary) their associated ODE with stochastic noise, in the same spirit as in Section~\ref{Sec:RorR}. Then, reducing the original system to its associated normal form enables the derivation of generic indicators of resilience loss, which can be applied to extract indicators of system resilience (and early warning signals to alert about loss of resilience and imminent critical transitions) from real-world time series \citep{dakos2012methods}.

\subsection{Resilience indicators}

Here, we survey several indicators suggested in the literature as proxies for resilience, whenever bifurcation normal forms augmented by stochastic terms are suitable minimal models for the system of interest.

\subsubsection{Recovery rate}

Engineering resilience is often measured in terms of recovery rates from perturbed states back to the equilibrium (see Section~\ref{sec:res_informal}), and can be understood in terms of bifurcations.

Given the system $\dot x(t)=f(x(t),p)$, where $p$ is a bifurcation parameter, let $\bar x(p)$ be an asymptotically stable equilibrium such that $f(\bar x(p), p)\equiv0$. Then, solutions emanating from initial conditions $z=\bar x(p) + u$, sufficiently close to $\bar x(p)$, converge to $\bar x(p)$ with an exponential decay rate of $\lambda_u <0$, which is the Lyapunov exponent in the $u$ direction. If the leading Lyapunov exponent (the largest, \ie the smallest in magnitude) is $\mathcal{O}(p^\alpha)$, then we call $\alpha$ the \emph{recovery exponent} \citep[Definition 2.9]{kuehn2011mathematical}.
The exponent $\alpha$ thus provides a measure of how quickly the trajectories converge to the equilibrium after a perturbation in the initial conditions, near a bifurcation, depending on the distance to the bifurcation point (corresponding to $p=p_0=0$ without loss of generality). Larger $\alpha$ corresponds to slower convergence. 
As shown by \cite{kuehn2011mathematical}, the value of the recovery exponent depends on the bifurcation type: $\alpha = 1/2$ for fold bifurcations, while $\alpha   = 1$ for pitchfork, transcritical and Hopf bifurcations.

In perturbed systems, the recovery rate can assess the proximity to catastrophic shifts \citep{van2007slow}, as demonstrated in ecological applications \citep{arnoldi2016resilience}.

\subsubsection{Escape rate}\label{subsec:escape}

The abundant literature on stochastic processes \citep{J.S.Allen2014, Gardiner1985, papoulis2002probability, van1992stochastic} often studies how likely a random walker is to jump from one basin of attraction to another due to noise. For double-well potentials, such as that in Figure~\ref{fig:pot_land}, mean passage time and Kramers' escape rate \citep{freidlin1998random,Gardiner1985,van1992stochastic} are popular measures for the likelihood of a regime shift. Escape measures quantify how long it takes to jump onto an alternative well for the first time, or how often such jumps are expected to occur. We briefly describe these measures for one-dimensional potentials and homogeneous processes; for extensions to multi-dimensional domains and non-homogeneous processes, see \cite{freidlin1998random,Gardiner1985}.

To derive the Kramers' escape rate, consider a landscape representation (analogous to the notion of Lyapunov function, \emph{cf.} \cite{conley1988}) of a system of interest, such as the one depicted in Figure~\ref{fig:pot_land}, associated with the potential function $V$. A collection of independent Brownian particles is initially located within a ``stable well'', \ie at a stable equilibrium point (\eg $x_1$ in Figure~\ref{fig:pot_land}). Being exposed to stochastic fluctuations, the particles may escape the well by reaching an unstable equilibrium (\eg $x_2$ in Figure~\ref{fig:pot_land}) and crossing the barrier towards an alternative stable equilibrium (\eg $x_3$ in Figure~\ref{fig:pot_land}). What is the rate at which this escape takes place?
Under suitable assumptions, the Kramers' escape rate is given by \citep{plischke1994equilibrium,van1992stochastic}
\begin{equation} \label{eq:kramers}
\tau_k = \frac{\sqrt{V''(x_1)|V''(x_2)|} e^{\frac{V(x_1)-V(x_2)}{D}}}{2 \pi},
\end{equation}
where $V''$ denotes the second derivative of $V$, $V(x_1)$ and $V(x_2)$ are the values of the potential function at the stable and unstable equilibrium points, and $D$ is the diffusion term of the stochastic Fokker-Planck equation \eqref{eq:fokker_plank}. The Kramers' escape rate has been used to characterise the resilience of chemical reactions across potential barriers \citep{hanggi1990reaction} and, recently, the resilience of gene regulation motifs (on average across cell populations) against noisy regime shifts \citep{proverbio2022buffering}.
For Ornstein-Uhlenbeck (OU) processes \eqref{eq:o-u}, associated with a linearised approach to one-dimensional bifurcations, it can be shown that the escape rate increases exponentially as a bifurcation point is approached (and, when the noise $D$ is large, it has non-negligible values even relatively far away from the bifurcation point, thus indicating the risk of noise-induced regime shifts); this makes it a valuable indicator of the system's resilience.

Alternatively, authors such as \cite{dennis2016allee} suggest, as a resilience indicator, the Mean First Passage Time (MFPT) $\mathcal{T}(x)$, \ie the time it takes, on average, for a single, randomly forced particle with state $x$ to randomly drift from a stable equilibrium to an unstable one.

Consider a SDE of the form \eqref{eq:langevin_eq}; in its associated Fokker-Planck representation \eqref{eq:fokker_plank}, the stochastic variable $x$ is associated with a stationary probability density function $p(x)$.
We denote by $x_1$ and $x_2$ the initial stable equilibrium and the unstable equilibrium (located at the boundary of the basin of attraction of $x_1$), respectively. Assume that $x_1 < x_2$. Then, the amount of time it takes for the system to transition from $x_1$ to $x_2$ is the mean first passage time:
\begin{equation}
\mathcal{T}(x) = 2 \int_{x_1}^{x_2} \frac{\int_0^yp(z)dz}{b(x)p(y)}dy \quad \text{ for } x_1 < x_2 \, .
\label{eq:MFPT_final}
\end{equation}
To estimate the MFPT from a stable equilibrium $x_3$ such that $x_3 > x_2$, it suffices to set the boundary conditions appropriately in \eqref{eq:MFPT_final}. 
In biology, $\mathcal{T}(x)$ was used by \cite{Sharma2016} to compare the resilience of gene regulatory systems of the form \eqref{eq:hill-equation}, as illustrated in Figure~\ref{fig:mfpt}.

\begin{figure}[ht!]
	\centering
	\includegraphics[width=0.5\textwidth]{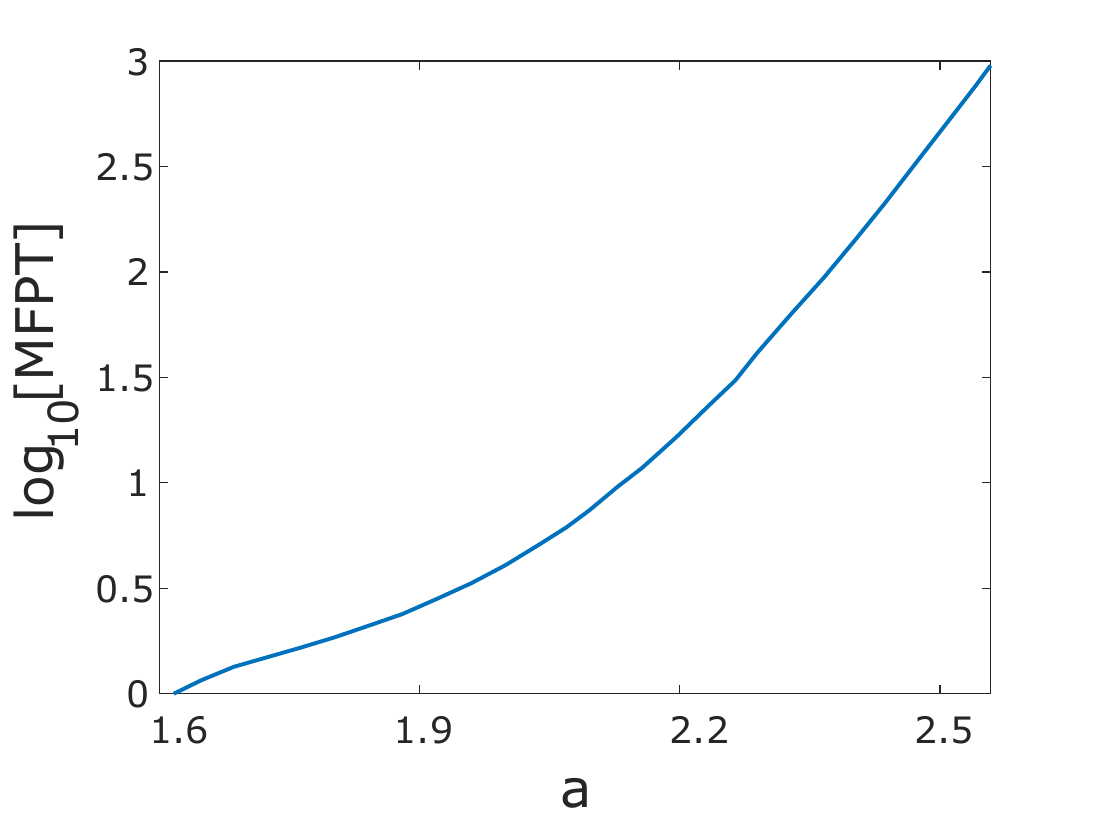}
	\caption{\footnotesize Logarithmic Mean First Passage Time \eqref{eq:MFPT_final} for the gene regulation system \eqref{eq:hill-equation} as a function of the bifurcation parameter $a$, with $k=0.1$ and $\sigma=0.2$. When $a$ is close to its bifurcation value $a \approx 1.6$, the MFPT becomes negligible, while it becomes higher for values further from the bifurcation.}
	\label{fig:mfpt}
\end{figure}

\subsubsection{Stochastic indicators}
\label{sec:monomodal_ind}

For bifurcation normal forms with added noise terms, summary statistic indicators that can be derived by studying normal forms \citep{kuehn2011mathematical} are particularly useful because their behaviour is correlated with the proximity to the bifurcation points, provided that the noise magnitude remains constant; hence, they serve as indicators for loss of resilience. In this section, we present summary statistic indicators and discuss recent results for their application to unknown biological and epidemiological systems. Currently, this theory is less developed and formal than the theory revolving around structural analysis; here, we aim to foster interest in connecting the concepts within an interdisciplinary research effort.

Stochastic indicators for systems near bifurcation points are usually derived by analysing linearised normal forms around equilibria, which can be mapped into Ornstein-Uhlenbeck processes \eqref{eq:o-u}. This mapping provides exact analytical solutions based on the theory of stochastic processes \citep{Gardiner1985}, and enables an immediate interpretation of sample summary statistics extracted from data, which can be thus compared with theoretical predictions to assess the validity of resilience indicators. 
We survey here the basic indicators \citep{scheffer2012anticipating} upon which many resilience studies are based. The results below assume quasi-steady state parameters (in the framework of timescale separation, see Appendix~\ref{sec:slowfast}) and white noise. We summarise the most relevant indicators and their functional forms; the proofs of the formal results are provided in Appendix~\ref{App:stoc-ind}, while Figure~\ref{fig:indicators} shows examples of the behaviour of these indicators as a critical point is approached.

More details on Ornstein-Uhlenbeck process are in Appendix~\ref{sec:ODESDE}.

\begin{proposition}\label{theo:var}
For an Ornstein-Uhlenbeck process
\begin{equation}\label{eq:OU}
dx = -k \ x \ dt + \sqrt{2D} dW_t,
\end{equation}
where $D$ is the diffusion intensity and $W_t$ a Wiener process, as a bifurcation point is approached (\ie as $k \to 0$), the solution is a stochastic stationary process with:
\begin{itemize}
\item variance $\text{Var}= \langle y_t^2 \rangle - \langle y_t \rangle ^2=\frac{D}{k}$ tending to infinity (namely, $\lim_{k \to 0} \frac{D}{k}=\infty$), as shown in Figure~\ref{fig:indicators}, upper left panel;
\item autocorrelation at lag-1 $AC(1)= e^{-k}$ tending to 1 (namely, $\lim_{k \to 0} e^{-k}=1$), as shown in Figure~\ref{fig:indicators}, upper central panel.
\end{itemize}
\end{proposition}

We recall that the autocorrelation is
\begin{equation*}
AC(t-t') = \lim_{t,t' \to \infty} \frac{\text{Cov}(y(t) y(t'))}{\sqrt{\text{Var}(y(t)) \text{Var}(y(t'))}} = e^{-k |t-t'|},
\end{equation*}
and the autocorrelation at lag-1 is $AC(t-t')$ computed for $t-t'=1$.

Other statistical moments like skewness, kurtosis and higher moments, for quasi-steady-state parameters, can be expressed as 
	\begin{equation}
		\langle \left(y - \langle y \rangle \right)^n \rangle  = \int_{-\infty}^{\infty} (y'-\mu)^np(y')dy' 
	\end{equation}
where $\mu = \langle y \rangle$ is the mean (expected value) and $p(x)$ is the probability density function. In case of white noise, $p(x)$ is Gaussian; in general, it can be derived from the Fokker-Planck equation \eqref{eq:fokker_plank}. 

Also the phenomenon known as \emph{spectral reddening} is sometimes used as a resilience indicator \citep{ Bury2020,Carpenter2011}.
\begin{proposition}[\textbf{Spectral reddening}]\label{theo:reddening}
For an Ornstein-Uhlenbeck process \eqref{eq:OU}, the solution is a stochastic stationary process with power spectral density function $S(\omega,k)= \frac{D}{\pi(k^2 + \omega^2)}$. As a bifurcation point is approached (\ie as $k \to 0$), its maximum value $\max_\omega S(\omega,k)$ increases, and it is always achieved for low frequencies $\omega$, as shown in Figure~\ref{fig:indicators}, lower panels.
\end{proposition}

Another relevant indicator showing an increase near bifurcation points is the Shannon entropy.
\begin{proposition}\label{theo:entropy}
For an Ornstein-Uhlenbeck process \eqref{eq:OU}, as a bifurcation point is approached (namely, as $k \to 0$), the solution is a stochastic stationary process whose (non-normalised) entropy $H_s(x) = \frac{1}{2}\left( \log (2 \pi \text{Var}) +1 \right)$ tends to infinity, as shown in Figure~\ref{fig:indicators}, upper right panel.
\end{proposition}

Figure~\ref{fig:indicators} shows the behaviour of these summary statistics, as functions of the distance $k$ from the bifurcation parameter, and for a fixed value of $D$ (in particular, we used $D = 0.1$).
According to the literature, the increase of each indicator as $k \to 0$ (clearly visible in Figure~\ref{fig:indicators}) indicates that the system is approaching the bifurcation point and therefore that there is a resilience loss.

\begin{figure}[ht!]
	\centering
	\includegraphics[width=\textwidth]{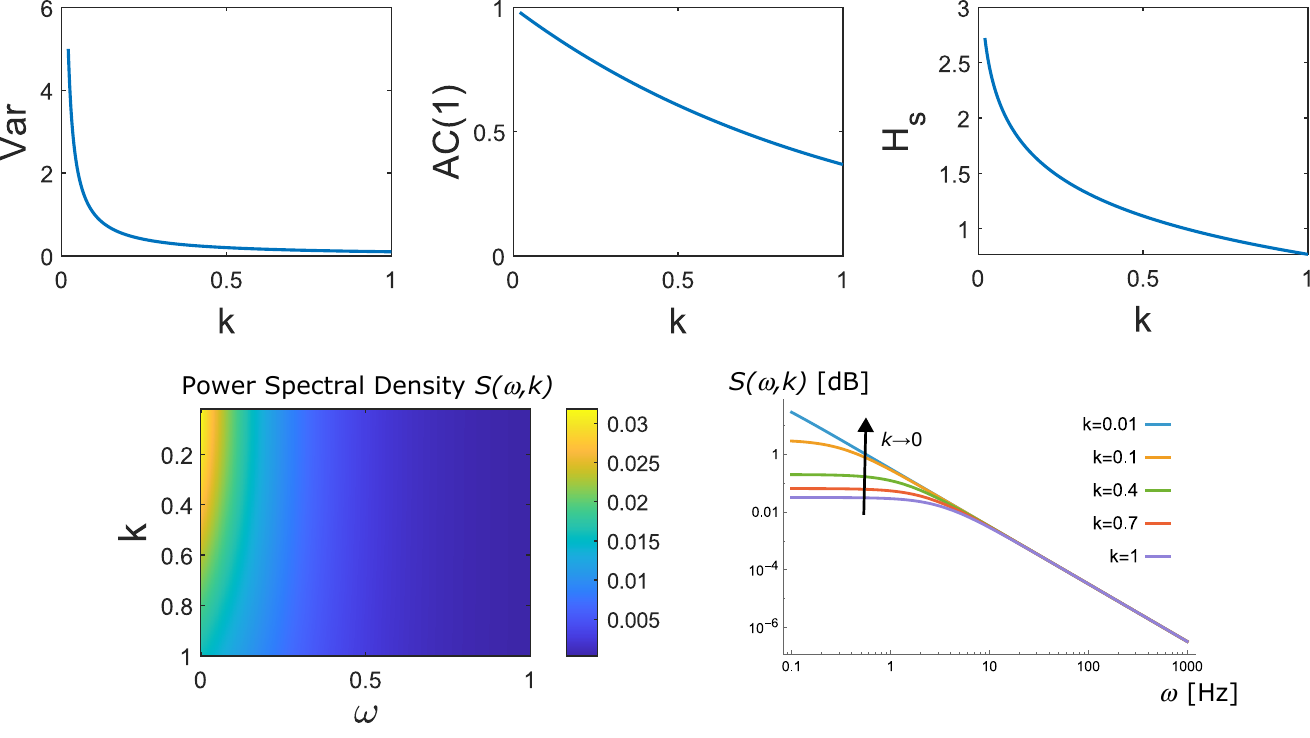}
	\caption{\footnotesize Behaviour of variance $\text{Var}$, autocorrelation at lag-1 $AC(1)$, Shannon entropy $H_s$ and power spectral density function $S(\omega,k)$ for an Ornstein-Uhlenbeck process \eqref{eq:OU} as the distance from the bifurcation parameter $k$ converges to $0$ from the right. The first three indicators increase when approaching the bifurcation point, while the spectral density function does so visibly only in the low frequencies (\emph{spectral reddening}), as is also visible in the graph on the bottom right, that shows $S(\omega,k)$ as a function of $\omega$ for different values of $k$.}
	\label{fig:indicators}
\end{figure}

The propositions above rely on surrogate models and are therefore valid regardless of the specific original system, provided that the modelling assumptions on stationarity and noise type are respected, and that the surrogate models provide a sufficiently accurate approximation of the original system, at least close to the  bifurcation point. The indicators are easily computed from distributions or from time series (using rolling windows and under the assumption of ergodic distribution), using their counterparts based on samples \citep{Kuehn2013a, proverbio2023systematic}. Hence, they can be easily extracted from simulated or empirical data. Figure~\ref{fig:indicator_var} shows an example of increasing variance before a tipping point and compares: the \emph{theoretical} expression of the stationary variance $\text{Var}=\frac{D}{k}$ as a function of the distance $k$ from the bifurcation value, as in Eq. \eqref{eq:variance}; the expression of the variance as a function of the distance $a-a_0$ from the bifurcation value, computed from \emph{numerical} simulations, for system \eqref{eq:hill-equation}; and the expression of the variance as a function of time based on \emph{empirical} results about COVID-19 data in Luxembourg \citep{proverbio2022performance}.

Resilience loss is associated with an \emph{increase} of the indicators, which needs to be quantified. When is the value of the indicator large enough to denote an imminent critical transition, or tipping point? Common methods employ optimised thresholds (to raise alerts in monitoring applications), or significant increases assessed via $\sigma$ or $p$-value analysis \citep{Boettiger2012b, proverbio2022classification, proverbio2022buffering}, or Kendall's $\tau$ scores to identify monotonic behaviour. In particular, Kendall's $\tau$ score \citep{kendall1938new} is defined as  
\begin{equation*}
    \tau = \frac{\#\text{concordant pairs} - \#\text{discordant pairs}}{M(M-1)/2},
\end{equation*}
where $M$ is the number of considered time points. Two generic points $(t_1,x_1)$ and $(t_2,x_2)$, with $t_1<t_2$, are a concordant pair if $x_1<x_2$ and a discordant pair otherwise. For constant sequences, $\tau=0$.

\begin{figure}[ht!]
	\centering
	\includegraphics[width=\textwidth]{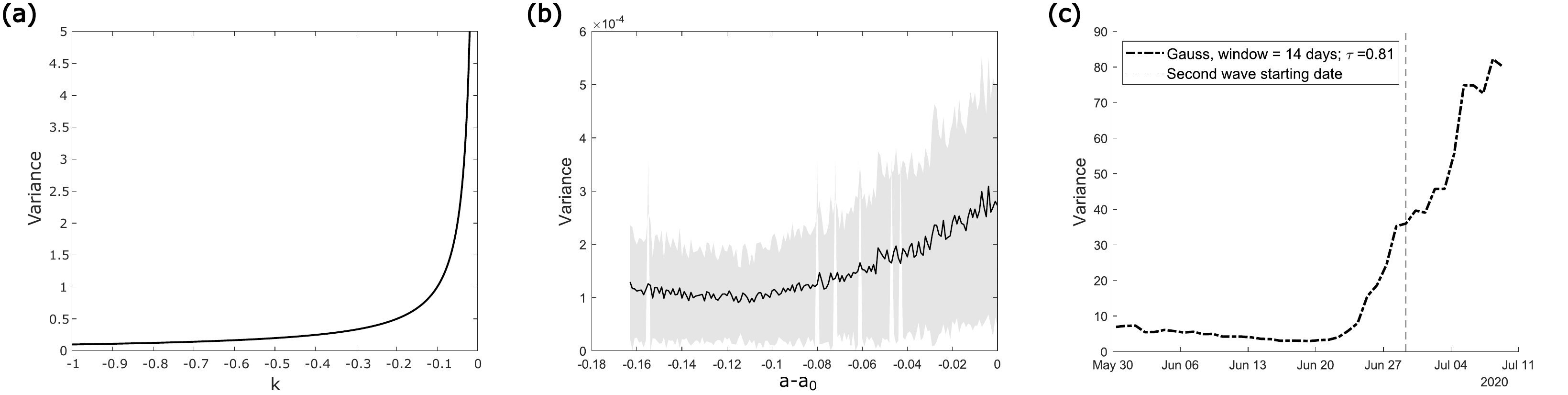}
	\caption{\footnotesize Increasing variance before a tipping point. (a) The theoretical expression of the variance as a function of the distance $k$ from the bifurcation value, shown in Eq. \eqref{eq:variance}, increases when approaching the bifurcation point ($k=0$). (b) The numerically computed expression of the variance as a function of the distance $a-a_0$ from the bifurcation value $a_0$, for the minimal model of gene activation in Eq. \eqref{eq:hill-equation}, increases when approaching the fold bifurcation point at $a_0 \approx 2.63$. The solid line corresponds to the mean over 500 repeated simulations, while the gray ribbon correspond to the mean $\pm 1$ standard deviation. (c) The variance computed from COVID-19 empirical data increases just before the second wave in Luxembourg (whose onset is marked with a dashed vertical line; adapted from \cite{proverbio2022performance}). The data were pre-processed with a Gaussian rolling window of 14 days. The increasing behaviour was assessed with a Kendall's $\tau$ score $>0.8$.}
	\label{fig:indicator_var}
\end{figure}

Stochastic indicators have been applied to case studies in systems biology, including phenomenological models for epithelial-mesenchymal transitions \citep{Sarkar2019}, biochemical reaction networks with feedback \citep{byrd2019critical}, microbiological yeast colonies \citep{dai2012generic}. In epidemiology, they have been tested on SIR- and SIS-like models \citep{brett2017anticipating,ORegan2013}, and on data from epidemics \citep{brett2020dynamical} and from the COVID-19 pandemic \citep{delecroix2023potential,proverbio2021dynamical}. Extensions to ecological systems have been proposed, \eg by \cite{scheffer2015generic}.

Their seemingly generic applicability and their ease of extraction from data supported the use of statistical indicators to gauge resilience loss in various systems. However, rigorous guarantees such as those offered by Propositions~\ref{theo:var}- \ref{theo:entropy} are valid only for stationary parameters and white noise. By leveraging theory by \cite{J.S.Allen2014, Gardiner1985, van1992stochastic}, these results have been partially generalised to state-dependent noise or ramping parameters \citep{proverbio2023systematic}, as well as non-Markovian noise \citep{Kuehn2021a}. However, in realistic scenarios encountered in natural systems \citep{Chen2017}, the results may not hold and thus the applicability of the indicators is hampered. Determining a safe space for operations, as well as finding the indicators that provide the best performance in terms of reliability and sufficient lead time, is still an open challenge. As a first step in that direction, \cite{proverbio2023systematic} developed an optimisation procedure to maximise the signal of resilience loss, and its lead time, in the case of state-dependent noise. Further refinements and extensions, including a general analysis of system structures to determine conditions for the applicability of various indicators, are left for future studies.

\subsection{Robust and structural analysis can frame resilience studies}
\label{sec:stab-for-res}

The resilience of a system is associated with a quantification of the possibility that the system undergoes a regime shift. Hence, assessing resilience loss to characterise the vulnerability of a system to a critical transition is only meaningful if such a critical transition is possible given the dynamics of the system of interest, \ie if the system \emph{can actually undergo a regime shift}, caused by stochastic noise or parametric uncertainties.
Otherwise, if we already know that no critical transition is possible (see, \eg Proposition 2.3 by \cite{Kuehn2013a}), resilience analysis may be neglected in favour of (structural) robustness considerations (see, \eg \cite{zhu2024disentangling}).

In the biological, ecological and epidemiological literature, it is often assumed that the system under study can -- and will -- undergo a critical transition \citep{Chen2012c, Liu2020b, Scheffer2009c}. However, it has been debated whether the burden of proof lies in supporting the hypothesis of unique or of alternative regimes \citep{estes2011trophic}. Alternatively, one may quantify the \emph{probability} of actually hitting a tipping point, given that a warning for resilience loss has been triggered. Under a Bayesian framework, such probability is \citep{boettiger2012early}:
\begin{equation}
    P(\text{CT} | \text{IRL}) = \frac{P(\text{IRL} | \text{CT}) P(\text{CT})}{P(\text{IRL})} \,,
    \label{eq:bayesian}
\end{equation}
where CT and IRL are the events associated with the occurrence of a Critical Transition (or, equivalently, of the system hitting a Tipping Point) and an Indication of Resilience Loss, respectively. According to Bayes' rule \eqref{eq:bayesian}, the quantity of interest $P(\text{CT} | \text{IRL})$ (\ie the probability of a critical transition actually occurring, given that an indication of resilience loss has been detected) depends on $P(\text{IRL} | \text{CT})$ (\ie the probability of detecting an indication of resilience loss given that a critical transition is about to occur), on $P(\text{IRL})$ (\ie the probability that an indication of resilience loss is detected, which is the subject of several studies inquiring the performance of related indicators, see, \eg \cite{Kuehn2021a,proverbio2023systematic}) and on $P(\text{CT})$ (\ie the \emph{a-priori} probability that the system undergoes a critical transition). 

Structural stability analysis and robust stability analysis (\emph{cf.} Chapter~\ref{ch:struct-an}) precisely address the latter point: they rely on qualitative, graph-theoretic information and on systems-and-control-theoretic methodologies to inquire whether alternative regimes (\ie attractors, or stable equilibria), including tipping points and critical transitions amongst them, may exist \emph{a priori} and therefore help rigorously assess $P(\text{CT})$. Structural analysis yields a binary outcome: it can only ensure that a critical transition is structurally impossible (hence, $P(\text{CT}) = 0$), or that a critical transition is structurally unavoidable (hence, $P(\text{CT}) = 1$), whereas it cannot offer a quantitative answer when a critical transition is possible with some probability.
Conversely, a probabilistic framework for robustness analysis, in the spirit of \cite{tempo2013randomized}, would return continuous values for $P(\text{CT}) \in [0, 1]$; see, \eg the recent preliminary work by \cite{sutulovic2024efficient} discussed in Section~\ref{sec:uq-gpc}. Robustness analysis is therefore often a necessary prerequisite to enable the correct interpretation and application of resilience results, and should be a fundamental preliminary step when studying the resilience properties of uncertain systems in nature.

\begin{example}[\textbf{Shallow lake ecosystem}]
Figure~\ref{fig:ecology} shows the reconstructed graph of interactions among species, for a biological system that captures population dynamics in a shallow lake.

\begin{figure}[ht!]
	\centering
	\includegraphics[width=\textwidth]{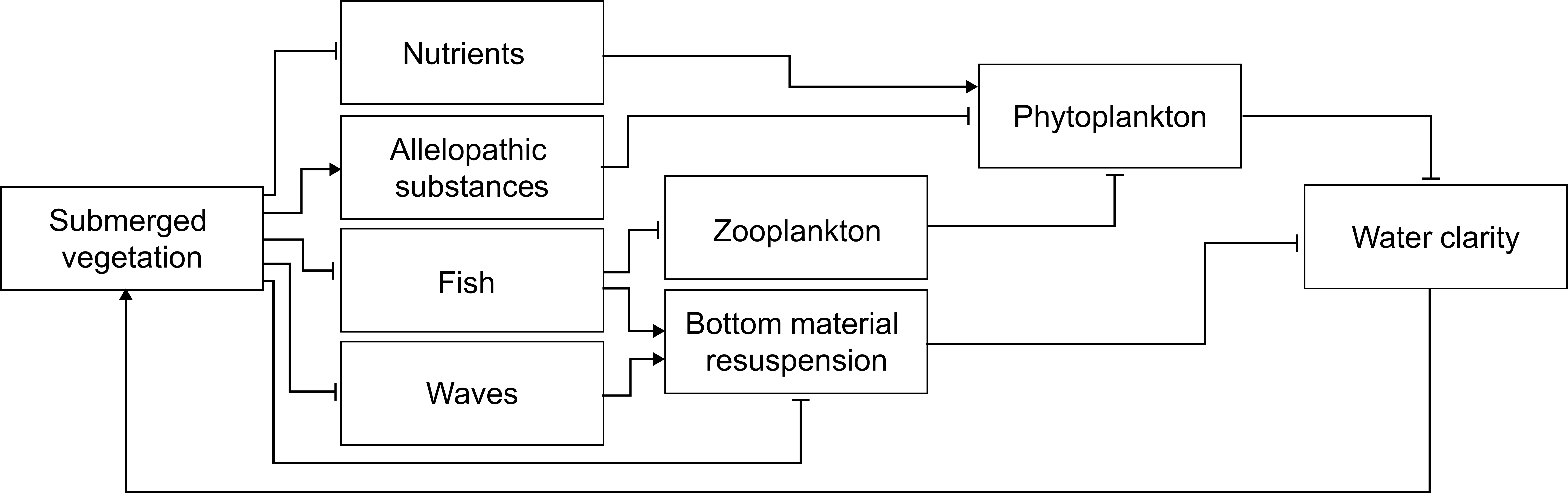}
	\caption{\footnotesize \small The reconstructed signed graph for a shallow lake system, adapted from \cite{Kefi2016}. Hammer-head arrows represent negative effects (\eg depletion, predation), while pointed arrows represent positive effects (\eg enhancement, sustainment).}
	\label{fig:ecology}
\end{figure}

Assuming knowledge of the graph, monotonicity of the transition rates between compartments (with the signs shown in the figure) and nonlinearity of the interaction from Submerged Vegetation (SV) to Water Clarity (WC) \citep{Kefi2016, scheffer2020critical}, structural methodologies from Chapter~\ref{ch:struct-an} can be applied to infer that the system can only exhibit bistability, whenever the equilibrium is destabilised, through a saddle-node bifurcation. In fact, since all the cycles in the graph are positive, the system is structurally a strong candidate bistable system according to the classification by \cite{Blanchini2014structural,Blanchini2015structuralclass,Blanchini2017e}; alternatively, checking that all the loops and paths from SV to WC are positive proves that the system can be bistable as per the results by \cite{Angeli2004}. Indeed, system bistability was empirically observed \citep{scheffer2020critical} and resilience indicators were correctly applied and employed to monitor the state of shallow lakes \citep{scheffer2012anticipating}. 
\end{example}

\begin{example}[\textbf{Brain network}]
The brain network in Figure~\ref{fig:parkinson} reconstructs cortex and basal ganglia interactions and has been suggested to be relevant to the development of Parkinson's disease \citep{Jones2014}. Robustness analysis can reveal whether its alteration may induce a regime shift from the healthy equilibrium to the disease equilibrium through bistability, thus triggering Parkinson's disease due to the bistable switch-like behaviour of the ganglia circuit.

\begin{figure}[ht!]
	\centering
	\includegraphics[width=0.8\textwidth]{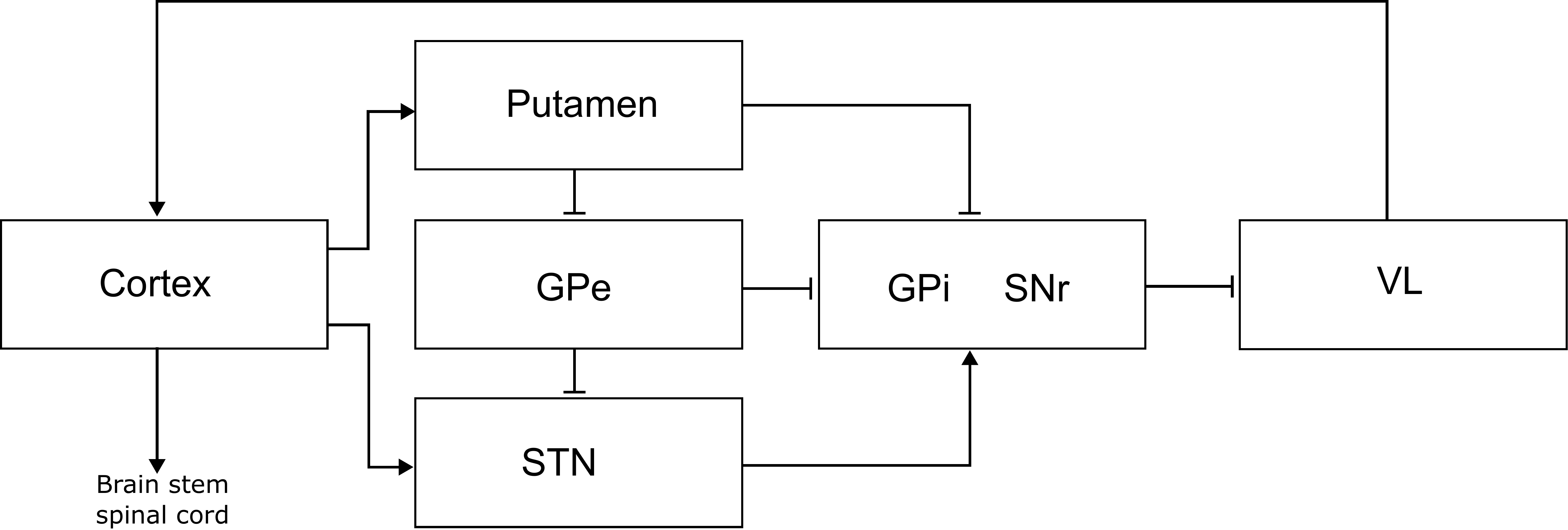}
	\caption{\footnotesize \small Reconstructed signed graph for dopaminergic pathways between the basal ganglia and the cortex (adapted from \cite{Jones2014}). The variables describe the activity of each brain region. Pointed arrows represent excitatory connections $(+)$, and hammer-head arrows inhibitory connections $(-)$. GPe: external segment of the globus pallidus; STN: subthalamic nuclus; GPi: internal segment of the globus pallidus; SNr: substantia nigra pars reticulata; VL: ventrolateral nucleus of the thalamus.}
	\label{fig:parkinson}
\end{figure}

Only one of the four cycles in the graph is positive (Cortex $\to$ Putamen $\to$ GPi SNr $\to$ VL $\to$ Cortex), hence the system cannot be structurally a strong candidate bistable (according to the classification by \cite{Blanchini2014structural,Blanchini2015structuralclass,Blanchini2017e}). Bistability phenomena may arise for specific values of the parameters, \ie the strength of the various interconnections (see, \eg \cite{Cloutier2012}), but oscillatory behaviour cannot be ruled out. For parameter values that do not allow for bistability, the computation of resilience indicators for this system would be misleading and not really informative about potential resilience loss, since structural stability analysis cannot guarantee the emergence of multiple stable regimes. 
\end{example}

The possibility to switch from a healthy equilibrium condition to a pathological equilibrium condition has been studied in two steps, first through the \emph{theoretical} structural or robust assessment of the \emph{possibility} of a bistable behaviour, and then through a \emph{computational} bifurcation analysis, in several other examples in medicine and physiology -- \eg for the onset of diseases such as Mal de Debarquement syndrome \citep{Burlando2022,Mucci2020}, amyotrophic lateral sclerosis \citep{Burlando2020} and fibromyalgia \citep{Demori2022,Demori2024} -- as well as in entomology -- \eg \cite{breda2022deeper} propose a model of honey bee health in the presence of several stressors, including parasites, pathogens and toxic compounds such as pesticides, and show that a bistability mechanism leads to two possible extreme outcomes in terms of survival or extinction of the bee colony.

Structural and robust analysis are thus precious not only for model falsification (see, \eg \cite{Angeli2012,Giordano2018}) but also to enable the correct application and interpretation of resilience indicators for unknown systems \citep{boettiger2012early}.
Going beyond a case-by-case assessment depending on the specific choice of parameters and extending formal results on the expected bifurcation type, based on qualitative and semi-quantitative information about the network structure and the properties of the system's functional expressions, is a challenging topic for future studies. Systematic methods to support informed modelling choices are fundamental in this context, since models (at least surrogate ones) are currently paramount to enable inferring bifurcations from data \citep{scheffer2015generic}. Integrating robustness and resilience would therefore offer a powerful and reliable framework to address open challenges in systems biology and epidemiology.

\subsection{Resilience indicators for complex networks}
\label{sec:ind-on-net}

The results discussed so far were focused on low-dimensional systems or on mean-field approximations of large networks. Here, we review resilience indicators tailored to encompass large-scale complex networks \citep{bianconi2021higher,Hsiao2018}. 

\subsubsection{Leading linearised eigenvalue}
For multidimensional systems, it is possible to look at the leading eigenvalue of the diagonalisation. Consider the system 
\begin{equation}\label{eq:MultiDimSyst}
    \dot{x} = f(x), \quad x\in \mathbb{R}^n
\end{equation}
and let  $\bar{x}\in \mathbb{R}^n$ be an equilibrium point; without loss of generality, $\bar{x} = 0$. We can study the local behaviour of the system by considering its linearisation around the equilibrium point, associated with the Jacobian matrix $J = \frac{\partial f}{\partial x}|_{x=\bar{x}}$. Then, the multivariate Ornstein-Uhlenbeck process corresponding to the linearisation of \eqref{eq:MultiDimSyst} can be written as
\begin{equation}
    d X(t) = J X(t) dt + \sigma d W \, ,
\end{equation}
where $X$ is an $n$-dimensional vector of stochastic processes.
Assuming that the system Jacobian is diagonalisable, we can enforce a coordinate transformation induced by the invertible square matrix $Q$, whose columns are eigenvectors of $J$, so that $X'(t) = Q^{-1}X$, $\sigma' = Q^{-1}\sigma$, and the diagonal matrix $\Lambda=Q^{-1} J Q$ has the eigenvalues of $J$ on its diagonal.
The coordinate transformation leads to the diagonal system
\begin{equation}
    d X'(t) = \Lambda X'(t) dt + \sigma' d W \, .
    \label{eq:multivarOU}
\end{equation}

If the diffusion terms have multivariate \emph{normal} distributions, statistical moments for Eq. \eqref{eq:multivarOU} can be computed exactly.
Multivariate Ornstein-Uhlenbeck processes are well studied and several results are proposed in the literature \citep{J.S.Allen2014,Gardiner1985,meucci2009review, vatiwutipong2019alternative}.
Notably, considering the leading eigenvalue in the multivariate case allows us to obtain results similar to those for the univariate case, and thus enables the use of the resilience indicators discussed in Section~\ref{sec:monomodal_ind}.

For simplicity, consider the case of $J$ with real eigenvalues, which are assumed to be ordered in matrix $\Lambda$ so that $\Lambda_{11} \geq \Lambda_{22} \geq \dots \geq \Lambda_{nn}$, and introduce 
    \begin{equation}\label{eq:CapitSigma}
        \text{vec}(\Sigma'_{\infty}) := [-(\Lambda \oplus \Lambda)]^{-1} \text{vec}(\Sigma') \, ,
    \end{equation}
    where $\Sigma'= \sigma'\sigma'^\top$, $\text{vec}$ is the vectorisation operator (which transforms matrices into column vectors by stacking the matrix columns), and $\oplus$ is the Kronecker sum, defined as $A \oplus B = A \otimes I_b + I_a \otimes B$ (where $a$ and $b$ are the sizes of the square matrices $A$ and $B$, respectively, $I_h$ is the identity matrix of size $h$, and $\otimes$ is the Kronecker product). Due to the properties of the Kronecker sum, $\Lambda \oplus \Lambda$ is diagonal, and it has been shown that only its first element approaches zero when the system gets close to a bifurcation point \citep{Gardiner1985}; its inverse is also diagonal, and its first element approaches infinity. Since
    \begin{equation}
     \text{vec}(\Sigma'_{\infty})   = \begin{bmatrix}
            \text{Var}(X') \\
            \text{Cov}(X'_1, X'_2) \\
            \vdots 
            \end{bmatrix} \, ,
    \end{equation}
where $X'_i$ is the $i$th entry of vector $X'$,
it follows that, as $t\to \infty$, the multivariate variance tends to infinity along the first eigendirection. This property is also conserved when reverting to the original coordinates and studying $\Sigma_\infty$, defined similarly to $\Sigma'_{\infty}$ in \eqref{eq:CapitSigma}. Along the same lines, the auto-covariance of $X'(t)$ can be computed as
    \begin{equation}
        \lim_{t \to \infty} \text{Cov}(X'(t), X'(t + \Delta t)) = e^{\Lambda \Delta t}\Sigma'_{\infty} \, .
    \end{equation}
    
Since matrix $\Lambda$ is diagonal,
    \begin{equation}
        e^{\Lambda \Delta t} = \text{diag}\left\{ e^{\Lambda_{11} \Delta t}, \dots e^{\Lambda_{nn} \Delta t}  \right\} \, .
    \end{equation}
    
The largest eigenvalue $\Lambda_1$ approaches 0 when the system gets close to a bifurcation point \citep{Gardiner1985}, and hence the first element of the auto-covariance matrix approaches $\text{Var}(X'_1)$, while the other elements do not display relevant changes in the vicinity of the bifurcation point. 

Stemming from these results, data-driven methods employing PCA (Principal Component Analysis) or MAF (Min/Max Autcorrelation Factor) have been proposed \citep{Weinans2021} to assess resilience in networks. Alternatively, \cite{morr2023anticipating} have directly gauged the first eigenvalue by interpolating Eq. \eqref{eq:multivarOU} to empirical data.

\subsubsection{Dimension reduction}\label{sec:ind-on-net-dimred}

Several network-based techniques have been suggested to reduce complex large-scale dynamical networks to low-dimensional systems, through aggregation procedures involving either the links or the vertices, and then extract from the low-dimensional model information related to the resilience properties of the original network. Most of these techniques are heuristic, and providing them with solid mathematical foundations is an essential open problem in the field. 

One of the first attempts is due to \cite{Gao2016}, who proposed a weighted dimension reduction on homogeneous networks. The method aims at ``collapsing'' a large system of equations into a scalar system, while retaining the qualitative stability properties of the original multidimensional system. The fundamental idea is to assign weights to each of the nodes, proportional to their degree, and thereby obtain new \emph{effective dynamics} in $\mathbb{R}$. Specifically, consider a system with state vector $x=[x_1\dots x_N]^\top \in \mathbb{R}^N$, whose scalar components have dynamics
\begin{equation}
	\dot{x_i} = F(x_i) + \sum^N_{j=1} A_{ij}G(x_i,x_j) \; \quad t\geq 0,
	\label{eq:multi_dim_eq}
\end{equation}
where $F$ describes the dynamics of each subsystem in isolation, $G$ describes the interaction between connected subsystems, and $A_{ij} \geq 0$ is an entry of the adjacency matrix expressing the strength of the interconnection (if any) between subsystems $i$ and $j$.
As an example, this model can capture a large-scale gene regulation system such as that in Eq. \eqref{eq:gene-regulation} where only activating functions are present. Assume that the out-degree (sum of the weights of the outgoing links) and the in-degree (sum of the weights of the incoming links) are the same for each node. The \emph{effective state} $x_{\text{eff}}$ of the system is obtained by the average nearest neighbours activity
\begin{equation}
	x_{\text{eff}} = \frac{\mathbf{1}^\top A \mathbf{x}}{\mathbf{1}^\top A \mathbf{1}} = \frac{\langle s^{\text{out}}x \rangle }{\langle s \rangle} \; ,
	\label{eq:xeff}
\end{equation}
where $\langle . \rangle$ is the averaging operation, $\mathbf{1}$ is the all-ones vector of the appropriate size, $s^{\text{out}}$ is the out-degree and $\langle s \rangle = \langle s^{\text{out}} \rangle = \langle s^{\text{in}} \rangle$ is the weighted average degree. Also, an \emph{effective coupling parameter} can be computed as
\begin{equation}
	\beta_{\text{eff}} = \frac{\mathbf{1}^\top A s^{in}}{\mathbf{1}^\top A \mathbf{1}} = \frac{\langle s^{\text{out}} s^{\text{in}} \rangle }{\langle s \rangle} \; .
\end{equation}
Then, the $N$-dimensional system in \eqref{eq:multi_dim_eq} can be reduced to a scalar equation 
\begin{equation}
	\dot{x}_{\text{eff}} = F(x_{\text{eff}}) + \beta_{\text{eff}}G(x_{\text{eff}},x_{\text{eff}}) \,,
 \label{eq:1d_gene-reg}
\end{equation}
and one can hope that studying the stability, robustness and resilience properties of \eqref{eq:1d_gene-reg} can shed light on the analogous properties of the original system \eqref{eq:multi_dim_eq}. 
In fact, the bifurcation diagram of $x_{\text{eff}}$ with respect to the bifurcation parameter $\beta_{\text{eff}}$ includes the transition point from one stable equilibrium to another, which is fully determined by the dynamics and not by the network topology. Mapping the network \eqref{eq:multi_dim_eq} into the aggregate scalar equation \eqref{eq:1d_gene-reg} allows us to apply the potential landscape analysis in Section~\ref{sec:potential_land}. Overall, the technique aims at estimating reduced/approximate \emph{resilience functions} \citep{li2021resilience}, which are bifurcation diagrams of bistable systems undergoing saddle-node bifurcations, akin to that in Figure~\ref{fig:res_func}; depending on the visualisation technique, they can also be easily mapped to diagrams such as that in Figure~\ref{fig:gene-reg-phase-portrait} or, by integration, to landscapes such as that in Figure~\ref{fig:pot_land}.  

In case of heterogeneous networks or of other types of resilience functions (hence, of bifurcations), alternative methods can be employed, involving, \eg spectral dimension reduction \citep{Laurence2019, Masuda2022}, decoupling methods \citep{Duan2022, Paul2017} or ad-hoc assumptions on reduction closure \citep{Jiang2018}. This research direction, relatively novel in network science, can fruitfully benefit from cross-interactions with the control and estimation community and its consolidate expertise on bistability detection \citep{Reyes2020} and model order reduction \citep{Cheng2021}, to propose methods that yield formal resilience indicators, accompanied by rigorous theoretical guarantees.

Other relevant approaches for model dimension reduction in the literature exploit parameter space compression \citep{machta2013parameter} and model manifold boundaries \citep{transtrum2014model,pare2019model}, where the latter has also been applied to Michaelis-Menten reaction dynamics \citep{pare2015mmr}.

\subsubsection{Dynamical network biomarkers}

Dynamical network biomarkers couple the theory of leading eigenvalues passing through bifurcation values and the knowledge of biological networks, so as to infer changes in the correlation between relevant nodes that constitute early warning signals for sudden alterations of the system behaviour, including deterioration of complex diseases \citep{Chen2012c}. The method involves the identification of core subnetworks, for which monitoring the evolution of variance correlation would alert for possible impending tipping points. Coupling formal methods and heuristics to derive resilience indicators has enabled important applications in systems biology and medicine \citep{Cohen2022, Yang2018a}, and initiated the study of network-specific indicators for system resilience. Although they were initially designed for distribution-type biological data, extensions to time-series networked data, including epidemiological ones, have been suggested \citep{chen2019detecting, delecroix2023potential}. A recent survey on dynamical network biomarkers is offered by \cite{Aihara2022}.
Other \emph{network biomarkers} have been proposed based on other summary statistic indicators, such as autocorrelation \citep{Mojtahedi2016b} or Jensen-Shannon divergence \citep{yan2021identifying}. Applications range from cell-fate decision, to diseases like prostate cancer, bladder urothelial carcinoma, or viral infections. However, such methods are skewed to data-driven approaches and would drastically benefit from theoretical validation, to improve their interpretability and verify their rigorous application domain.

\subsubsection{Global network-based methods}

The methods described so far rely on the linearisation of a complex network around an equilibrium. However, capturing the intrinsic nonlinearities of systems that experience sudden tipping points can provide further insight and global methods have been proposed, albeit limited to specific cases. \cite{Loppini2019} suggest that loss of low-degree nodes may be an early warning signal of resilience loss for the whole networked system, and employ algorithmic techniques to gauge structure-based resilience (akin to the concept of robustness in network theory).
\cite{masuda2024anticipating} mix various signals from multiple nodes, obtaining a computational average over the complete network. \cite{morr2024internal} discuss the role of observables in multidimensional systems, to guarantee the effective extraction of early warning signals.
\cite{Montanari2022} extend functional observability analysis (assessing whether a functional of the state variables can be reconstructed from measurements) of complex nonlinear systems to link observability, topological features of system attractors and resilience indicators so as to offer early warnings for epileptic seizures; the smallest singular value derived from time-series data is taken as a proxy for the system observability, and is shown to serve as a resilience indicator predicting switches from resting to epileptic brain states. The method still requires extensive validation on multiple time series and other models, but suggests a promising direction for the study of strongly nonlinear phenomena.

\subsection{Resilience indicators in biology and epidemiology}

Indicator-based resilience studies are particularly popular in ecology \citep{bathiany2024resilience,Carpenter2011,jiang2019harnessing,scheffer2012anticipating} and climate science \citep{ben2023uncertainties,Boers2021, boulton2022pronounced, hummel2024inconclusive,ritchie2021overshooting, thompson2011predicting}, where the interested reader may find most applications. However, in recent years they have also been investigated in systems biology and epidemiology, mostly with the aim of anticipating the emergence of undesired regimes. 

In systems biology and medicine, indicator-based resilience methods were used to anticipate tipping-like behaviours in (almost) model-free systems \citep{Korolev2014, Trefois2015a}. Theoretical studies addressing gene regulation and epithelial-mesenchimal transitions showed that indicators can be used to monitor and alert \citep{Sharma2016}, and optimisation protocols have been developed to anticipate biological regime shifts from noisy and distribution-based data \citep{proverbio2023systematic}. Hypotheses have been built around their use in assessing brain dynamical states, \eg to predict epileptic seizures \citep{Maturana2020}, although with varying success \citep{Wilkat2019}. As mentioned, dynamical network biomarkers \citep{Aihara2022} are an active area of research combining indicators and heuristics. Overall, a systematic characterisation of system features and indicator performance, together with formal characterisation and applicability assessment, is needed to advance this promising and important research, and build a complete framework that supports estimation, prediction, targeted interventions and control. 

The interest in epidemiological applications, further propelled by the COVID-19 pandemic, relies on the fact that the transition from a disease-free state to an outbreak, or to endemicity, can be understood as a bifurcation phenomenon (\cite{ORegan2013}; see also the examples provided in Section~\ref{sec:bif-normal-forms}). Theoretical and applied research in this area often aims at enriching monitoring toolboxes for prompt responses in the detection of epidemic outbreaks \citep{brett2017anticipating, brett2020detecting, Horstmeyer2018}. The effectiveness of resilience indicators has been assessed for various types of disease transmission, network topology and data sources \citep{Alonso2019, southall2020prospects}, showing varying success in correctly identifying resilience loss. This issue is primarily due to the relative time scales of contagion and noise type in data, as pointed out by \cite{Dablander2021, proverbio2022performance}; hence, more robust and transferable methods are needed. Prospects and applicability of resilience indicators in epidemiology have been surveyed by \cite{Southall2021}, from a theoretical standpoint, and \cite{delecroix2023potential}, with an application-oriented perspective. They highlight the importance of developing generic indicators that can be applied even with little knowledge about the system, but recognise the potential limitations when the available knowledge about the system is insufficient. Systematic studies that quantify or list the necessary information, including (partial) knowledge of the network structure, are still needed. Leveraging resilience indicators for (epidemic) control is also an uncharted research area \citep{Alibakhshi2025}.

\subsection{Open challenges}

The challenge of developing a coherent framework to investigate resilience, and loss thereof, has recently started being tackled. This quest shows promising avenues to which the expertise of the control community can contribute significantly; as an example, the recent work by \cite{pirani2023network} expands the investigation of critical transitions with notions from the theory of control over networks. The scarce information available about complex systems, as well as the need for generic and data-driven methods, fuelled the development of data-based indicators for resilience loss, which do not require full mechanistic and quantitative models, but are based on a careful assessment of the system underlying structure and robustness, hence pointing at the necessity to integrate the two concepts and methodological paradigms.

In this chapter, we have considered resilience loss driven by parameter variations inducing bifurcations, in the presence of stochastic fluctuations or noise. We have surveyed several approaches, their potential and limitations, focusing in particular on applications to systems biology and epidemiology, and we have pointed at opportunities for cross-pollination between network science and systems theory. Further efforts are still required in order to build a unique systematic framework that tackles system robustness and resilience. 

Moreover, it is still unclear how one can determine which bifurcations best describe the regime shift that may occur in a networked system, given its structure or relevant time-series data. It has been suggested \citep{Bury2020, ditlevsen2007climate, Meisel2015a, proverbio2022classification} that resilience indicators and stochastic indicators based on surrogate models (see Section~\ref{sec:monomodal_ind}) may be extracted from data subsets to validate bifurcation modelling choices. Also in this direction there is vast room for improvement, to systematise model selection and identification. 

Finally, as briefly mentioned in Section~\ref{sec:preliminar_res_ind}, other types of regime shifts may occur, different from noisy bifurcations. In some cases, including noise-induced phenomena, rate-induced instabilities in non-autonomous systems, or structural network modifications, combining robustness and resilience definitions could enable significant advances towards the understanding, prediction and control of complex networked systems.

\newpage
\section{Concluding discussion}
\label{ch:conclusion}

The interplay of mathematical disciplines with biology and epidemiology is unravelling new knowledge, in terms of both theoretical methodologies and algorithms, and their applications to the life sciences.
Complementing empirical insight with mathematical modelling and analysis enables novel scientific discoveries, and allows us to address old and new questions from different perspectives.
Capturing the essence of phenomena using interpretable models, and providing sharper and testable definitions to fundamental concepts, allows us to provide deeper insights into natural principles and enables both theoretical predictions and the design of control strategies to enforce the desired properties and emergent behaviours.

\subsection{Uncharted directions}
\label{ch:unch-dir}

Within the broad field of dynamical systems, this monograph focused on the narrower scope of understanding structural, robust and resilient properties of systems in biology and epidemiology (with additional hints to systems in ecology, neuroscience and medicine).
The concepts we discussed can be easily extended to other disciplines.
In fact, studying the dynamic behaviour of uncertain nonlinear interconnected systems is fundamental to understand and control phenomena not only in the life sciences, but also in climatology, sociology, opinion dynamics, economy, finance, physics and engineering (electronics, transportation, communication networks, embedded systems, multi-agent robotics), as long as they can be cast within the same modelling framework (which is often the case, provided that suitable assumptions hold). An example is the dynamic spreading of diseases, cyberattacks or opinions: predicting the evolution and containing contagion, or fake news, is crucial in all these cases and the respective dynamics are akin. Due to the strong analogies between phenomena in different contexts, analysis and control results can be then immediately applicable, or easily generalisable, to various systems across science and engineering, to gain new knowledge and better understand real phenomena, as well as to devise new paradigms for control, coordination, system design and optimisation.

Many of the topics we have discussed in this monograph offer several open problems for cutting-edge research; we have already outlined in each chapter possible promising directions for future research aimed at expanding our toolkit of mathematical theories and algorithms, including the integration of structural and probabilistic methodologies, and the development of a unified formal theory for resilience definitions and related indicators. We briefly mention here issues that we could not examine more in depth, so as to sketch other prominent research questions in the thriving field revolving around robustness and resilience for dynamical networks in nature.

\subsubsection{Bifurcations and normal forms for control}

Bifurcation normal forms are known mathematical constructs, whose relevance for control is still undermined. After developing systematic taxonomies and exploring the related properties, a general framework is necessary to distinguish and validate them as minimal models for unknown dynamics and perform model selection \citep{Bates2011, Bury2020}. 

When validated and tractable models are available, the associated normal forms can be derived by employing various techniques, such as general center manifold calculation \citep{Guckenheimer2009a,kuehn2021universal} or Taylor expansion around the ``critical'' values $x = x_c$ and $p = p_c$ up to second or third order \citep{strogatz2018nonlinear}: 
\begin{eqnarray*}
\dot{x} &=& f(x,p) = f(x_c, p_c) +(x-x_c) \left. \frac{\partial f(x,p)}{\partial x} \right\rvert _{(x_c,p_c)} \\
    & +& (p-p_c)\left. \frac{\partial f (x,p)}{\partial p} \right\rvert _{(x_c,p_c)} + \frac{1}{2} (x-x_c)^2 \left. \frac{\partial^2 f(x,p)}{\partial x^2} \right\rvert _{(x_c,p_c)} + \dots \, .
\end{eqnarray*}
However, this procedure is not always possible and it is often necessary to hypothesise the bifurcation to decide at which order the expansion should be truncated. Graphical inspection might suggest the expansion order: one can check how the diagram $(f(x,p),x)$ changes when varying $p$, and compare its local properties next to the $x$-axis with those of normal forms. Other approaches rely on spectral analysis of the Jacobian and verification of the conditions defining each normal form, and thus can rarely be applied when models are complex. 

Methods to infer bifurcation dynamics from structural information or from data (of various quality) are thus necessary, and can also foster the design of control strategies for the system. Not only attractors can be considered, but also limit cycles and pseudo-regimes \citep{lemmon2020achieving}, strange attractors \citep{strogatz2018nonlinear} or global bifurcations \citep{izhikevich2000neural}. Deriving performing resilience indicators for these cases is an intriguing challenge, as well as verifying methods for multivariate systems without resorting to one-dimensionalisation through normal forms. 
This interesting research avenue is also suggested by a recent paper focusing on resilience of dynamical systems \citep{Krakovska2024}, which provides an ample classification of resilience indicators.

Finally, several papers in control \citep{bizyaeva2022nonlinear, cathcart2023proactive,leonard2024fast} are building on the concept of complex bifurcations to steer systems towards consensus in a fast and flexible manner. Leveraging data-driven indicators that alert about an approaching bifurcation would provide information about the proximity to consensus-breaking, and would provide useful information to control systems at the edge of critical transitions.

\subsubsection{Dimension reduction and model order reduction}\label{sec:dimred}

The control community has widely studied, also in the context of systems biology, methods for model order reduction \citep{Feliu2012,Rao2013,schaft2015} and decomposition of dynamical systems \citep{Soranzo2012}, aimed at constructing models of lower dimensionality, which are more amenable to analytical study and at the same time retain key properties, such as asymptotic stability, monotonicity, or positivity.

The network community, conversely, pursues the development of dimension reduction methods aimed at preserving and predicting resilience properties and tipping mechanisms \citep{Gao2016,Kundu2022, Laurence2019, Thibeault2020, Tu2021, Vegue2023, Wu2023, Yu2022}. 
As discussed in Section~\ref{sec:ind-on-net}, the goal is to find a reduction of system dimension, based on networks properties such as mean degree, spectral eigenvalues or other network metrics, to obtain a low-dimensional system to which resilience indicators can be more easily applied. Other approaches use a data-driven version of the Central Limit Theorem \citep{Golubitsky2003a} to constrain time series around low-dimensional dynamics, and then study their resilience properties. Potential extensions are further reviewed by \cite{Liu2020b}. 

Combining the approaches for model order reduction and dimension reduction could pave the way for improved system understanding and control, thus allowing an analytical assessment of robustness and resilience (possibly with integrated structural and probabilistic approaches) in real-world applications.

\subsubsection{Higher-order networks}

Beyond classical networks \citep{Albert2002, Barabasi2013, Boccaletti2006, Newman2003}, an array of extensions have been recently proposed to study systems at all scales \citep{Klein2020}. Multilayer networks \citep{Boccaletti2014} describe multiple intertwined modes of contact formation within a set of nodes, with out-layer edges between the same nodes (\emph{multiplex}, see for instance \cite{Battiston2017}) or different ones (\emph{interconnected networks}, see for instance \cite{kinsley2020multilayer}), static or dynamic \citep{bianconi2018multilayer, Nicosia2013}, homeogenous or weighted \citep{Menichetti2014}. Higher-order networks, as well as simplicial complexes and triadic interactions \citep{Battiston2021, bianconi2021higher, Bick2023, cisneros2021multigroup, Majhi2022, muhammad2006control}, describe group interactions that are more complex than pair-wise ones \citep{Battiston2020}. These models are increasingly used to model biological regulation \citep{salnikov2018simplicial} and epidemic spreading phenomena \citep{iacopini2019simplicial,li2021contagion}. 

Unexpected and explosive transitions have been recently observed in higher-order networks \citep{millan2020explosive}. Future research may investigate their robustness and resilience, as well as derive indicators to efficiently and systematically assess such properties. Such a research endeavour requires refining and developing methods to identify higher-order networks in biological systems \citep{Lotito2022}, to derive dimension reduction techniques \citep{Ghosh2023} and to solve structural stability problems, as well as systematic studies on their efficacy with limited available information. Assessing the robustness of higher-order networks with multiple spreading phenomena would be fundamental to advance our understanding of coupled diffusion, \eg in the context of epidemics evolving together with information spreading \citep{Wang2022}. 
 
\subsubsection{Combining analytical insights and data-driven methods} 

Data-driven and machine learning approaches are starting to be used to learn resilience properties from data \citep{brett2020detecting, Bury2021, chakraborty2024early, grassia2021machine, huang2024deep}.
The limitations of machine learning-based techniques, due to data scarcity and poor interpretability, could be circumvented by coupling them with models based on bifurcations, to create training sets while going beyond the linearisation of normal forms. However, preliminary studies are currently restricted to easy cases, and aimed at anticipating critical transitions rather than fully characterising the system's resilience properties. 
Despite the growing interest, abundant, curated and systematised data on regime shifts are still missing, and the low agreement between resilience definitions (see Chapter~\ref{ch:rob-res-sta-mod}) hinders the development of more complex algorithms across fields. While this monograph paves the way to addressing the issue related to resilience definitions, developing datasets and validating algorithms is a long-term endeavour. Data-driven and model-driven approaches are expected to complement and support each other, towards a better understanding and assessment of the robustness and resilience of dynamical systems.

\subsubsection{Model validation and model falsification}

Qualitative and semi-quantitative information from structural analysis is key for model validation or falsification \citep{Anderson2009,Angeli2012,Bates2011,Porreca2012}, by comparing structural predictions to experimental results. The fundamental role of such a preliminary analysis for the subsequent application of resilience indicators has been discussed in Section~\ref{sec:stab-for-res} and can be further developed. Moreover, structural approaches help predict necessary outcomes, stemming from the system structure, or rule out incompatible behaviours, to be tested against experimental results. If behaviours that are structurally incompatible with the considered model are nonetheless observed in the data, then the model is \emph{structurally unsuitable} and deserves radical amendment. Otherwise, the model survives the checkpoint and can be further tested using additional methods. 
For example, \cite{Salerno2013} used robustness analysis to validate a model of galactose uptake genes for \emph{Saccharomyces cerevisiae}, which has been shown to exhibit structural bistability \citep{Cosentino2012b}; another example of tests on regulatory networks is provided by \cite{Batt2005}, while epidemic models are discussed \eg by \cite{Pare2020analysis}.
\cite{Angeli2012} propose an approach to invalidate biological models using monotone systems theory; conversely, \cite{Ascensao2016} provide examples of model falsification based on non-monotonicity.
Analogous techniques can be used to compare ``competitive'' alternative models and perform counterfactual analysis, as suggested by \cite{Anderson2009,Bates2011,Giordano2018,Hamadeh2011,Polynikis2009,Porreca2012}.

\subsubsection{Control}

Modelling systems in biology and epidemiology as dynamical networks allows us enforce a desired global property through targeted local actions inspired by strategies developed in engineering, such as pinning control (controlling a whole networked system through strong local feedback control actions that only affect some of the nodes, see \cite{WangSu2014} and also \cite{BDG2021}) and network-decentralised control (controlling a whole networked system through local control actions that affect the interactions among the nodes, for instance by deciding the flows associated with the links, see \cite{IftarDavison2002} and also \cite{Blanchini2013,Blanchini2015b,Blanchini2016compartmental,Giordano2016smallest}), whose concept is visualised in Figure~\ref{fig:pinn_dec}.
In this respect, a fundamental first step is to identify the dominant nodes (or links) that can be controlled in order to govern the behaviour of the whole system; integrated structural and probabilistic approaches for control and decision could be applied to whole families of uncertain networked systems, targeting the crucial parameters and key network motifs identified in the analysis phase.

\begin{figure}[ht!]
	\centering
		\includegraphics[width=0.4\textwidth]{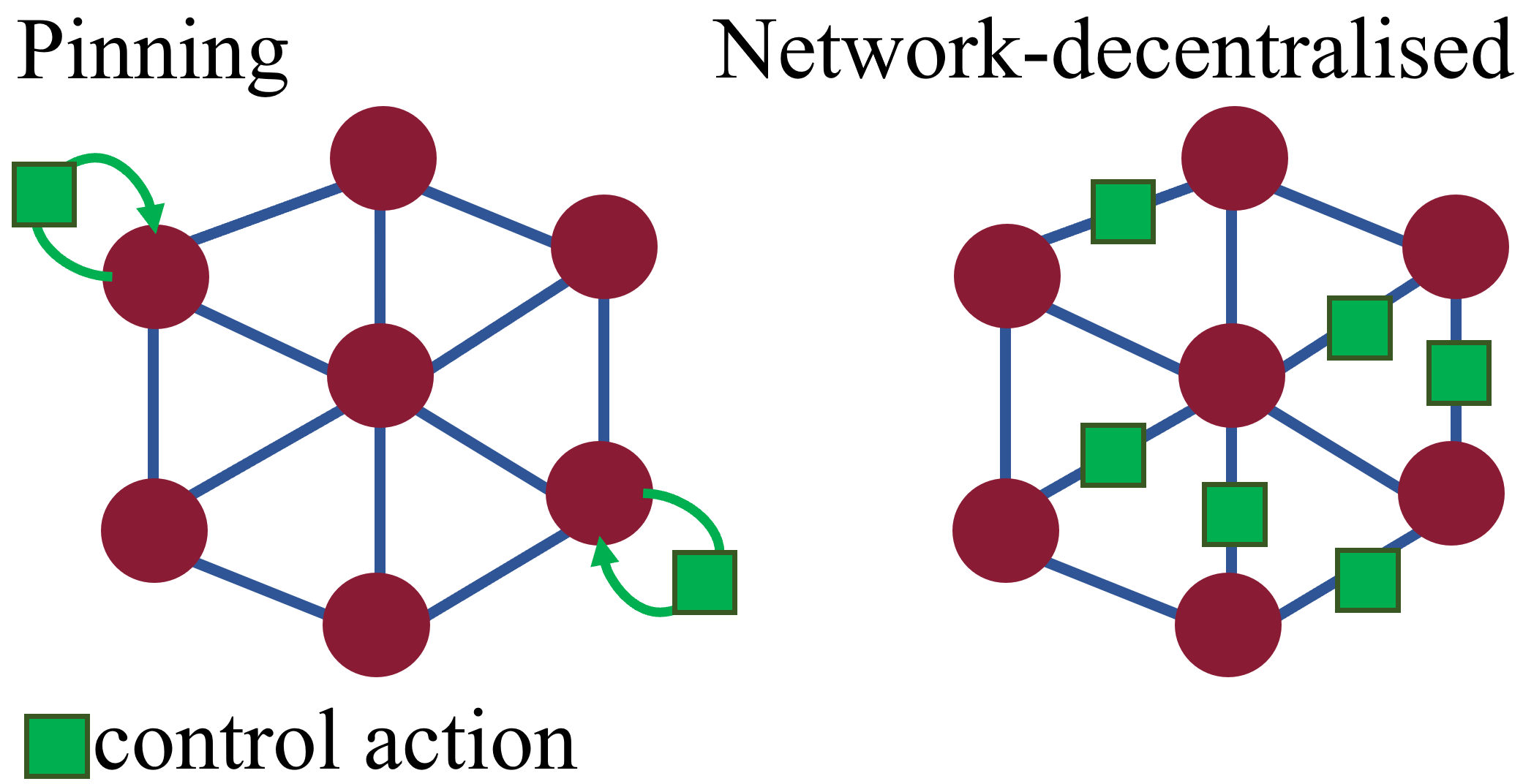}
	\caption{\footnotesize \small Pinning control and network-decentralised control strategies. }
	\label{fig:pinn_dec}
\end{figure}

Pinning and decentralised strategies are particularly relevant for natural systems, where a different paradigm is needed with respect to the classic sensor-controller-actuator feedback paradigm employed in engineering: for systems in biology and epidemiology, control decisions often need to be made based on limited information (sensors may be impossible to place, and only some specific quantities can be measured or observed) and control actuation is strongly limited in nature. Hence, only some specific types of actions are possible.
Control and decision making is particularly crucial in the case of epidemics and infectious diseases \citep{Alamo2021,Donofrio2023,HernandezVargas2022,Nowzari2016,Rowthorn2009,Sélley2015,Sharomi2015,Zaric2002}, and
pinning and network-decentralised control strategies are well-suited to represent the types of actions that can be enforced for epidemic control: some affect the individual hosts (nodes), through pharmaceutical interventions such as vaccination, therapies, treatment \citep{CalaCampana2024,Giordano2021,Hansen2010,Hernandez-Vargas2019,Perez2022} and others affect the interactions (links), through non-pharmaceutical interventions such as use of personal protective equipment that reduces pathogen transmission, physical distancing, travel restrictions, lockdown \citep{Angulo2021,Bin2021,Giordano2020,Gumel2004,Kantner2020,Köhler2020,Morato2020,Sontag2021}. Awareness of the structural robustness of the system can thus support decision-making strategies for the multi-pronged control of epidemics.
In this context, the availability of huge amounts of data (although often difficult to handle, see \cite{Alamo2021,Alamo2022}) and the inherent stochasticity related, \eg to waning immunity
and pathogen mutations suggest the importance of leveraging probabilistic information and approaches and
embedding them into structure-based control approaches.
For instance, for systems admitting a $BDC$-decomposition, such control strategies can be structurally designed based on
control polyhedral Lyapunov functions \citep{BDG2021}.

A promising research direction is to build an integrated structural and probabilistic paradigm for control and decision making, which exploits the insight obtained with the analysis approaches discussed in this monograph to provide rigorous theoretical guarantees in probability. Such approaches would also enforce properties that go beyond stability and aim at some form of optimality, as guaranteed by optimisation-based and model-predictive control approaches \citep{Alamo2021,Hansen2010,Köhler2020,Morato2020}. For instance, multi-pronged epidemic control requires multi-objective optimisation with possibly conflicting goals and constraints, related to public health, social and mental well-being, costs and economic factors, limited available resources; embedding all such requirements within a unified (holistic) formal framework would allow for precise and reliable interventions.

In a control theoretic context, and in particular in subfields of control theory such as cybernetics \citep{ishii2022security}, system resilience is being increasingly investigated, also building on game-theoretic formulations \citep{chen2019game,zhu2015game}. The formalism of game theory has not been addressed in this monograph; we just note in passing that a promising direction for future investigations is to delve deeper into the contact points between game-theoretic approaches and the framework presented in this monograph, so as to develop a more comprehensive quantitative theory of robustness and resilience.

\subsubsection{Towards real-world applications}

We finally recall the challenges of connecting theoretical results with real-world applications \citep{Bussell2019, Carli2020}, which require user-friendly formulations, reliable and consistent methods, and performing metrics. Models and analytical methods ought to offer clear insight and encompass assumptions that are consistent with real-world scenarios, so as to overcome the challenges of actual and effective implementation \citep{Alamo2022,Hansen2010, samad2020industry}. In spite of the analogies between the dynamics and the structure of models describing mechanisms in engineering and in nature, biological and epidemiological systems are inherently different from engineered systems, in that their design principles were not human-made and are not mechanistically known \emph{a priori}.
In addition, their interplay with other open natural and human-made system (\eg \cite{burzynski2021covid,iqbal2020effects, Lep2020, moran2016epidemic,Proverbio2025,PTG2025,zhou2020covid}) makes it more challenging to disentangle effects and design optimal strategies. Network models are suitable candidates for such tasks, but research in this direction is still at its infancy \citep{Fan2022, Fontaine2011}.
Importantly, models of different nature can be combined to provide independent, robust and reliable results. In order to be relevant for real applications, they should be tuned and parametrised according to available data, to make sense of what is possible to observe and collect \citep{ashwin2025early,chapman2025quantifying,liu2018measurability,Roosa2019,sparks2013challenges}; in this way, theoretical advancements can be fruitfully coupled with real-world monitoring capacity, and thus enable effective interventions.

\subsection{Conclusion}

Biomolecular systems composed of nonlinear interconnected dynamic entities are commonly encountered in cell biology, and the synthesis of biomolecular circuits with complex desired behaviours is a sought-after objective: it is thus fundamental to better understand and engineer the building blocks of life, by streamlining the analysis and synthesis of biomolecular feedback systems, and tightly controlling their function in spite of uncertainty and variability. The analysis of physiological models is crucial to elucidate the onset of diseases, and hence identify therapeutic targets; the achieved insight can pave the way for smart drugs, innovative therapies and biotechnologies to improve human health and quality of life. Studying how infectious diseases spread in an interacting population, possibly including for each individual the dynamics of pathogens and immune system, helps predict the emergence and time evolution of contagion, and contrast the epidemic spread by combining non-pharmaceutical interventions and optimal vaccination plans.

To gain rigorous theoretical insight in spite of huge uncertainties, parameter-free structural methods \citep{blanchini2021structural} have been developed for systems in the life sciences, to predict the parameter-independent, \emph{structural} properties of whole families of uncertain systems, exclusively based on the analysis of their \emph{structure}.
However, although structural methods can yield very strong conclusions, their ``yes'' versus ``no'' qualitative outcome leaves us clueless when the property under investigation does not hold structurally: does it hold for some systems in the family? If so, with which probability? Or for which ones, and what do they have in common? Probabilistic approaches \citep{Proverbio2024,tempo2013randomized,Sagar1998,Sagar2011}, including those based on uncertainty quantification and generalised polynomial chaos \citep{fisher2011optimal,kim2013generalised,pepper2019multiscale,sutulovic2024efficient,wan2006multi,xiu2002wiener,zou2023uncertainty}, can offer flexible and computationally efficient answers to these questions, allowing us to explore a whole quantitative spectrum of outcomes, gain both qualitative and quantitative insight, and reveal the sources of indeterminacy for non-structural behaviours.

Complementing the investigation of the robustness of systems to parametric uncertainties in their dynamics with a formal analysis of the preservation of properties in the face of stochastic external perturbations enables a rigorous and thorough assessment of the resilience of systems.
A combined study of the resilience and the robustness of complex systems in nature unveils their conserved mechanisms, identifies core evolutionary strategies rooted in their structure and enables targeted interventions to enforce desired properties. It gives us precious qualitative (and quantitative) information about highly uncertain models, when the underlying mechanisms are not fully known, and when complete information about the dynamics, parameters and disturbances is missing.

In this monograph, we have offered a broad and thorough overview of different types of approaches to assess robustness and resilience in dynamical networks, which are focused on challenges in the life sciences (in particular, biology and epidemiology), but have the strong potential to benefit applications in several other disciplines, and society at large \citep{helbing2015saving}. Through close collaboration with experimentalists and other theoreticians, the systems-and-control community can provide a fundamental contribution to the development of new theories, frameworks and applications that allow us to understand, predict and control natural phenomena.

\newpage
\appendix

\section*{Appendices}
\section{Mathematical background and definitions}

\subsection{Dynamical systems and ODEs}\label{App:A0}

Here, we summarise fundamental definitions and results related to dynamical systems and ordinary differential equations (ODEs); we refer to \cite{Berglund2006,dobrushkin2014applied,Golubitsky2003a, kuznetsov2013elements, strogatz2018nonlinear} and  for a more thorough exposition.

\begin{definition}[\textbf{Dynamical system}]\label{Def:DynSys}
Let $\mathcal{T}$ denote one of the additive topological groups $\mathbb{R}$ or $\mathbb{Z}$, and let  $\Omega$ be a topological space, referred to as the state space, or \emph{phase} space. A \emph{dynamical system} is the triple $\mathcal{D}=\left(\Omega, \left\{\phi_{t,s}\right\}_{t,s\in \mathcal{T}}, \mathcal{T}\right)$, where for all $t,s\in \mathcal{T}$,  $\phi_{t,s}\colon\Omega \to \Omega$ are continuous maps such that, for all $x\in \Omega$ and $u,t,s\in \mathcal{T}$, (i) $\phi_{t,t}(x) = x$, and (ii) $\phi_{u,s}(x) = \phi_{u,t} (\phi_{t,s}(x))$.
A dynamical system is \emph{homogeneous} if, for all $u,t\in \mathcal{T}$, it holds that $\phi_{u,t}=\phi_{u-t,0}$; then, we can consider the one-parameter family of maps $\phi_t := \phi_{t,0}$.
\end{definition}

\begin{definition}[\textbf{Trajectory}]\label{Def:Trajectory}
Given the homogeneous dynamical system $\mathcal{D}$, an \emph{orbit}, or \emph{trajectory}, corresponding to the initial condition $x_0\in \Omega$ is the image of $\phi_{\cdot}(x_0) \colon \mathcal{T}\to \Omega$, namely,
$\text{Or}(x_0) = \{\phi_t(x_0)  \colon t \in \mathcal{T} \}$.
\end{definition}
Definition~\ref{Def:DynSys} implies that the relation $x\sim y$, $x,y\in \Omega$ $\iff$ $x\in \text{Or}(y)$ is an equivalence relation, and motivates the following definition.
\begin{definition}[\textbf{Phase portrait}]\label{Def:phaseport}
Let $\mathcal{D}$ be a homogeneous dynamical system. The \emph{phase portrait} of  $\mathcal{D}$ is the quotient set $\Omega \big /_{\sim}$, corresponding to the relation $\sim$.
\end{definition}

The main class of dynamical systems we consider are those induced by ODEs of the form
\begin{equation}
		\begin{aligned}
			& \dot{x}(t) = f(t,x(t)),\ \ x(t_0) = x_0, \quad t\geq 0,
			\label{eq:ode_general}
		\end{aligned}
	\end{equation}
where $(t_0,x_0)\in U$ and  $f: U \to \mathbb{R}^{n}$ is sufficiently smooth. Throughout this monograph, we assume that $U=\mathbb{R}\times U'$, with $U'\subseteq \mathbb{R}^n$, and that for every $(t_0,x_0)\in U$ we have that the interval of existence and uniqueness of the solution is $I_{(t_0,x_0)}=\mathbb{R}$ (\ie the solutions are extensible to all times).

The following standard existence and uniqueness theorem for \eqref{eq:ode_general} guarantees that an ODE induces a (local) dynamical system.
\begin{theorem}[Existence and uniqueness] \label{th:ODE}
Let $U=\mathbb{R}\times U'$, with $U'\subseteq \mathbb{R}^n$, and consider the system of ordinary differential equations \eqref{eq:ode_general}, where $(t_0,x_0)\in U$ and  $f: U \to \mathbb{R}^{n}$ is continuous. If $f$ is locally Lipschitz in the second argument, uniformly with respect to the first, then there exists some open interval $t_0\in I_{(t_0,x_0)}$ and a unique local solution $x(\cdot;t_0,x_0)\in C^1(I_{(t_0,x_0)})$ for the ODE \eqref{eq:ode_general}.
\end{theorem}

We focus on \emph{autonomous} ODEs \eqref{eq:ode_general}, where $f$ does not depend on time, \ie $f = f(x)$. Under these assumptions, we associate a dynamical system with the ODEs \eqref{eq:ode_general} by defining, for all $x_0\in U'$ and $t\in \mathbb{R}$, the flow
\begin{equation*}
\phi_t(x_0):=x(t;t_0=0,x_0)
\end{equation*}
with the corresponding orbits obtained from Definition~\ref{Def:Trajectory}. 

\subsubsection{Equilibria, cycles, invariant sets and stability}
\label{app:equilibria}
Consider an autonomous ODE system \eqref{eq:ode_general} (\ie $f$ therein is independent of $t$) with $(t,x)\in \mathbb{R}\times U'\subseteq \mathbb{R}^{n+1}$.

\begin{definition}[\textbf{Equilibrium}]
	A point $\hat{x} \in U'$ is called an \emph{equilibrium point}, or fixed point, or steady state, if $f(\hat{x}) = 0$.
\end{definition}

\begin{remark}[\textbf{Equilibrium versus steady-state}]
In the spirit of bridging these concepts to other disciplines like statistical physics, we point out that, when referring to ``equilibrium'', we use here the notion of ``steady-state'', given by the no-change condition $\frac{dx}{dt} = 0$ in the associated ODE. This concept is different from the mechanical definition of equilibrium as balance of forces acting on the system, $\sum_{i} F_i = 0$, or the thermodynamic requirement of absence of net flux of energy through the system, $\sum_{i} J_i = 0$ (adiabatic regime).   
\end{remark}

In addition to equilibrium points, periodic trajectories (cycles) are also of immense interest in systems and control.

\begin{definition}[\textbf{Cycle}]\label{def:Cycle}
The orbit $\text{Or}(x_0) = \left\{x(t;x_0) :  t\in \mathbb{R}  \right\}$ is a \emph{cycle} of a dynamical system if it a periodic trajectory, namely, $\text{Or}(x_0)$  is not a singleton and there exists some $\tau>0$ such that   $x(\tau;x_0) =   x_0$. The smallest $\tau > 0$ with this property is called the \emph{period} of the cycle/trajectory. If there exists $\delta>0$ such that in the $\delta$-neighbourhood of $\text{Or}(x_0)$ there are no other cycles, and there exists at least one $x_1\notin \text{Or}(x_0)$ such that $x(t; x_1)$ tends to $\text{Or}(x_0)$ as $t\to \infty$, then we say that $\text{Or}(x_0)$  is a \emph{limit cycle}.
\end{definition}

\begin{definition}[\textbf{Invariant set}]
A set $\mathcal{S}\subseteq \mathbb{R}^n$ is \emph{forward invariant}, or \emph{positively invariant}, for the dynamical system if $x_0 \in \mathcal{S}$ implies $ x(t;x_0) \in \mathcal{S}$ for all $t\geq 0$.
\end{definition}

\begin{definition}[\textbf{Stability}]
An invariant set $\mathcal{S}$ is \emph{Lyapunov stable} if, for any neighbourhood $X$, with $\mathcal{S} \subset X $, there exist a neighbourhood $V$, with $ \mathcal{S} \subset V$, such that $x(t;x_0) \in X$ for all $x_0 \in V$ and all $t>0$.  $\mathcal{S}$ is \emph{asymptotically stable} if, in addition, there exists a neighbourhood $W$, with $ \mathcal{S}\subset W$, such that
$ \lim_{t\to \infty} \operatorname{dist}\left(x(t;x_0),\mathcal{S}\right)=0$ for all $x_0 \in W$.
 \end{definition}
For autonomous ODE systems \eqref{eq:ode_general} with absolutely continuous right-hand side, the following sufficient  condition for stability is known as Lyapunov's indirect method.
\begin{theorem}
An equilibrium point $\hat{x}$ is asymptotically stable if all eigenvalues of the Jacobian matrix $\frac{\partial f}{\partial x}(\hat{x})$ have a negative real part.
	\label{theorem:linear_stability}
\end{theorem}

\subsubsection{Attractors and regimes}
\label{sec:attractor}
\begin{definition}[\textbf{Attractor and basin of attraction}]
Let $A$ be a stable invariant set and let $B(A)$ be the set of all $c\in U'$ such that 
\begin{equation*}
 \lim_{t\to \infty} \operatorname{dist}\left(x(t;c),A\right)=0,
\end{equation*}
namely, the set of all initial conditions for which the trajectory emanating from the initial condition converges to the set $A$. If $B(A)\neq \emptyset$, we call $A$ an \emph{attractor} and $B(A)$ its corresponding \emph{basin of attraction}. 
\end{definition}
In addition to model-based techniques (analytical or computational) for finding stable invariant sets \citep{BM2015}, time-series driven methods \citep{Montanari2022} to reconstruct attractors have been developed building on Takens' theorem \citep{Takens1981}, that identifies conditions under which an attractor of a discrete-time system can be reconstructed through a delay embedding.   Other methods, using the concepts of quasi-potential, were mentioned in Section~\ref{sec:potential_land}. \cite{sutulovic2025differentiator} recently proposed a method for efficient and faithful reconstruction of dynamical attractors from time series based on the homogeneous differentiator introduced by \cite{levant2003higher,levant2017sliding,levant2020robust,hanan2021low}.

\begin{definition}[\textbf{Hyperbolic and elliptic equilibria}]
Consider an autonomous system \eqref{eq:ode_general} and let $f$ be continuously differentiable. An equilibrium point $\hat{x}$ is \emph{hyperbolic} (or generic) if $\frac{\partial f}{\partial x}(\hat{x}) $ has no eigenvalues on the imaginary axis. Otherwise, the equilibrium is \emph{elliptic}.
\end{definition}
The behaviour of a dynamical system in a domain near a hyperbolic equilibrium point is qualitatively the same as the behaviour of its linearisation near this equilibrium point (Hartman-Grobman theorem; see, \eg \cite{wiggins2003equilibrium}). This fact justifies the use of linearised Langevin equations to derive indicators for resilience loss in Chapter~\ref{ch:res-ind}.

\subsection{Slow-fast systems on different time scales}\label{sec:slowfast}

We often consider ODE systems that depend on uncertain parameters $p\in \mathbb{R}^m$. In natural systems, such parameters (which are often the expression of modelling choices associated with functional dependencies among variables) may be time-varying and evolve on a different time scale with respect to the dynamics of the state $x(t)\in \mathbb{R}^n$. In the literature, noise has been approximated by rapidly-varying deterministic variables \citep{proverbio2023systematic}, employing a hierarchy of time scales. Here, we briefly recall the notion of ODEs on multiple time scales \citep{DelVecchio2016}, often written in the form
\begin{equation}
\begin{aligned}
	\epsilon \frac{dx}{dt} & = f(x,y),\qquad 
	\frac{dy}{dt} = g(x,y) ,
	\label{eq:slowfast}
\end{aligned} 
\end{equation}
where $\epsilon>0$ is a small parameter. The first equation captures fast dynamics, whereas the second captures slow dynamics. System \eqref{eq:slowfast} can be re-written in terms of the fast timescale $s=\frac{t}{\epsilon}$, $t>0$, as
\begin{equation}
\begin{aligned}
	\frac{dx^s}{ds} & = f(x^s,y^s),\qquad 
	\frac{dy^s}{ds} &= \epsilon g(x^s,y^s) \, ,
\end{aligned} 
\end{equation}
where $x^s(s) = x(\epsilon s)$ and $y^s(s) = y(\epsilon s)$. Such timescale separation is motivated by the observation of real-life phenomena evolving on significantly different time scales. Examples include fast promoter activation versus slow protein degradation, or fast neural spiking versus slow signalling control \citep{drion2018cellular}, or even viral dynamics versus population-wide epidemic spreading of influenza \citep{Yan2016} and replicator-mutator dynamics in large populations of individuals \citep{Dey2018}. Different methods to verify timescale separation in models exist, including spectral theory and large deviation theory \citep{Berglund2006}. Empirically, one normally considers the rates (or times) in which a process and its coupled processes occur and compares, \eg their orders of magnitude. For instance, in chemical neuroscience, it is known that the action of ion channels across neural membranes happens at a different (slower) timescale than neural firing \citep{Drion2012b, franci2012organizing}, and this slow-fast activation at the feedback level can be exploited to develop refined models for neural dynamics \citep{franci2014modeling,franci2016three,franci2018robust}. In these cases, the slow variable can be interpreted as a quasi-steady-state parameter, and methods from singularity theory can be readily applied (see, \eg \cite{franci2014modeling} and Chapter~\ref{ch:res-ind}). Additional methods to extract slow-fast information from data are currently being developed also in other fields such as climatology; see, \eg \cite{chekroun2021stochastic}.

For a rigorous theoretical treatment of slow-fast dynamics, the reader is referred to the foundational work on singular perturbation theory, also known as Tikhonov theory \citep{tikhonov1948dependence,tikhonov1950systems,tikhonov1952systems,wasow2018asymptotic}. For a modern treatment of quasi-steady-state approximation (QSSA) and double timescale approximations, a relevant reference is provided by \cite{bersani2020uniform}.

\subsection{From ODEs to SDEs}\label{sec:ODESDE}

Stochastic processes model random phenomena and can be leveraged to introduce parametric uncertainty into a dynamical system. Notably, stochastic differential equations (SDEs) have been related to slow-fast systems (see Appendix~\ref{sec:slowfast}) via the method of averaging  \citep{khas1966limit}. Furthermore, the numerical simulation of stochastic processes via highly oscillatory dynamics (\ie dynamics of fast variables that evolve on much shorter time scales than state variables)  has been well studied in the literature \citep{kasdin1995discrete}. Moreover, stochastic averaging was proven to be equivalent to stochastic normal forms (\cite{namachchivaya1990equivalence}; see also Section~\ref{sec:bif-normal-forms}). Such normal forms are thus widely employed first approximations \citep{Berglund2006} to model fast fluctuations on top of slowly varying state variables.
Here, we recall some of the basic definitions concerning probability theory and stochastic systems; for a more thorough discussion, see, \eg \cite{J.S.Allen2014,Berglund2006,kuehn2011mathematical,papoulis2002probability}. 

Henceforth we let $\left(\Omega, \mathcal{A},\mathbb{P} \right)$ be a probability space, where the sample space $\Omega$ is the set of all possible outcomes of the random process, the event space $\mathcal{A}$ is a $\sigma$-algebra on $\Omega$ (\ie a family of subsets of $\Omega$, $\mathcal{A} \subseteq 2^\Omega$) and $\mathbb{P}$ is a probability measure. We further denote by $\mathcal{B}(\mathbb{R}^k)$ the Borel $\sigma$-algebra on $\mathbb{R}^k,k \geq 1$.
\begin{definition}[\textbf{Random variable}]
Given $k\geq 1$, a $k$-dimensional \emph{random variable}, or \emph{stochastic variable}, is a measurable function $\check{x} \colon \Omega \to \mathbb{R}^k$; \ie the inverse image of any Borel set in $\mathbb{R}^k$ must be an event, $\check{x}^{-1}(B)\in \mathcal{A}$ for all $B\in \mathcal{B}(\mathbb{R}^k)$.
\end{definition}

\begin{definition}[\textbf{Cumulative distribution function}]
Let $\check{x}$ be a random variable. Consider the pushforward measure $\mathbb{P}\circ \check{x}^{-1}$ on $(\mathbb{R},\mathcal{B}(\mathbb{R}))$ . The \emph{cumulative distribution function} (CDF) of $\check{x}$ is the function $F_{\check{x}} \colon \mathbb{R}\to [0,1]$ defined via
$F_{\check{x}}(z) = \mathbb{P}\circ \check{x}^{-1}\left((-\infty,z] \right) = \mathbb{P}\left(\left\{\omega \in \Omega \ \colon \ \check{x}(\omega)\leq z \right\} \right)$.
The standard properties of $F_{\check{x}}$ are: 
(i) $F_{\check{x}}$  is right-continuous and non-decreasing, and (ii) $\lim_{z\to -\infty}F_{\check{x}}(z)=0$ and $\lim_{z\to \infty}F_{\check{x}}(z) = 1$.
\end{definition}

We consider random variables $\check{x}$ for which $F_{\check{x}}$ is absolutely continuous, thereby giving rise to a probability density function. 
\begin{definition}[\textbf{Probability density function}]
	The \emph{probability density function} (PDF) $f_{\check{x}}:\mathbb{R}\to \mathbb{R}$ is a Lebesgue integrable function defined as 
	\begin{equation}
		F_{\check{x}}(z) = \int_{-\infty}^z f_{\check{x}}(t) \mathrm{d}t.
	\end{equation}
We often denote $f_{\check{x}}(t) =\frac{\mathrm{d}F_{\check{x}}}{\mathrm{d}z}(t),\ t\in \mathbb{R}$. Note that $f_{\check{x}}$ is well defined and non-negative almost everywhere.
\end{definition}

Let us fix $k=1$ and refer to $\check{x}$ as a \emph{real} random variable. Henceforth, we proceed under this assumption.
\begin{definition}[\textbf{Stochastic process and sample path}]
	A \emph{stochastic process} is a family of random variables $\left\{\check{x}(t,\cdot) \right\}_{t\in \mathfrak{T}}$, indexed by some set $\mathfrak{T}$. Given $\omega \in \Omega$, we will call $\left\{ \check{x}(t,\omega)\right\}_{t\in \mathfrak{T}}$ a \emph{sample path}.
\end{definition}
We consider $\mathfrak{T}\subseteq \mathbb{R}$, whence we refer to $t$ as a time parameter. In general, both discrete and continuous stochastic processes exist, depending on the domain of $t$ and range of $\check{x}(t,\cdot), \ t\in \mathfrak{T}$, with different properties \citep{Borzì2020}. We only consider continuous stochastic processes. Continuous models can be connected with other, fine-grained, descriptions, such as the Master equations \citep{Gillespie2000}. To study the evolution of continuous processes, Langevin equations and their corresponding Fokker-Planck equations are often employed \citep{Gardiner1985}. In particular, the \emph{Langevin equation} is an SDE, widely used to describe phenomena that are subject to combinations of deterministic and stochastic processes, whose general form is 
\begin{equation}
		\dot{x}_i(t) = \frac{dx_i(t)}{dt} = a_i({x(t)}) + \sum_{m=1}^{n}b^m_i({x(t)})g_m(t),
\end{equation}
where $i=1,\dots,n$, $x = [x_1 \dots x_n ]^{\top}\in \mathbb{R}^n$ is the state vector and $\left\{g_m(t)\right\}_{m=1}^n$ are fluctuation forces (random forces), which obey a Gaussian distribution and represent the effects of \emph{white noise}.
The fluctuation forces are often assumed to be Wiener processes.

\begin{definition}[\textbf{Wiener process}]\label{Def:Wiener}
A stochastic process $\left\{W_t:=W(t,\cdot) \right\}_{t\geq 0}$ is a (one-dimensional) \emph{Wiener process}, or \emph{Brownian motion}, if  
\begin{enumerate}
    \item  $\left\{W(t,\omega) \right\}_{t\geq 0}$ are continuous in $t$ for almost all $\omega\in \Omega$,
    \item $W(0,\omega) = 0$ for almost all $\omega\in \Omega$,
    \item For $0\leq s \leq t$, the increment $W(t,\cdot)-W(s,\cdot)$ is a Gaussian random variable with zero mean and variance $t-s$.
    \item The random variables $W(t_0,\cdot), W(t_1,\cdot)-W(t_0,\cdot),\dots, W(t_k,\cdot)-W(t_{k-1},\cdot)$, for every $k\geq 1$ and $0=t_0\leq t_1\leq \cdots \leq t_k$, are independent.
\end{enumerate}
\end{definition}

The Langevin equation is then recast as   
\begin{equation}
dx  = a(x) \, dt + b(x) \, dW_t \, ,
\label{eq:langevin_eq}
\end{equation}
where $b(x)$, which can be either a state-independent (constant) or a state-dependent function, dictates the noise type. Alternatively, one can also write the Langevin equation in the form
\begin{equation}\label{eq:Langev!stord}
    \frac{dx(t)}{dt} = a(x) + b(x) \eta(t) \, ,
\end{equation}
where $\eta(t)$ is a white noise process such that $\langle \eta \rangle = 0$ and $\langle \eta(t) \eta(t') \rangle = 2 \sigma \delta(t-t')$, where $\sigma$ is the noise intensity, $\delta$ is the Dirac delta and $\langle \cdot \rangle$ is the expected value. In many cases, state-dependent noise fits empirical biological data better than additive white noise \citep{Sidney2010, Wang2012b}, and is a good representation of biological coloured noise \citep{liu2009effect}.

The Langevin equation describes the evolution of a stochastic process; the evolution of its associated probability density function is described by the associated Fokker-Planck equation. In particular, consider the one-dimensional SDE
\begin{equation*}
dx = A(x,t)dt+\sigma(x,t)dW_t
\end{equation*}
and denote by $p(x,t)$ the probability density function of $x$ at time $t$, which evolves according to the  \emph{Fokker-Planck equation}
\begin{equation}
\partial_t p(x,t) = - \partial_x \left[A(x,t) p(x,t) \right] + \frac{1}{2} \partial^2_x \left[B^2(x,t) p(x,t) \right] \, ,
\label{eq:fokker_plank}
\end{equation}
with appropriate initial conditions, where $B(x,t) = \left|\sigma(x,t)\right|$ is related to the diffusion coefficient and $A(x,t)$ is a \emph{drift} term that accounts for the deterministic forces acting on the system.
In case of multiplicative noise, with $b(x)$ non-constant in Eq. \eqref{eq:langevin_eq}, $A$ and $B$ can be combinations of $a$ and $b$ \citep{Risken1991}.

The linear case, also multidimensional \citep{Gardiner1985}, is obtained from the following second-order model for Brownian motion, in Langevin form
\begin{equation*}
\ddot{u}(t) = -k \dot{u}(t)+\sqrt{2D}\eta(t),
\end{equation*}
with appropriate (Gaussian) initial conditions. Here $u(t)$ is the position of the Brownian particle, $k$ is the drift coefficient, $D$ is the diffusion coefficient and $\eta$ is a white noise process, as in equation \eqref{eq:Langev!stord}. Setting $y =\dot{u}$ to be the velocity process yields the \emph{Ornstein-Uhlenbeck process} (OU process)
\begin{equation}
dy = -k \, y \, dt + \sqrt{2D} \, dW_t,
\label{eq:o-u}
\end{equation}
where $k$ and $D$ are drift and diffusion parameters and $\sqrt{2D} = \sigma$.  An OU process can be alternatively represented by its Fokker-Planck equation
\begin{equation}
\partial_t p(y) = \partial_y (k y p(y)) + D \, \partial^2_y p(y) \, .
\end{equation}
OU processes are central to our discussion of resilience indicators in Chapter~\ref{ch:res-ind}.

\subsection{Bifurcation normal forms}\label{App:Normal}
We provide here a more extensive description of five bifurcation normal forms, frequently encountered in systems biology and epidemiology.

\subsubsection{Fold (saddle-node) bifurcation}
\label{subsec:fold}
The basic bifurcation by which point attractors are created or destroyed occurs when the leading negative eigenvalue crosses 0. A fold (or saddle-node, or blue-sky) bifurcation is a dimension-1, codimension-1 bifurcation\footnote{The codimension is the number of parameters that need to be varied for the bifurcation to occur.} with normal form:
\begin{equation}
	\dot{x}(t) = p+x(t)^2 \,,
	\label{eq:fold}
\end{equation}
with $x \in \mathbb{R}$.
For $p<0$, the system admits two equilibria $\hat{x}_{1,2} = \pm \sqrt{-p}$; $\hat{x}_1$ is stable and $\hat{x}_2$ is unstable. At the critical value $p=p_0 =0$, the two equilibria collide and vanish; for $p>0$, the system does not admit any equilibrium. The mechanism captured in Figure~\ref{fig:gene-reg-phase-portrait} for the genetic regulation system, where two equilibria collide and disappear when $f_2(x)$ becomes tangent to $f_1(x)$, corresponds to a fold bifurcation. In the potential landscape representation, a fold bifurcation occurs when the concavity/convexity of the basin of attraction near an equilibrium point flattens and then disappears (Figure~\ref{fig:pot_land}). Fold-induced bistability is also present in epidemic models with vaccination \citep{Martins2017}, and this bifurcation typically allows us to reproduce switching behaviours, often observed in cell development \citep{Bhattacharya2010,lebar2014bistable}.
The saddle-node bifurcation diagram is shown in Figure~\ref{fig:fold_trans_diag}a. 

Importantly, any system derived from \eqref{eq:fold} through a mapping $x \mapsto -x$ or $p \mapsto -p$, such as $\dot{x}(t) = p-x(t)^2$, is also a saddle-node normal form. Moreover, near the origin, any system $\dot{x}(t) = p+x(t)^2 + O(x(t)^3)$ is locally topologically equivalent to the system $\dot{x}(t) = p+x(t)^2 $ \citep{strogatz2018nonlinear}. The saddle-node bifurcation is thus generic, since additional parameters do not change its structure. 

\subsubsection{Transcritical bifurcation}
A transcritical bifurcation captures the scenario in which a fixed point always exists (\eg populations with zero infectious individuals never experience outbreaks, regardless of the infection rate and the other system parameters), but its stability properties change when a relevant parameter is varied.
This situation is typical of SIR-like models where, in the absence of seasonality, behavioural or vaccine feedback, the endemic equilibrium and the disease-free equilibrium exchange their stability properties (\cite{Brauer2012}; see also \cite{proverbio2021dynamical} for the effect of various parameters on transcritical bifurcations). The associated normal form,
\begin{equation}
	\dot{x}(t) = px(t)-x(t)^2 \, ,
	\label{eq:transcritical}
\end{equation}
has two equilibria: $\hat{x}_1 = 0$ is stable when $p<0$ and unstable when $p>0$; $\hat{x}_2 =p$ is stable when $p>0$ and unstable when $p<0$. The critical value is thus $p=p_0=0$. The transcritical bifurcation diagram is shown in Figure~\ref{fig:fold_trans_diag}b. 

\begin{figure}[ht!]
	\centering
		\includegraphics[width=0.7\textwidth]{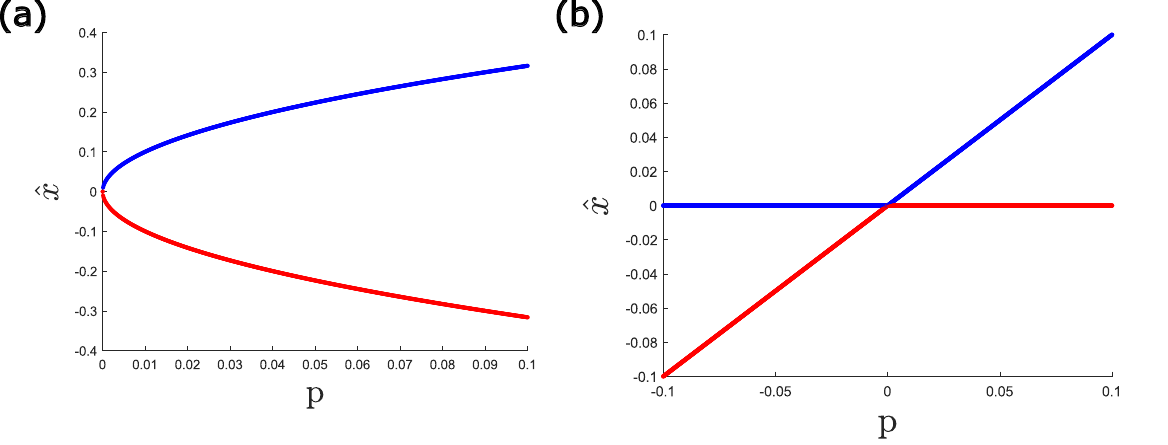}
	\caption{\footnotesize \small Bifurcation diagrams, where $\hat{x}$ is associated with equilibrium points ($f(\hat{x})=0$) and $p$ with the bifurcation parameter. (a) Saddle-node bifurcation diagram, see \eqref{eq:fold}. Stable (blue) and unstable (red) equilibria collide for a critical value of $p$ and vanish. (b) Transcritical bifurcation diagram, see \eqref{eq:transcritical}. Stable (blue) and unstable (red) equilibria intersect and exchange stability at a critical value of $p$. }
	\label{fig:fold_trans_diag}
\end{figure}

\subsubsection{Pitchfork bifurcation}
The pitchfork bifurcation leads to the emergence of two new equilibria (which are symmetric, and then the system is driven towards either one or the other), or to their disappearance, and is common in models of cell-fate decision, when new phenotype pathways become available, \eg during stem cell development \citep{Eugenio2014, Moris2016}. 
It is customary to distinguish between supercritical and subcritical pitchfork bifurcations: in the former case, two new equilibria appear, while in the latter case, two equilibria vanish. The normal form is

\begin{equation}
	\dot{x}(t) = px(t) \pm x(t)^3 \, ,
	\label{eq:subpitch}
\end{equation}
with $+$ for the subcritical bifurcation and $-$ for the supercritical.

Figure~\ref{fig:pitch_diagram_bif}a shows the bifurcation diagram in the subcritical case. The system has three equilibria for $p<0$: the equilibrium at the origin $\hat{x}_3 =0$ is stable and the two symmetric equilibria $\hat{x}_{1,2} = \pm \sqrt{-p}$ are unstable. After the critical value $p=p_0=0$, the unstable equilibria collide with the stable one and vanish, while the equilibrium $\hat{x}_3$ becomes unstable. This bifurcation is found, \eg in models of cell-fate decision, with abrupt switching \citep{Andrecut2011a}. The instability is often counterbalanced by the stabilising influence of higher order terms.

In the case of a supercritical pitchfork bifurcation, the system admits a unique equilibrium $\hat{x}_3 = 0$, which is stable, for $p<0$. For $p=p_0=0$,  $\hat{x}_3$ becomes unstable, while two other stable equilibria $\hat{x}_{1,2}= \pm \sqrt{-p}$ appear, thus leading to bistability. The bifurcation diagram is shown in Figure~\ref{fig:pitch_diagram_bif}b. Extra parameters can be added \citep{nyman2020bifurcation}, producing cusp-like bifurcations.

\begin{figure}[ht]
	\centering
		\includegraphics[width=0.7\textwidth]{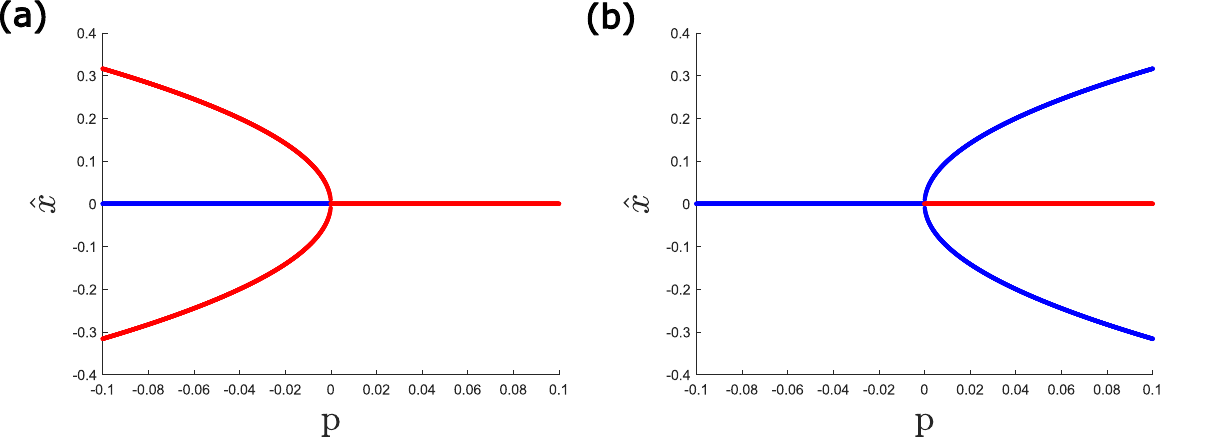}
	\caption{\footnotesize \small Pitchfork bifurcation diagrams, showing stable (blue) and unstable (red) equilibria. (a) Subcritical case: \eqref{eq:subpitch} with $+$. (b) Supercritical case: \eqref{eq:subpitch} with $-$.}
	\label{fig:pitch_diagram_bif}
\end{figure}

\subsubsection{Cusp bifurcation}
\label{subsec:cusp}
The cusp is a generic dimension-1, codimension-2 bifurcation. It is often called ``organising centre'' as it governs bistable regimes and switches across them \citep{franci2012organizing, Thom2554}, 
and its normal form is
\begin{equation}
	\dot{x}(t) = a + b x(t) - x(t)^3 \, .
	\label{eq:cusp}
\end{equation}
The bifurcation mechanism is visualised in Figure~\ref{fig:cusp}a. The projection of its stable region in the $(a,b)$ plane is shown in Figure~\ref{fig:cusp}b; Figure~\ref{fig:cusp-gen} is the cusp diagram for system \eqref{eq:hill-equation}. The cusp can be interpreted as an ``intersection'' of supercritical pitchfork and fold bifurcations, or as an unfolding around a supercritical pitchfork bifurcation (Figure~\ref{fig:cusp}c), and induces symmetry breaking in the dynamics, a known mechanism for, \eg asymmetric cellular differentiation \citep{Stanoev2021}. 

\begin{figure}[ht!]
	\centering
		\includegraphics[width=\textwidth]{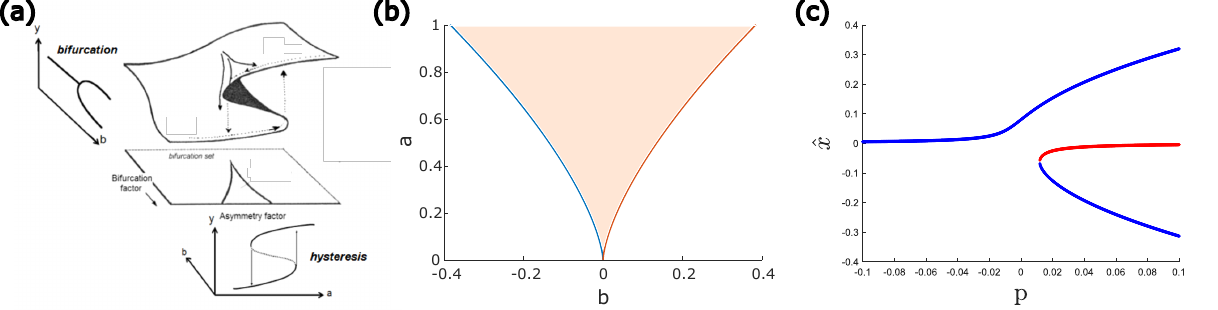}
	\caption{\footnotesize \small (a) Representation of a cusp stable manifold, its projection on the $(a,b)$ plane and the ``intersection'' of supercritical pitchfork and fold bifurcations; figure from \cite{stamovlasis2014bifurcation}. (b) Projection of the bistable region in the $(a,b)$ plane, corresponding to the shaded area. Outside the shaded area, the system is monostable. On the lines, saddle-node bifurcations occur. (c) Cusp plot as an unfolding of the pitchfork bifurcation, $\dot{x}(t) = 0.0005 + p x(t) - x(t)^3 $.}
	\label{fig:cusp}
\end{figure}

\subsubsection{Hopf bifurcation}
The Hopf bifurcation is a generic dimension-2, codimension-1 bifurcation from a stable equilibrium to a limit cycle. 
Its normal form is:
\begin{equation}
		\left\{
		\begin{array}{ll}
	           \dot{x}_1 &= p x_1 - x_2 + \ell x_1(x_1^2 + x_2^2)\,,\\
	           \dot{x}_2 &= x_1 + p x_2 + \ell x_2(x_1^2 + x_2^2) \, ,
		\end{array}
		\right.
  \label{eq:subhopf}
	\end{equation}
where $\ell$ is the \emph{Lyapunov coefficient}; the bifurcation is subcritical if $\ell>0$ and supercritical if $\ell<0$. It is characterised by complex eigenvalues crossing the imaginary axis, yielding oscillations.
The bifurcation diagram for the supercritical case is shown in Figure~\ref{fig:Hopf}.
\begin{figure}[ht!]
	\centering
		\includegraphics[width=0.5\textwidth]{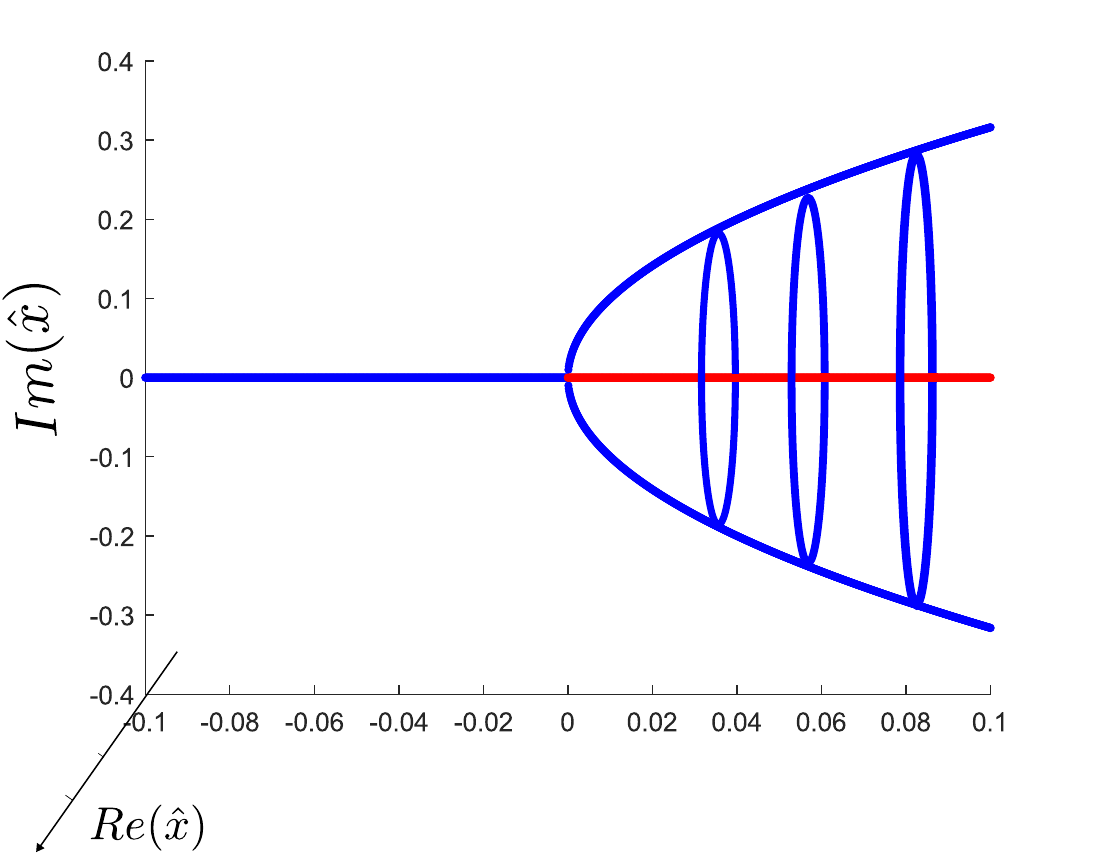}
	\caption{\footnotesize \small Supercritical Hopf bifurcation diagram. The subcritical case similarly follows the bifurcation diagram in Figure~\ref{fig:pitch_diagram_bif}a.}
	\label{fig:Hopf}
\end{figure}

\section{Derivation of stochastic indicators}\label{App:stoc-ind}

Here we provide the proofs for the results stated in Section~\ref{sec:monomodal_ind}.

\subsection{Proof of Proposition~\ref{theo:var}}

To derive expressions for stationary variance and autocorrelation, we solve Eq. \eqref{eq:o-u} with an Itô-like integration and the method of variation of constants for the solution of differential equations. In particular, we consider $y = e^{kt} y_t$ as an Ansatz of solution and we proceed with Itô differentiation: 
	\begin{align*}
d \left( e^{kt} y_t \right) &= k \, e^{kt} y_t \, dt + e^{kt} dy_t \\ &= ke^{kt} y_t \, dt + e^{kt} \left( -k \, y_t \, dt + \sqrt{2D} \, dW_t \right) = e^{kt} \, \sqrt{2D} \, dW_t.
	\end{align*}
Integrating yields $e^{kt} y_t - y_0 = \sqrt{2D} \int^t_0 e^{sk} dW_s$, which can be re-arranged to get the stochastic trajectory
\begin{equation*}
y_t = y_0 e^{-kt} + \sqrt{2D} \int_{0}^{t} e^{k(s-t)} dW_s \, .
\end{equation*}
	If the initial condition is deterministic or Gaussian, we can estimate the stationary mean
	\begin{align*}
		\langle y_t \rangle &= \langle y_0 e^{-kt} + \sqrt{2D} \int_{0}^{t} e^{k(s-t)} dW_s \rangle = \langle y_0 e^{-kt}  \rangle + \langle  \sqrt{2D} \int_{0}^{t} e^{k(s-t)} dW_s \rangle \\
		&= y_0 e^{-kt} + \sqrt{2D} \int_{0}^{t} e^{k(s-t)} \langle dW_s \rangle = y_0 e^{-kt} + 0 \xrightarrow[]{t \to \infty} 0
	\end{align*}
and the covariance $\text{Cov}(y(t) y(t')) = \langle y_t y_{t'} \rangle$, which is
	\begin{align*}
		\langle y_t y_{t'} \rangle &= y_0^2 e^{-k(t-t')} + \left( \sqrt{2D} \right) ^2 \int_{0}^{t} \int_{0}^{t'} e^{k(s-t)} e^{k(s'-t')} \langle dW_s \, dW_{s'} \rangle\\
		&= y_0^2 e^{-k(t-t')} + 2D \, e^{-k(t+t')} \int_{o}^{t} \int_{0}^{t'} e^{k(s+s')} \delta(s-s') \, ds \, ds',
	\end{align*}
where we exploit the property of Wiener processes that $ \langle dW_s \, dW_{s'} \rangle = \delta(s-s') \, ds \, ds' $, with $\delta(s-s')$ the Dirac delta. Then,
	\begin{align*}
		\langle y_t y_{t'} \rangle &=  y_0^2 e^{-k(t+t')} + 2D \, e^{-k(t+t')} \left. \frac{e^{2k - \min(s,s')}}{2k} \right\rvert_0^{\min(t,t')}\\
		&=  y_0^2 e^{-k(t+t')} + \frac{D}{k} e^{-k(t+t')} \left( e^{2k - \min(t,t') } -1  \right) \\
		&= y_0^2 e^{-k(t+t')} + \frac{D}{k}e^{-k|t-t'|} - \frac{D}{k}e^{-k(t+t')}
	\end{align*}
and, by neglecting the rapid initial transient, we finally get:
	\begin{equation} \label{eq:covar}
		\langle y_t y_{t'} \rangle  \xrightarrow[]{t,t' \to \infty} \frac{D}{k}e^{-k|t-t'|} 
	\end{equation}
	
	The stationary variance $\text{Var}=\langle y_t^2 \rangle - \langle y_t \rangle ^2$ can be computed as:
	\begin{equation} \label{eq:variance}
		\text{Var}= \langle y_t y_{t'} \rvert_{t'=t} \rangle -  \langle y_t \rangle ^2 = \frac{D}{k} \left( 1- e^{-2kt} \right) \xrightarrow{t \to \infty} \frac{D}{k}.
	\end{equation}
The autocorrelation is
$AC(t-t') = \lim_{t,t' \to \infty} \frac{\text{Cov}(y(t) y(t'))}{\sqrt{\text{Var}(y(t)) \text{Var}(y(t'))}} = e^{-k |t-t'|}$,
 hence
$AC(1) = e^{-k}$.

\subsection{Proof of Theorem~\ref{theo:reddening}}

Given a stationary stochastic process, its stationary spectral density function is defined as the Fourier transform of the Covariance function \citep{papoulis2002probability}:
	\begin{equation}
		S(\omega) = \mathcal{F}\left[ \langle y(t+\tau)y(t)  \rangle \right] = \frac{1}{2 \pi}\int_{-\infty}^{\infty}R(\tau)e^{-i\omega t} dt\,,
		\label{eq:spectrum}
	\end{equation}
	where $R(\tau)$ is the covariance function. For $t,t' \to \infty$, the covariance function has the expression in \eqref{eq:covar}, so the calculation is straightforward as long as we perform the change of coordinates $|t'-t| \to |t|$, which is admissible for stationary processes:
	\begin{align*}
		S(\omega) &= \frac{1}{2 \pi} \int_{-\infty}^{\infty} e^{-i \omega t} \frac{D}{k} e^{-k|t|} = \frac{D}{2 \pi k} \left(\int_{-\infty}^{0} e^{t(-i\omega + k)} dt + \int_{0}^{\infty}e^{-t(i\omega +k)} dt \right) \\
		&= \frac{D}{2 \pi k} \left( \frac{1}{-i\omega + k} + \frac{1}{i\omega +k}   \right) =\frac{D}{\pi k}\frac{k}{k^2 + \omega^2} = \frac{D}{\pi(k^2 + \omega^2)}.
	\end{align*}

\subsection{Proof of Theorem~\ref{theo:entropy}}

Deriving the Shannon entropy for a Gaussian distributed variable $y \sim \mathcal{N}(\mu, \text{Var})$ yields
    
	\begin{align*}
		H_s(y) & = - \int p(y')\log p(y')dy'  \\
		& = - \mathbb{E}\left[ \log \mathcal{N}(\mu, \text{Var})  \right] =\mathbb{E} \left[ \log \left[ \frac{1}{\sqrt{2 \pi \text{Var}}} \exp \left(-\frac{1}{2 \text{Var}}(x-\mu)^2 \right)  \right]  \right] \\
		& = \frac{1}{2}\log\left(2\pi\text{Var}\right) + \frac{1}{2\text{Var}}\mathbb{E}\left[(x-\mu)^2\right]  = \frac{1}{2}\left( \log (2 \pi \text{Var}) +1 \right).
	\end{align*}

\bibliography{sample-now.bib}

\begin{thebibliography}{831}
\providecommand{\natexlab}[1]{#1}
\providecommand{\url}[1]{\texttt{#1}}
\expandafter\ifx\csname urlstyle\endcsname\relax
  \providecommand{\doi}[1]{doi: #1}\else
  \providecommand{\doi}{doi: \begingroup \urlstyle{rm}\Url}\fi

\bibitem[Alon(2019)]{alon2019introduction}
Uri Alon.
\newblock \emph{An Introduction to Systems Biology: Design Principles of
  Biological Circuits}.
\newblock CRC press; 2nd ed., 2019.

\bibitem[Doyle and Stelling(2006)]{Doyle2006b}
Francis~J. Doyle and Jörg Stelling.
\newblock Systems interface biology.
\newblock \emph{J Roy Soc Interface}, 3:\penalty0 603--616, 2006.

\bibitem[Kitano(2002)]{Kitano2002}
Hiroaki Kitano.
\newblock Systems biology: A brief overview.
\newblock \emph{Science}, 295:\penalty0 1662--1664, 2002.

\bibitem[Wellstead et~al.(2008)Wellstead, Bullinger, Kalamatianos, Mason, and
  Verwoerd]{Wellstead2008}
Peter Wellstead, Eric Bullinger, Dimitrios Kalamatianos, Oliver Mason, and Mark
  Verwoerd.
\newblock The role of control and system theory in systems biology.
\newblock \emph{Annu Rev Control}, 32:\penalty0 33--47, 2008.

\bibitem[Keener and Sneyd(2009)]{keener2009mathematical}
James Keener and James Sneyd.
\newblock \emph{Mathematical Physiology: II: Systems Physiology}.
\newblock Springer, 2009.

\bibitem[Barabási et~al.(2011)Barabási, Gulbahce, and Loscalzo]{Barabasi2011}
Albert-László Barabási, Natali Gulbahce, and Joseph Loscalzo.
\newblock Network medicine: a network-based approach to human disease.
\newblock \emph{Nat Rev Genet}, 12:\penalty0 56--68, 2011.

\bibitem[Bélair et~al.(1995)Bélair, Glass, Heiden, and Milton]{Belair1995}
Jacques Bélair, Leon Glass, Uwe An~Der Heiden, and John Milton.
\newblock Dynamical disease: Identification, temporal aspects and treatment
  strategies of human illness.
\newblock \emph{Chaos}, 5:\penalty0 1--7, 1995.

\bibitem[Piro(2012)]{Piro2012}
Rosario~M. Piro.
\newblock Network medicine: Linking disorders.
\newblock \emph{Human Genet}, 131:\penalty0 1811--1820, 2012.

\bibitem[Anderson and May(1991)]{anderson1991infectious}
Roy~M Anderson and Robert~M May.
\newblock \emph{Infectious Diseases of Humans: Dynamics and Control}.
\newblock Oxford UP, 1991.

\bibitem[Brauer and Castillo-Chavez(2012)]{Brauer2012}
Fred Brauer and Carlos Castillo-Chavez.
\newblock \emph{Mathematical Models in Population Biology and Epidemiology},
  volume~40.
\newblock Springer, 2 edition, 2012.

\bibitem[Edelstein-Keshet(2005)]{edelstein2005}
Leah Edelstein-Keshet.
\newblock \emph{Mathematical Models in Biology}.
\newblock SIAM, 2005.

\bibitem[Khammash(2016)]{Khammash2016}
Mustafa Khammash.
\newblock An engineering viewpoint on biological robustness.
\newblock \emph{BMC Biol}, 14, 2016.

\bibitem[Kitano(2004{\natexlab{a}})]{Kitano2004b}
Hiroaki Kitano.
\newblock Biological robustness.
\newblock \emph{Nat Rev Genet}, 5:\penalty0 826--837, 2004{\natexlab{a}}.

\bibitem[Dai et~al.(2015)Dai, Korolev, Gore, and Carpenter]{Dai2015}
Lei Dai, Kirill~S. Korolev, Jeff Gore, and Stephen~R. Carpenter.
\newblock Relation between stability and resilience determines the performance
  of early warning signals under different environmental drivers.
\newblock \emph{P Natl Acad Sci Usa}, 112:\penalty0 10056--10061, 2015.

\bibitem[Alon(2003)]{Alon2003}
U.~Alon.
\newblock Biological networks: The tinkerer as an engineer.
\newblock \emph{Science}, 301:\penalty0 1866--1867, 2003.

\bibitem[Chen et~al.(2007)Chen, Wu, Wu, and Li]{Chen2007}
Bor-Sen Chen, Wan-Shian Wu, Wei-Sheng Wu, and Wen-Hsiung Li.
\newblock On the adaptive design rules of biochemical networks in evolution.
\newblock \emph{Evolut Bioinf}, 2:\penalty0 27--39, 2007.

\bibitem[Sneppen et~al.(2010)Sneppen, Krishna, and Semsey]{Sneppen2010}
Kim Sneppen, Sandeep Krishna, and Szabolcs Semsey.
\newblock Simplified models of biological networks.
\newblock \emph{Annu Rev Biophys}, 39:\penalty0 43--59, 2010.

\bibitem[Siami et~al.(2020)Siami, Motee, Buzi, Bamieh, Khammash, and
  Doyle]{Siami2020}
Milad Siami, Nader Motee, Gentian Buzi, Bassam Bamieh, Mustafa Khammash, and
  John~C. Doyle.
\newblock Fundamental limits and tradeoffs in autocatalytic pathways.
\newblock \emph{IEEE T Automat Contr}, 65:\penalty0 733--740, 2020.

\bibitem[Aoki et~al.(2019)Aoki, Lillacci, Gupta, Baumschlager, Schweingruber,
  and Khammash]{Aoki2019}
Stephanie~K. Aoki, Gabriele Lillacci, Ankit Gupta, Armin Baumschlager, David
  Schweingruber, and Mustafa Khammash.
\newblock A universal biomolecular integral feedback controller for robust
  perfect adaptation.
\newblock \emph{Nature}, 570:\penalty0 533--537, 2019.

\bibitem[Giordano et~al.(2019)Giordano, De~Graaf, Vasilakou, and
  Wahl]{Giordano2019}
G.~Giordano, L.M. De~Graaf, E.~Vasilakou, and S.A. Wahl.
\newblock Unraveling energy homeostasis in a dynamic model of glycolysis in
  escherichia coli.
\newblock In \emph{Proc ECC}, pages 2140--2145, 2019.

\bibitem[Palumbo et~al.(2013)Palumbo, Ditlevsen, Bertuzzi, and
  Gaetano]{Palumbo2013}
Pasquale Palumbo, Susanne Ditlevsen, Alessandro Bertuzzi, and Andrea~De
  Gaetano.
\newblock Mathematical modeling of the glucose-insulin system: A review.
\newblock \emph{Math Biosci}, 244:\penalty0 69--81, 2013.

\bibitem[Tang and McMillen(2016)]{Tang2016}
Zhe~F. Tang and David~R. McMillen.
\newblock Design principles for the analysis and construction of robustly
  homeostatic biological networks.
\newblock \emph{J Theor Biol}, 408:\penalty0 274--289, 2016.

\bibitem[Shinar et~al.(2007)Shinar, Milo, Martínez, and Alon]{Shinar2007}
Guy Shinar, Ron Milo, María~Rodríguez Martínez, and Uri Alon.
\newblock Input-output robustness in simple bacterial signaling systems.
\newblock \emph{P Natl Acad Sci Usa}, 104:\penalty0 19931--19935, 2007.

\bibitem[Steuer et~al.(2011)Steuer, Waldherr, Sourjik, and
  Kollmann]{Steuer2011}
Ralf Steuer, Steffen Waldherr, Victor Sourjik, and Markus Kollmann.
\newblock Robust signal processing in living cells.
\newblock \emph{PLoS Comput Biol}, 7, 2011.

\bibitem[Giordano(2018)]{Giordano2018}
G.~Giordano.
\newblock Ceramide-transfer protein-mediated ceramide transfer is a
  structurally tunable flow-inducing mechanism with structural feed-forward
  loops.
\newblock \emph{J Roy Soc Open Sci}, 5, 2018.

\bibitem[Angeli and Sontag(2008)]{Angeli2008}
David Angeli and Eduardo~D. Sontag.
\newblock Oscillations in i/o monotone systems under negative feedback.
\newblock \emph{IEEE T Automat Contr}, 53:\penalty0 166--176, 2008.

\bibitem[Atkinson et~al.(2003)Atkinson, Savageau, Myers, and
  Ninfa]{Atkinson2003}
Mariette~R Atkinson, Michael~A Savageau, Jesse~T Myers, and Alexander~J Ninfa.
\newblock Development of genetic circuitry exhibiting toggle switch or
  oscillatory behavior in escherichia coli.
\newblock \emph{Cell}, 113:\penalty0 597--607, 2003.

\bibitem[Blanchini et~al.(2014{\natexlab{a}})Blanchini, Franco, and
  Giordano]{Blanchini2014structural}
F~Blanchini, E~Franco, and G~Giordano.
\newblock A structural classification of candidate oscillatory and
  multistationary biochemical systems.
\newblock \emph{B Math Biol}, 76:\penalty0 2542--2569, 2014{\natexlab{a}}.

\bibitem[Blanchini et~al.(2015{\natexlab{a}})Blanchini, Franco, and
  Giordano]{Blanchini2015structuralclass}
F.~Blanchini, E.~Franco, and G.~Giordano.
\newblock Structural conditions for oscillations and multistationarity in
  aggregate monotone systems.
\newblock In \emph{Proc IEEE CDC}, 2015{\natexlab{a}}.

\bibitem[Blanchini et~al.(2018{\natexlab{a}})Blanchini, {Cuba Samaniego},
  Franco, and Giordano]{Blanchini2018homogeneous}
F.~Blanchini, C.~{Cuba Samaniego}, E.~Franco, and G.~Giordano.
\newblock Homogeneous time constants promote oscillations in negative feedback
  loops.
\newblock \emph{ACS Synth Biol}, 2018{\natexlab{a}}.

\bibitem[Hori et~al.(2013)Hori, Takada, and Hara]{Hori2013}
Yutaka Hori, Masaaki Takada, and Shinji Hara.
\newblock Biochemical oscillations in delayed negative cyclic feedback:
  Existence and profiles.
\newblock \emph{Automatica}, 49:\penalty0 2581--2590, 2013.

\bibitem[Katz et~al.({\natexlab{a}})Katz, Giordano, and Margaliot]{Katz2025osc}
Rami Katz, Giulia Giordano, and Michael Margaliot.
\newblock Instability of equilibrium and convergence to periodic orbits in
  strongly 2-cooperative systems.
\newblock \emph{Journal of Differential Equations}, 444, {\natexlab{a}}.

\bibitem[Leite and Wang(2010)]{Leite2010}
Maria Conceição~A. Leite and Yunjiao Wang.
\newblock Multistability, oscillations and bifurcations in feedback loops.
\newblock \emph{Math Biosci Eng}, 7:\penalty0 83--97, 2010.

\bibitem[Proverbio et~al.(2020)Proverbio, Gallo, Passalacqua, Destefanis,
  Maggiora, and Pellegrino]{proverbio2020assessing}
Daniele Proverbio, Luca Gallo, Barbara Passalacqua, Marco Destefanis, Marco
  Maggiora, and Jacopo Pellegrino.
\newblock Assessing the robustness of decentralized gathering: a multi-agent
  approach on micro-biological systems.
\newblock \emph{Swarm Int}, 14:\penalty0 313--331, 2020.

\bibitem[Proverbio(2024)]{proverbio2024chemotaxis}
Daniele Proverbio.
\newblock Chemotaxis in heterogeneous environments: A multi-agent model of
  decentralized gathering past obstacles.
\newblock \emph{J Theo Biol}, 586:\penalty0 111820, 2024.

\bibitem[Proverbio and Giordano(2025{\natexlab{a}})]{proverbio2025threshold}
Daniele Proverbio and Giulia Giordano.
\newblock Threshold sensing yields optimal path formation in physarum
  polycephalum--but the mould does not know.
\newblock \emph{arXiv}, 2025{\natexlab{a}}.

\bibitem[Reginato et~al.(2025)Reginato, Proverbio, and Giordano]{Reginato2025}
Damiano Reginato, Daniele Proverbio, and Giulia Giordano.
\newblock Bottom-up robust modelling for the foraging behaviour of physarum
  polycephalum.
\newblock \emph{J Roy Soc Interface}, 22\penalty0 (223):\penalty0 20240701,
  2025.

\bibitem[Romeralo et~al.(2013)Romeralo, Baldauf, and Escalante]{romeralo}
M.~Romeralo, S.~Baldauf, and R.~Escalante.
\newblock \emph{Dictyostelids: Evolution, Genomics and Cell Biology}.
\newblock Springer Sci \& Bus Media, 2013.

\bibitem[Doyle et~al.(2006)Doyle, Gunawan, Bagheri, Mirsky, and To]{Doyle2006}
Francis~J. Doyle, Rudiyanto Gunawan, Neda Bagheri, Henry Mirsky, and Tsz~Leung
  To.
\newblock Circadian rhythm: A natural, robust, multi-scale control system.
\newblock \emph{Comput Chem Eng}, 30:\penalty0 1700--1711, 2006.

\bibitem[Doyle(2008)]{Doyle2008}
Francis~Joseph Doyle.
\newblock Robust control in biology: From genes to cells to systems.
\newblock \emph{IFAC Proc}, 17, 2008.

\bibitem[Goldbeter et~al.(2012)Goldbeter, Gérard, Gonze, Leloup, and
  Dupont]{Goldbeter2012}
A.~Goldbeter, C.~Gérard, D.~Gonze, J.~C. Leloup, and G.~Dupont.
\newblock Systems biology of cellular rhythms.
\newblock \emph{FEBS Lett}, 586:\penalty0 2955--2965, 2012.

\bibitem[Li and Fang(2007)]{Li2007}
Shao Li and Youhan Fang.
\newblock Modelling circadian rhythms of protein kaia, kaib and kaic
  interactions in cyanobacteria.
\newblock \emph{Biol Rhythm Res}, 38:\penalty0 43--53, 2007.

\bibitem[Stelling et~al.(2004)Stelling, Sauer, Szallasi, Doyle, and
  Doyle]{stelling2004robustness}
Jorg Stelling, Uwe Sauer, Zoltan Szallasi, Francis~J Doyle, and John Doyle.
\newblock Robustness of cellular functions.
\newblock \emph{Cell}, 118\penalty0 (6):\penalty0 675--685, 2004.

\bibitem[Burlando et~al.(2019)Burlando, Blanchini, and Giordano]{Burlando2019}
Bruno Burlando, Franco Blanchini, and Giulia Giordano.
\newblock Loop analysis of blood pressure/volume homeostasis.
\newblock \emph{PLOS Comput Biol}, 15\penalty0 (9):\penalty0 1--24, 09 2019.

\bibitem[Burlando et~al.(2020)Burlando, Milanese, Giordano, Bonifacino, Ravera,
  Blanchini, and Bonanno]{Burlando2020}
Bruno Burlando, Marco Milanese, Giulia Giordano, Tiziana Bonifacino, Silvia
  Ravera, Franco Blanchini, and Giambattista Bonanno.
\newblock A multistationary loop model of als unveils critical molecular
  interactions involving mitochondria and glucose metabolism.
\newblock \emph{PLOS ONE}, 15\penalty0 (12):\penalty0 1--17, 12 2020.

\bibitem[Burlando et~al.(2022)Burlando, Mucci, Browne, Losacco, Indovina,
  Marinelli, Blanchini, and Giordano]{Burlando2022}
Bruno Burlando, Viviana Mucci, Cherylea~J Browne, Serena Losacco, Iole
  Indovina, Lucio Marinelli, Franco Blanchini, and Giulia Giordano.
\newblock {Mal de Debarquement Syndrome explained by a vestibulo–cerebellar
  oscillator}.
\newblock \emph{IMA J Math Med Biol}, 40\penalty0 (1):\penalty0 96--110, 12
  2022.

\bibitem[Demori et~al.(2022)Demori, Giordano, Mucci, Losacco, Marinelli,
  MAssobrio, Blanchini, and Burlando]{Demori2022}
Ilaria Demori, Giulia Giordano, Viviana Mucci, Serena Losacco, Lucio Marinelli,
  Paolo MAssobrio, Franco Blanchini, and Bruno Burlando.
\newblock Thalamocortical bistable switch as a theoretical model of
  fibromyalgia pathogenesis inferred from a literature survey.
\newblock \emph{J Comput Neurosc}, 50:\penalty0 471--484, 2022.

\bibitem[Demori et~al.(2024)Demori, Losacco, Giordano, Mucci, Blanchini, and
  Burlando]{Demori2024}
Ilaria Demori, Serena Losacco, Giulia Giordano, Viviana Mucci, Franco
  Blanchini, and Bruno Burlando.
\newblock Fibromyalgia pathogenesis explained by a neuroendocrine multistable
  model.
\newblock \emph{PLOS ONE}, 19\penalty0 (7):\penalty0 1--32, 07 2024.

\bibitem[Mucci et~al.(2020)Mucci, Indovina, Browne, Blanchini, Giordano,
  Marinelli, and Burlando]{Mucci2020}
Viviana Mucci, Iole Indovina, Cherylea~J. Browne, Franco Blanchini, Giulia
  Giordano, Lucio Marinelli, and Bruno Burlando.
\newblock Mal de debarquement syndrome: A matter of loops?
\newblock \emph{Frontiers in Neurology}, 11, 2020.

\bibitem[Arcak(2013)]{Arcak2013}
Murat Arcak.
\newblock Pattern formation by lateral inhibition in large-scale networks of
  cells.
\newblock \emph{IEEE T Automat Contr}, 58:\penalty0 1250--1262, 2013.

\bibitem[Giordano and Altafini(2017{\natexlab{a}})]{Giordano2017qualitative}
G~Giordano and C~Altafini.
\newblock Qualitative and quantitative responses to press perturbations in
  ecological networks.
\newblock \emph{Sci Rep}, 7, 2017{\natexlab{a}}.

\bibitem[Alamo et~al.(2021)Alamo, Reina, Gata, Preciado, and
  Giordano]{Alamo2021}
Teodoro Alamo, Daniel~Gutiérrez Reina, Pablo~Millán Gata, Victor~M. Preciado,
  and Giulia Giordano.
\newblock Data-driven methods for present and future pandemics: Monitoring,
  modelling and managing.
\newblock \emph{Annu Rev Control}, 52:\penalty0 448--464, 2021.

\bibitem[Blanchini et~al.(2021{\natexlab{a}})Blanchini, Bolzern, Colaneri,
  De~Nicolao, and Giordano]{Blanchini2021}
Franco Blanchini, Paolo Bolzern, Patrizio Colaneri, Giuseppe De~Nicolao, and
  Giulia Giordano.
\newblock Generalized epidemiological compartmental models: guaranteed bounds
  via optimal control.
\newblock In \emph{Proc 60th IEEE CDC}, pages 3532--3537, 2021{\natexlab{a}}.

\bibitem[Giordano et~al.(2020)Giordano, Blanchini, Bruno, Colaneri, Filippo,
  Matteo, and Colaneri]{Giordano2020}
Giulia Giordano, Franco Blanchini, Raffaele Bruno, Patrizio Colaneri,
  Alessandro~Di Filippo, Angela~Di Matteo, and Marta Colaneri.
\newblock Modelling the covid-19 epidemic and implementation of population-wide
  interventions in italy.
\newblock \emph{Nat Med}, 26\penalty0 (6):\penalty0 855--860, 2020.

\bibitem[Giordano et~al.(2021)Giordano, Colaneri, Filippo, Blanchini, Bolzern,
  Nicolao, Sacchi, Colaneri, and Bruno]{Giordano2021}
Giulia Giordano, Marta Colaneri, Alessandro~Di Filippo, Franco Blanchini, Paolo
  Bolzern, Giuseppe~De Nicolao, Paolo Sacchi, Patrizio Colaneri, and Raffaele
  Bruno.
\newblock Modeling vaccination rollouts, sars-cov-2 variants and the
  requirement for non-pharmaceutical interventions in italy.
\newblock \emph{Nat Med}, 27\penalty0 (6):\penalty0 993--998, 2021.

\bibitem[Hernandez-Vargas et~al.(2022)Hernandez-Vargas, González, Beck, Bi,
  Calà~Campana, and Giordano]{HernandezVargas2022}
Esteban~A. Hernandez-Vargas, Alejandro~H. González, Carolyn~L. Beck, Xiaoqi
  Bi, Francesca Calà~Campana, and Giulia Giordano.
\newblock Modelling and control of epidemics across scales.
\newblock In \emph{Proc 61st IEEE CDC}, pages 4963--4980, 2022.

\bibitem[Almocera et~al.(2018)Almocera, Nguyen, and
  Hernandez-Vargas]{Almocera2018}
A.E.S. Almocera, V.K. Nguyen, and E.A. Hernandez-Vargas.
\newblock Multiscale model within-host and between-host for viral infectious
  diseases.
\newblock \emph{J Math Biol}, 77:\penalty0 1035--1057, 2018.

\bibitem[Axer and Amunts(2022)]{Axer2022}
Markus Axer and Katrin Amunts.
\newblock Scale matters: The nested human connectome.
\newblock \emph{Science}, 378:\penalty0 500--504, 2022.

\bibitem[Hart et~al.(2020)Hart, Maini, Yates, and Thompson]{Hart2020}
William~S Hart, Philip~K. Maini, Christian~A. Yates, and Robin~N Thompson.
\newblock A theoretical framework for transitioning from patient-level to
  population-scale epidemiological dynamics: influenza a as a case study.
\newblock \emph{J Roy Soc Interface}, 17:\penalty0 20200230, 2020.

\bibitem[Tegner et~al.(2016)Tegner, Zenil, Kiani, Ball, and
  Gomez-Cabrero]{Tegner2016}
Jesper Tegner, Hector Zenil, Narsis~A. Kiani, Gordon Ball, and David
  Gomez-Cabrero.
\newblock A perspective on bridging scales and design of models using
  low-dimensional manifolds and data-driven model inference.
\newblock \emph{Philos T R Soc A}, 374:\penalty0 20160144, 2016.

\bibitem[Anderson and Papachristodoulou(2009)]{Anderson2009}
James Anderson and Antonis Papachristodoulou.
\newblock On validation and invalidation of biological models.
\newblock \emph{BMC Bioinf}, 10, 2009.

\bibitem[Bates and Cosentino(2011)]{Bates2011}
D.~G. Bates and C.~Cosentino.
\newblock Validation and invalidation of systems biology models using
  robustness analysis.
\newblock \emph{IET Sys Biol}, 5:\penalty0 229--244, 2011.

\bibitem[Proverbio(2022)]{proverbio2022classification}
Daniele Proverbio.
\newblock \emph{Classification and detection of Critical Transitions: from
  theory to data}.
\newblock PhD thesis, University of Luxembourg, Luxembourg, 2022.

\bibitem[Conrad et~al.(2008)Conrad, Mayo, Ninfa, and Forger]{Conrad2008}
Emery Conrad, Avraham~E. Mayo, Alexander~J. Ninfa, and Daniel~B. Forger.
\newblock Rate constants rather than biochemical mechanism determine behaviour
  of genetic clocks.
\newblock \emph{J Roy Soc Interface}, 5, 2008.

\bibitem[Calà~Campana et~al.(2024)Calà~Campana, Katz, and
  Giordano]{CalaCampana2024}
Francesca Calà~Campana, Rami Katz, and Giulia Giordano.
\newblock Sequential-quadratic-hamiltonian optimal control of epidemic models
  with an arbitrary number of infected and non-infected compartments.
\newblock \emph{IEEE Contr Sys Lett}, 8:\penalty0 1805--1810, 2024.

\bibitem[Cosentino and Bates(2012)]{Cosentino2012}
Carlo. Cosentino and Declan. Bates.
\newblock \emph{Feedback Control in Systems Biology}.
\newblock CRC Press, 2012.

\bibitem[Giordano et~al.(2016{\natexlab{a}})Giordano, Rantzer, and
  Jonsson]{Giordano2016convex}
G~Giordano, A~Rantzer, and V~D Jonsson.
\newblock A convex optimization approach to cancer treatment to address tumor
  heterogeneity and imperfect drug penetration in physiological compartments.
\newblock In \emph{Proc IEEE 55th CDC}, pages 2494--2500, 2016{\natexlab{a}}.

\bibitem[Hernandez-Vargas et~al.(2010)Hernandez-Vargas, Colaneri, Middleton,
  and Blanchini]{Hernandez-Vargas2010}
Esteban~A. Hernandez-Vargas, Patrizio Colaneri, Richard~H. Middleton, and
  Franco Blanchini.
\newblock Discrete‐time control for switched positive systems with
  application to mitigating viral escape.
\newblock \emph{Int J Robust Nonlin}, 21:\penalty0 1093--1111, 2010.

\bibitem[Lenhart and Workman(2007)]{Lenhart2007}
Suzanne Lenhart and John~T Workman.
\newblock \emph{Optimal Control Applied to Biological Models}.
\newblock CRC press, 2007.

\bibitem[Morris et~al.(2021)Morris, Rossine, Plotkin, and Levin]{Morris2021}
Dylan~H. Morris, Fernando~W. Rossine, Joshua~B. Plotkin, and Simon~A. Levin.
\newblock Optimal, near-optimal, and robust epidemic control.
\newblock \emph{Commun Phys}, 4\penalty0 (1):\penalty0 78, 2021.

\bibitem[Andrecut et~al.(2011)Andrecut, Halley, Winkler, and
  Huang]{Andrecut2011a}
Mircea Andrecut, Julianne~D. Halley, David~A. Winkler, and Sui Huang.
\newblock A general model for binary cell fate decision gene circuits with
  degeneracy: Indeterminacy and switch behavior in the absence of
  cooperativity.
\newblock \emph{PLoS ONE}, 6:\penalty0 e19358, 2011.

\bibitem[Pio et~al.(2022)Pio, Mignone, Magazzù, Zampieri, Ceci, and
  Angione]{Pio2022}
Gianvito Pio, Paolo Mignone, Giuseppe Magazzù, Guido Zampieri, Michelangelo
  Ceci, and Claudio Angione.
\newblock Integrating genome-scale metabolic modelling and transfer learning
  for human gene regulatory network reconstruction.
\newblock \emph{Bioinformatics}, 38:\penalty0 487--493, 2022.

\bibitem[Ghaffarizadeh et~al.(2014)Ghaffarizadeh, Flann, and
  Podgorski]{Ghaffarizadeh2014}
Ahmadreza Ghaffarizadeh, Nicholas~S. Flann, and Gregory~J. Podgorski.
\newblock Multistable switches and their role in cellular differentiation
  networks.
\newblock \emph{BMC Bioinf}, 15:\penalty0 S7, 2014.

\bibitem[Amara et~al.(2022)Amara, Frainay, Jourdan, Naake, Neumann, del Toro,
  et~al.]{Amara2022}
Adam Amara, Clément Frainay, Fabien Jourdan, Thomas Naake, Steffen Neumann,
  Elva María~Novoa del Toro, et~al.
\newblock Networks and graphs discovery in metabolomics data analysis and
  interpretation.
\newblock \emph{Front Mol Biosci}, 9, 2022.

\bibitem[Arcak and Sontag(2007)]{Arcak2007}
Murat Arcak and Eduardo~D Sontag.
\newblock A passivity-based stability criterion for a class of interconnected
  systems and applications to biochemical reaction networks.
\newblock In \emph{46th IEEE CDC}, pages 4477--4482, 2007.

\bibitem[Liu and Locasale(2017)]{Liu2017}
Xiaojing Liu and Jason~W. Locasale.
\newblock Metabolomics: A primer.
\newblock \emph{Trends Biochem Sci}, 42:\penalty0 274--284, 2017.

\bibitem[Hoffman et~al.(2017)Hoffman, Lyu, Pletcher, and
  Promislow]{Hoffman2017}
Jessica~M. Hoffman, Yang Lyu, Scott~D. Pletcher, and Daniel~E.L. Promislow.
\newblock Proteomics and metabolomics in ageing research: from biomarkers to
  systems biology.
\newblock \emph{Essays Biochem}, 61:\penalty0 379--388, 2017.

\bibitem[Alhourani et~al.(2015)Alhourani, McDowell, Randazzo, Wozny, Kondylis,
  Lipski, et~al.]{alhourani2015network}
Ahmad Alhourani, Michael~M McDowell, Michael~J Randazzo, Thomas~A Wozny,
  Efstathios~D Kondylis, Witold~J Lipski, et~al.
\newblock Network effects of deep brain stimulation.
\newblock \emph{J Neurophysiol}, 114\penalty0 (4):\penalty0 2105--2117, 2015.

\bibitem[Elam et~al.(2021)Elam, Glasser, Harms, Sotiropoulos, Andersson,
  Burgess, et~al.]{Elam2021}
Jennifer~Stine Elam, Matthew~F. Glasser, Michael~P. Harms, Stamatios~N.
  Sotiropoulos, Jesper~L.R. Andersson, Gregory~C. Burgess, et~al.
\newblock The human connectome project: A retrospective.
\newblock \emph{NeuroImage}, 244:\penalty0 118543, 2021.

\bibitem[Lynn and Bassett(2019)]{Lynn2019}
Christopher~W. Lynn and Danielle~S. Bassett.
\newblock The physics of brain network structure, function and control.
\newblock \emph{Nat Rev Phys}, 1:\penalty0 318--332, 2019.

\bibitem[Goh et~al.(2007)Goh, Cusick, Valle, Childs, Vidal, and
  Barabási]{Goh2007}
Kwang-Il Goh, Michael~E. Cusick, David Valle, Barton Childs, Marc Vidal, and
  Albert-László Barabási.
\newblock The human disease network.
\newblock \emph{P Natl Acad Sci Usa}, 104:\penalty0 8685--8690, 2007.

\bibitem[Hammond et~al.(2007)Hammond, Bergman, and Brown]{Hammond2007}
Constance Hammond, Hagai Bergman, and Peter Brown.
\newblock Pathological synchronization in parkinson's disease: networks, models
  and treatments.
\newblock \emph{Trends Neurosci}, 30:\penalty0 357--364, 2007.

\bibitem[Royer et~al.(2022)Royer, Bernhardt, Larivière, Gleichgerrcht,
  Vorderwülbecke, Vulliémoz, and Bonilha]{Royer2022}
Jessica Royer, Boris~C. Bernhardt, Sara Larivière, Ezequiel Gleichgerrcht,
  Bernd~J. Vorderwülbecke, Serge Vulliémoz, and Leonardo Bonilha.
\newblock Epilepsy and brain network hubs.
\newblock \emph{Epilepsia}, 63:\penalty0 537--550, 2022.

\bibitem[Alutto et~al.(2024)Alutto, Cianfanelli, Como, and
  Fagnani]{alutto2024dynamic}
Martina Alutto, Leonardo Cianfanelli, Giacomo Como, and Fabio Fagnani.
\newblock On the dynamic behavior of the network sir epidemic model.
\newblock \emph{IEEE Transactions on Control of Network Systems}, 2024.

\bibitem[Arino et~al.(2005)Arino, Davis, Hartley, Jordan, Miller, and
  Driessche]{arino2005multi}
Julien Arino, Jonathan~R Davis, David Hartley, Richard Jordan, Joy~M Miller,
  and P~Van~Den Driessche.
\newblock A multi-species epidemic model with spatial dynamics.
\newblock \emph{Math Med Biol}, 22:\penalty0 129--142, 2005.

\bibitem[Bansal et~al.(2007)Bansal, Grenfell, and Meyers]{bansal2007individual}
Shweta Bansal, Bryan~T Grenfell, and Lauren~Ancel Meyers.
\newblock When individual behaviour matters: homogeneous and network models in
  epidemiology.
\newblock \emph{J Roy Soc Interface}, 4\penalty0 (16):\penalty0 879--891, 2007.

\bibitem[Keeling and Eames(2005)]{keeling2005networks}
Matt~J Keeling and Ken~TD Eames.
\newblock Networks and epidemic models.
\newblock \emph{J Roy Soc Interface}, 2\penalty0 (4):\penalty0 295--307, 2005.

\bibitem[Alon et~al.(1999)Alon, Surette, Barkai, and Leibler]{Alon1999}
Uri Alon, Michael~G Surette, Naama Barkai, and Stanislas Leibler.
\newblock Robustness in bacterial chemotaxis.
\newblock \emph{Nature}, 397:\penalty0 168--171, 1999.

\bibitem[Blanchini and Franco(2011)]{Blanchini2011}
Franco Blanchini and Elisa Franco.
\newblock Structurally robust biological networks.
\newblock \emph{BMC Sys Biol}, 5:\penalty0 74, 2011.

\bibitem[Chesi(2011)]{Chesi2011b}
Graziano Chesi.
\newblock Robustness analysis of genetic regulatory networks affected by model
  uncertainty.
\newblock \emph{Automatica}, 47:\penalty0 1131--1138, 2011.

\bibitem[Kim et~al.(2006)Kim, Bates, Postlethwaite, Ma, and Iglesias]{Kim2006}
J.~Kim, D.~G. Bates, I.~Postlethwaite, L.~Ma, and P.~A. Iglesias.
\newblock Robustness analysis of biochemical network models.
\newblock \emph{IEEE Proc Sys Biol}, 153:\penalty0 96--104, 2006.

\bibitem[Shinar et~al.(2009)Shinar, Alon, and Feinberg]{Shinar2009}
Guy Shinar, Uri Alon, and Martin Feinberg.
\newblock Sensitivity and robustness in chemical reaction networks.
\newblock \emph{Siam J Appl Math}, 69:\penalty0 977--998, 2009.

\bibitem[Shinar and Feinberg(2010)]{Shinar2010}
Guy Shinar and Martin Feinberg.
\newblock Structural sources of robustness in biochemical reaction networks.
\newblock \emph{Science}, 327:\penalty0 1389--1391, 2010.

\bibitem[Streif et~al.(2016)Streif, Kim, Rumschinski, Kishida, Shen, Findeisen,
  and Braatz]{Streif2016}
Stefan Streif, Kwang Ki~K. Kim, Philipp Rumschinski, Masako Kishida,
  Dongying~Erin Shen, Rolf Findeisen, and Richard~D. Braatz.
\newblock Robustness analysis, prediction, and estimation for uncertain
  biochemical networks: An overview.
\newblock \emph{J Proc Control}, 42:\penalty0 14--34, 2016.

\bibitem[Waldherr and Allgöwer(2011)]{Waldherr2011}
Steffen Waldherr and Frank Allgöwer.
\newblock Robust stability and instability of biochemical networks with
  parametric uncertainty.
\newblock \emph{Automatica}, 47:\penalty0 1139--1146, 2011.

\bibitem[Blanchini and Giordano(2021)]{blanchini2021structural}
Franco Blanchini and Giulia Giordano.
\newblock Structural analysis in biology: A control-theoretic approach.
\newblock \emph{Automatica}, 126:\penalty0 109376, 2021.

\bibitem[Aldana and Cluzel(2003)]{Aldana2003}
Maximino Aldana and Philippe Cluzel.
\newblock A natural class of robust networks.
\newblock \emph{P Natl Acad Sci Usa}, 100:\penalty0 8710--8714, 2003.

\bibitem[Li et~al.(2022)Li, Wang, Zhong, Sun, Guo, Chen, and Fu]{Li2022}
Jie Li, Ying Wang, Jilong Zhong, Yun Sun, Zhijun Guo, Zhiwei Chen, and Chaoqi
  Fu.
\newblock Network resilience assessment and reinforcement strategy against
  cascading failure.
\newblock \emph{Chaos Soliton Fract}, 160:\penalty0 112271, 2022.

\bibitem[Liu et~al.(2022)Liu, Li, Ma, Szymanski, Stanley, and Gao]{Liu2020b}
Xueming Liu, Daqing Li, Manqing Ma, Boleslaw~K Szymanski, H~Eugene Stanley, and
  Jianxi Gao.
\newblock Network resilience.
\newblock \emph{Phys Rep}, 971:\penalty0 1--108, 2022.

\bibitem[Vitkup and Koonin(2004)]{Vitkup2004}
Dennis Vitkup and Eugene Koonin.
\newblock Biological networks: from physical principles to biological insights.
\newblock \emph{Genome Biol}, 5:\penalty0 313, 2004.

\bibitem[Atkinson(1965)]{Atkinson1965}
Daniel~E Atkinson.
\newblock \emph{Biological Feedback Control at the Molecular Level}, volume
  150.
\newblock Science, New Series, 1965.

\bibitem[Blanchini et~al.(2018{\natexlab{b}})Blanchini, {El-Samad}, Giordano,
  and Sontag]{Blanchini2018}
Franco Blanchini, Hana {El-Samad}, Giulia Giordano, and Eduardo~D Sontag.
\newblock Control-theoretic methods for biological networks.
\newblock In \emph{Proc IEEE CDC}, pages 466--483, 2018{\natexlab{b}}.

\bibitem[{El-Samad} et~al.(2006){El-Samad}, Prajna, Papachristodoulou, Doyle,
  and Khammash]{samad2006}
Hana {El-Samad}, Stephen Prajna, Antonis Papachristodoulou, John Doyle, and
  Mustafa Khammash.
\newblock Advanced methods and algorithms for biological networks analysis.
\newblock \emph{Proc IEEE}, 94:\penalty0 832--852, 2006.

\bibitem[Stelling et~al.(2002)Stelling, Klamt, Bettenbrock, Schuster, and
  Gilles]{stelling2002}
Jörg Stelling, Steffen Klamt, Katja Bettenbrock, Stefan Schuster, and
  Ernst~Dieter Gilles.
\newblock Metabolic network structure determines key aspects of functionality
  and regulation.
\newblock \emph{Nature}, 420:\penalty0 190--193, 2002.

\bibitem[Barmish(1994)]{Barmish1994}
B~Ross Barmish.
\newblock \emph{New Tools for Robustness of Linear Systems}.
\newblock Macmillan Publishing Company, 1994.

\bibitem[Sánchez-Peña and Sznaier(1998)]{sanchezpena1998}
Ricardo~S. Sánchez-Peña and Mario Sznaier.
\newblock \emph{Robust Systems: Theory and Applications}.
\newblock Wiley, 1998.

\bibitem[Zhou and Doyle(1998)]{zhou1998essentials}
Kemin Zhou and John~Comstock Doyle.
\newblock \emph{Essentials of Robust Control}, volume 104.
\newblock Prentice hall Upper Saddle River, NJ, 1998.

\bibitem[Baetica et~al.(2019)Baetica, Westbrook, and {El-Samad}]{Baetica2019}
Ania~Ariadna Baetica, Alexandra Westbrook, and Hana {El-Samad}.
\newblock Control theoretical concepts for synthetic and systems biology.
\newblock \emph{Curr Opin Sys Biol}, 14:\penalty0 50--57, 2019.

\bibitem[{Del Vecchio} et~al.(2016){Del Vecchio}, Dy, and Qian]{DelVecchio2016}
Domitilla {Del Vecchio}, Aaron~J. Dy, and Yili Qian.
\newblock Control theory meets synthetic biology.
\newblock \emph{J Roy Soc Interface}, 13, 2016.

\bibitem[Garabed et~al.(2019)Garabed, Jolles, Garira, Lanzas, Gutierrez, and
  Rempala]{Garabed2019}
Rebecca Garabed, Anna~E. Jolles, Winston Garira, Cristina Lanzas, Juan~B.
  Gutierrez, and Grzegorz~A. Rempala.
\newblock Multi-scale dynamics of infectious diseases.
\newblock \emph{Interface Focus}, 10:\penalty0 20190118--NA, 2019.

\bibitem[{Del Vecchio} et~al.(2008){Del Vecchio}, Ninfa, and
  Sontag]{DelVecchio2008}
Domitilla {Del Vecchio}, Alexander~J. Ninfa, and Eduardo~D. Sontag.
\newblock Modular cell biology: Retroactivity and insulation.
\newblock \emph{Mol Sys Biol}, 4, 2008.

\bibitem[{Del Vecchio} and Sontag(2009)]{DelVecchio2009}
Domitilla {Del Vecchio} and Eduardo~D. Sontag.
\newblock Engineering principles in bio-molecular systems: From retroactivity
  to modularity.
\newblock \emph{Europ J Contr}, 15:\penalty0 389--397, 2009.

\bibitem[Jayanthi and {Del Vecchio}(2011)]{Jayanthi2011}
Shridhar Jayanthi and Domitilla {Del Vecchio}.
\newblock Retroactivity attenuation in bio-molecular systems based on timescale
  separation.
\newblock \emph{IEEE T Automat Contr}, 56:\penalty0 748--761, 2011.

\bibitem[McBride et~al.(2019)McBride, Shah, and {Del Vecchio}]{McBride2019}
Cameron McBride, Rushina Shah, and Domitilla {Del Vecchio}.
\newblock The effect of loads in molecular communications.
\newblock \emph{Proc IEEE}, 107:\penalty0 1369--1386, 2019.

\bibitem[Bailey(2001)]{Bailey2001}
James~E Bailey.
\newblock Complex biology with no parameters.
\newblock \emph{Nat Biotechnol}, 19:\penalty0 503--504, 2001.

\bibitem[Liu et~al.(2013)Liu, Slotine, and Barabási]{Liu2013}
Yang-Yu Liu, Jean-Jacques Slotine, and Albert-László Barabási.
\newblock Observability of complex systems.
\newblock \emph{P Natl Acad Sci Usa}, 110:\penalty0 2460--2465, 2013.

\bibitem[Roda et~al.(2020)Roda, Varughese, Han, and Li]{Roda2020}
Weston~C. Roda, Marie Varughese, Donglin Han, and Michael~Y. Li.
\newblock Why is it difficult to accurately predict the covid-19 epidemic?
\newblock \emph{Infec Dis Mod}, 5\penalty0 (5):\penalty0 271--281, 2020.

\bibitem[Bagheri et~al.(2008)Bagheri, Taylor, Meeker, Petzold, and
  Doyle]{Bagheri2008}
Neda Bagheri, Stephanie~R. Taylor, Kirsten Meeker, Linda~R. Petzold, and
  Francis~J. Doyle.
\newblock Synchrony and entrainment properties of robust circadian oscillators.
\newblock \emph{J Roy Soc Interface}, 5, 2008.

\bibitem[Barkai and Leibler(1997)]{Barkai1997}
N~Barkai and S~Leibler.
\newblock Robustness in simple biochemical networks.
\newblock \emph{Nature}, 387:\penalty0 913--917, 1997.

\bibitem[Alon(2007)]{Alon2007}
Uri Alon.
\newblock Network motifs: Theory and experimental approaches.
\newblock \emph{Nat Rev Genet}, 8:\penalty0 450--461, 2007.

\bibitem[Milo et~al.(2002)Milo, Shen-Orr, Itzkovitz, Kashtan, Chkolovskii, and
  Alon]{Milo2002a}
R.~Milo, S.~Shen-Orr, S.~Itzkovitz, N.~Kashtan, D.~Chkolovskii, and U.~Alon.
\newblock Network motifs: simple building blocks of complex networks.
\newblock \emph{Science}, 298:\penalty0 824--827, 2002.

\bibitem[Stone et~al.(2019)Stone, Simberloff, and Artzy-Randrup]{Stone2019}
Lewi Stone, Daniel Simberloff, and Yael Artzy-Randrup.
\newblock Network motifs and their origins.
\newblock \emph{PLoS Comput Biol}, 15:\penalty0 1--7, 2019.

\bibitem[Yeger-Lotem et~al.(2004)Yeger-Lotem, Sattath, Kashtan, Itzkovitz,
  Milo, Pinter, Alon, and Margalit]{yeger2004}
Esti Yeger-Lotem, Shmuel Sattath, Nadav Kashtan, Shalev Itzkovitz, Ron Milo,
  Ron~Y Pinter, Uri Alon, and Hanah Margalit.
\newblock Network motifs in integrated cellular networks of
  transcription-regulation and protein-protein interaction.
\newblock \emph{P Natl Acad Sci Usa}, 101:\penalty0 5934--5939, 2004.

\bibitem[Angeli and Sontag(2012)]{Angeli2012}
David Angeli and Eduardo Sontag.
\newblock Remarks on the invalidation of biological models using monotone
  systems theory.
\newblock In \emph{Proc 51st IEEE CDC}, pages 2989--2994, 2012.

\bibitem[Blanchini et~al.(2012)Blanchini, Franco, and Giordano]{Blanchini2012}
F~Blanchini, E~Franco, and G~Giordano.
\newblock Determining the structural properties of a class of biological
  models.
\newblock In \emph{Proc IEEE CDC}, pages 5505--5510, 2012.

\bibitem[Jacquez and Simon(1993)]{Jacquezt1993}
John~A Jacquez and Carl~P Simon.
\newblock Qualitative theory of compartmental systems.
\newblock \emph{SIAM Rev}, 35:\penalty0 43--79, 1993.

\bibitem[Kaltenbach et~al.(2009)Kaltenbach, Dimopoulos, and
  Stelling]{Kaltenbach2009}
Hans~Michael Kaltenbach, Sotiris Dimopoulos, and Jörg Stelling.
\newblock Systems analysis of cellular networks under uncertainty.
\newblock \emph{FEBS Letters}, 583:\penalty0 3923--3930, 2009.

\bibitem[Reder(1988)]{Reder1988}
Christine Reder.
\newblock Metabolic control theory: A structural approach.
\newblock \emph{J Theo Biol}, 135:\penalty0 175--201, 1988.

\bibitem[Hara et~al.(2019)Hara, Iwasaki, and Hori]{Hara2019}
Shinji Hara, Tetsuya Iwasaki, and Yutaka Hori.
\newblock Robust stability analysis for lti systems with generalized frequency
  variables and its application to gene regulatory networks.
\newblock \emph{Automatica}, 105:\penalty0 96--106, 2019.

\bibitem[Ma and Iglesias(2002)]{Ma2002}
Lan Ma and Pablo~A Iglesias.
\newblock Quantifying robustness of biochemical network models.
\newblock \emph{BMC Bioinf}, 3:\penalty0 1--13, 2002.

\bibitem[Bhamra et~al.(2011)Bhamra, Dani, and Burnard]{bhamra2011resilience}
Ran Bhamra, Samir Dani, and Kevin Burnard.
\newblock Resilience: the concept, a literature review and future directions.
\newblock \emph{Int J Prod Res}, 49:\penalty0 5375--5393, 2011.

\bibitem[Carpenter et~al.(2012)Carpenter, Arrow, Barrett, Biggs, Brock,
  Cr{\'e}pin, Engstr{\"o}m, Folke, Hughes, Kautsky,
  et~al.]{carpenter2012general}
Stephen~R Carpenter, Kenneth~J Arrow, Scott Barrett, Reinette Biggs, William~A
  Brock, Anne-Sophie Cr{\'e}pin, Gustav Engstr{\"o}m, Carl Folke, Terry~P
  Hughes, Nils Kautsky, et~al.
\newblock General resilience to cope with extreme events.
\newblock \emph{Sustainability}, 4\penalty0 (12):\penalty0 3248--3259, 2012.

\bibitem[Dakos et~al.(2015)Dakos, Carpenter, van Nes, and Scheffer]{Dakos2015}
Vasilis Dakos, Stephen~R. Carpenter, Egbert~H. van Nes, and Maórten Scheffer.
\newblock Resilience indicators: Prospects and limitations for early warnings
  of regime shifts.
\newblock \emph{Philos T R Soc B}, 370:\penalty0 1--10, 2015.

\bibitem[Fraccascia et~al.(2018)Fraccascia, Giannoccaro, and
  Albino]{fraccascia2018resilience}
Luca Fraccascia, Ilaria Giannoccaro, and Vito Albino.
\newblock Resilience of complex systems: state of the art and directions for
  future research.
\newblock \emph{Complexity}, 2018, 2018.

\bibitem[Fisher(2015)]{fisher2015more}
Len Fisher.
\newblock More than 70 ways to show resilience.
\newblock \emph{Nature}, 518:\penalty0 35, 2015.

\bibitem[Gao et~al.(2016)Gao, Barzel, and Barabási]{Gao2016}
Jianxi Gao, Baruch Barzel, and Albert~László Barabási.
\newblock Universal resilience patterns in complex networks.
\newblock \emph{Nature}, 530:\penalty0 307--312, 2016.

\bibitem[Gunderson(2000)]{gunderson2000ecological}
Lance~H Gunderson.
\newblock Ecological resilience—in theory and application.
\newblock \emph{Annu Rev Eco System}, 31\penalty0 (1):\penalty0 425--439, 2000.

\bibitem[Holling(1996)]{holling1996engineering}
Crawford~Stanley Holling.
\newblock Engineering resilience versus ecological resilience.
\newblock \emph{Engineering within ecological constraints}, 31:\penalty0 32,
  1996.

\bibitem[Meyer(2016)]{meyer2016mathematical}
Katherine Meyer.
\newblock A mathematical review of resilience in ecology.
\newblock \emph{Nat Res Mod}, 29\penalty0 (3):\penalty0 339--352, 2016.

\bibitem[Weerakkody et~al.(2019)Weerakkody, Ozel, Mo, Sinopoli,
  et~al.]{weerakkody2019resilient}
Sean Weerakkody, Omur Ozel, Yilin Mo, Bruno Sinopoli, et~al.
\newblock Resilient control in cyber-physical systems: Countering uncertainty,
  constraints, and adversarial behavior.
\newblock \emph{Found Trends Sys Contr}, 7\penalty0 (1-2):\penalty0 1--252,
  2019.

\bibitem[Arpino et~al.(2013)Arpino, Hancock, Anderson, Barahona, Stan,
  Papachristodoulou, and Polizzi]{Arpino2013}
James~A.J. Arpino, Edward~J. Hancock, James Anderson, Mauricio Barahona, Guy
  Bart~V. Stan, Antonis Papachristodoulou, and Karen Polizzi.
\newblock Tuning the dials of synthetic biology.
\newblock \emph{Microbiology}, 159:\penalty0 1236--1253, 2013.

\bibitem[Batt et~al.(2007)Batt, Yordanov, Weiss, and Belta]{Batt2007}
Grégory Batt, Boyan Yordanov, Ron Weiss, and Calin Belta.
\newblock Robustness analysis and tuning of synthetic gene networks.
\newblock \emph{Bioinformatics}, 23:\penalty0 2415--2422, 2007.

\bibitem[Benner and Sismour(2005)]{benner2005synthetic}
Steven~A Benner and A~Michael Sismour.
\newblock Synthetic biology.
\newblock \emph{Nat Rev Genet}, 6:\penalty0 533--543, 2005.

\bibitem[Csete and Doyle(2002)]{Csete2002}
Marie~E Csete and John~C Doyle.
\newblock Reverse engineering of biological complexity.
\newblock \emph{Science}, 295:\penalty0 1664--1669, 2002.

\bibitem[{Del Vecchio} and Murray(2015)]{DelVecchio2015}
Domitilla {Del Vecchio} and Richard Murray.
\newblock \emph{Biomolecular Feedback Systems}.
\newblock Princeton UP, 2015.

\bibitem[{Del Vecchio} et~al.(2018){Del Vecchio}, Qian, Murray, and
  Sontag]{DelVecchio2018}
Domitilla {Del Vecchio}, Yili Qian, Richard~M. Murray, and Eduardo~D. Sontag.
\newblock Future systems and control research in synthetic biology.
\newblock \emph{Annu Rev Control}, 45:\penalty0 5--17, 2018.

\bibitem[Hsiao et~al.(2018)Hsiao, Swaminathan, and Murray]{Hsiao2018}
Victoria Hsiao, Anandh Swaminathan, and Richard~M. Murray.
\newblock Control theory for synthetic biology: recent advances in system
  characterization, control design, and controller implementation for synthetic
  biology.
\newblock \emph{IEEE Control Sys}, 38:\penalty0 32--62, 2018.

\bibitem[Koeppl et~al.(2011)Koeppl, Densmore, Setti, and
  di~Bernardo]{Koeppl2011}
Heinz Koeppl, Douglas Densmore, Gianluca Setti, and Mario di~Bernardo.
\newblock \emph{Design and Analysis of Biomolecular Circuits}.
\newblock Springer, 2011.

\bibitem[Sootla et~al.(2016)Sootla, Oyarzún, Angeli, and Stan]{Sootla2016}
Aivar Sootla, Diego Oyarzún, David Angeli, and Guy~Bart Stan.
\newblock Shaping pulses to control bistable systems: Analysis, computation and
  counterexamples.
\newblock \emph{Automatica}, 63:\penalty0 254--264, 2016.

\bibitem[Schön et~al.(2011)Schön, Wills, and Ninness]{Schön2011}
Thomas~B. Schön, Adrian Wills, and Brett Ninness.
\newblock System identification of nonlinear state-space models.
\newblock \emph{Automatica}, 47:\penalty0 39--49, 2011.

\bibitem[Chiuso and Pillonetto(2012)]{chiuso2012bayesian}
Alessandro Chiuso and Gianluigi Pillonetto.
\newblock A bayesian approach to sparse dynamic network identification.
\newblock \emph{Automatica}, 48\penalty0 (8):\penalty0 1553--1565, 2012.

\bibitem[Papachristodoulou and Recht(2007)]{Papachristodoulou2007}
Antonis Papachristodoulou and Ben Recht.
\newblock Determining interconnections in chemical reaction networks.
\newblock In \emph{Proc ACC}, pages 4872--4877, 2007.

\bibitem[Porreca et~al.(2008)Porreca, Drulhe, Jong, and
  Ferrari-Trecate]{Porreca2008}
Riccardo Porreca, Samuel Drulhe, Hidde~De Jong, and Giancarlo Ferrari-Trecate.
\newblock Structural identification of piecewise-linear models of genetic
  regulatory networks.
\newblock \emph{Int Comput Biol}, 15:\penalty0 1365--1380, 2008.

\bibitem[Porreca et~al.(2012)Porreca, Cinquemani, Lygeros, and
  Ferrari-Trecate]{Porreca2012}
Riccardo Porreca, Eugenio Cinquemani, John Lygeros, and Giancarlo
  Ferrari-Trecate.
\newblock Invalidation of the structure of genetic network dynamics: A
  geometric approach.
\newblock \emph{Int J Robust Nonlin}, 22:\penalty0 1140--1156, 2012.

\bibitem[Yu et~al.(2007)Yu, Dong, Altimus, Tang, Griffith, Morello, Dudek,
  Arnold, and Schüttler]{Yu2007}
Yihai Yu, Wubei Dong, Cara Altimus, Xiaojia Tang, James Griffith, Melissa
  Morello, Lisa Dudek, Jonathan Arnold, and Heinz-Bernd Schüttler.
\newblock A genetic network for the clock of neurospora crassa.
\newblock \emph{P Natl Acad Sci Usa}, 104:\penalty0 2809--2814, 2007.

\bibitem[Gon{\c{c}}alves and Warnick(2008)]{gonccalves2008necessary}
Jorge Gon{\c{c}}alves and Sean Warnick.
\newblock Necessary and sufficient conditions for dynamical structure
  reconstruction of lti networks.
\newblock \emph{IEEE T Autom Contr}, 53\penalty0 (7):\penalty0 1670--1674,
  2008.

\bibitem[Materassi and Innocenti(2010)]{materassi2010topological}
Donatello Materassi and Giacomo Innocenti.
\newblock Topological identification in networks of dynamical systems.
\newblock \emph{IEEE T Autom Contr}, 55\penalty0 (8):\penalty0 1860--1871,
  2010.

\bibitem[Materassi and Salapaka(2012)]{materassi2012problem}
Donatello Materassi and Murti~V Salapaka.
\newblock On the problem of reconstructing an unknown topology via locality
  properties of the wiener filter.
\newblock \emph{IEEE T Autom Contr}, 57\penalty0 (7):\penalty0 1765--1777,
  2012.

\bibitem[Aalto et~al.(2020)Aalto, Viitasaari, Ilmonen, Mombaerts, and
  Gonçalves]{aalto2020gene}
Atte Aalto, Lauri Viitasaari, Pauliina Ilmonen, Laurent Mombaerts, and Jorge
  Gonçalves.
\newblock Gene regulatory network inference from sparsely sampled noisy data.
\newblock \emph{Nat Commun}, 11:\penalty0 1--9, 2020.

\bibitem[Peixoto(2019)]{peixoto2019network}
Tiago~P Peixoto.
\newblock Network reconstruction and community detection from dynamics.
\newblock \emph{Phys Rev Lett}, 123\penalty0 (12):\penalty0 128301, 2019.

\bibitem[Sutulovic et~al.(2025{\natexlab{a}})Sutulovic, Proverbio, Katz, and
  Giordano]{sutulovic2025differentiator}
Uros Sutulovic, Daniele Proverbio, Rami Katz, and Giulia Giordano.
\newblock Efficient and faithful reconstruction of dynamical attractors using
  homogeneous differentiators.
\newblock \emph{Chaos Soliton Fract}, 3\penalty0 (199):\penalty0 116798,
  2025{\natexlab{a}}.

\bibitem[Runge(2018)]{runge2018causal}
Jakob Runge.
\newblock Causal network reconstruction from time series: From theoretical
  assumptions to practical estimation.
\newblock \emph{Chaos}, 28\penalty0 (7), 2018.

\bibitem[Albert et~al.(2018)Albert, Baillieul, and Motter]{Albert2018}
Reka Albert, John Baillieul, and Adilson~E. Motter.
\newblock Introduction to the special issue on approaches to control biological
  and biologically inspired networks.
\newblock \emph{IEEE T Contr Network Sys}, 5:\penalty0 690--693, 2018.

\bibitem[Allgöwer and Doyle(2011)]{Allgower2011}
Frank Allgöwer and Frank Doyle.
\newblock Introduction to the special issue on systems biology.
\newblock \emph{Automatica}, 47:\penalty0 1095--1096, 2011.

\bibitem[Arcak et~al.(2019)Arcak, Blanchini, and Vidyasagar]{Arcak2019}
Murat Arcak, Franco Blanchini, and M.~Vidyasagar.
\newblock Editorial to the special issue of l-css on control and network theory
  for biological systems.
\newblock \emph{IEEE Contr Sys Lett}, 3:\penalty0 228--229, 2019.

\bibitem[Chesi and Chen(2011)]{Chesi2011}
Graziano Chesi and Luonan Chen.
\newblock Systems biology.
\newblock \emph{Int J Robust Nonlin}, 21:\penalty0 1729, 2011.

\bibitem[Khammash et~al.(2008)Khammash, Tomlin, and Vidyasagar]{Khammash2008}
Mustafa Khammash, Claire~Jennifer Tomlin, and M.~Vidyasagar.
\newblock Guest editorial - special issue on systems biology.
\newblock \emph{IEEE T Automat Contr}, 53:\penalty0 4--7, 2008.

\bibitem[Yue and Bullinger(2010)]{Yue2010}
Hong Yue and Eric Bullinger.
\newblock Editorial for special issue on 'robustness in systems biology:
  Methods and applications'.
\newblock \emph{Int J Robust Nonlin}, 20:\penalty0 1015--1016, 2010.

\bibitem[Ashwin et~al.(2016)Ashwin, Coombes, and Nicks]{Ashwin2016d}
Peter Ashwin, Stephen Coombes, and Rachel Nicks.
\newblock Mathematical frameworks for oscillatory network dynamics in
  neuroscience.
\newblock \emph{J Math Neurosci}, 6:\penalty0 1--92, 2016.

\bibitem[Brockmann and Helbing(2013)]{Brockmann2013}
Dirk Brockmann and Dirk Helbing.
\newblock The hidden geometry of complex, network-driven contagion phenomena.
\newblock \emph{Science}, 342\penalty0 (6164):\penalty0 1337--1342, 2013.

\bibitem[Hethcote(2000)]{Hethcote2000}
Herbert~W Hethcote.
\newblock The mathematics of infectious diseases.
\newblock \emph{SIAM Rev}, 42:\penalty0 599--653, 2000.

\bibitem[Iglesias et~al.(2007)Iglesias, Khammash, Munsky, Sontag, and {Del
  Vecchio}]{Iglesias2007}
Pablo~A Iglesias, Mustafa Khammash, Brian Munsky, Eduardo~D Sontag, and
  Domitilla {Del Vecchio}.
\newblock Systems biology and control - a tutorial.
\newblock In \emph{46th IEEE CDC}, pages 1--12, 2007.

\bibitem[Pastor-Satorras et~al.(2015)Pastor-Satorras, Castellano, Mieghem, and
  Vespignani]{Pastor2015}
Romualdo Pastor-Satorras, Claudio Castellano, Piet~Van Mieghem, and Alessandro
  Vespignani.
\newblock Epidemic processes in complex networks.
\newblock \emph{Rev Mod Phys}, 87:\penalty0 925--979, 2015.

\bibitem[Sontag(2005)]{Sontag2005}
Eduardo~D Sontag.
\newblock Molecular systems biology and control.
\newblock \emph{Eur J Control}, 11\penalty0 (4-5):\penalty0 396--435, 2005.

\bibitem[Allen(2014)]{J.S.Allen2014}
J.~S.~L. Allen.
\newblock \emph{An Introduction to Stochastic Process with Application in
  Biology}, volume~25.
\newblock CRC press, 2014.

\bibitem[Allman and Rhodes(2004)]{allman2004mathematical}
Elizabeth~S Allman and John~A Rhodes.
\newblock \emph{Mathematical Models in Biology: an Introduction}.
\newblock Cambridge UP, 2004.

\bibitem[Blackwood and Childs(2018)]{Blackwood2018}
Julie~C Blackwood and Lauren~M Childs.
\newblock An introduction to compartmental modeling for the budding infectious
  disease modeler.
\newblock \emph{Lett Biomath}, 5:\penalty0 195--221, 2018.

\bibitem[Brauer(2008)]{Brauer2008}
Fred Brauer.
\newblock Compartmental models in epidemiology.
\newblock \emph{Mathematical epidemiology}, pages 19--79, 2008.

\bibitem[Iglesias and Ingalls(2010)]{Iglesias2010}
B.P. Iglesias and P.A. Ingalls.
\newblock \emph{Control Theory and Systems Biology}.
\newblock MIT Press, 2010.

\bibitem[Ingalls(2013)]{Ingalls2013}
Brian Ingalls.
\newblock \emph{Mathematical Modelling in Systems Biology: An Introduction}.
\newblock MIT Press, 2013.

\bibitem[Jayaraman and Hahn(2009)]{Jayaraman2009}
Arul Jayaraman and Juergen Hahn.
\newblock \emph{Methods in Bioengineering: Systems Analysis of Biological
  Networks}.
\newblock Artech House, 2009.

\bibitem[Martcheva(2015)]{martcheva2015introduction}
Maia Martcheva.
\newblock \emph{An Introduction to Mathematical Epidemiology}, volume~61.
\newblock Springer, 2015.

\bibitem[Nowzari et~al.(2016)Nowzari, Preciado, and Pappas]{Nowzari2016}
Cameron Nowzari, Victor~M. Preciado, and George~J. Pappas.
\newblock Analysis and control of epidemics: A survey of spreading processes on
  complex networks.
\newblock \emph{IEEE Control Sys}, 36\penalty0 (1):\penalty0 26--46, 2016.

\bibitem[Queinnec et~al.(2007)Queinnec, Tarbouriech, Garcia, and
  Niculescu]{Queinnec2007}
Isabelle Queinnec, Sophie Tarbouriech, Germain Garcia, and Silviu-Iulian
  Niculescu.
\newblock \emph{Biology and Control Theory: Current Challenges}.
\newblock Springer, 2007.

\bibitem[Segel and Edelstein-Keshet(2013)]{Segel2013}
Lee~A. Segel and Leah Edelstein-Keshet.
\newblock \emph{A Primer on Mathematical Models in Biology}, volume 129.
\newblock SIAM, 2013.

\bibitem[Strogatz(2018)]{strogatz2018nonlinear}
Steven~H Strogatz.
\newblock \emph{Nonlinear Dynamics and Chaos with Applications to Physics,
  Biology, Chemistry, and Engineering}.
\newblock CRC press, 2018.

\bibitem[Szallasi et~al.(2006)Szallasi, Stelling, and Periwal]{Szallasi2006}
Zoltan Szallasi, Jorg Stelling, and Vipul Periwal.
\newblock \emph{System Modeling in Cell Biology: from Concepts to Nuts and
  Bolts}.
\newblock MIT Press, 2006.

\bibitem[Alberts et~al.(2015)Alberts, Bray, Hopkin, Johnson, Lewis, Raff,
  et~al.]{Alberts2015}
Bruce Alberts, Dennis Bray, Karen Hopkin, Alexander~D Johnson, Julian Lewis,
  Martin Raff, et~al.
\newblock \emph{Essential Cell Biology}.
\newblock Garland Science, 4th edition, 2015.

\bibitem[Bertolaso et~al.(2019)Bertolaso, Caianiello, and
  Serrelli]{Bertolaso2019}
Marta Bertolaso, Silvia Caianiello, and Emanuele Serrelli.
\newblock \emph{Biological Robustness: Emerging Perspectives from within the
  Life Sciences}, volume~23.
\newblock Springer, 1 edition, 2019.

\bibitem[Breindl et~al.(2011)Breindl, Waldherr, Wittmann, Theis, and
  Allgöwer]{Breindl2011}
C.~Breindl, S.~Waldherr, D.~M. Wittmann, F.~J. Theis, and F.~Allgöwer.
\newblock Steady-state robustness of qualitative gene regulation networks.
\newblock \emph{Int J Robust Nonlin}, 21:\penalty0 1742--1758, 2011.

\bibitem[Dethier et~al.(2015)Dethier, Drion, Franci, and
  Sepulchre]{dethier2015positive}
Julie Dethier, Guillaume Drion, Alessio Franci, and Rodolphe Sepulchre.
\newblock A positive feedback at the cellular level promotes robustness and
  modulation at the circuit level.
\newblock \emph{Am J Physio - Heart Circ Physio}, 114:\penalty0 2472--2484,
  2015.

\bibitem[El-Samad et~al.(2005)El-Samad, Kurata, Doyle, Gross, and
  Khammash]{samad2005b}
HJCM El-Samad, H~Kurata, JC~Doyle, CA~Gross, and M~Khammash.
\newblock Surviving heat shock: control strategies for robustness and
  performance.
\newblock \emph{Proc Natl Acad Sci USA}, 102\penalty0 (8):\penalty0 2736--2741,
  2005.

\bibitem[Faulon et~al.(2004)Faulon, Martin, and Carr]{Faulon2004}
J-L Faulon, Shawn Martin, and Robert~D Carr.
\newblock Dynamical robustness in gene regulatory networks.
\newblock In \emph{Proc IEEE CSB}, pages 595--596, 2004.

\bibitem[Félix and Barkoulas(2015)]{felix2015}
Marie~Anne Félix and Michalis Barkoulas.
\newblock Pervasive robustness in biological systems.
\newblock \emph{Nat Rev Genet}, 16:\penalty0 483--496, 2015.

\bibitem[Ghaemi et~al.(2009)Ghaemi, Sun, Iglesias, and {Del
  Vecchio}]{Ghaemi2009}
Reza Ghaemi, Jing Sun, Pablo~A. Iglesias, and Domitilla {Del Vecchio}.
\newblock A method for determining the robustness of bio-molecular oscillator
  models.
\newblock \emph{BMC Sys Biol}, 3:\penalty0 95, 2009.

\bibitem[Hernandez(2009)]{Hernandez2009}
Maria~Josefina Hernandez.
\newblock Disentangling nature, strength and stability issues in the
  characterization of population interactions.
\newblock \emph{J Theor Biol}, 261:\penalty0 107--119, 2009.

\bibitem[Kitano(2007{\natexlab{a}})]{Kitano2007}
Hiroaki Kitano.
\newblock A robustness-based approach to systems-oriented drug design.
\newblock \emph{Nature}, 6:\penalty0 202--210, 2007{\natexlab{a}}.

\bibitem[Little et~al.(1999)Little, Shepley, and Wert]{Little1999}
John~W Little, Donald~P Shepley, and David~W Wert.
\newblock Robustness of a gene regulatory circuit.
\newblock \emph{EMBO J}, 18:\penalty0 4299--4307, 1999.

\bibitem[Paulino et~al.(2019)Paulino, Foo, Kim, and Bates]{Paulino2019}
Nuno~M.G. Paulino, Mathias Foo, Jongmin Kim, and Declan~G. Bates.
\newblock Robustness analysis of a nucleic acid controller for a dynamic
  biomolecular process using the structured singular value.
\newblock \emph{J Proc Contr}, 78:\penalty0 34--44, 2019.

\bibitem[Qian and {Del Vecchio}(2021)]{Qian2021}
Yili Qian and Domitilla {Del Vecchio}.
\newblock Robustness of networked systems to unintended interactions with
  application to engineered genetic circuits.
\newblock \emph{IEEE T Contr Network Sys}, 8:\penalty0 1705--1716, 2021.

\bibitem[Venturelli et~al.(2012)Venturelli, {El-Samad}, and
  Murray]{Venturelli2012}
Ophelia~S. Venturelli, Hana {El-Samad}, and Richard~M. Murray.
\newblock Synergistic dual positive feedback loops established by molecular
  sequestration generate robust bimodal response.
\newblock \emph{P Natl Acad Sci Usa}, 109, 2012.

\bibitem[Barabási(2013)]{Barabasi2013}
Albert-László Barabási.
\newblock Network science.
\newblock \emph{Philos T R Soc A}, 371:\penalty0 20120375, 2013.

\bibitem[Prill et~al.(2005)Prill, Iglesias, and Levchenko]{Prill2005}
Robert~J. Prill, Pablo~A. Iglesias, and Andre Levchenko.
\newblock Dynamic properties of network motifs contribute to biological network
  organization.
\newblock \emph{PLoS Biol}, 3:\penalty0 1881--1892, 2005.

\bibitem[Battiston et~al.(2021)Battiston, Amico, Barrat, Bianconi, de~Arruda,
  Franceschiello, et~al.]{Battiston2021}
Federico Battiston, Enrico Amico, Alain Barrat, Ginestra Bianconi,
  Guilherme~Ferraz de~Arruda, Benedetta Franceschiello, et~al.
\newblock The physics of higher-order interactions in complex systems.
\newblock \emph{Nat Phys}, 17:\penalty0 1093--1098, 2021.

\bibitem[Bianconi(2021)]{bianconi2021higher}
Ginestra Bianconi.
\newblock \emph{Higher-order Networks}.
\newblock Cambridge UP, 2021.

\bibitem[Cisneros-Velarde and Bullo(2021)]{cisneros2021multigroup}
Pedro Cisneros-Velarde and Francesco Bullo.
\newblock Multigroup sis epidemics with simplicial and higher order
  interactions.
\newblock \emph{IEEE T Contr Network Sys}, 9\penalty0 (2):\penalty0 695--705,
  2021.

\bibitem[Feliu et~al.(2012)Feliu, Knudsen, Andersen, and Wiuf]{Feliu2012}
Elisenda Feliu, Michael Knudsen, Lars~N. Andersen, and Carsten Wiuf.
\newblock An algebraic approach to signaling cascades with n layers.
\newblock \emph{B Math Biol}, 74:\penalty0 45--72, 2012.

\bibitem[Rao et~al.(2013)Rao, {van der Schaft}, and Jayawardhana]{Rao2013}
Shodhan Rao, Arjan {van der Schaft}, and Bayu Jayawardhana.
\newblock A graph-theoretical approach for the analysis and model reduction of
  complex-balanced chemical reaction networks.
\newblock \emph{J Math Chem}, 51:\penalty0 2401--2422, 2013.

\bibitem[{van der Schaft} et~al.(2015){van der Schaft}, Rao, and
  Jayawardhana]{schaft2015}
Arjan {van der Schaft}, Shodhan Rao, and Bayu Jayawardhana.
\newblock Complex and detailed balancing of chemical reaction networks
  revisited.
\newblock \emph{J Math Chem}, 53:\penalty0 1445--1458, 2015.

\bibitem[Snowden et~al.(2017)Snowden, {van der Graaf}, and
  Tindall]{Snowden2017}
T.~J. Snowden, P.~H. {van der Graaf}, and M.~J. Tindall.
\newblock Methods of model reduction for large-scale biological systems: A
  survey of current methods and trends.
\newblock \emph{Bull Math Biol}, 79:\penalty0 1449--1486, 2017.

\bibitem[Laurence et~al.(2019)Laurence, Doyon, Dubé, and
  Desrosiers]{Laurence2019}
Edward Laurence, Nicolas Doyon, Louis~J. Dubé, and Patrick Desrosiers.
\newblock Spectral dimension reduction of complex dynamical networks.
\newblock \emph{Phys Rev X}, 9:\penalty0 1--17, 2019.

\bibitem[Vegué et~al.(2023)Vegué, Thibeault, Desrosiers, and
  Allard]{Vegue2023}
Marina Vegué, Vincent Thibeault, Patrick Desrosiers, and Antoine Allard.
\newblock Dimension reduction of dynamics on modular and heterogeneous directed
  networks.
\newblock \emph{PNAS Nexus}, 2, 2023.

\bibitem[Chaves et~al.(2005)Chaves, Albert, and Sontag]{Chaves2005}
Madalena Chaves, Réka Albert, and Eduardo~D. Sontag.
\newblock Robustness and fragility of boolean models for genetic regulatory
  networks.
\newblock \emph{J Theor Biol}, 235:\penalty0 431--449, 2005.

\bibitem[Chaves et~al.(2006)Chaves, Sontag, and Albert]{Chaves2006}
M.~Chaves, E.~D. Sontag, and R.~Albert.
\newblock Methods of robustness analysis for boolean models of gene control
  networks.
\newblock \emph{IEEE Proc Sys Biol}, 153:\penalty0 154--167, 2006.

\bibitem[Saadatpour et~al.(2011)Saadatpour, Wang, Liao, Liu, Loughran, Albert,
  and Albert]{Saadatpour2011}
Assieh Saadatpour, Rui~Sheng Wang, Aijun Liao, Xin Liu, Thomas~P. Loughran,
  István Albert, and Réka Albert.
\newblock Dynamical and structural analysis of a t cell survival network
  identifies novel candidate therapeutic targets for large granular lymphocyte
  leukemia.
\newblock \emph{PLoS Comput Biol}, 7, 2011.

\bibitem[Angeli et~al.(2007)Angeli, Leenheer, and Sontag]{Angeli2007}
David Angeli, Patrick~De Leenheer, and Eduardo~D. Sontag.
\newblock A petri net approach to the study of persistence in chemical reaction
  networks.
\newblock \emph{Math Biosci}, 210:\penalty0 598--618, 2007.

\bibitem[Angeli(2011)]{Angeli2011}
David Angeli.
\newblock Boundedness analysis for open chemical reaction networks with
  mass-action kinetics.
\newblock \emph{Nat Comput}, 10:\penalty0 751--774, 2011.

\bibitem[Chaouiya(2007)]{Chaouiya2007}
Claudine Chaouiya.
\newblock Petri net modelling of biological networks.
\newblock \emph{Brief Bioinform}, 8:\penalty0 210--219, 2007.

\bibitem[Koch(2015)]{Koch2015}
Ina Koch.
\newblock Petri nets in systems biology.
\newblock \emph{Soft Sys Mod}, 14:\penalty0 703--710, 2015.

\bibitem[Borshchev and Filippov(2004)]{Borshchev2004a}
Andrei Borshchev and Alexei Filippov.
\newblock From system dynamics and discrete event to practical agent based
  modeling: Reasons, techniques, tools.
\newblock \emph{Simulation}, 66:\penalty0 25–29, 2004.

\bibitem[Devia and Giordano(2023{\natexlab{a}})]{DG2023model}
Carlos~Andrés Devia and Giulia Giordano.
\newblock Classification-based opinion formation model embedding agents'
  psychological traits.
\newblock \emph{J Art Soc Social Sim}, 26\penalty0 (3):\penalty0 1,
  2023{\natexlab{a}}.

\bibitem[Devia and Giordano(2023{\natexlab{b}})]{DG2023}
Carlos~Andrés Devia and Giulia Giordano.
\newblock Probabilistic analysis of agent-based opinion formation models.
\newblock \emph{Sci Rep}, 13:\penalty0 20152, 2023{\natexlab{b}}.

\bibitem[Devia and Giordano(2024)]{DG2024}
Carlos~Andrés Devia and Giulia Giordano.
\newblock Graphical analysis of agent-based opinion formation models.
\newblock \emph{PLOS ONE}, 19\penalty0 (5):\penalty0 1--27, 05 2024.

\bibitem[Fan et~al.(2022)Fan, Zhao, Xia, and Tanimoto]{Fan2022}
Junfeng Fan, Dawei Zhao, Chengyi Xia, and Jun Tanimoto.
\newblock Coupled spreading between information and epidemics on multiplex
  networks with simplicial complexes.
\newblock \emph{Chaos}, 32, 2022.

\bibitem[Kerr et~al.(2021)Kerr, Stuart, Mistry, Abeysuriya, Rosenfeld, Hart,
  Núñez, Cohen, Selvaraj, Hagedorn, et~al.]{kerr2021covasim}
Cliff~C Kerr, Robyn~M Stuart, Dina Mistry, Romesh~G Abeysuriya, Katherine
  Rosenfeld, Gregory~R Hart, Rafael~C Núñez, Jamie~A Cohen, Prashanth
  Selvaraj, Brittany Hagedorn, et~al.
\newblock Covasim: an agent-based model of covid-19 dynamics and interventions.
\newblock \emph{PLoS Comput Biol}, 17:\penalty0 e1009149, 2021.

\bibitem[Leonard et~al.(2024)Leonard, Bizyaeva, and Franci]{leonard2024fast}
Naomi~Ehrich Leonard, Anastasia Bizyaeva, and Alessio Franci.
\newblock Fast and flexible multiagent decision-making.
\newblock \emph{Annu Rev Contr Rob Auton Sys}, 7, 2024.

\bibitem[Pagliara(2021)]{Pagliara2021}
Naomi~Ehrich Pagliara, Renato;~Leonard.
\newblock Adaptive susceptibility and heterogeneity in contagion models on
  networks.
\newblock \emph{IEEE T Automat Contr}, 66\penalty0 (2):\penalty0 581--594,
  2021.

\bibitem[Proverbio et~al.(2025{\natexlab{a}})Proverbio, Tessarin, and
  Giordano]{PTG2025}
Daniele Proverbio, Riccardo Tessarin, and Giulia Giordano.
\newblock Recent trends in socio-epidemic modelling: behaviours and their
  determinants.
\newblock \emph{Boll Unione Mat Ital}, 2025{\natexlab{a}}.

\bibitem[Sturniolo et~al.(2021)Sturniolo, Waites, Colbourn, Manheim, and
  Panovska-Griffiths]{Sturniolo2021}
Simone Sturniolo, William Waites, Tim Colbourn, David Manheim, and Jasmina
  Panovska-Griffiths.
\newblock Testing, tracing and isolation in compartmental models.
\newblock \emph{PLoS Comput Biol}, 17:\penalty0 e1008633--NA, 2021.

\bibitem[Angeli(2009)]{Angeli2009c}
David Angeli.
\newblock A tutorial on chemical reaction network dynamics.
\newblock \emph{Europ J Control}, 15:\penalty0 398--406, 2009.

\bibitem[Feinberg(2019)]{Feinberg2019}
Martin Feinberg.
\newblock \emph{Foundations of Chemical Reaction Network Theory}, volume 202.
\newblock Springer, 2019.

\bibitem[Goentoro et~al.(2009)Goentoro, Shoval, Kirschner, and
  Alon]{Goentoro2009}
Lea Goentoro, Oren Shoval, Marc~W. Kirschner, and Uri Alon.
\newblock The incoherent feedforward loop can provide fold-change detection in
  gene regulation.
\newblock \emph{Molecular Cell}, 36\penalty0 (5):\penalty0 894--899, 2009.

\bibitem[Kaplan et~al.(2008)Kaplan, Bren, Dekel, and Alon]{Kaplan2008}
Shai Kaplan, Anat Bren, Erez Dekel, and Uri Alon.
\newblock The incoherent feed‐forward loop can generate non‐monotonic input
  functions for genes.
\newblock \emph{Mol Sys Biol}, 4\penalty0 (1):\penalty0 203, 2008.

\bibitem[Epstein and Pojman(1998)]{Epstein1998}
Irving~R Epstein and John~A Pojman.
\newblock \emph{An Introduction to Nonlinear Chemical Dynamics: Oscillations,
  Waves, Patterns, and Chaos}.
\newblock Oxford UP, 1998.

\bibitem[Angeli et~al.(2004)Angeli, Ferrell, and Sontag]{Angeli2004detection}
David Angeli, James~E. Ferrell, and Eduardo~D Sontag.
\newblock Detection of multistability, bifurcations, and hysteresis in a large
  class of biological positive-feedback systems.
\newblock \emph{P Natl Acad Sci Usa}, 101:\penalty0 1822--1827, 2004.

\bibitem[Cloutier et~al.(2012)Cloutier, Middleton, and Wellstead]{Cloutier2012}
M.~Cloutier, R.~Middleton, and P.~Wellstead.
\newblock Feedback motif for the pathogenesis of parkinson's disease.
\newblock \emph{IET Sys Biol}, 6:\penalty0 86--93(7), 2012.

\bibitem[Capasso and Serio(1978)]{Capasso1978}
Vincenzo Capasso and Gabriella Serio.
\newblock A generalization of the kermack-mckendrick deterministic epidemic
  model.
\newblock \emph{Math Biosci}, 42:\penalty0 43--61, 1978.

\bibitem[Chen et~al.(2010)Chen, Wang, Li, and Aihara]{Chen2010}
Luonan Chen, Ruiqi Wang, Chunguang Li, and Kazuyuki Aihara.
\newblock \emph{Modeling Biomolecular Networks in Cells: Structures and
  Dynamics}.
\newblock Springer, 1 edition, 2010.

\bibitem[Santill{\'a}n(2008)]{santillan2008use}
Moises Santill{\'a}n.
\newblock On the use of the hill functions in mathematical models of gene
  regulatory networks.
\newblock \emph{Math Mod Nat Phen}, 3\penalty0 (2):\penalty0 85--97, 2008.

\bibitem[Kermack and McKendrick(1927)]{kermack1927contribution}
William~Ogilvy Kermack and Anderson~G McKendrick.
\newblock A contribution to the mathematical theory of epidemics.
\newblock \emph{Proc Roy Soc A}, 115\penalty0 (772):\penalty0 700--721, 1927.

\bibitem[Breda et~al.(2012)Breda, Diekmann, de~Graaf, Pugliese, and
  Vermiglio]{Breda2012}
Dimitri Breda, Odo Diekmann, W.F. de~Graaf, Andrea Pugliese, and Rossana
  Vermiglio.
\newblock On the formulation of epidemic models (an appraisal of kermack and
  mckendrick).
\newblock \emph{J Biol Dynam}, 6:\penalty0 103--117, 2012.

\bibitem[Arino et~al.(2007)Arino, Brauer, van~den Driessche, Watmough, and
  Wu]{Arino2007}
Julien Arino, Fred Brauer, P.~van~den Driessche, James Watmough, and Jianhong
  Wu.
\newblock A final size relation for epidemic models.
\newblock \emph{Math Biosci Eng}, 4\penalty0 (2):\penalty0 159--175, 2007.

\bibitem[Proverbio et~al.(2021)Proverbio, Kemp, Magni, Husch, Aalto, Mombaerts,
  et~al.]{proverbio2021dynamical}
Daniele Proverbio, Françoise Kemp, Stefano Magni, Andreas Husch, Atte Aalto,
  Laurent Mombaerts, et~al.
\newblock Dynamical spqeir model assesses the effectiveness of
  non-pharmaceutical interventions against covid-19 epidemic outbreaks.
\newblock \emph{PLoS ONE}, 16:\penalty0 e0252019, 2021.

\bibitem[Li et~al.(2020)Li, Yang, and Martcheva]{Martcheva2020}
Xue-Zhi Li, Junyuan Yang, and Maia Martcheva.
\newblock \emph{Age Structured Epidemic Modeling}.
\newblock Springer, 2020.

\bibitem[Safi et~al.(2013)Safi, Gumel, and Elbasha]{Safi2013}
Mohammad~A. Safi, Abba~B. Gumel, and Elamin~H. Elbasha.
\newblock Qualitative analysis of an age-structured seir epidemic model with
  treatment.
\newblock \emph{Appl Math Comput}, 219:\penalty0 10627--10642, 2013.

\bibitem[Kemp et~al.(2021)Kemp, Proverbio, Aalto, Mombaerts, d’Herouel,
  Husch, et~al.]{herdImmunity}
Francoise Kemp, Daniele Proverbio, Atte Aalto, Laurent Mombaerts,
  Aymeric~Fouquier d’Herouel, Andreas Husch, et~al.
\newblock Modelling covid-19 dynamics and potential for herd immunity by
  vaccination in austria, luxembourg and sweden.
\newblock \emph{J Theor Biol}, 530:\penalty0 110874, 2021.

\bibitem[Brauer(2011)]{brauer2011backward}
Fred Brauer.
\newblock Backward bifurcations in simple vaccination/treatment models.
\newblock \emph{J Biol Dyn}, 5\penalty0 (5):\penalty0 410--418, 2011.

\bibitem[Krueger et~al.(2022)Krueger, Gogolewski, Bodych, Gambin, Giordano,
  Cuschieri, Czypionka, Perc, Petelos, Rosińska, and Szczurek]{Krueger2022}
Tyll Krueger, Krzysztof Gogolewski, Marcin Bodych, Anna Gambin, Giulia
  Giordano, Sarah Cuschieri, Thomas Czypionka, Matjaz Perc, Elena Petelos,
  Magdalena Rosińska, and Ewa Szczurek.
\newblock Risk assessment of covid-19 epidemic resurgence in relation to
  sars-cov-2 variants and vaccination passes.
\newblock \emph{Commun Med}, 2:\penalty0 23, 2022.

\bibitem[Zheng et~al.(2022)Zheng, Bao, Li, de~Vries, Giordano, and
  Pan]{PanPox2022}
Qinyue Zheng, Chunbing Bao, Pengfei Li, Annemarie~C de~Vries, Giulia Giordano,
  and Qiuwei Pan.
\newblock {Projecting the impact of testing and vaccination on the transmission
  dynamics of the 2022 monkeypox outbreak in the USA}.
\newblock \emph{J Travel Med}, 29\penalty0 (8):\penalty0 taac101, 09 2022.

\bibitem[Proverbio et~al.(2022{\natexlab{a}})Proverbio, Kemp, Magni, Ogorzaly,
  Cauchie, Gonçalves, et~al.]{Proverbio2022}
Daniele Proverbio, Françoise Kemp, Stefano Magni, Leslie Ogorzaly,
  Henry-Michel Cauchie, Jorge Gonçalves, et~al.
\newblock Model-based assessment of covid-19 epidemic dynamics by wastewater
  analysis.
\newblock \emph{Sci Total Environ}, 827:\penalty0 154235, 2022{\natexlab{a}}.

\bibitem[Aron and Schwartz(1984)]{aron1984seasonality}
Joan~L Aron and Ira~B Schwartz.
\newblock Seasonality and period-doubling bifurcations in an epidemic model.
\newblock \emph{J Theor Biol}, 110\penalty0 (4):\penalty0 665--679, 1984.

\bibitem[Bichara and Iggidr(2018)]{bichara2018multi}
Derdei Bichara and Abderrahman Iggidr.
\newblock Multi-patch and multi-group epidemic models: a new framework.
\newblock \emph{J Math Biol}, 77\penalty0 (1):\penalty0 107--134, 2018.

\bibitem[{Della Rossa} et~al.(2020){Della Rossa}, Salzano, {Di Meglio}, {De
  Lellis}, Coraggio, Calabrese, et~al.]{DellaRossa2020}
Fabio {Della Rossa}, Davide Salzano, Anna {Di Meglio}, Francesco {De Lellis},
  Marco Coraggio, Carmela Calabrese, et~al.
\newblock A network model of italy shows that intermittent regional strategies
  can alleviate the covid-19 epidemic.
\newblock \emph{Nat Commun}, 11:\penalty0 5106, 2020.

\bibitem[Grenfell and Harwood(1997)]{Grenfell1997}
Bryan Grenfell and John Harwood.
\newblock (meta)population dynamics of infectious diseases.
\newblock \emph{Trends in ecology \& evolution}, 12:\penalty0 395--399, 1997.

\bibitem[Bogun{\'a} et~al.(2013)Bogun{\'a}, Castellano, and
  Pastor-Satorras]{boguna2013nature}
Marian Bogun{\'a}, Claudio Castellano, and Romualdo Pastor-Satorras.
\newblock Nature of the epidemic threshold for the
  susceptible-infected-susceptible dynamics in networks.
\newblock \emph{Phys Rev Lett}, 111\penalty0 (6):\penalty0 068701, 2013.

\bibitem[Aalto et~al.(2025)Aalto, Proverbio, Giordano, Skupin, and Gon{\c
  c}alves]{Aalto2025}
Atte Aalto, Daniele Proverbio, Giulia Giordano, Alexander Skupin, and Jorge
  Gon{\c c}alves.
\newblock Networked sirs model with kalman filter state estimation for epidemic
  monitoring in europe.
\newblock \emph{medRxiv}, page 2025.02.06.25321780, 2025.

\bibitem[Ball et~al.(2014)Ball, Britton, House, Isham, Mollison, Pellis, and
  Scalia~Tomba]{Ball2014}
Frank Ball, Tom Britton, Thomas House, Valerie Isham, Denis Mollison, Lorenzo
  Pellis, and Gianpaolo Scalia~Tomba.
\newblock Seven challenges for metapopulation models of epidemics, including
  households models.
\newblock \emph{Epidemics}, 10:\penalty0 63--67, 2014.

\bibitem[Bertuzzo et~al.(2020)Bertuzzo, Mari, Pasetto, Miccoli, Casagrandi,
  Gatto, and Rinaldo]{Bertuzzo2020}
Enrico Bertuzzo, Lorenzo Mari, Damiano Pasetto, Stefano Miccoli, Renato
  Casagrandi, Marino Gatto, and Andrea Rinaldo.
\newblock The geography of covid-19 spread in italy and implications for the
  relaxation of confinement measures.
\newblock \emph{Nat Commun}, 11:\penalty0 4264, 2020.

\bibitem[Gatto et~al.(2020)Gatto, Bertuzzo, Mari, Miccoli, Carraro, Casagrandi,
  and Rinaldo]{Gatto2020}
Marino Gatto, Enrico Bertuzzo, Lorenzo Mari, Stefano Miccoli, Luca Carraro,
  Renato Casagrandi, and Andrea Rinaldo.
\newblock Spread and dynamics of the covid-19 epidemic in italy: Effects of
  emergency containment measures.
\newblock \emph{P Natl Acad Sci Usa}, 117:\penalty0 10484--10491, 2020.

\bibitem[Rowthorn et~al.(2009)Rowthorn, Laxminarayan, and
  Gilligan]{Rowthorn2009}
Robert Rowthorn, Ramanan Laxminarayan, and Christopher~A. Gilligan.
\newblock Optimal control of epidemics in metapopulations.
\newblock \emph{J Roy Soc Interface}, 6\penalty0 (41):\penalty0 1135--1144,
  2009.

\bibitem[Priesemann et~al.(2021{\natexlab{a}})Priesemann, Brinkmann, Ciesek,
  Cuschieri, Czypionka, Giordano, et~al.]{Priesemann2020}
Viola Priesemann, Melanie~M. Brinkmann, Sandra Ciesek, Sarah Cuschieri, Thomas
  Czypionka, Giulia Giordano, et~al.
\newblock Calling for pan-european commitment for rapid and sustained reduction
  in sars-cov-2 infections.
\newblock \emph{Lancet}, 397\penalty0 (10269):\penalty0 92--93,
  2021{\natexlab{a}}.

\bibitem[Priesemann et~al.(2021{\natexlab{b}})Priesemann, Balling, Brinkmann,
  Ciesek, Czypionka, Eckerle, et~al.]{Priesemann2021}
Viola Priesemann, Rudi Balling, Melanie~M Brinkmann, Sandra Ciesek, Thomas
  Czypionka, Isabella Eckerle, et~al.
\newblock An action plan for pan-european defence against new sars-cov-2
  variants.
\newblock \emph{Lancet}, 397\penalty0 (10273):\penalty0 469--470,
  2021{\natexlab{b}}.

\bibitem[Valdez et~al.(2022)Valdez, Iftekhar, Oliu-Barton, Böhm, Cuschieri,
  Czypionka, et~al.]{Valdez2022}
André~Calero Valdez, Emil~N. Iftekhar, Miquel Oliu-Barton, Robert Böhm, Sarah
  Cuschieri, Thomas Czypionka, et~al.
\newblock Europe must come together to confront omicron.
\newblock \emph{BMJ}, 376:\penalty0 o90, 2022.

\bibitem[Almocera and Hernandez-Vargas(2019)]{Almocera2019}
Alexis Erich~S Almocera and Esteban~A. Hernandez-Vargas.
\newblock Coupling multiscale within-host dynamics and between-host
  transmission with recovery (sir) dynamics.
\newblock \emph{Math Biosci}, 309:\penalty0 34--41, 2019.

\bibitem[Abuin et~al.(2020)Abuin, Anderson, Ferramosca, Hernandez-Vargas, and
  González]{Abuin2020}
Pablo Abuin, Alejandro Anderson, Antonio Ferramosca, Esteban~A.
  Hernandez-Vargas, and Alejandro~H. González.
\newblock Characterization of sars-cov-2 dynamics in the host.
\newblock \emph{Annu Rev Control}, 50:\penalty0 457--468, 2020.

\bibitem[Hernandez-Vargas and Velasco-Hernandez(2020)]{Hernandez-Vargas2020}
Esteban~A. Hernandez-Vargas and Jorge~X. Velasco-Hernandez.
\newblock In-host mathematical modelling of covid-19 in humans.
\newblock \emph{Annu Rev Control}, 50:\penalty0 448--456, 2020.

\bibitem[{de Jong} et~al.(2023){de Jong}, Calà~Campana, Li, Pan, and
  Giordano]{deJong2023}
Maarten {de Jong}, Francesca Calà~Campana, Pengfei Li, Qiuwei Pan, and Giulia
  Giordano.
\newblock A novel viral infection model to guide optimal mpox treatment.
\newblock \emph{IEEE Contr Sys Lett}, 7:\penalty0 3145--3150, 2023.

\bibitem[Anderson et~al.(2025)Anderson, Katz, Calà~Campana, and
  Giordano]{AKCG2025}
Alejandro Anderson, Rami Katz, Francesca Calà~Campana, and Giulia Giordano.
\newblock Failure and success in single-drug control of antimicrobial
  resistance.
\newblock \emph{IEEE Contr Sys Lett}, 9:\penalty0 991--996, 2025.

\bibitem[Hernandez-Vargas et~al.(2014)Hernandez-Vargas, Colaneri, and
  Middleton]{Hernandez2014}
Esteban~A. Hernandez-Vargas, Patrizio Colaneri, and Richard~H. Middleton.
\newblock Switching strategies to mitigate hiv mutation.
\newblock \emph{IEEE T Cont Sys Tech}, 22:\penalty0 1623--1628, 2014.

\bibitem[Blanchini et~al.(2021{\natexlab{b}})Blanchini, Bolzern, Colaneri,
  De~Nicolao, and Giordano]{BBCDG2021}
Franco Blanchini, Paolo Bolzern, Patrizio Colaneri, Giuseppe De~Nicolao, and
  Giulia Giordano.
\newblock Generalized epidemiological compartmental models: guaranteed bounds
  via optimal control.
\newblock In \emph{Proc 60th IEEE CDC}, pages 3532--3537, 2021{\natexlab{b}}.

\bibitem[Blanchini et~al.(2023{\natexlab{a}})Blanchini, Bolzern, Colaneri,
  Nicolao, and Giordano]{Blanchini2023}
Franco Blanchini, Paolo Bolzern, Patrizio Colaneri, Giuseppe~De Nicolao, and
  Giulia Giordano.
\newblock Optimal control of compartmental models: The exact solution.
\newblock \emph{Automatica}, 147:\penalty0 110680, 2023{\natexlab{a}}.

\bibitem[Devia and Giordano(2019)]{DG2019}
Carlos~Andrés Devia and Giulia Giordano.
\newblock Optimal duration and planning of switching treatments taking drug
  toxicity into account: a convex optimisation approach.
\newblock In \emph{Proc IEEE 58th CDC}, pages 5674--5679, 2019.

\bibitem[Bellomo et~al.(2020)Bellomo, Bingham, Chaplain, Dosi, Forni, Knopoff,
  Lowengrub, Twarock, and Virgillito]{Bellomo2020}
Nicola Bellomo, Richard~J. Bingham, Mark A.~J. Chaplain, Giovanni Dosi, Guido
  Forni, Damián~Alejandro Knopoff, John Lowengrub, Reidun Twarock, and
  Maria~Enrica Virgillito.
\newblock A multiscale model of virus pandemic: Heterogeneous interactive
  entities in a globally connected world.
\newblock \emph{Mathematical models \& methods in applied sciences : M3AS},
  30:\penalty0 1591--1651, 2020.

\bibitem[Feng et~al.(2011)Feng, Velasco-Hernandez, Tapia-Santos, and
  Leite]{Feng2011}
Zhilan Feng, Jorge~X. Velasco-Hernandez, Brenda Tapia-Santos, and Maria
  Conceição~A. Leite.
\newblock A model for coupling within-host and between-host dynamics in an
  infectious disease.
\newblock \emph{Nonlinear Dyn}, 68:\penalty0 401--411, 2011.

\bibitem[Feng et~al.(2015)Feng, Cen, Zhao, and Velasco-Hernandez]{Feng2015}
Zhilan Feng, Xiuli Cen, Yulin Zhao, and Jorge~X. Velasco-Hernandez.
\newblock Coupled within-host and between-host dynamics and evolution of
  virulence.
\newblock \emph{Math Biosci}, 270:\penalty0 204--212, 2015.

\bibitem[Gandolfi et~al.(2014)Gandolfi, Pugliese, and Sinisgalli]{Gandolfi2014}
Alberto Gandolfi, Andrea Pugliese, and Carmela Sinisgalli.
\newblock Epidemic dynamics and host immune response: a nested approach.
\newblock \emph{J Math Biol}, 70:\penalty0 399--435, 2014.

\bibitem[Murillo et~al.(2013)Murillo, Murillo, and Perelson]{Murillo2013}
Lisa~N. Murillo, Michael~S. Murillo, and Alan~S. Perelson.
\newblock Towards multiscale modeling of influenza infection.
\newblock \emph{J Theor Biol}, 332:\penalty0 267--290, 2013.

\bibitem[Cai et~al.(2017)Cai, Tuncer, and Martcheva]{Cai2017}
Liming Cai, Necibe Tuncer, and Maia Martcheva.
\newblock How does within-host dynamics affect population-level dynamics?
  insights from an immuno-epidemiological model of malaria.
\newblock \emph{Math Meth Appl Sci}, 40:\penalty0 6424--6450, 2017.

\bibitem[Nemudryi et~al.(2020)Nemudryi, Nemudraia, Wiegand, Surya, Buyukyoruk,
  Cicha, et~al.]{nemudryi2020temporal}
Artem Nemudryi, Anna Nemudraia, Tanner Wiegand, Kevin Surya, Murat Buyukyoruk,
  Calvin Cicha, et~al.
\newblock Temporal detection and phylogenetic assessment of {SARS-CoV-2} in
  municipal wastewater.
\newblock \emph{Cell Rep. Med.}, 1\penalty0 (6):\penalty0 100098, 2020.

\bibitem[Boccaletti et~al.(2006)Boccaletti, Latora, Moreno, Chavez, and
  Hwang]{Boccaletti2006}
Stefano Boccaletti, V.~Latora, Y.~Moreno, M.~Chavez, and D.~U. Hwang.
\newblock Complex networks: Structure and dynamics.
\newblock \emph{Phys Rep}, 424:\penalty0 175--308, 2006.

\bibitem[Gross and Yellen(2005)]{gross2005graph}
Jonathan~L Gross and Jay Yellen.
\newblock \emph{{Graph Theory and its Applications}}.
\newblock CRC press, 2005.

\bibitem[Anderson(2011)]{Anderson2011}
David~F. Anderson.
\newblock Boundedness of trajectories for weakly reversible, single linkage
  class reaction systems.
\newblock \emph{J Math Chem}, 49:\penalty0 2275--2290, 2011.

\bibitem[Sharma et~al.(2016)Sharma, Dutta, and Gupta]{Sharma2016}
Yogita Sharma, Partha~Sharathi Dutta, and A.~K. Gupta.
\newblock Anticipating regime shifts in gene expression: The case of an
  autoactivating positive feedback loop.
\newblock \emph{Phys Rev E}, 93:\penalty0 1--13, 2016.

\bibitem[Edwards et~al.(2015)Edwards, Machina, McGregor, and van~den
  Driessche]{Edwards2015}
R.~Edwards, A.~Machina, G.~McGregor, and P.~van~den Driessche.
\newblock A modelling framework for gene regulatory networks including
  transcription and translation.
\newblock \emph{B Math Biol}, 77:\penalty0 953--983, 2015.

\bibitem[Lyapunov(1992)]{lyapunov1992general}
Aleksandr~Mikhailovich Lyapunov.
\newblock \emph{General Problem of the Stability of Motion}, volume~55.
\newblock CRC Press, 1992.

\bibitem[Carlson(1998)]{carlson1998signal}
Gordon~E Carlson.
\newblock \emph{Signal and Linear System Analysis}.
\newblock John Wiley Hoboken, NJ, 1998.

\bibitem[Sontag(2013)]{sontag2013mathematical}
Eduardo~D Sontag.
\newblock \emph{Mathematical Control Theory: Deterministic Finite Dimensional
  Systems}, volume~6.
\newblock Springer Science \& Business Media, 2013.

\bibitem[Arnoldi et~al.(2016)Arnoldi, Loreau, and
  Haegeman]{arnoldi2016resilience}
Jean-Fran{\c{c}}ois Arnoldi, Michel Loreau, and Bart Haegeman.
\newblock Resilience, reactivity and variability: A mathematical comparison of
  ecological stability measures.
\newblock \emph{J Theo Biol}, 389:\penalty0 47--59, 2016.

\bibitem[Grimm and Wissel(1997)]{grimm1997babel}
Volker Grimm and Christian Wissel.
\newblock Babel, or the ecological stability discussions: an inventory and
  analysis of terminology and a guide for avoiding confusion.
\newblock \emph{Oecologia}, 109:\penalty0 323--334, 1997.

\bibitem[Pimm(1984)]{pimm1984complexity}
Stuart~L Pimm.
\newblock The complexity and stability of ecosystems.
\newblock \emph{Nature}, 307\penalty0 (5949):\penalty0 321--326, 1984.

\bibitem[Zhang(2010)]{ZHANG201041}
Peng Zhang.
\newblock Industrial control engineering.
\newblock In \emph{Advanced Industrial Control Technology}, pages 41--70.
  William Andrew Pub, 2010.

\bibitem[Alamo et~al.(2022)Alamo, Millán, Gutiérrez~Reina, Preciado, and
  Giordano]{Alamo2022}
Teodoro Alamo, Pablo Millán, Daniel Gutiérrez~Reina, Victor~M. Preciado, and
  Giulia Giordano.
\newblock Challenges and future directions in pandemic control.
\newblock \emph{IEEE Control Sys Lett}, 6:\penalty0 722--727, 2022.

\bibitem[Lee and Leitmann(1995)]{Lee1995}
C.~S. Lee and George Leitmann.
\newblock Control strategies for an endemic disease in the presence of
  uncertainty.
\newblock \emph{Recent Trends in Optimization Theory and Applications},
  NA:\penalty0 221--238, 1995.

\bibitem[Leitmann(1998)]{Leitmann1998}
Georg Leitmann.
\newblock The use of screening for the control of an endemic disease.
\newblock \emph{Var Calc Opt Contr Appl}, pages 291--300, 1998.

\bibitem[Kitano(2004{\natexlab{b}})]{Kitano2004}
Hiroaki Kitano.
\newblock Cancer as a robust system: implications for anticancer therapy.
\newblock \emph{Nat Rev Cancer}, 4:\penalty0 227--235, 2004{\natexlab{b}}.

\bibitem[Kitano(2007{\natexlab{b}})]{Kitano2007towards}
Hiroaki Kitano.
\newblock Towards a theory of biological robustness.
\newblock \emph{Mol Sys Biol}, 3, 2007{\natexlab{b}}.

\bibitem[Bloom et~al.(2006)Bloom, Labthavikul, Otey, and
  Arnold]{bloom2006protein}
Jesse~D Bloom, Sy~T Labthavikul, Christopher~R Otey, and Frances~H Arnold.
\newblock Protein stability promotes evolvability.
\newblock \emph{P Natl Acad Sci Usa}, 103\penalty0 (15):\penalty0 5869--5874,
  2006.

\bibitem[Kwon and Cho(2008)]{Kwon2008}
Yung~Keun Kwon and Kwang~Hyun Cho.
\newblock Quantitative analysis of robustness and fragility in biological
  networks based on feedback dynamics.
\newblock \emph{Bioinformatics}, 24:\penalty0 987--994, 2008.

\bibitem[Masel and Trotter(2010)]{masel2010robustness}
Joanna Masel and Meredith~V Trotter.
\newblock Robustness and evolvability.
\newblock \emph{Trends Gen}, 26\penalty0 (9):\penalty0 406--414, 2010.

\bibitem[Nikolov et~al.(2007)Nikolov, Yankulova, Wolkenhauer, and
  Petrov]{Nikolov2007}
Svetoslav Nikolov, Elka Yankulova, Olaf Wolkenhauer, and Valko Petrov.
\newblock Principal difference between stability and structural stability
  (robustness) as used in systems biology.
\newblock \emph{Nonlinear Dyn Psychol Life Sci}, 11:\penalty0 413--433, 2007.

\bibitem[Shil'nikov(2001)]{shil2001methods}
Leonid~P Shil'nikov.
\newblock \emph{Methods of Qualitative Theory in Nonlinear Dynamics}, volume~5.
\newblock World Scientific, 2001.

\bibitem[Önder Kartal and Ebenhöh(2009)]{Kartal2009}
Önder Kartal and Oliver Ebenhöh.
\newblock Ground state robustness as an evolutionary design principle in
  signaling networks.
\newblock \emph{PLoS ONE}, 4, 2009.

\bibitem[Lesne(2008)]{Lesne2008}
Annick Lesne.
\newblock Robustness: Confronting lessons from physics and biology.
\newblock \emph{Biol Rev}, 83:\penalty0 509--532, 2008.

\bibitem[Whitacre(2012)]{Whitacre2012}
James~Michael Whitacre.
\newblock Biological robustness: Paradigms, mechanisms, systems principles.
\newblock \emph{Front Gen}, 3, 2012.

\bibitem[Artime et~al.(2024)Artime, Grassia, De~Domenico, Gleeson, Makse,
  Mangioni, Perc, and Radicchi]{artime2024robustness}
Oriol Artime, Marco Grassia, Manlio De~Domenico, James~P Gleeson, Hern{\'a}n~A
  Makse, Giuseppe Mangioni, Matja{\v{z}} Perc, and Filippo Radicchi.
\newblock Robustness and resilience of complex networks.
\newblock \emph{Nature Rev Phys}, pages 1--18, 2024.

\bibitem[Liu et~al.(2015)Liu, Peng, and Gao]{liu2015vulnerability}
Xueming Liu, Hao Peng, and Jianxi Gao.
\newblock Vulnerability and controllability of networks of networks.
\newblock \emph{Chaos Soliton Fract}, 80:\penalty0 125--138, 2015.

\bibitem[Ren et~al.(2019)Ren, Gleinig, Helbing, and
  Antulov-Fantulin]{ren2019generalized}
Xiao-Long Ren, Niels Gleinig, Dirk Helbing, and Nino Antulov-Fantulin.
\newblock Generalized network dismantling.
\newblock \emph{P Natl Acad Sci Usa}, 116\penalty0 (14):\penalty0 6554--6559,
  2019.

\bibitem[Zdeborov{\'a} et~al.(2016)Zdeborov{\'a}, Zhang, and
  Zhou]{zdeborova2016fast}
Lenka Zdeborov{\'a}, Pan Zhang, and Hai-Jun Zhou.
\newblock Fast and simple decycling and dismantling of networks.
\newblock \emph{Sci Rep}, 6\penalty0 (1):\penalty0 37954, 2016.

\bibitem[Jeong et~al.(2000)Jeong, Tombor, Albert, Oltvai, and
  Barab{\'a}si]{jeong2000large}
Hawoong Jeong, B{\'a}lint Tombor, R{\'e}ka Albert, Zoltan~N Oltvai, and A-L
  Barab{\'a}si.
\newblock The large-scale organization of metabolic networks.
\newblock \emph{Nature}, 407\penalty0 (6804):\penalty0 651--654, 2000.

\bibitem[Crespi et~al.(2021)Crespi, Burnap, Chen, Das, Gassman, Rosa, Simmons,
  Wada, Wang, Xiao, Yang, Yin, and Goldstone]{Crespi2021}
Erica Crespi, Robert Burnap, Jing Chen, Moumita Das, Natalie Gassman,
  Epaminondas Rosa, Rebecca Simmons, Haruka Wada, Zhen~Q Wang, Jie Xiao, Bing
  Yang, John Yin, and Jared~V Goldstone.
\newblock Resolving the rules of robustness and resilience in biology across
  scales.
\newblock \emph{Integrative and Comparative Biology}, 61\penalty0 (6):\penalty0
  2163–2179, 2021.

\bibitem[Krakovská et~al.(2024)Krakovská, Kuehn, and Longo]{Krakovska2024}
Hana Krakovská, Christian Kuehn, and Iacopo~P. Longo.
\newblock Resilience of dynamical systems.
\newblock \emph{Europ J Appl Math}, 35\penalty0 (1):\penalty0 155--200, 2024.

\bibitem[Vitousek et~al.(2025)Vitousek, Taff, and Williamson]{Vitousek2025}
Maren~N. Vitousek, Conor~C. Taff, and Jessie~L. Williamson.
\newblock Resilience and robustness: from sub-organismal responses to
  communities.
\newblock \emph{Trends in Ecology \& Evolution}, 40\penalty0 (5):\penalty0
  468--478, 2025.

\bibitem[Baggio et~al.(2015)Baggio, Brown, and Hellebrandt]{baggio2015boundary}
Jacopo~A Baggio, Katrina Brown, and Denis Hellebrandt.
\newblock Boundary object or bridging concept? a citation network analysis of
  resilience.
\newblock \emph{Ecology and Society}, 20\penalty0 (2), 2015.

\bibitem[Standish et~al.(2014)Standish, Hobbs, Mayfield, Bestelmeyer, Suding,
  Battaglia, et~al.]{standish2014resilience}
Rachel~J Standish, Richard~J Hobbs, Margaret~M Mayfield, Brandon~T Bestelmeyer,
  Katherine~N Suding, Loretta~L Battaglia, et~al.
\newblock Resilience in ecology: Abstraction, distraction, or where the action
  is?
\newblock \emph{Biol Conserv}, 177:\penalty0 43--51, 2014.

\bibitem[Holling(1973)]{Holling1973}
C~S Holling.
\newblock Resilience and stability of ecological systems.
\newblock \emph{Annu Rev Eco Evol System}, 4:\penalty0 1--23, 1973.

\bibitem[Dakos and K{\'e}fi(2022)]{dakos2022ecological}
Vasilis Dakos and Sonia K{\'e}fi.
\newblock Ecological resilience: What to measure and how.
\newblock \emph{Environ Res Lett}, 17\penalty0 (4):\penalty0 043003, 2022.

\bibitem[Doyle(1982)]{doyle1982analysis}
John Doyle.
\newblock Analysis of feedback systems with structured uncertainties.
\newblock \emph{IEEE Proc D Contr Theo Appl}, 129\penalty0 (6):\penalty0
  242--250, 1982.

\bibitem[Clement et~al.(2021)Clement, Wioland, Govaere, Gourc, Cegarra,
  Marmier, and Kamissoko]{Clement2021}
Antoine Clement, Li\^{e}n Wioland, Virginie Govaere, Didier Gourc, Julien
  Cegarra, Fran\c{c}ois Marmier, and Daouda Kamissoko.
\newblock Robustness, resilience: typology of definitions through a
  multidisciplinary structured analysis of the literature.
\newblock \emph{European Journal of Industrial Engineering}, 15\penalty0
  (4):\penalty0 487--513, 2021.

\bibitem[Carpenter et~al.(2001)Carpenter, Walker, Anderies, and
  Abel]{carpenter2001metaphor}
Steve Carpenter, Brian Walker, J~Marty Anderies, and Nick Abel.
\newblock From metaphor to measurement: resilience of what to what?
\newblock \emph{Ecosystems}, 4:\penalty0 765--781, 2001.

\bibitem[Scheffer et~al.(2015)Scheffer, Carpenter, Dakos, and van
  Nes]{scheffer2015generic}
Marten Scheffer, Stephen~R Carpenter, Vasilis Dakos, and Egbert~H van Nes.
\newblock Generic indicators of ecological resilience: inferring the chance of
  a critical transition.
\newblock \emph{Annu Rev Ecol Evol Syst}, 46:\penalty0 145--167, 2015.

\bibitem[Lemmon(2020)]{lemmon2020achieving}
MD~Lemmon.
\newblock Achieving ecological resilience through regime shift management.
\newblock \emph{Found Trends Sys Contr}, 7\penalty0 (4):\penalty0 384--499,
  2020.

\bibitem[Loppini et~al.(2019)Loppini, Filippi, and Stanley]{Loppini2019}
Alessandro Loppini, Simonetta Filippi, and H.~Eugene Stanley.
\newblock Critical transitions in heterogeneous networks: Loss of low-degree
  nodes as an early warning signal.
\newblock \emph{Phys Rev E}, 99:\penalty0 1--5, 2019.

\bibitem[Mugisha and Zhou(2016)]{mugisha2016identifying}
Salomon Mugisha and Hai-Jun Zhou.
\newblock Identifying optimal targets of network attack by belief propagation.
\newblock \emph{Phys Rev E}, 94\penalty0 (1):\penalty0 012305, 2016.

\bibitem[Zhang et~al.(2020)Zhang, Shao, He, and Gao]{zhang2020resilience}
Yongtao Zhang, Cunqi Shao, Shibo He, and Jianxi Gao.
\newblock Resilience centrality in complex networks.
\newblock \emph{Phys Rev E}, 101:\penalty0 22304, 2020.

\bibitem[Bhandary et~al.(2021)Bhandary, Kaur, Banerjee, and
  Dutta]{Bhandary2021}
Subhendu Bhandary, Taranjot Kaur, Tanmoy Banerjee, and Partha~Sharathi Dutta.
\newblock Network resilience of fitzhugh-nagumo neurons in the presence of
  nonequilibrium dynamics.
\newblock \emph{Phys Rev E}, 103:\penalty0 1--12, 2021.

\bibitem[Zhu and Ba{\c{s}}ar(2024)]{zhu2024disentangling}
Quanyan Zhu and Tamer Ba{\c{s}}ar.
\newblock Disentangling resilience from robustness: Contextual dualism,
  interactionism, and game-theoretic paradigms.
\newblock \emph{IEEE Control Systems Magazine}, 44\penalty0 (3):\penalty0
  95--103, 2024.

\bibitem[Blanchet et~al.(2017)Blanchet, Nam, Ramalingam, and
  Pozo-Martin]{Blanchet2017}
Karl Blanchet, Sara~L. Nam, Ben Ramalingam, and Francisco Pozo-Martin.
\newblock Governance and capacity to manage resilience of health systems:
  towards a new conceptual framework.
\newblock \emph{Int J Health Pol Manag}, 6:\penalty0 431--435, 2017.

\bibitem[Munoz et~al.(2022)Munoz, Billsberry, and
  Ambrosini]{munoz2022resilience}
Albert Munoz, Jon Billsberry, and Veronique Ambrosini.
\newblock Resilience, robustness, and antifragility: Towards an appreciation of
  distinct organizational responses to adversity.
\newblock \emph{Int J Manag Rev}, 24\penalty0 (2):\penalty0 181--187, 2022.

\bibitem[Al-Radhawi and Angeli(2016)]{Al-Radhawi2016}
Muhammad Al-Radhawi and David Angeli.
\newblock New approach to the stability of chemical reaction networks:
  Piecewise linear in rates lyapunov functions.
\newblock \emph{IEEE T Automat Contr}, 61:\penalty0 76--89, 2016.

\bibitem[Blanchini and Giordano(2014)]{Blanchini2014d}
F~Blanchini and G~Giordano.
\newblock Piecewise-linear lyapunov functions for structural stability of
  biochemical networks.
\newblock \emph{Automatica}, 50:\penalty0 2482--2493, 2014.

\bibitem[Blanchini and Giordano(2017)]{Blanchini2017f}
F~Blanchini and G~Giordano.
\newblock Polyhedral lyapunov functions structurally ensure global asymptotic
  stability of dynamical networks iff the jacobian is non-singular.
\newblock \emph{Automatica}, 86:\penalty0 183--191, 2017.

\bibitem[Clarke(1980)]{Clarke1980}
Bruce~L Clarke.
\newblock Stability of complex reaction networks.
\newblock \emph{Advances Chem Phys}, pages 1--215, 1980.

\bibitem[Cosentino et~al.(2012)Cosentino, Salerno, Passanti, Merola, Bates, and
  Amato]{Cosentino2012b}
Carlo Cosentino, Luca Salerno, Antonio Passanti, Alessio Merola, Declan~G.
  Bates, and Francesco Amato.
\newblock Structural bistability of the gal regulatory network and
  characterization of its domains of attraction.
\newblock \emph{Int Comput Biol}, 19:\penalty0 148--162, 2012.

\bibitem[Kaufman et~al.(2007)Kaufman, Soulé, and Thomas]{Kaufman2007}
M.~Kaufman, C.~Soulé, and R.~Thomas.
\newblock A new necessary condition on interaction graphs for
  multistationarity.
\newblock \emph{J Theor Biol}, 248:\penalty0 675--685, 2007.

\bibitem[Mincheva and Craciun(2008)]{Mincheva2008}
Maya Mincheva and Gheorghe Craciun.
\newblock Multigraph conditions for multistability, oscillations and pattern
  formation in biochemical reaction networks.
\newblock \emph{Proc IEEE}, 96:\penalty0 1281--1291, 2008.

\bibitem[Snoussi(1998)]{Snoussi1998}
El~Houssine Snoussi.
\newblock Necessary conditions for multistationarity and stable periodicity.
\newblock \emph{J Biol Sys}, 6:\penalty0 3--9, 1998.

\bibitem[Soulé(2003)]{soule2003}
Christophe Soulé.
\newblock Graphic requirements for multistationarity.
\newblock \emph{ComPlexUs}, 1:\penalty0 123--33, 2003.

\bibitem[Culos et~al.(2016)Culos, Olesky, and van~den Driessche]{Culos2016}
G.~J. Culos, D.~D. Olesky, and P.~van~den Driessche.
\newblock Using sign patterns to detect the possibility of periodicity in
  biological systems.
\newblock \emph{J Math Biol}, 72:\penalty0 1281--1300, 2016.

\bibitem[Feliu(2019)]{Feliu2019}
Elisenda Feliu.
\newblock Sign-sensitivities for reaction networks: An algebraic approach.
\newblock \emph{Math Biosc and Eng}, 16:\penalty0 8195--8213, 2019.

\bibitem[Giordano et~al.(2016{\natexlab{b}})Giordano, {Cuba Samaniego}, Franco,
  and Blanchini]{Giordano2016}
G~Giordano, C~{Cuba Samaniego}, E~Franco, and F~Blanchini.
\newblock Computing the structural influence matrix for biological systems.
\newblock \emph{J Math Biol}, 72:\penalty0 1927--1958, 2016{\natexlab{b}}.

\bibitem[Mochizuki and Fiedler(2015)]{Mochizuki2015}
Atsushi Mochizuki and Bernold Fiedler.
\newblock Sensitivity of chemical reaction networks: A structural approach. 1.
  examples and the carbon metabolic network.
\newblock \emph{J Theor Biol}, 367:\penalty0 189--202, 2015.

\bibitem[Sontag(2014)]{Sontag2014}
Eduardo~D. Sontag.
\newblock A technique for determining the signs of sensitivities of steady
  states in chemical reaction networks.
\newblock \emph{IET Sys Biol}, 8:\penalty0 251--267, 2014.

\bibitem[Vassena(2023)]{Vassena2023}
Nicola Vassena.
\newblock Sign-sensitivity of metabolic networks: Which structures determine
  the sign of the responses.
\newblock \emph{Int J Rob Nonlinear Contr}, 33\penalty0 (9):\penalty0
  4843--4866, 2023.

\bibitem[Hori et~al.(2015)Hori, Miyazako, Kumagai, and Hara]{Hori2015}
Yutaka Hori, Hiroki Miyazako, Soichiro Kumagai, and Shinji Hara.
\newblock Coordinated spatial pattern formation in biomolecular communication
  networks.
\newblock \emph{IEEE T Mol Biol Multi-Scale Commun}, 1:\penalty0 111--121,
  2015.

\bibitem[Scardovi et~al.(2010)Scardovi, Arcak, and Sontag]{Scardovi2010}
Luca Scardovi, Murat Arcak, and Eduardo~D. Sontag.
\newblock Synchronization of interconnected systems with applications to
  biochemical networks: An input-output approach.
\newblock \emph{IEEE T Automat Contr}, 55:\penalty0 1367--1379, 2010.

\bibitem[Brualdi and Shader(1997)]{Brualdi1997}
Richard~A Brualdi and Bryan~L Shader.
\newblock \emph{Matrices of Sign-Solvable Linear Systems}.
\newblock Cambridge UP, 1997.

\bibitem[Radde et~al.(2009)Radde, Bar, and Banaji]{Radde2009}
Nicole Radde, Nadav~S Bar, and Murad Banaji.
\newblock Graphical methods for analysing feedback in biological networks-a
  survey.
\newblock \emph{Int J Sys Sci}, 00:\penalty0 1--16, 2009.

\bibitem[Radde(2010)]{Radde2010}
N.~Radde.
\newblock Fixed point characterization of biological networks with complex
  graph topology.
\newblock \emph{Bioinformatics}, 26:\penalty0 2874--2880, 2010.

\bibitem[{van der Schaft} et~al.(2013){van der Schaft}, Rao, and
  Jayawardhana]{schaft2013}
Arjan {van der Schaft}, Shodhan Rao, and Bayu Jayawardhana.
\newblock On the mathematical structure of balanced chemical reaction networks
  governed by mass action kinetics.
\newblock \emph{Siam J Appl Math}, 73:\penalty0 953--973, 2013.

\bibitem[Angeli et~al.(2009{\natexlab{a}})Angeli, Hirsch, and
  Sontag]{Angeli2009b}
David Angeli, Morris~W. Hirsch, and Eduardo~D. Sontag.
\newblock Attractors in coherent systems of differential equations.
\newblock \emph{J Diff Eq}, 246:\penalty0 3058--3076, 2009{\natexlab{a}}.

\bibitem[Gouzé(1998)]{Gouze1998}
Jean-Luc Gouzé.
\newblock Positive and negative circuits in dynamical systems.
\newblock \emph{J Biol Sys}, 6:\penalty0 11--15, 1998.

\bibitem[Maybee and Quirk(1969)]{Maybee1969}
John Maybee and James Quirk.
\newblock Qualitative problems in matrix theory, 1969.

\bibitem[Domijan and Pécou(2012)]{Domijan2012}
Mirela Domijan and Elisabeth Pécou.
\newblock The interaction graph structure of mass-action reaction networks.
\newblock \emph{J Math Biol}, 65:\penalty0 375--402, 2012.

\bibitem[Lloyd(1978)]{Lloyd1978}
N.~G. Lloyd.
\newblock \emph{{Degree Theory}}.
\newblock Cambridge University Press, London, 1978.

\bibitem[Blanchini(1999)]{Blanchini1999}
Franco Blanchini.
\newblock Set invariance in control.
\newblock \emph{Automatica}, 35:\penalty0 1747--1767, 1999.

\bibitem[Blanchini and Miani(2015)]{BM2015}
Franco Blanchini and Stefano Miani.
\newblock \emph{Set-theoretic Methods in Control; 2nd ed.}
\newblock Birkh\"auser, Basel, 2015.

\bibitem[Aubin and Cellina(1984)]{AubinCellina}
J.~P. Aubin and A.~Cellina.
\newblock \emph{Differential Inclusions. Set-Valued Maps and Viability Theory}.
\newblock Springer, 1984.

\bibitem[Berman and Plemmons(1994)]{BermanPlemmons1994}
Abraham Berman and Robert~J. Plemmons.
\newblock \emph{Nonnegative Matrices in the Mathematical Sciences}.
\newblock SIAM, 1994.

\bibitem[Hirsch and Smith(2006)]{Hirsch2006b}
M~W Hirsch and Hal Smith.
\newblock \emph{Monotone Dynamical Systems}, volume~2.
\newblock 2006.

\bibitem[Leenheer et~al.(2007)Leenheer, Angeli, and Sontag]{Leenheer2007}
Patrick~De Leenheer, David Angeli, and Eduardo~D. Sontag.
\newblock Monotone chemical reaction networks.
\newblock \emph{J Math Chem}, 41:\penalty0 295--314, 2007.

\bibitem[Smith(1995)]{Smith1995}
Hal Smith.
\newblock \emph{Monotone Dynamical Systems: An Introduction to the Theory of
  Competitive and Cooperative Systems}.
\newblock American Mathematical Society, 1995.

\bibitem[Horn and Jackson(1972)]{Horn1972}
F~Horn and R~Jackson.
\newblock General mass action kinetics.
\newblock \emph{Archive for rational mechanics and analysis}, 47:\penalty0
  81--116, 1972.

\bibitem[Horn(1973{\natexlab{a}})]{Horn1973}
F~Horn.
\newblock On a connexion between stability and graphs in chemical kinetics. i.
  stability and the reaction diagram.
\newblock \emph{Proc Roy Soc A}, 334:\penalty0 299--312, 1973{\natexlab{a}}.

\bibitem[Feinberg(1987)]{Feinberg1987}
Martin Feinberg.
\newblock Chemical reaction network structure and the stability of complex
  isothermal reactors—i. the deficiency zero and deficiency one theorems.
\newblock \emph{Chem Eng Sci}, 42:\penalty0 2229--2268, 1987.

\bibitem[Metropolis et~al.(1953)Metropolis, Rosenbluth, Rosenbluth, Teller, and
  Teller]{Metropolis1953}
Nicholas Metropolis, Arianna~W. Rosenbluth, Marshall~N. Rosenbluth, Augusta~H.
  Teller, and Edward Teller.
\newblock Equation of state calculations by fast computing machines.
\newblock \emph{J. Chem. Phys.}, 21\penalty0 (6):\penalty0 1087--1091, 1953.

\bibitem[Tempo et~al.(2013)Tempo, Calafiore, and Dabbene]{tempo2013randomized}
Roberto Tempo, Giuseppe Calafiore, and Fabrizio Dabbene.
\newblock \emph{Randomized Algorithms for Analysis and Control of Uncertain
  Systems: with Applications}, volume~7.
\newblock Springer, 2013.

\bibitem[Vidyasagar(1998)]{Sagar1998}
M.~Vidyasagar.
\newblock Statistical learning theory and randomized algorithms for control.
\newblock \emph{IEEE Control Systems Magazine}, 18:\penalty0 69--85, 1998.

\bibitem[Vidyasagar(2011)]{Sagar2011}
M.~Vidyasagar.
\newblock Probabilistic methods in cancer biology.
\newblock \emph{Europ J Contr}, 17:\penalty0 483--511, 2011.

\bibitem[Palumbo et~al.(2005)Palumbo, Colosimo, Giuliani, and
  Farina]{Palumbo2005}
Maria~Concetta Palumbo, Alfredo Colosimo, Alessandro Giuliani, and Lorenzo
  Farina.
\newblock Functional essentiality from topology features in metabolic networks:
  A case study in yeast.
\newblock \emph{FEBS Letters}, 579:\penalty0 4642--4646, 2005.

\bibitem[Hamadeh et~al.(2011)Hamadeh, Roberts, August, McSharry, Maini,
  Armitage, and Papachristodoulou]{Hamadeh2011}
Abdullah Hamadeh, Mark~A.J. Roberts, Elias August, Patrick~E. McSharry,
  Philip~K. Maini, Judith~P. Armitage, and Antonis Papachristodoulou.
\newblock Feedback control architecture and the bacterial chemotaxis network.
\newblock \emph{PLoS Comput Biol}, 7, 2011.

\bibitem[Polynikis et~al.(2009)Polynikis, Hogan, and
  di~Bernardo]{Polynikis2009}
A.~Polynikis, S.~J. Hogan, and M.~di~Bernardo.
\newblock Comparing different ode modelling approaches for gene regulatory
  networks.
\newblock \emph{J Theor Biol}, 261:\penalty0 511--530, 2009.

\bibitem[Thomaseth et~al.(2017)Thomaseth, Kuritz, Allg{\"o}wer, and
  Radde]{Thomaseth2017}
Caterina Thomaseth, Karsten Kuritz, Frank Allg{\"o}wer, and Nicole Radde.
\newblock The circuit-breaking algorithm for monotone systems.
\newblock \emph{Math Biosci}, 284:\penalty0 80--91, 2017.

\bibitem[Elowitz and Leibler(2000)]{Elowitz2000}
Michael~B Elowitz and Stanislas Leibler.
\newblock A synthetic oscillatory network of transcriptional regulators.
\newblock \emph{Nature}, 403:\penalty0 335--338, 2000.

\bibitem[{El-Samad} et~al.(2005){El-Samad}, {Del Vecchio}, and
  Khammash]{samad2005}
H.~{El-Samad}, D.~{Del Vecchio}, and M.~Khammash.
\newblock Repressilators and promotilators: Loop dynamics in synthetic gene
  networks.
\newblock In \emph{Proc ACC}, volume~6, pages 4405--4410, 2005.

\bibitem[Blanchini and Giordano(2019)]{BG2019}
Franco Blanchini and Giulia Giordano.
\newblock Bdc-decomposition for global influence analysis.
\newblock \emph{IEEE Contr Sys Lett}, 3:\penalty0 260--265, 2019.

\bibitem[Giordano(2016)]{GiuliaPhD2016}
Giulia Giordano.
\newblock \emph{Structural Analysis and Control of Dynamical Networks}.
\newblock PhD thesis, University of Udine, Italy, 2016.

\bibitem[Blanchini and Giordano(2015{\natexlab{a}})]{Blanchini2016}
F.~Blanchini and G.~Giordano.
\newblock Polyhedral lyapunov functions for structural stability of biochemical
  systems in concentration and reaction coordinates.
\newblock In \emph{Proc IEEE CDC}, 2015{\natexlab{a}}.

\bibitem[Blanchini and Giordano(2024)]{BG2024}
Franco Blanchini and Giulia Giordano.
\newblock Spiking systems in population-infection dynamics.
\newblock \emph{IEEE Contr Sys Lett}, 8:\penalty0 688--693, 2024.

\bibitem[Levine et~al.(2013)Levine, Lin, and Elowitz]{Levine2013}
J.~H. Levine, Y.~Lin, and M.~B. Elowitz.
\newblock Functional roles of pulsing in genetic circuits.
\newblock \emph{Science}, 342\penalty0 (6163):\penalty0 1193–1200,, 2013.

\bibitem[Nakamura et~al.(2024)Nakamura, {Cuba Samaniego}, Blanchini, Giordano,
  and Franco]{Nakamura2024}
E.~Nakamura, C.~{Cuba Samaniego}, F.~Blanchini, G.~Giordano, and E.~Franco.
\newblock Design of a sequestration-based network with tunable pulsing
  dynamics.
\newblock In \emph{Proc IEEE CDC}, 2024.

\bibitem[Izhikevich(2000)]{izhikevich2000neural}
Eugene~M Izhikevich.
\newblock Neural excitability, spiking and bursting.
\newblock \emph{Int J Bifurcat Chaos}, 10\penalty0 (06):\penalty0 1171--1266,
  2000.

\bibitem[Izhikevich(2007)]{Izhikevich2007}
Eugene~M. Izhikevich.
\newblock \emph{Dynamical Systems in Neuroscience: The Geometry of Excitability
  and Bursting}.
\newblock MIT Press, 2007.

\bibitem[Meisel et~al.(2015)Meisel, Klaus, Kuehn, and Plenz]{Meisel2015a}
Christian Meisel, Andreas Klaus, Christian Kuehn, and Dietmar Plenz.
\newblock {Critical slowing down governs the transition to neuron spiking}.
\newblock \emph{PLoS Comput Biol}, 11\penalty0 (2):\penalty0 1--21, 2015.

\bibitem[August and Barahona(2010)]{August2010}
Elias August and Mauricio Barahona.
\newblock Solutions of weakly reversible chemical reaction networks are bounded
  and persistent.
\newblock \emph{IFAC Proc Volumes}, 43:\penalty0 42--47, 2010.

\bibitem[Nagumo(1942)]{Nagumo1942}
Mitio Nagumo.
\newblock Über die lage der integralkurven gewöhnlicher
  differentialgleichungen.
\newblock \emph{Proc Phys Math Soc Japan}, 24:\penalty0 551--559, 1942.

\bibitem[Srzednicki(1985)]{Srzednicki1985}
Roman Srzednicki.
\newblock On the rest points of dynamical systems, 1985.

\bibitem[Richeson and Wiseman(2002)]{Richeson2002}
David Richeson and Jim Wiseman.
\newblock A fixed point theorem for bounded dynamical systems.
\newblock \emph{Illinois J Math}, 46:\penalty0 491--495, 2002.

\bibitem[Richeson and Wiseman(2004)]{Richeson2004}
David Richeson and Jim Wiseman.
\newblock Addendum to “a fixed point theorem for bounded dynamical
  systems”.
\newblock \emph{Illinois J Math}, 48:\penalty0 1079--1080, 2004.

\bibitem[Kellogg(1976)]{Kellogg1976}
R.~B. Kellogg.
\newblock Uniqueness in the schauder fixed point theorem.
\newblock \emph{Proc Am Math Soc}, 60:\penalty0 207--210, 1976.

\bibitem[Zampieri(1992)]{Zampieri1992}
Gaetano Zampieri.
\newblock Diffeomorphisms with banach space domains.
\newblock \emph{Nonlin An Theo Meth Appl}, 19:\penalty0 923--932, 1992.

\bibitem[Gantmacher(1959)]{Gantmacher}
F.~R. Gantmacher.
\newblock \emph{The Theory of Matrices}.
\newblock Chelsea Publishing Co., New York, NY, 1959.

\bibitem[Feinberg(1995{\natexlab{a}})]{Feinberg1995}
Martin Feinberg.
\newblock The existence and uniqueness of steady states for a class of chemical
  reaction networks.
\newblock \emph{Arch. Rational Mech. Anal}, 132:\penalty0 311--370,
  1995{\natexlab{a}}.

\bibitem[Feinberg(1995{\natexlab{b}})]{Feinberg1995multiple}
Martin Feinberg.
\newblock The existence and uniqueness of steady states for a class of chemical
  reaction networks.
\newblock \emph{Arch. Rational Mech. Anal}, 132:\penalty0 371--406,
  1995{\natexlab{b}}.

\bibitem[Banaji et~al.(2007)Banaji, Donnell, and Baigent]{Banaji2007}
Murad Banaji, Pete Donnell, and Stephen Baigent.
\newblock P matrix properties, injectivity, and stability in chemical reaction
  systems.
\newblock \emph{Siam J Appl Math}, 67:\penalty0 1523--1547, 2007.

\bibitem[Banaji and Craciun(2009)]{Banaji2009}
Murad Banaji and Gheorghe Craciun.
\newblock Graph-theoretic criteria for injectivity and unique equilibria in
  general chemical reaction systems.
\newblock \emph{Adv Appl Math}, 44:\penalty0 168--184, 2009.

\bibitem[Craciun and Feinberg(2005)]{Craciun2005}
Gheorghe Craciun and Martin Feinberg.
\newblock Multiple equilibria in complex chemical reaction networks: I. the
  injectivity property.
\newblock \emph{Siam J Appl Math}, 65:\penalty0 1526--1546, 2005.

\bibitem[Craciun and Feinberg(2006)]{Craciun2006}
Gheorghe Craciun and Martin Feinberg.
\newblock Multiple equilibria in complex chemical reaction networks: Ii. the
  species-reaction graph.
\newblock \emph{Siam J Appl Math}, 66:\penalty0 1321--1338, 2006.

\bibitem[Joshi and Shiu(2015)]{Joshi2015}
Badal Joshi and Anne Shiu.
\newblock A survey of methods for deciding whether a reaction network is
  multistationary.
\newblock \emph{Math Mod Natur Phenom}, 10:\penalty0 47--67, 2015.

\bibitem[Clarke(1975)]{Clarke1975}
Bruce~L Clarke.
\newblock Theorems on chemical network stability.
\newblock \emph{J Chem Phys}, 62:\penalty0 773--775, 1975.

\bibitem[Jeffries(1974)]{Jeffries1974}
C.~Jeffries.
\newblock Qualitative stability and digraphs in model ecosystems.
\newblock \emph{Ecology}, 55\penalty0 (6):\penalty0 1415--1419, 1974.

\bibitem[Levins(1974)]{Levins1974}
R.~Levins.
\newblock The qualitative analysis of partially specified systems.
\newblock \emph{Ann. N. Y. Acad. Sci.}, 231:\penalty0 123--138, 1974.

\bibitem[Levins(1977)]{Levins1977}
R.~Levins.
\newblock Qualitative analysis of complex systems.
\newblock \emph{Mathematics and the Life Sciences. Lec Notes Biomath},
  18:\penalty0 153--199, 1977.

\bibitem[May(1973)]{May1973}
Robert~M May.
\newblock Qualitative stability in model ecosystems.
\newblock \emph{Ecology}, 54:\penalty0 638--641, 1973.

\bibitem[May(1974)]{May1974}
Robert~M May.
\newblock \emph{Stability and Complexity in Model Ecosystems}.
\newblock Princeton University Press, Princeton, New Jersey, 1974.

\bibitem[Quirk and Ruppert(1965)]{Quirk1965}
James Quirk and Richard Ruppert.
\newblock Qualitative economics and the stability of equilibrium.
\newblock \emph{Rev of Econ Stud}, 32:\penalty0 311--326, 1965.

\bibitem[Jeffries et~al.(1977)Jeffries, Klee, and van~den
  Driessche]{Jeffries1977}
Clark Jeffries, Victor Klee, and Pauline van~den Driessche.
\newblock When is a matrix sign stable?
\newblock \emph{Canadian J Math}, 29:\penalty0 315--326, 1977.

\bibitem[Logofet and Ulianov(1982)]{Logofet1982}
D.O. Logofet and N.B. Ulianov.
\newblock Sign stability in model ecosystems: A complete class of sign-stable
  patterns.
\newblock \emph{Eco Mod}, 16\penalty0 (2):\penalty0 173--189, 1982.

\bibitem[Dambacher et~al.(2002)Dambacher, Li, and Rossignol]{Dambacher2002}
Jeffrey~M Dambacher, Hiram~W Li, and Philippe~A Rossignol.
\newblock Relevance of community structure in assessing indeterminacy of
  ecological predictions.
\newblock \emph{Ecology}, 83:\penalty0 1372--1385, 2002.

\bibitem[Dambacher et~al.(2003{\natexlab{a}})Dambacher, Luh, Li, and
  Rossignol]{Dambacher2003stability}
Jeffrey~M. Dambacher, Hang‐Kwang Luh, Hiram~W. Li, and Philippe~A. Rossignol.
\newblock Qualitative stability and ambiguity in model ecosystems.
\newblock \emph{The American Naturalist}, 161\penalty0 (6):\penalty0 876--888,
  2003{\natexlab{a}}.

\bibitem[Dambacher et~al.(2003{\natexlab{b}})Dambacher, Li, and
  Rossignol]{Dambacher2003}
Jeffrey~M Dambacher, Hiram~W Li, and Philippe~A Rossignol.
\newblock Qualitative predictions in model ecosystems.
\newblock \emph{Eco Mod}, 161:\penalty0 79--93, 2003{\natexlab{b}}.

\bibitem[Dambacher et~al.(2005)Dambacher, Levins, and Rossignol]{Dambacher2005}
Jeffrey~M. Dambacher, Richard Levins, and Philippe~A. Rossignol.
\newblock Life expectancy change in perturbed communities: Derivation and
  qualitative analysis.
\newblock \emph{Math Biosci}, 197:\penalty0 1--14, 2005.

\bibitem[Dambacher and Ramos-Jiliberto(2007)]{Dambacher2007}
Jeffrey~M Dambacher and Rodrigo Ramos-Jiliberto.
\newblock Understanding and predicting effects of modified interactions through
  a qualitative analysis of community structure.
\newblock \emph{Quart Rev Biol}, 82, 2007.

\bibitem[Dambacher et~al.(2009)Dambacher, Gaughan, Rochet, Rossignol, and
  Trenkel]{Dambacher2009}
Jeffrey~M. Dambacher, Daniel~J. Gaughan, Marie~Joëlle Rochet, Philippe~A.
  Rossignol, and Verena~M. Trenkel.
\newblock Qualitative modelling and indicators of exploited ecosystems.
\newblock \emph{Fish and Fisheries}, 10:\penalty0 305--322, 2009.

\bibitem[Marzloff et~al.(2011)Marzloff, Dambacher, Johnson, Little, and
  Frusher]{Marzloff2011}
Martin~P. Marzloff, Jeffrey~M. Dambacher, Craig~R. Johnson, L.~Richard Little,
  and Stewart~D. Frusher.
\newblock Exploring alternative states in ecological systems with a qualitative
  analysis of community feedback.
\newblock \emph{Eco Mod}, 222:\penalty0 2651--2662, 2011.

\bibitem[Maeda et~al.(1978)Maeda, Kodama, and Ohta]{Maeda1978}
Hajime Maeda, Shinzo Kodama, and Wzo Ohta.
\newblock Asymptotic behavior of nonlinear compartmental systems:
  nonoscillation and stability.
\newblock \emph{IEEE T Circuits Sys}, 1978.

\bibitem[Horn(1973{\natexlab{b}})]{Horn1973b}
F~Horn.
\newblock On a connexion between stability and graphs in chemical kinetics ii.
  stability and the complex graph.
\newblock \emph{Proc Roy Soc A}, 334:\penalty0 313--330, 1973{\natexlab{b}}.

\bibitem[Sontag(2001)]{Sontag2001}
Eduardo~D Sontag.
\newblock Structure and stability of certain chemical networks and applications
  to the kinetic proofreading model of t-cell receptor signal transduction.
\newblock \emph{IEEE T Automat Contr}, 46, 2001.

\bibitem[Hangos(2010)]{Hangos2010}
Katalin~M. Hangos.
\newblock Engineering model reduction and entropy-based lyapunov functions in
  chemical reaction kinetics.
\newblock \emph{Entropy}, 12:\penalty0 772--797, 2010.

\bibitem[Ke et~al.(2019)Ke, Fang, and Gao]{Ke2019}
Min Ke, Zhou Fang, and Chuanhou Gao.
\newblock Complex balancing reconstructed to the asymptotic stability of
  mass-action chemical reaction networks with conservation laws.
\newblock \emph{Siam J Appl Math}, 79:\penalty0 55--74, 2019.

\bibitem[Anderson(2008)]{Anderson2008}
David~F. Anderson.
\newblock Global asymptotic stability for a class of nonlinear chemical
  equations.
\newblock \emph{Siam J Appl Math}, 68:\penalty0 1464--1476, 2008.

\bibitem[Rao(2017)]{Rao2017}
Shodhan Rao.
\newblock Global stability of a class of futile cycles.
\newblock \emph{J Math Biol}, 74:\penalty0 709--726, 2017.

\bibitem[Chaves(2006)]{Chaves2006b}
Madalena Chaves.
\newblock Stability of rate-controlled zero-deficiency networks.
\newblock In \emph{Proc 45th IEEE CDC}, pages 5766--5771, 2006.

\bibitem[Chaves et~al.(2008)Chaves, Eissing, and Allgöwer]{Chaves2008}
Madalena Chaves, Thomas Eissing, and Frank Allgöwer.
\newblock Bistable biological systems: A characterization through local compact
  input-to-state stability.
\newblock \emph{IEEE T Automat Contr}, 53:\penalty0 87--100, 2008.

\bibitem[Nilsson and Giordano(2025)]{nilssongiordano2025}
Gustav Nilsson and Giulia Giordano.
\newblock Stability analysis of biological networks through a dynamical flow
  network modeling approach.
\newblock In \emph{Proc 2025 European Control Conference}, 2025.

\bibitem[Chesi and Hung(2008)]{Chesi2008}
G.~Chesi and Y.~S. Hung.
\newblock Stability analysis of uncertain genetic sum regulatory networks.
\newblock \emph{Automatica}, 44:\penalty0 2298--2305, 2008.

\bibitem[Franco and Blanchini(2013)]{Franco2013}
Elisa Franco and Franco Blanchini.
\newblock Structural properties of the mapk pathway topologies in pc12 cells.
\newblock \emph{J Math Biol}, 67:\penalty0 1633--1668, 2013.

\bibitem[Grognard et~al.(2005)Grognard, Mazenc, and Rapaport]{Grognard2005}
Frédéric Grognard, Frédéric Mazenc, and Alain Rapaport.
\newblock Polytopic lyapunov functions for the stability analysis of
  persistence of competing species.
\newblock In \emph{Proc 44th IEEE CDC}, pages 3699--3704, 2005.

\bibitem[Pasquini and Angeli(2020)]{Pasquini2020}
Mirko Pasquini and David Angeli.
\newblock On convergence for piecewise affine models of gene regulatory
  networks via a lyapunov approach.
\newblock \emph{IEEE T Automat Contr}, 65:\penalty0 3333--3348, 2020.

\bibitem[Blanchini et~al.(2020{\natexlab{a}})Blanchini, Chesi, Colaneri, and
  Giordano]{BCCG2020}
Franco Blanchini, Graziano Chesi, Patrizio Colaneri, and Giulia Giordano.
\newblock Checking structural stability of bdc-decomposable systems via convex
  optimisation.
\newblock \emph{IEEE Contr Sys Lett}, 4\penalty0 (1):\penalty0 205--210,
  2020{\natexlab{a}}.

\bibitem[Molchanov and Pyatnitskiy(1986)]{Molchanov1986}
A.P. Molchanov and Ye.S. Pyatnitskiy.
\newblock Lyapunov functions that define necessary and sufficient conditions
  for absolute stability of nonlinear nonstationary control systems.
\newblock \emph{Autom. Remote Control. I: 47(3):344-354; II: 47(4):443-451;
  III: 47(5):620-630}, 1986.

\bibitem[Molchanov and Pyatnitskiy(1989)]{Molchanov1989}
A.P. Molchanov and Ye.S. Pyatnitskiy.
\newblock Criteria of asymptotic stability of differential and difference
  inclusions encountered in control theory.
\newblock \emph{Sys Contr Lett}, 13\penalty0 (1):\penalty0 59--64, 1989.

\bibitem[Chesi et~al.(2011)Chesi, Chen, and Aihara]{Chesi2011c}
Graziano Chesi, Luonan Chen, and Kazuyuki Aihara.
\newblock On the robust stability of time-varying uncertain genetic regulatory
  networks.
\newblock \emph{Int J Robust Nonlin}, 21:\penalty0 1778--1790, 2011.

\bibitem[Blanchini and Giordano(2015{\natexlab{b}})]{Blanchini2015c}
Franco Blanchini and Giulia Giordano.
\newblock Structural stability of biochemical networks: Quadratic vs.
  polyhedral lyapunov functions.
\newblock \emph{IFAC-PapersOnLine}, 28:\penalty0 278--283, 2015{\natexlab{b}}.

\bibitem[Pasquini and Angeli(2018)]{Pasquini2018}
Mirko Pasquini and David Angeli.
\newblock On piecewise quadratic lyapunov functions for piecewise affine models
  of gene regulatory networks.
\newblock In \emph{Proc IEEE CDC}, pages 1071--1076, 2018.

\bibitem[Angeli et~al.(2022)Angeli, Al-Radhawi, and Sontag]{Angeli2022}
David Angeli, Muhammad~Ali Al-Radhawi, and Eduardo~D. Sontag.
\newblock A robust lyapunov criterion for nonoscillatory behaviors in
  biological interaction networks.
\newblock \emph{IEEE T Automat Contr}, 67:\penalty0 3305--3320, 2022.

\bibitem[Angeli et~al.(2009{\natexlab{b}})Angeli, Leenheer, and
  Sontag]{Angeli2009}
David Angeli, Patrick~De Leenheer, and Eduardo~D Sontag.
\newblock Chemical networks with inflows and outflows: A positive linear
  differential inclusions approach.
\newblock \emph{Am Inst Chem Eng Biotechnol Prog}, 25:\penalty0 632--642,
  2009{\natexlab{b}}.

\bibitem[Al-Radhawi et~al.(2020)Al-Radhawi, Angeli, and Sontag]{Al-Radhawi2020}
M~Ali Al-Radhawi, David Angeli, and Eduardo~D Sontag.
\newblock A computational framework for a lyapunov-enabled analysis of
  biochemical reaction networks.
\newblock \emph{PLoS Comput Biol}, 16:\penalty0 e1007681, 2020.

\bibitem[Blanchini and Giordano(2022)]{BG2022}
Franco Blanchini and Giulia Giordano.
\newblock Dual chemical reaction networks and implications for lyapunov-based
  structural stability.
\newblock \emph{IEEE Contr Sys Lett}, 6:\penalty0 488--493, 2022.

\bibitem[Blanchini et~al.(2020{\natexlab{b}})Blanchini, Colaneri, Giordano, and
  Zorzan]{BCGZ2020}
Franco Blanchini, Patrizio Colaneri, Giulia Giordano, and Irene Zorzan.
\newblock Predicting adaptation for uncertain systems with robust real plots.
\newblock In \emph{Proc 59th IEEE CDC}, pages 5861--5866, 2020{\natexlab{b}}.

\bibitem[Blanchini et~al.(2022)Blanchini, Colaneri, Giordano, and
  Zorzan]{BCGZ2022}
Franco Blanchini, Patrizio Colaneri, Giulia Giordano, and Irene Zorzan.
\newblock Vertex results for the robust analysis of uncertain biochemical
  systems.
\newblock \emph{J Math Biol}, 85:\penalty0 35, 2022.

\bibitem[Blanchini et~al.(2017{\natexlab{a}})Blanchini, Casagrande, Giordano,
  and Viaro]{Blanchini2017}
Franco Blanchini, Daniele Casagrande, Giulia Giordano, and Umberto Viaro.
\newblock A bounded complementary sensitivity function ensures
  topology-independent stability of homogeneous dynamical networks.
\newblock \emph{IEEE T Automat Contr}, 63:\penalty0 1140 -- 1146,
  2017{\natexlab{a}}.

\bibitem[Devia and Giordano(2020)]{DG2020}
Carlos~Andrés Devia and Giulia Giordano.
\newblock Topology-independent robust stability for networks of homogeneous
  mimo systems.
\newblock \emph{IFAC-PapersOnLine}, 53\penalty0 (2):\penalty0 3379--3384, 2020.

\bibitem[Devia and Giordano(2021{\natexlab{a}})]{DG2021}
Carlos~Andrés Devia and Giulia Giordano.
\newblock Topology-independent robust stability conditions for uncertain mimo
  networks.
\newblock \emph{IEEE Contr Sys Lett}, 5\penalty0 (1):\penalty0 325--330,
  2021{\natexlab{a}}.

\bibitem[Devia and Giordano(2021{\natexlab{b}})]{DG2021ecc}
Carlos~Andrés Devia and Giulia Giordano.
\newblock Mimo networks with heterogeneous uncertainties: Topology-independent
  robust stability and $\alpha$-convergence.
\newblock In \emph{2021 ECC}, pages 884--889, 2021{\natexlab{b}}.

\bibitem[Blanchini et~al.(2023{\natexlab{b}})Blanchini, Breda, Giordano, and
  Liessi]{BLANCHINI2023110683}
Franco Blanchini, Dimitri Breda, Giulia Giordano, and Davide Liessi.
\newblock Michaelis-menten networks are structurally stable.
\newblock \emph{Automatica}, 147:\penalty0 110683, 2023{\natexlab{b}}.

\bibitem[Fiedler and Mochizuki(2015)]{Fiedler2015}
Bernold Fiedler and Atsushi Mochizuki.
\newblock Sensitivity of chemical reaction networks: A structural approach. 2.
  regular monomolecular systems.
\newblock \emph{Math Meth Appl Sci}, 38:\penalty0 3519--3537, 2015.

\bibitem[Brehm and Fiedler(2018)]{Brehm2018}
Bernhard Brehm and Bernold Fiedler.
\newblock Sensitivity of chemical reaction networks: A structural approach.
  3. regular multimolecular systems.
\newblock \emph{Math Meth Appl Sci}, 41:\penalty0 1344--1376, 2018.

\bibitem[Araujo and Liotta(2018)]{Araujo2018}
Robyn~P. Araujo and Lance~A. Liotta.
\newblock The topological requirements for robust perfect adaptation in
  networks of any size.
\newblock \emph{Nat Commun}, 9, 2018.

\bibitem[Drengstig et~al.(2008)Drengstig, Ueda, and Ruoff]{Drengstig2008}
Tormod Drengstig, Hiroki~R. Ueda, and Peter Ruoff.
\newblock Predicting perfect adaptation motifs in reaction kinetic networks.
\newblock \emph{J Phys Chem B}, 112:\penalty0 16752--16758, 2008.

\bibitem[Ma et~al.(2009)Ma, Trusina, {El-Samad}, Lim, and Tang]{Ma2009}
Wenzhe Ma, Ala Trusina, Hana {El-Samad}, Wendell~A. Lim, and Chao Tang.
\newblock Defining network topologies that can achieve biochemical adaptation.
\newblock \emph{Cell}, 138:\penalty0 760--773, 2009.

\bibitem[Sontag(2003)]{Sontag2003}
Eduardo~D. Sontag.
\newblock Adaptation and regulation with signal detection implies internal
  model.
\newblock \emph{Sys Contr Lett}, 50:\penalty0 119--126, 2003.

\bibitem[Waldherr et~al.(2012)Waldherr, Streif, and Allgöwer]{Waldherr2012}
S.~Waldherr, S.~Streif, and F.~Allgöwer.
\newblock Design of biomolecular network modifications to achieve adaptation.
\newblock \emph{IET Sys Biol}, 6:\penalty0 223--231, 2012.

\bibitem[Muzzey et~al.(2009)Muzzey, Gómez-Uribe, Mettetal, and van
  Oudenaarden]{Muzzey2009}
Dale Muzzey, Carlos~A. Gómez-Uribe, Jerome~T. Mettetal, and Alexander van
  Oudenaarden.
\newblock A systems-level analysis of perfect adaptation in yeast
  osmoregulation.
\newblock \emph{Cell}, 138:\penalty0 160--171, 2009.

\bibitem[{El-Samad} et~al.(2002){El-Samad}, Goff, and Khammash]{samad2002}
H.~{El-Samad}, J.~P. Goff, and M.~Khammash.
\newblock Calcium homeostasis and parturient hypocalcemia: An integral feedback
  perspective.
\newblock \emph{J Theor Biol}, 214:\penalty0 17--29, 2002.

\bibitem[Clausznitzer et~al.(2010)Clausznitzer, Oleksiuk, Løvdok, Sourjik, and
  Endres]{Clausznitzer2010}
Diana Clausznitzer, Olga Oleksiuk, Linda Løvdok, Victor Sourjik, and Robert~G.
  Endres.
\newblock Chemotactic response and adaptation dynamics in escherichia coli.
\newblock \emph{PLoS Comput Biol}, 6:\penalty0 1--11, 2010.

\bibitem[Levchenko and Iglesias(2002)]{Levchenko2002}
Andre Levchenko and Pablo~A. Iglesias.
\newblock Models of eukaryotic gradient sensing: Application to chemotaxis of
  amoebae and neutrophils.
\newblock \emph{Biophys J}, 82:\penalty0 50--63, 2002.

\bibitem[Spiro et~al.(1997)Spiro, Parkinson, and Othmer]{Spiro1997}
Peter~A Spiro, John~S Parkinson, and Hans~G Othmer.
\newblock A model of excitation and adaptation in bacterial chemotaxis.
\newblock \emph{Biochemistry}, 94:\penalty0 7263--7268, 1997.

\bibitem[Yi et~al.(2000)Yi, Huang, Simon, and Doyle]{Yi2000}
Tau-Mu Yi, Yun Huang, Melvin~I Simon, and John Doyle.
\newblock Robust perfect adaptation in bacterial chemotaxis through integral
  feedback control.
\newblock \emph{P Natl Acad Sci Usa}, 97:\penalty0 4649--4653, 2000.

\bibitem[Briat et~al.(2016)Briat, Gupta, and Khammash]{Briat2016}
Corentin Briat, Ankit Gupta, and Mustafa Khammash.
\newblock Antithetic integral feedback ensures robust perfect adaptation in
  noisy bimolecular networks.
\newblock \emph{Cell Systems}, 2:\penalty0 15--26, 2016.

\bibitem[Gupta and Khammash(2022)]{Gupta2022}
Ankit Gupta and Mustafa Khammash.
\newblock Universal structural requirements for maximal robust perfect
  adaptation in biomolecular networks.
\newblock \emph{P Natl Acad Sci Usa}, 119, 2022.

\bibitem[Steel and Papachristodoulou(2018)]{Steel2018}
Harrison Steel and Antonis Papachristodoulou.
\newblock Design constraints for biological systems that achieve adaptation and
  disturbance rejection.
\newblock \emph{IEEE T Contr Network Sys}, 5:\penalty0 807--817, 2018.

\bibitem[Bender et~al.(1984)Bender, Case, and Gilpin]{Bender1984}
Edward~A. Bender, Ted~J. Case, and Michael~E. Gilpin.
\newblock Perturbation experiments in community ecology: Theory and practice.
\newblock \emph{Ecology}, 65\penalty0 (1):\penalty0 1--13, 1984.

\bibitem[Giordano and Altafini(2017{\natexlab{b}})]{Giordano2017interaction}
Giulia Giordano and Claudio Altafini.
\newblock Interaction sign patterns in biological networks: from qualitative to
  quantitative criteria.
\newblock In \emph{Proc 56th IEEE CDC}, pages 5348--5353, 2017{\natexlab{b}}.

\bibitem[Montoya et~al.(2009)Montoya, Woodward, Emmerson, and
  Solé]{Montoya2009}
José~M. Montoya, Guy Woodward, Mark~C. Emmerson, and Ricard~V. Solé.
\newblock Press perturbations and indirect effects in real food webs.
\newblock \emph{Ecology}, 90\penalty0 (9):\penalty0 2426--2433, 2009.

\bibitem[Novak et~al.(2011)Novak, Wootton, Doak, Emmerson, Estes, and
  Tinker]{Novak2011}
Mark Novak, J.~Timothy Wootton, Daniel~F. Doak, Mark Emmerson, James~A. Estes,
  and M.~Timothy Tinker.
\newblock Predicting community responses to perturbations in the face of
  imperfect knowledge and network complexity.
\newblock \emph{Ecology}, 92\penalty0 (4):\penalty0 836--846, 2011.

\bibitem[Novak et~al.(2016)Novak, Yeakel, Noble, Doak, Emmerson, Estes, Jacob,
  Tinker, and Wootton]{Novak2016}
Mark Novak, Justin~D. Yeakel, Andrew~E. Noble, Daniel~F. Doak, Mark Emmerson,
  James~A. Estes, Ute Jacob, M.~Timothy Tinker, and J.~Timothy Wootton.
\newblock Characterizing species interactions to understand press
  perturbations: What is the community matrix?
\newblock \emph{Annu Rev Eco Evol System}, 47:\penalty0 409--432, 2016.

\bibitem[Schmitz(1997)]{Schmitz1997}
Oswald~J. Schmitz.
\newblock Press perturbations and the predictability of ecological interactions
  in a food web.
\newblock \emph{Ecology}, 78\penalty0 (1):\penalty0 55--69, 1997.

\bibitem[Yodzis(1988)]{Yodzis1988}
Peter Yodzis.
\newblock The indeterminacy of ecological interactions as perceived through
  perturbation experiments.
\newblock \emph{Ecology}, 69\penalty0 (2):\penalty0 508--515, 1988.

\bibitem[Giordano and Franco(2016)]{Giordano2016negative}
G~Giordano and E~Franco.
\newblock Negative feedback enables structurally signed steady-state influences
  in artificial biomolecular networks.
\newblock In \emph{IEEE 55th Conference on Decision and Control, CDC 2016},
  pages 3369--3374, 2016.

\bibitem[Franco et~al.(2014)Franco, Giordano, Forsberg, and Murray]{Franco2014}
Elisa Franco, Giulia Giordano, Per-Ola Forsberg, and Richard~M Murray.
\newblock Negative autoregulation matches production and demand in synthetic
  transcriptional networks.
\newblock \emph{ACS Synth Biol}, 3:\penalty0 589–599, 2014.

\bibitem[Giordano et~al.(2013)Giordano, Franco, and Murray]{Giordano2013}
G~Giordano, E~Franco, and R~M Murray.
\newblock Feedback architectures to regulate flux of components in artificial
  gene networks.
\newblock In \emph{Proc ACC}, pages 4747--4752, 2013.

\bibitem[Blanchini et~al.(2023{\natexlab{c}})Blanchini, Franco, Giordano, and
  Osmanović]{BFGO2023}
Franco Blanchini, Elisa Franco, Giulia Giordano, and Dino Osmanović.
\newblock Robust microphase separation through chemical reaction networks.
\newblock \emph{IEEE Contr Sys Lett}, 7:\penalty0 2281--2286,
  2023{\natexlab{c}}.

\bibitem[Hanada(2010)]{Hanada2010}
Kentaro Hanada.
\newblock Intracellular trafficking of ceramide by ceramide transfer protein.
\newblock \emph{Proc Japan Acad B, Series B}, 86\penalty0 (4):\penalty0
  426--437, 2010.

\bibitem[Weber et~al.(2015)Weber, Hornjik, Olayioye, Hausser, and
  Radde]{Weber2015}
Patrick Weber, Mariana Hornjik, Monilola~A. Olayioye, Angelika Hausser, and
  Nicole~E. Radde.
\newblock A computational model of pkd and cert interactions at the trans-golgi
  network of mammalian cells.
\newblock \emph{BMC Sys Biol}, 9\penalty0 (1):\penalty0 9, 2015.

\bibitem[Giordano and Blanchini(2017)]{Giordano2017}
Giulia Giordano and Franco Blanchini.
\newblock Flow-inducing networks.
\newblock \emph{IEEE Contr Sys Lett}, 1:\penalty0 44--49, 2017.

\bibitem[Cournac and Sepulchre(2009)]{Cournac2009}
Axel Cournac and Jacques~Alexandre Sepulchre.
\newblock Simple molecular networks that respond optimally to time-periodic
  stimulation.
\newblock \emph{BMC Sys Biol}, 3, 2009.

\bibitem[Fiore et~al.(2018)Fiore, Guarino, and Bernardo]{Fiore2018}
Davide Fiore, Agostino Guarino, and Mario~Di Bernardo.
\newblock Analysis and control of genetic toggle switches subject to periodic
  multi-input stimulation.
\newblock \emph{IEEE Contr Sys Lett}, 3:\penalty0 278--283, 2018.

\bibitem[Ingalls(2004)]{Ingalls2004}
Brian~P. Ingalls.
\newblock A frequency domain approach to sensitivity analysis of biochemical
  networks.
\newblock \emph{J Phys Chem B}, 108:\penalty0 1143--1152, 2004.

\bibitem[Rahi et~al.(2017)Rahi, Larsch, Pecani, Katsov, Mansouri,
  Tsaneva-Atanasova, Sontag, and Cross]{Rahi2017}
Sahand~Jamal Rahi, Johannes Larsch, Kresti Pecani, Alexander~Y. Katsov, Nahal
  Mansouri, Krasimira Tsaneva-Atanasova, Eduardo~D. Sontag, and Frederick~R.
  Cross.
\newblock Oscillatory stimuli differentiate adapting circuit topologies.
\newblock \emph{Nature Methods}, 14:\penalty0 1010--1016, 2017.

\bibitem[Russo et~al.(2010)Russo, di~Bernardo, and Sontag]{Russo2010}
Giovanni Russo, Mario di~Bernardo, and Eduardo~D. Sontag.
\newblock Global entrainment of transcriptional systems to periodic inputs.
\newblock \emph{PLoS Comput Biol}, 6, 2010.

\bibitem[Okada et~al.(2018)Okada, Tsai, and Mochizuki]{Okada2018}
Takashi Okada, Je~Chiang Tsai, and Atsushi Mochizuki.
\newblock Structural bifurcation analysis in chemical reaction networks.
\newblock \emph{Phys Rev E}, 98, 2018.

\bibitem[Mengel et~al.(2010)Mengel, Hunziker, Pedersen, Trusina, Jensen, and
  Krishna]{MENGEL2010}
Benedicte Mengel, Alexander Hunziker, Lykke Pedersen, Ala Trusina, Mogens~H
  Jensen, and Sandeep Krishna.
\newblock Modeling oscillatory control in {NF-$\kappa$B}, {p53} and {Wnt}
  signaling.
\newblock \emph{Curr Op Gen Develop}, 20\penalty0 (6):\penalty0 656--664, 2010.
\newblock Genetics of system biology.

\bibitem[Zak et~al.(2005)Zak, Stelling, and Doyle]{Zak2005}
Daniel~E. Zak, Jörg Stelling, and Francis~J. Doyle.
\newblock Sensitivity analysis of oscillatory (bio)chemical systems.
\newblock \emph{Comput Chem Eng}, 29:\penalty0 663--673, 2005.

\bibitem[Ebenhöh and Hazlerigg(2013)]{Ebenhoh2013}
Oliver Ebenhöh and David Hazlerigg.
\newblock Modelling a molecular calendar: The seasonal photoperiodic response
  in mammals.
\newblock \emph{Chaos Soliton Fract}, 50:\penalty0 39--47, 2013.

\bibitem[Balázsi et~al.(2003)Balázsi, Cornell-Bell, and
  Moss]{balazsi2003increased}
Gábor Balázsi, Ann~H. Cornell-Bell, and Frank Moss.
\newblock Increased phase synchronization of spontaneous calcium oscillations
  in epileptic human versus normal rat astrocyte cultures.
\newblock \emph{Chaos}, 13:\penalty0 515--518, 2003.

\bibitem[Gunawan and Doyle(2006)]{Gunawan2006}
Rudiyanto Gunawan and Francis~J. Doyle.
\newblock Isochron-based phase response analysis of circadian rhythms.
\newblock \emph{Biophys J}, 91:\penalty0 2131--2141, 2006.

\bibitem[Katz et~al.(2024)Katz, Kriecherbauer, Gr\"{u}ne, and
  Margaliot]{Katz2024}
Rami Katz, Thomas Kriecherbauer, Lars Gr\"{u}ne, and Michael Margaliot.
\newblock On the gain of entrainment in a class of weakly contractive bilinear
  control systems.
\newblock \emph{SIAM J Contr Opt}, 62\penalty0 (5):\penalty0 2723--2749, 2024.

\bibitem[Kuramoto(1984)]{Kuramoto1984}
Yoshiki Kuramoto.
\newblock \emph{Chemical Oscillations, Waves, and Turbulence}.
\newblock Springer-Verlag, 1984.

\bibitem[Li et~al.(2006)Li, Chen, and Aihara]{Li2006}
Chunguang Li, Luonan Chen, and Kazuyuki Aihara.
\newblock Synchronization of coupled nonidentical genetic oscillators.
\newblock \emph{Phys Bioly}, 3:\penalty0 37--44, 2006.

\bibitem[Uriu et~al.(2012)Uriu, Ares, Oates, and Morelli]{Uriu2012}
Koichiro Uriu, Saúl Ares, Andrew~C. Oates, and Luis~G. Morelli.
\newblock Optimal cellular mobility for synchronization arising from the
  gradual recovery of intercellular interactions.
\newblock \emph{Phys Bioly}, 9, 2012.

\bibitem[Craciun et~al.(2006)Craciun, Tang, and Feinberg]{Craciun2006b}
Gheorghe Craciun, Yangzhong Tang, and Martin Feinberg.
\newblock Understanding bistability in complex enzyme-driven reaction networks.
\newblock \emph{P Natl Acad Sci Usa}, 103:\penalty0 8697--8702, 2006.

\bibitem[Eissing et~al.(2004)Eissing, Conzelmann, Gilles, Allgöwer, Bullinger,
  and Scheurich]{Eissing2004}
Thomas Eissing, Holger Conzelmann, Ernst~D. Gilles, Frank Allgöwer, Eric
  Bullinger, and Peter Scheurich.
\newblock Bistability analyses of a caspase activation model for
  receptor-induced apoptosis.
\newblock \emph{J Biol Chem}, 279:\penalty0 36892--36897, 2004.

\bibitem[Ferrell(2012)]{Ferrell2012}
James~E. Ferrell.
\newblock Bistability, bifurcations, and waddington's epigenetic landscape.
\newblock \emph{Curr Biol}, 22:\penalty0 R458--R466, 2012.

\bibitem[Xiong and Ferrell(2003)]{Xiong2003}
W.~Xiong and J.~Ferrell.
\newblock A positive-feedback-based bistable `memory module' that governs a
  cell fate decision.
\newblock \emph{Nature}, 426:\penalty0 460--465, 2003.

\bibitem[Zorzan et~al.(2019)Zorzan, {Del Favero}, {Di Camillo}, and
  Schenato]{Zorzan2019}
Irene Zorzan, Simone {Del Favero}, Barbara {Di Camillo}, and Luca Schenato.
\newblock Analysis of a minimal gene regulatory network for cell
  differentiation.
\newblock \emph{IEEE Contr Sys Lett}, 3:\penalty0 302--307, 2019.

\bibitem[Blanchini et~al.(2014{\natexlab{b}})Blanchini, {Cuba Samaniego},
  Franco, and Giordano]{Blanchini2014c}
F~Blanchini, C~{Cuba Samaniego}, E~Franco, and G~Giordano.
\newblock Design of a molecular clock with rna-mediated regulation.
\newblock In \emph{Proc IEEE CDC}, volume 2015-February, pages 4611--4616,
  2014{\natexlab{b}}.

\bibitem[{Cuba Samaniego} et~al.(2016){Cuba Samaniego}, Giordano, Kim,
  Blanchini, and Franco]{Samaniego2016}
C~{Cuba Samaniego}, G~Giordano, J~Kim, F~Blanchini, and E~Franco.
\newblock Molecular titration promotes oscillations and bistability in minimal
  network models with monomeric regulators.
\newblock \emph{ACS Synth Biol}, 5:\penalty0 321--333, 2016.

\bibitem[{Cuba Samaniego} et~al.(2017){Cuba Samaniego}, Giordano, Blanchini,
  and Franco]{Samaniego2017b}
C~{Cuba Samaniego}, G~Giordano, F~Blanchini, and E~Franco.
\newblock Stability analysis of an artificial biomolecular oscillator with
  non-cooperative regulatory interactions.
\newblock \emph{J Biol Dynam}, 11:\penalty0 102--120, 2017.

\bibitem[Cuba~Samaniego et~al.(2020)Cuba~Samaniego, Giordano, and
  Franco]{CGF2020}
Christian Cuba~Samaniego, Giulia Giordano, and Elisa Franco.
\newblock Periodic switching in a recombinase-based molecular circuit.
\newblock \emph{IEEE Contr Sys Lett}, 4\penalty0 (1):\penalty0 241--246, 2020.

\bibitem[Jayanthi and {Del Vecchio}(2012)]{Jayanthi2012}
Shridhar Jayanthi and Domitilla {Del Vecchio}.
\newblock Tuning genetic clocks employing dna binding sites.
\newblock \emph{PLoS ONE}, 7, 2012.

\bibitem[Landau et~al.(2023)Landau, {Cuba Samaniego}, Giordano, and
  Franco]{Landau2023}
Judith Landau, Christian {Cuba Samaniego}, Giulia Giordano, and Elisa Franco.
\newblock Computational characterization of recombinase circuits for periodic
  behaviors.
\newblock \emph{iScience}, 26\penalty0 (1):\penalty0 105624, 2023.

\bibitem[Montagne et~al.(2011)Montagne, Plasson, Sakai, Fujii, and
  Rondelez]{Montagne2011}
Kevin Montagne, Raphael Plasson, Yasuyuki Sakai, Teruo Fujii, and Yannick
  Rondelez.
\newblock Programming an in vitro dna oscillator using a molecular networking
  strategy.
\newblock \emph{Mol Sys Biol}, 7, 2011.

\bibitem[Novák and Tyson(2008)]{novak2008}
Béla Novák and John~J. Tyson.
\newblock Design principles of biochemical oscillators.
\newblock \emph{Nature Rev Mol Cell Biol}, 9:\penalty0 981--991, 2008.

\bibitem[Purcell et~al.(2010)Purcell, Savery, Grierson, and
  Bernardo]{Purcell2010}
Oliver Purcell, Nigel~J. Savery, Claire~S. Grierson, and Mario~Di Bernardo.
\newblock A comparative analysis of synthetic genetic oscillators.
\newblock \emph{J Roy Soc Interface}, 7:\penalty0 1503--1524, 2010.

\bibitem[Tsai et~al.(2008)Tsai, Yoon, Ma, Pomerening, Tang, and
  Ferrell]{Tsai2008}
Tony Yu~Chen Tsai, Sup~Choi Yoon, Wenzhe Ma, Joseph~R. Pomerening, Chao Tang,
  and James~E. Ferrell.
\newblock Robust, tunable biological oscillations from interlinked positive and
  negative feedback loops.
\newblock \emph{Science}, 321:\penalty0 126--139, 2008.

\bibitem[Wang et~al.(2006)Wang, Chen, and Aihara]{Wang2006}
Ruiqi Wang, Luonan Chen, and Kazuyuki Aihara.
\newblock Construction of genetic oscillators with interlocked feedback
  networks.
\newblock \emph{J Theor Biol}, 242:\penalty0 454--463, 2006.

\bibitem[Wilhelm and Heinrich(1995)]{Wilhelm1995}
Thomas Wilhelm and Reinhart Heinrich.
\newblock Smallest chemical reaction system with hopf bifurcation.
\newblock \emph{J Math Chem}, 17:\penalty0 1--14, 1995.

\bibitem[Wilhelm(1996)]{Wilhelm1996}
T~Wilhelm.
\newblock Mathematical analysis of the smallest chemical reaction system with
  hopf bifurcation.
\newblock \emph{J Math Chem}, 19:\penalty0 111--130, 1996.

\bibitem[Gardner et~al.(2000)Gardner, Cantor, and Collins]{Gardner2000}
Timothy~S. Gardner, Charles~R. Cantor, and James~J. Collins.
\newblock Construction of a genetic toggle switch in escherichia coli.
\newblock \emph{Nature}, 403:\penalty0 339--342, 2000.

\bibitem[Sabouri-Ghomi et~al.(2008)Sabouri-Ghomi, Ciliberto, Kar, Novak, and
  Tyson]{sabouri2008}
Mohsen Sabouri-Ghomi, Andrea Ciliberto, Sandip Kar, Bela Novak, and John~J.
  Tyson.
\newblock Antagonism and bistability in protein interaction networks.
\newblock \emph{J Theor Biol}, 250:\penalty0 209--218, 2008.

\bibitem[Wilhelm(2009)]{Wilhelm2009}
Thomas Wilhelm.
\newblock The smallest chemical reaction system with bistability.
\newblock \emph{BMC Sys Biol}, 3:\penalty0 90, 2009.

\bibitem[Thomas(1981)]{Thomas1981}
Ren\{\'e\} Thomas.
\newblock On the relation between the logical structure of systems and their
  ability to generate multiple steady states or sustained oscillations.
\newblock \emph{Numerical Methods in the Study of Critical Phenomena:
  Proceedings of a Colloquium}, pages 180--193, 1981.

\bibitem[Domijan and Kirkilionis(2009)]{Domijan2009}
Mirela Domijan and Markus Kirkilionis.
\newblock Bistability and oscillations in chemical reaction networks.
\newblock \emph{J Math Biol}, 59:\penalty0 467--501, 2009.

\bibitem[Mincheva(2011)]{Mincheva2011}
Maya Mincheva.
\newblock Oscillations in biochemical reaction networks arising from pairs of
  subnetworks.
\newblock \emph{B Math Biol}, 73:\penalty0 2277--2304, 2011.

\bibitem[Richard and Comet(2011)]{Richard2011}
Adrien Richard and Jean~Paul Comet.
\newblock Stable periodicity and negative circuits in differential systems.
\newblock \emph{J Math Biol}, 63:\penalty0 593--600, 2011.

\bibitem[Angeli et~al.(2014)Angeli, Banaji, and Pantea]{Angeli2014}
David Angeli, Murad Banaji, and Casian Pantea.
\newblock Combinatorial approaches to hopf bifurcations in systems of
  interacting elements.
\newblock \emph{Commun Math Sci}, 12:\penalty0 1101--1133, 2014.

\bibitem[Bullo(2024)]{Bullo2024}
F.~Bullo.
\newblock \emph{Lectures on Network Systems}.
\newblock Kindle Direct Publishing, {1.7} edition, 2024.
\newblock ISBN 978-1986425643.

\bibitem[Blanchini et~al.(2017{\natexlab{b}})Blanchini, {Cuba Samaniego},
  Franco, and Giordano]{Blanchini2017e}
F~Blanchini, C~{Cuba Samaniego}, E~Franco, and G~Giordano.
\newblock Aggregates of monotonic step response systems: a structural
  classification.
\newblock \emph{IEEE T Contr Network Sys}, 2017{\natexlab{b}}.

\bibitem[Breda et~al.(2022)Breda, Frizzera, Giordano, Seffin, Zanni, Annoscia,
  and other]{breda2022deeper}
Dimitri Breda, Davide Frizzera, Giulia Giordano, Elisa Seffin, Virginia Zanni,
  Desiderato Annoscia, and other.
\newblock A deeper understanding of system interactions can explain
  contradictory field results on pesticide impact on honey bees.
\newblock \emph{Nat Commun}, 13\penalty0 (1):\penalty0 5720, 2022.

\bibitem[Smith(1988)]{Smith1988}
Hal~L. Smith.
\newblock Systems of ordinary differential equations which generate an order
  preserving flow. a survey of results.
\newblock \emph{SIAM Review}, 30\penalty0 (1):\penalty0 87--113, 1988.

\bibitem[Sontag(2007)]{Sontag2007}
Eduardo~D. Sontag.
\newblock Monotone and near-monotone biochemical networks.
\newblock \emph{Sys Synth Biol}, 1:\penalty0 59--87, 2007.

\bibitem[Angeli et~al.(2006)Angeli, Leenheer, and Sontag]{Angeli2006}
David Angeli, Patrick~De Leenheer, and Eduardo~D Sontag.
\newblock On the structural monotonicity of chemical reaction networks.
\newblock In \emph{Proc 45th IEEE CDC}, pages 7--12, 2006.

\bibitem[Angeli et~al.(2010)Angeli, de~Leenheer, and Sontag]{Angeli2010}
David Angeli, Patrick de~Leenheer, and Eduardo Sontag.
\newblock Graph-theoretic characterizations of monotonicity of chemical
  networks in reaction coordinates.
\newblock \emph{J Math Biol}, 61:\penalty0 581--616, 2010.

\bibitem[Angeli and Sontag(2003)]{Angeli2003b}
David Angeli and Eduardo~D. Sontag.
\newblock Monotone control systems.
\newblock \emph{IEEE T Automat Contr}, 48:\penalty0 1684--1698, 2003.

\bibitem[Enciso and Sontag(2005)]{Enciso2005b}
German Enciso and Eduardo~D. Sontag.
\newblock Monotone systems under positive feedback: Multistability and a
  reduction theorem.
\newblock \emph{Sys Contr Lett}, 54:\penalty0 159--168, 2005.

\bibitem[Angeli and Sontag(2004)]{Angeli2004}
David Angeli and Eduardo~D. Sontag.
\newblock Multi-stability in monotone input/output systems.
\newblock \emph{Sys Contr Lett}, 51:\penalty0 185--202, 2004.

\bibitem[Katz et~al.({\natexlab{b}})Katz, Giordano, and Margaliot]{Katz2024osc}
Rami Katz, Giulia Giordano, and Michael Margaliot.
\newblock Convergence to periodic orbits in 3-dimensional strongly
  2-cooperative~systems.
\newblock \emph{arXiv}, {\natexlab{b}}.

\bibitem[Goodwin(1965)]{Goodwin1965}
B.~C. Goodwin.
\newblock Oscillatory behavior in enzymatic control processes.
\newblock \emph{Adv Enzyme Reg}, 3:\penalty0 425--438, 1965.

\bibitem[Gonze and Ruoff(2021)]{gonze2021}
D.~Gonze and P.~Ruoff.
\newblock The {Goodwin} oscillator and its legacy.
\newblock \emph{Acta Biotheor.}, 6\penalty0 (4):\penalty0 857--874, 2021.

\bibitem[Griffith(1968)]{Griffith1968}
J.~S. Griffith.
\newblock Mathematics of cellular control processes. {I.} negative feedback to
  one gene.
\newblock \emph{J Theo Biol}, 20\penalty0 (2):\penalty0 202--208, 1968.

\bibitem[Sanchez(2009)]{sanchez2009}
L.~A. Sanchez.
\newblock Global asymptotic stability of the {Goodwin} system with repression.
\newblock \emph{Nonlinear analysis: real world applications}, 10\penalty0
  (4):\penalty0 2151--2156, 2009.

\bibitem[Tyson(1975)]{Tyson1975}
J.~J. Tyson.
\newblock On the existence of oscillatory solutions in negative feedback
  cellular control processes.
\newblock \emph{J. Math. Biology}, 1:\penalty0 311--315, 1975.

\bibitem[Field and Noyes(1974)]{FieldNoyes1974}
Richard~J. Field and Richard~M. Noyes.
\newblock Oscillations in chemical systems. {IV.} limit cycle behavior in a
  model of a real chemical reaction.
\newblock \emph{J Chem Phys}, 60\penalty0 (5):\penalty0 1877--1884, 1974.

\bibitem[Hastings and Murray(1975)]{hastings1975existence}
S.~P. Hastings and J.~D. Murray.
\newblock The existence of oscillatory solutions in the {F}ield--{N}oyes model
  for the {B}elousov--{Z}habotinskii reaction.
\newblock \emph{SIAM J. Applied Math.}, 28\penalty0 (3):\penalty0 678--688,
  1975.

\bibitem[Murray(1974)]{murray1974}
J.~D. Murray.
\newblock On a model for the temporal oscillations in the
  {Belousov‐Zhabotinsky} reaction.
\newblock \emph{J Chem Phys}, 61\penalty0 (9):\penalty0 3610--3613, 1974.

\bibitem[Simmons et~al.(2021)Simmons, Blyth, Blanchard, Clegg, Delmas, Garnier,
  Griffiths, Jacob, Pennekamp, Petchey, Poisot, Webb, and
  Beckerman]{Simmons2021}
Benno~I Simmons, Penelope S~A Blyth, Julia~L Blanchard, Tom Clegg, Eva Delmas,
  Aurélie Garnier, Christopher~A Griffiths, Ute Jacob, Frank Pennekamp, Owen~L
  Petchey, Timothée Poisot, Thomas~J Webb, and Andrew~P Beckerman.
\newblock Refocusing multiple stressor research around the targets and scales
  of ecological impacts.
\newblock \emph{Nature Eco Evol}, 5\penalty0 (11):\penalty0 1478--1489,
  November 2021.

\bibitem[Catford et~al.(2022)Catford, Wilson, Pyšek, Hulme, and
  Duncan]{Catford2022}
Jane~A. Catford, John~R.U. Wilson, Petr Pyšek, Philip~E. Hulme, and Richard~P.
  Duncan.
\newblock Addressing context dependence in ecology.
\newblock \emph{Trends Ecol Evol}, 37\penalty0 (2):\penalty0 158--170, 2022.

\bibitem[Blokhuis et~al.(2020)Blokhuis, Lacoste, and Nghe]{Blokhuis2020}
Alex Blokhuis, David Lacoste, and Philippe Nghe.
\newblock Universal motifs and the diversity of autocatalytic systems.
\newblock \emph{P Natl Acad Sci Usa}, 117:\penalty0 25230--25236, 2020.

\bibitem[Wong and Huck(2017)]{Wong2017}
Albert~S.Y. Wong and Wilhelm~T.S. Huck.
\newblock Grip on complexity in chemical reaction networks.
\newblock \emph{Beilstein J Organic Chem}, 13:\penalty0 1486--1497, 2017.

\bibitem[Khalil(2002)]{KhalilBook2002}
Hassan~K Khalil.
\newblock \emph{{Nonlinear Systems; 3rd ed.}}
\newblock Prentice-Hall, Upper Saddle River, NJ, 2002.

\bibitem[Chaves and Oyarzún(2019)]{Chaves2019}
Madalena Chaves and Diego~A. Oyarzún.
\newblock Dynamics of complex feedback architectures in metabolic pathways.
\newblock \emph{Automatica}, 99:\penalty0 323--332, 2019.

\bibitem[Franci and Sepulchre(2016)]{franci2016three}
Alessio Franci and Rodolphe Sepulchre.
\newblock A three-scale model of spatio-temporal bursting.
\newblock \emph{Siam J Appl Dyn Syst}, 15:\penalty0 2143--2175, 2016.

\bibitem[Meyer-Baese et~al.(2003)Meyer-Baese, Pilyugin, and Chen]{meyer2003}
A.~Meyer-Baese, S.~S. Pilyugin, and Y.~Chen.
\newblock Global exponential stability of competitive neural networks with
  different time scales.
\newblock \emph{IEEE T Neural Net}, 14:\penalty0 716--719, 2003.

\bibitem[DasGupta et~al.(2007)DasGupta, Enciso, Sontag, and
  Zhang]{DasGupta2007}
Bhaskar DasGupta, German~Andres Enciso, Eduardo Sontag, and Yi~Zhang.
\newblock Algorithmic and complexity results for decompositions of biological
  networks into monotone subsystems.
\newblock \emph{BioSystems}, 90:\penalty0 161--178, 2007.

\bibitem[Blanchini et~al.(2017{\natexlab{c}})Blanchini, {Cuba Samaniego},
  Franco, and Giordano]{Blanchini2018aggregates}
F.~Blanchini, C.~{Cuba Samaniego}, E.~Franco, and G.~Giordano.
\newblock Aggregates of positive impulse response systems: A decomposition
  approach for complex networks.
\newblock In \emph{2017 IEEE 56th Annual Conference on Decision and Control,
  CDC 2017}, 2017{\natexlab{c}}.

\bibitem[Barmish and Lagoa(1997)]{BarmishLagoa1997}
B.~R. Barmish and C.~M. Lagoa.
\newblock Uniform distribution: a rigorous justification for its use in
  robustness analysis.
\newblock \emph{Mathematics of Control, Signals and Systems}, 10\penalty0
  (3):\penalty0 203--222, 1997.

\bibitem[Chernoff(1952)]{Chernoff1952}
H.~Chernoff.
\newblock A measure of asymptotic efficiency for tests of a hypothesis based on
  the sum of observations.
\newblock \emph{Ann Math Stat}, 23:\penalty0 493--507, 1952.

\bibitem[Komurov et~al.(2010)Komurov, White, and Ram]{Komurov2010}
Kakajan Komurov, Michael~A. White, and Prahlad~T. Ram.
\newblock Use of data-biased random walks on graphs for the retrieval of
  context-specific networks from genomic data.
\newblock \emph{PLOS Comput Biol}, 6\penalty0 (8):\penalty0 1--10, 08 2010.

\bibitem[Zaki et~al.(2012)Zaki, Berengueres, and Efimov]{Zaki2012}
Nazar Zaki, Jose Berengueres, and Dmitry Efimov.
\newblock Detection of protein complexes using a protein ranking algorithm.
\newblock \emph{Proteins Struct Func Bioinf}, 80\penalty0 (10):\penalty0
  2459--2468, 2012.

\bibitem[Liebermeister and Klipp(2005)]{Liebermeister2005}
W.~Liebermeister and E.~Klipp.
\newblock Biochemical networks with uncertain parameters.
\newblock \emph{IEE Proc Sys Biol}, 152:\penalty0 97--107, 2005.

\bibitem[Weiße et~al.(2010)Weiße, Middleton, and Huisinga]{Weisse2010}
A.Y. Weiße, R.H. Middleton, and W.~Huisinga.
\newblock Quantifying uncertainty, variability and likelihood for ordinary
  differential equation models.
\newblock \emph{BMC Syst Biol}, 4\penalty0 (144), 2010.

\bibitem[Sutulovic et~al.(2025{\natexlab{b}})Sutulovic, Proverbio, Katz, and
  Giordano]{sutulovic2024efficient}
Uros Sutulovic, Daniele Proverbio, Rami Katz, and Giulia Giordano.
\newblock Efficient gpc-based quantification of probabilistic robustness for
  systems in neuroscience.
\newblock In \emph{Proc 2025 European Control Conference}, 2025{\natexlab{b}}.

\bibitem[Devia and Giordano(2022)]{DG2022}
Carlos~Andrés Devia and Giulia Giordano.
\newblock A framework to analyze opinion formation models.
\newblock \emph{Sci Rep}, 12:\penalty0 13441, 2022.

\bibitem[Azuma(2022)]{Azuma2022}
Shun-Ichi Azuma.
\newblock Structural equilibrium control of network systems.
\newblock \emph{IEEE T Automat Contr}, 67:\penalty0 3621--3626, 2022.

\bibitem[Zañudo et~al.(2017)Zañudo, Yang, Albert, and Levine]{zanudo2017}
Jorge Gomez~Tejeda Zañudo, Gang Yang, Réka Albert, and Herbert Levine.
\newblock Structure-based control of complex networks with nonlinear dynamics.
\newblock \emph{P Natl Acad Sci Usa}, 114:\penalty0 7234--7239, 2017.

\bibitem[Pepper et~al.(2019)Pepper, Montomoli, and
  Sharma]{pepper2019multiscale}
Nick Pepper, Francesco Montomoli, and Sanjiv Sharma.
\newblock Multiscale uncertainty quantification with arbitrary polynomial
  chaos.
\newblock \emph{Comput Method Appl Mech Eng}, 357:\penalty0 112571, 2019.

\bibitem[Xiu and Karniadakis(2002)]{xiu2002wiener}
Dongbin Xiu and George Karniadakis.
\newblock The {W}iener--{A}skey polynomial chaos for stochastic differential
  equations.
\newblock \emph{SIAM J Sci Comp}, 24\penalty0 (2):\penalty0 619--644, 2002.

\bibitem[Kim and Braatz(2013)]{kim2013generalised}
Kwang-Ki~K Kim and Richard~D Braatz.
\newblock Generalised polynomial chaos expansion approaches to approximate
  stochastic model predictive control.
\newblock \emph{Int J Control}, 86\penalty0 (8):\penalty0 1324--1337, 2013.

\bibitem[Fisher and Bhattacharya(2011)]{fisher2011optimal}
James Fisher and Raktim Bhattacharya.
\newblock Optimal trajectory generation with probabilistic system uncertainty
  using polynomial chaos.
\newblock \emph{J Dyn Syst-T Asme}, 133:\penalty0 014501, 2011.

\bibitem[Anderson et~al.(2015)Anderson, Craciun, Gopalkrishnan, and
  Wiuf]{Anderson2015}
David~F. Anderson, Gheorghe Craciun, Manoj Gopalkrishnan, and Carsten Wiuf.
\newblock Lyapunov functions, stationary distributions, and non-equilibrium
  potential for reaction networks.
\newblock \emph{B Math Biol}, 77:\penalty0 1744--1767, 2015.

\bibitem[Anderson et~al.(2017)Anderson, Cappelletti, and Kurtz]{Anderson2017}
David~F. Anderson, Daniele Cappelletti, and Thomas~G. Kurtz.
\newblock Finite time distributions of stochastically modeled chemical systems
  with absolute concentration robustness.
\newblock \emph{Siam J Appl Dyn Syst}, 16:\penalty0 1309--1339, 2017.

\bibitem[Anderson and Kurtz(2015)]{Anderson2015b}
David~F Anderson and Thomas~G Kurtz.
\newblock \emph{Stochastic Analysis of Biochemical Systems}.
\newblock Springer, 2015.

\bibitem[Britto~Bisso et~al.(2025)Britto~Bisso, Giordano, and
  Cuba~Samaniego]{BrittoBisso2025}
Frank Britto~Bisso, Giulia Giordano, and Christian Cuba~Samaniego.
\newblock Engineering cell fate with adaptive feedback control.
\newblock \emph{ACS Synth. Biol.}, 14\penalty0 (8):\penalty0 3163--3176, 08
  2025.

\bibitem[Chevalier et~al.(2019)Chevalier, Gómez-Schiavon, Ng, and
  {El-Samad}]{Chevalier2019}
Michael Chevalier, Mariana Gómez-Schiavon, Andrew~H. Ng, and Hana {El-Samad}.
\newblock Design and analysis of a proportional-integral-derivative controller
  with biological molecules.
\newblock \emph{Cell Systems}, 9:\penalty0 338--353.e10, 2019.

\bibitem[Lugagne et~al.(2017)Lugagne, Carrillo, Kirch, Köhler, Batt, and
  Hersen]{Lugagne2017}
Jean~Baptiste Lugagne, Sebastián~Sosa Carrillo, Melanie Kirch, Agnes Köhler,
  Gregory Batt, and Pascal Hersen.
\newblock Balancing a genetic toggle switch by real-time feedback control and
  periodic forcing.
\newblock \emph{Nat Commun}, 8, 2017.

\bibitem[Pandi et~al.(2019)Pandi, Koch, Voyvodic, Soudier, Bonnet, Kushwaha,
  and Faulon]{Pandi2019}
A.~Pandi, M.~Koch, P.~L. Voyvodic, P.~Soudier, J.~Bonnet, M.~Kushwaha, and
  J.-L. Faulon.
\newblock Metabolic perceptrons for neural computing in biological systems.
\newblock \emph{Nat Commun}, 10\penalty0 (1):\penalty0 3880, 2019.

\bibitem[Okumura et~al.(2022)Okumura, Gines, Lobato-Dauzier, Baccouche, Deteix,
  Fujii, Rondelez, and Genot]{Okumura2022}
S.~Okumura, G.~Gines, N.~Lobato-Dauzier, A.~Baccouche, R.~Deteix, T.~Fujii,
  Y.~Rondelez, and A.~J. Genot.
\newblock Nonlinear decision-making with enzymatic neural networks.
\newblock \emph{Nature}, 610\penalty0 (7932):\penalty0 496--501, 2022.

\bibitem[{Cuba Samaniego} et~al.(2024){Cuba Samaniego}, Wallace, Blanchini,
  Franco, and Giordano]{Cuba2024}
C.~{Cuba Samaniego}, E.~Wallace, F.~Blanchini, E.~Franco, and G.~Giordano.
\newblock Neural networks built from enzymatic reactions can operate as linear
  and nonlinear classifiers.
\newblock In \emph{Proc IEEE CDC}, 2024.

\bibitem[Jansen and Rit(1995)]{jansen1995electroencephalogram}
Ben~H Jansen and Vincent~G Rit.
\newblock Electroencephalogram and visual evoked potential generation in a
  mathematical model of coupled cortical columns.
\newblock \emph{Biological cybernetics}, 73\penalty0 (4):\penalty0 357--366,
  1995.

\bibitem[Hindmarsh and Rose(1984)]{hindmarsh1984model}
James~L Hindmarsh and RM~Rose.
\newblock A model of neuronal bursting using three coupled first order
  differential equations.
\newblock \emph{P Roy Soc B}, 221\penalty0 (1222):\penalty0 87--102, 1984.

\bibitem[Innocenti et~al.(2007)Innocenti, Morelli, Genesio, and
  Torcini]{innocenti2007dynamical}
Giacomo Innocenti, Alice Morelli, Roberto Genesio, and Alessandro Torcini.
\newblock Dynamical phases of the {H}indmarsh-{R}ose neuronal model: Studies of
  the transition from bursting to spiking chaos.
\newblock \emph{Chaos}, 17\penalty0 (4), 2007.

\bibitem[Montanari et~al.(2022)Montanari, Freitas, Proverbio, and
  Gonçalves]{Montanari2022}
Arthur~N. Montanari, Leandro Freitas, Daniele Proverbio, and Jorge Gonçalves.
\newblock Functional observability and subspace reconstruction in nonlinear
  systems.
\newblock \emph{Phys Rev Research}, 4:\penalty0 043195, 2022.

\bibitem[Jirsa et~al.(2014)Jirsa, Stacey, Quilichini, Ivanov, and
  Bernard]{jirsa2014nature}
Viktor~K Jirsa, William~C Stacey, Pascale~P Quilichini, Anton~I Ivanov, and
  Christophe Bernard.
\newblock On the nature of seizure dynamics.
\newblock \emph{Brain}, 137\penalty0 (8):\penalty0 2210--2230, 2014.

\bibitem[{de Jong} et~al.(2024){de Jong}, {McAllister}, and
  Giordano]{deJong2024}
M.~N. {de Jong}, R.~D. {McAllister}, and G.~Giordano.
\newblock Optimal periodic crop rotations under weed seed bank dynamics.
\newblock In \emph{Proc IEEE CDC}, 2024.

\bibitem[Nowzari et~al.(2014)Nowzari, Preciado, and Pappas]{Nowzari2014}
Cameron Nowzari, Victor~M Preciado, and George~J Pappas.
\newblock Stability analysis of generalized epidemic models over directed
  networks.
\newblock In \emph{Proc 53rd IEEE CDC}, pages 6197--6202. IEEE, 2014.

\bibitem[Hernandez-Vargas(2019)]{Hernandez-Vargas2019}
E.~A. Hernandez-Vargas.
\newblock \emph{Modeling and Control of Infectious Diseases in the Host}.
\newblock Academic Press, 2019.

\bibitem[Scheffer et~al.(2012)Scheffer, Carpenter, Lenton, Bascompte, Brock,
  Dakos, et~al.]{scheffer2012anticipating}
Marten Scheffer, Stephen~R Carpenter, Timothy~M Lenton, Jordi Bascompte,
  William Brock, Vasilis Dakos, et~al.
\newblock Anticipating critical transitions.
\newblock \emph{Science}, 338\penalty0 (6105):\penalty0 344--348, 2012.

\bibitem[Ashwin et~al.(2012)Ashwin, Wieczorek, Vitolo, and
  Cox]{ashwin2012tipping}
Peter Ashwin, Sebastian Wieczorek, Renato Vitolo, and Peter Cox.
\newblock Tipping points in open systems: bifurcation, noise-induced and
  rate-dependent examples in the climate system.
\newblock \emph{Philos T R Soc A}, 370\penalty0 (1962):\penalty0 1166--1184,
  2012.

\bibitem[Proverbio et~al.(2024{\natexlab{a}})Proverbio, Katz, and
  Giordano]{Proverbio2024}
Daniele Proverbio, Rami Katz, and Giulia Giordano.
\newblock Bridging robustness and resilience for dynamical systems in nature.
\newblock \emph{IFAC-PapersOnLine}, 58\penalty0 (17):\penalty0 43--48,
  2024{\natexlab{a}}.

\bibitem[Proverbio et~al.(2024{\natexlab{b}})Proverbio, Katz, and
  Giordano]{PROVERBIO2024445}
Daniele Proverbio, Rami Katz, and Giulia Giordano.
\newblock Corrigendum to: “bridging robustness and resilience for dynamical
  systems in nature” ifac-papersonline volume 58, issue 17, 2024, pages
  43-48.
\newblock \emph{IFAC-PapersOnLine}, 58\penalty0 (17):\penalty0 445,
  2024{\natexlab{b}}.

\bibitem[Kuehn(2011)]{kuehn2011mathematical}
Christian Kuehn.
\newblock {A mathematical framework for critical transitions: Bifurcations,
  fast--slow systems and stochastic dynamics}.
\newblock \emph{Physica D}, 240\penalty0 (12):\penalty0 1020--1035, 2011.

\bibitem[Proverbio et~al.(2023)Proverbio, Skupin, and
  Gon{\c{c}}alves]{proverbio2023systematic}
Daniele Proverbio, Alexander Skupin, and Jorge Gon{\c{c}}alves.
\newblock Systematic analysis and optimization of early warning signals for
  critical transitions using distribution data.
\newblock \emph{iScience}, 26:\penalty0 107156, 2023.

\bibitem[Berglund and Gentz(2006)]{Berglund2006}
Nils Berglund and Barbara Gentz.
\newblock \emph{Noise-Induced Phenomena in Slow-Fast Dynamical Systems: a
  Sample-Paths Approach}.
\newblock Springer Sci \& Bus Media, 2006.

\bibitem[Freidlin and Wentzell(1998)]{freidlin1998random}
Mark~Iosifovich Freidlin and Alexander~D Wentzell.
\newblock \emph{Random Perturbations of Dynamical Systems}.
\newblock Springer, 1998.

\bibitem[Wang et~al.(2010)Wang, Li, and Wang]{Wang2010}
Jin Wang, Chunhe Li, and Erkang Wang.
\newblock Potential and flux landscapes quantify the stability and robustness
  of budding yeast cell cycle network.
\newblock \emph{P Natl Acad Sci Usa}, 107:\penalty0 8195--8200, 2010.

\bibitem[Nolting and Abbott(2016)]{Nolting2016}
B.C. Nolting and K.C. Abbott.
\newblock Balls, cups, and quasi-potentials: quantifying stability in
  stochastic systems.
\newblock \emph{Ecology}, 97:\penalty0 850--864, 2016.

\bibitem[Zhou et~al.(2012)Zhou, Aliyu, Aurell, and Huang]{Zhou2012}
Joseph~Xu Zhou, D.~S.M. Aliyu, Erik Aurell, and Sui Huang.
\newblock {Quasi-potential landscape in complex multi-stable systems}.
\newblock \emph{J Roy Soc Interface}, 9\penalty0 (77):\penalty0 3539--3553,
  2012.

\bibitem[Eugenio et~al.(2014)Eugenio, Karp, Guo, Robson, Hart, Trippa, and
  Yuan]{Eugenio2014}
Marco Eugenio, Robert~L. Karp, Guoji Guo, Paul Robson, Adam~H. Hart, Lorenzo
  Trippa, and Guo~Cheng Yuan.
\newblock {Bifurcation analysis of single-cell gene expression data reveals
  epigenetic landscape}.
\newblock \emph{P Natl Acad Sci Usa}, 111\penalty0 (52):\penalty0 E5643--E5650,
  2014.

\bibitem[Moris et~al.(2016)Moris, Pina, and Arias]{Moris2016}
Naomi Moris, Cristina Pina, and Alfonso~Martinez Arias.
\newblock Transition states and cell fate decisions in epigenetic landscapes.
\newblock \emph{Nat Rev Genet}, 17:\penalty0 693--703, 2016.

\bibitem[Proverbio et~al.(2022{\natexlab{b}})Proverbio, Montanari, Skupin, and
  Gon{\c{c}}alves]{proverbio2022buffering}
Daniele Proverbio, Arthur~N Montanari, Alexander Skupin, and Jorge
  Gon{\c{c}}alves.
\newblock Buffering variability in cell regulation motifs close to criticality.
\newblock \emph{Phys Rev E}, 106\penalty0 (3):\penalty0 L032402,
  2022{\natexlab{b}}.

\bibitem[Su et~al.(2019{\natexlab{a}})Su, Bintz, Yang, Robert, Ng,
  et~al.]{Su2019}
Yapeng Su, Marcus Bintz, Yezi Yang, Lidia Robert, Alphonsus~H.C. Ng, et~al.
\newblock {Phenotypic heterogeneity and evolution of melanoma cells associated
  with targeted therapy resistance}.
\newblock \emph{PLoS Comput Biol}, 15\penalty0 (6):\penalty0 1--22,
  2019{\natexlab{a}}.

\bibitem[Khona and Fiete(2022)]{khona2022attractor}
Mikail Khona and Ila~R Fiete.
\newblock Attractor and integrator networks in the brain.
\newblock \emph{Nat Rev Neurosci}, 23\penalty0 (12):\penalty0 744--766, 2022.

\bibitem[Grassly and Fraser(2006)]{Grassly2006}
Nicholas~C. Grassly and Christophe Fraser.
\newblock Seasonal infectious disease epidemiology.
\newblock \emph{Proc Biol Sci}, 273\penalty0 (1600):\penalty0 2541--2550, 2006.

\bibitem[Heino et~al.(2022)Heino, Proverbio, Marchand, Resnicow, and
  Hankonen]{heino2022attractor}
Matti~TJ Heino, Daniele Proverbio, Gwen Marchand, Kenneth Resnicow, and Nelli
  Hankonen.
\newblock Attractor landscapes: A unifying conceptual model for understanding
  behaviour change across scales of observation.
\newblock \emph{Health Psychol Rev}, pages 1--18, 2022.

\bibitem[Schultz et~al.(2017)Schultz, Menck, Heitzig, and
  Kurths]{schultz2017potentials}
Paul Schultz, Peter~J Menck, Jobst Heitzig, and J{\"u}rgen Kurths.
\newblock Potentials and limits to basin stability estimation.
\newblock \emph{New J Phys}, 19\penalty0 (2):\penalty0 023005, 2017.

\bibitem[Proverbio and Giordano(2025{\natexlab{b}})]{Proverbio2025cdc}
Daniele Proverbio and Giulia Giordano.
\newblock Resilience of the positive gene autoregulation loop.
\newblock In \emph{Proc IEEE CDC}, 2025{\natexlab{b}}.

\bibitem[Weber and Buceta(2013)]{weber2013stochastic}
Marc Weber and Javier Buceta.
\newblock {Stochastic stabilization of phenotypic states: the genetic bistable
  switch as a case study}.
\newblock \emph{PLoS ONE}, 8\penalty0 (9):\penalty0 e73487, 2013.

\bibitem[Gardiner(1985)]{Gardiner1985}
Crispin~W Gardiner.
\newblock \emph{{Handbook of Stochastic Methods}}.
\newblock Springer, 1985.

\bibitem[Lundstr{\"o}m and Aidanp{\"a}{\"a}(2007)]{lundstrom2007dynamic}
Niklas~LP Lundstr{\"o}m and Jan-Olov Aidanp{\"a}{\"a}.
\newblock Dynamic consequences of electromagnetic pull due to deviations in
  generator shape.
\newblock \emph{J Sound Vibration}, 301\penalty0 (1-2):\penalty0 207--225,
  2007.

\bibitem[Ullah and Wolkenhauer(2011)]{Ullah2011}
Mukhtar Ullah and Olaf Wolkenhauer.
\newblock \emph{Stochastic Approaches for Systems Biology}.
\newblock Springer Sci \& Bus Media, 2011.

\bibitem[Su et~al.(2019{\natexlab{b}})Su, Bintz, Yang, Robert, Ng, Liu,
  et~al.]{su2019phenotypic}
Yapeng Su, Marcus Bintz, Yezi Yang, Lidia Robert, Alphonsus H~C Ng, Victoria
  Liu, et~al.
\newblock {Phenotypic heterogeneity and evolution of melanoma cells associated
  with targeted therapy resistance}.
\newblock \emph{PLoS Comput Biol}, 15\penalty0 (6):\penalty0 e1007034,
  2019{\natexlab{b}}.

\bibitem[Tsimring(2014)]{Tsimring2015a}
Lev~S Tsimring.
\newblock Noise in biology.
\newblock \emph{Rep Progs Phys}, 77.2:\penalty0 026601, 2014.

\bibitem[Hasty et~al.(2000)Hasty, Pradines, Dolnik, and Collins]{Hasty2000}
Jeff Hasty, Joel Pradines, Milos Dolnik, and J.~J. Collins.
\newblock {Noise-based switches and amplifiers forgene expression}.
\newblock \emph{P Natl Acad Sci Usa}, 97\penalty0 (5):\penalty0 2075--2080,
  2000.

\bibitem[Norman et~al.(2015)Norman, Lord, Paulsson, and
  Losick]{norman2015stochastic}
Thomas~M Norman, Nathan~D Lord, Johan Paulsson, and Richard Losick.
\newblock {Stochastic switching of cell fate in microbes}.
\newblock \emph{Annu Rev Microbiol}, 69:\penalty0 381--403, 2015.

\bibitem[Zhang et~al.(2012)Zhang, Chen, and Chen]{Zhang2012}
Hui Zhang, Yueling Chen, and Yong Chen.
\newblock {Noise Propagation in Gene Regulation Networks Involving Interlinked
  Positive and Negative Feedback Loops}.
\newblock \emph{PLoS ONE}, 7\penalty0 (12):\penalty0 1--8, 2012.

\bibitem[Gillespie(2000)]{Gillespie2000}
Daniel~T. Gillespie.
\newblock {Chemical Langevin equation}.
\newblock \emph{J Chem Phys}, 113\penalty0 (1):\penalty0 297--306, 2000.

\bibitem[O'Regan and Burton(2018)]{o2018stochasticity}
Suzanne~M. O'Regan and Danielle~L. Burton.
\newblock {How stochasticity influences leading indicators of critical
  transitions}.
\newblock \emph{B Math Biol}, 80\penalty0 (6):\penalty0 1630--1654, 2018.
\newblock ISSN 15229602.

\bibitem[Tian(2004)]{Tian2004}
Tianhai Tian.
\newblock Robustness of mathematical models for biological systems.
\newblock \emph{ANZIAM Journal}, 45:\penalty0 565--577, 2004.

\bibitem[Aleta and Moreno(2020)]{Aleta2020}
Alberto Aleta and Yamir Moreno.
\newblock Evaluation of the potential incidence of covid-19 and effectiveness
  of containment measures in spain: a data-driven approach.
\newblock \emph{BMC Med}, 18:\penalty0 157, 2020.

\bibitem[Becker(1977)]{Becker1977}
Niels~G. Becker.
\newblock On a general stochastic epidemic model.
\newblock \emph{Theor Popul Biol}, 11\penalty0 (1):\penalty0 23--36, 1977.

\bibitem[Brett et~al.(2017)Brett, Drake, and Rohani]{brett2017anticipating}
Tobias~S Brett, John~M Drake, and Pejman Rohani.
\newblock Anticipating the emergence of infectious diseases.
\newblock \emph{J Roy Soc Interface}, 14:\penalty0 20170115, 2017.

\bibitem[Diekmann and Heesterbeek(2000)]{Diekmann2000}
Odo Diekmann and Johan Andre~Peter Heesterbeek.
\newblock \emph{Mathematical Epidemiology of Infectious Diseases: Model
  Building, Analysis and Interpretation}, volume~5.
\newblock John Wiley \& Sons, 2000.

\bibitem[Hellewell et~al.(2020)Hellewell, Abbott, Gimma, Bosse, Jarvis,
  Russell, et~al.]{Hellewell2020}
Joel Hellewell, Sam Abbott, Amy Gimma, Nikos~I Bosse, Christopher~I Jarvis,
  Timothy~W Russell, et~al.
\newblock Feasibility of controlling covid-19 outbreaks by isolation of cases
  and contacts.
\newblock \emph{Lancet Glob Health}, 8\penalty0 (4):\penalty0 e488--e496, 2020.

\bibitem[Buonomo et~al.(2018)Buonomo, Manfredi, and d’Onofrio
  Alberto]{Buonomo2018}
Bruno Buonomo, Piero Manfredi, and d’Onofrio Alberto.
\newblock Optimal time-profiles of public health intervention to shape
  voluntary vaccination for childhood diseases.
\newblock \emph{J Math Biol}, 78:\penalty0 1089--1113, 2018.

\bibitem[Oksendal(2013)]{oksendal2013stochastic}
Bernt Oksendal.
\newblock \emph{Stochastic differential equations: an introduction with
  applications}.
\newblock Springer Science \& Business Media, 2013.

\bibitem[Khasminskii(2012)]{khasminskii2012stochastic}
Rafail Khasminskii.
\newblock \emph{Stochastic Stability of Differential Equations -- With
  contributions by G.N. Milstein and M.B. Nevelson -- Completely Revised and
  Enlarged 2nd Edition}.
\newblock Stochastic Modelling and Applied Probability. Springer-Verlag Berlin
  Heidelberg, 2012.

\bibitem[Bashkirtseva et~al.(2021)Bashkirtseva, Kolinichenko, and
  Ryashko]{bashkirtseva2021stochastic}
Irina Bashkirtseva, Alexander Kolinichenko, and Lev Ryashko.
\newblock Stochastic sensitivity of turing patterns: methods and applications
  to the analysis of noise-induced transitions.
\newblock \emph{Chaos, Solitons \& Fractals}, 153:\penalty0 111491, 2021.

\bibitem[Levin and Segel(1976)]{levin1976hypothesis}
Simon~A Levin and Lee~A Segel.
\newblock Hypothesis for origin of planktonic patchiness.
\newblock \emph{Nature}, 259\penalty0 (5545):\penalty0 659--659, 1976.

\bibitem[Wieczorek et~al.(2011)Wieczorek, Ashwin, Luke, and
  Cox]{wieczorek2011excitability}
Sebastian Wieczorek, Peter Ashwin, Catherine~M Luke, and Peter~M Cox.
\newblock Excitability in ramped systems: the compost-bomb instability.
\newblock \emph{Proc Roy Soc A}, 467\penalty0 (2129):\penalty0 1243--1269,
  2011.

\bibitem[Bayram et~al.(2018)Bayram, Partal, and
  Orucova~Buyukoz]{bayram2018numerical}
Mustafa Bayram, Tugcem Partal, and Gulsen Orucova~Buyukoz.
\newblock Numerical methods for simulation of stochastic differential
  equations.
\newblock \emph{Advances in Difference Equations}, 2018\penalty0 (1):\penalty0
  17, 2018.

\bibitem[Tamberg et~al.(2022)Tamberg, Heitzig, and Donges]{tamberg2022modeler}
Lea~A Tamberg, Jobst Heitzig, and Jonathan~F Donges.
\newblock A modeler’s guide to studying the resilience of
  social-technical-environmental systems.
\newblock \emph{Environ Res Lett}, 17\penalty0 (5):\penalty0 055005, 2022.

\bibitem[Hinrichsen and Pritchard(2005)]{HinrichsenPritchard2005}
Diederich Hinrichsen and Anthony~J. Pritchard.
\newblock \emph{Mathematical Systems Theory I}.
\newblock Springer Berlin, Heidelberg, 2005.

\bibitem[Trefois et~al.(2015)Trefois, Antony, Goncalves, Skupin, and
  Balling]{Trefois2015a}
Christophe Trefois, Paul M~A Antony, Jorge Goncalves, Alexander Skupin, and
  Rudi Balling.
\newblock {Critical transitions in chronic disease: Transferring concepts from
  ecology to systems medicine}.
\newblock \emph{Curr Opin Biotech}, 34:\penalty0 48--55, 2015.

\bibitem[Chen et~al.(2012)Chen, Liu, Liu, Li, and Aihara]{Chen2012c}
Luonan Chen, Rui Liu, Zhi~Ping Liu, Meiyi Li, and Kazuyuki Aihara.
\newblock {Detecting early-warning signals for sudden deterioration of complex
  diseases by dynamical network biomarkers}.
\newblock \emph{Sci Rep}, 2:\penalty0 18--20, 2012.

\bibitem[Aguirre and Manrubia(2015)]{Aguirre2015c}
Jacobo Aguirre and Susanna Manrubia.
\newblock Tipping points and early warning signals in the genomic composition
  of populations induced by environmental changes.
\newblock \emph{Sci Rep}, 5:\penalty0 1--7, 2015.

\bibitem[Maturana et~al.(2020)Maturana, Meisel, Dell, Karoly, D'Souza, Grayden,
  et~al.]{Maturana2020}
Matias~I. Maturana, Christian Meisel, Katrina Dell, Philippa~J. Karoly, Wendyl
  D'Souza, David~B. Grayden, et~al.
\newblock {Critical slowing down as a biomarker for seizure susceptibility}.
\newblock \emph{Nat Commun}, 11\penalty0 (1):\penalty0 1--12, 2020.

\bibitem[O'Regan et~al.(2013)O'Regan, Drake, O'Regan, and Drake]{ORegan2013}
Suzanne~M. O'Regan, John~M. Drake, Suzanne~M O'Regan, and John~M. Drake.
\newblock {Theory of early warning signals of disease emergenceand leading
  indicators of elimination}.
\newblock \emph{Theo Ecol}, 6\penalty0 (3):\penalty0 333--357, 2013.

\bibitem[Boettiger et~al.(2013)Boettiger, Ross, and Hastings]{Boettiger2013}
Carl Boettiger, Noam Ross, and Alan Hastings.
\newblock Early warning signals: The charted and uncharted territories.
\newblock \emph{Theor Ecol}, 6:\penalty0 255--264, 2013.

\bibitem[Brock et~al.(2008)Brock, Carpenter, and Scheffer]{brock2008regime}
William~A Brock, Stephen~R Carpenter, and Marten Scheffer.
\newblock Regime shifts, environmental signals, uncertainty, and policy choice.
\newblock \emph{Complex Theo Sustain Future}, pages 180--206, 2008.

\bibitem[Brett et~al.(2020)Brett, Ajelli, Liu, Krauland, Grefenstette, van
  Panhuis, et~al.]{brett2020detecting}
Tobias Brett, Marco Ajelli, Quan-Hui Liu, Mary~G Krauland, John~J Grefenstette,
  Willem~G van Panhuis, et~al.
\newblock Detecting critical slowing down in high-dimensional epidemiological
  systems.
\newblock \emph{PLoS Comput Biol}, 16:\penalty0 e1007679, 2020.

\bibitem[Bury et~al.(2021)Bury, Sujith, Pavithran, Scheffer, Lenton, Anand, and
  Bauch]{Bury2021}
Thomas~M. Bury, R.~I. Sujith, Induja Pavithran, Marten Scheffer, Timothy~M.
  Lenton, Madhur Anand, and Chris~T. Bauch.
\newblock {Deep learning for early warning signals of tipping points}.
\newblock \emph{P Natl Acad Sci Usa}, 118\penalty0 (39):\penalty0 e2106140118,
  2021.

\bibitem[Kostoulas et~al.(2021)Kostoulas, Meletis, Pateras, Eusebi, Kostoulas,
  Furuya-Kanamori, Speybroeck, Denwood, Doi, Althaus,
  et~al.]{kostoulas2021epidemic}
Polychronis Kostoulas, Eletherios Meletis, Konstantinos Pateras, Paolo Eusebi,
  Theodoros Kostoulas, Luis Furuya-Kanamori, Niko Speybroeck, Matthew Denwood,
  Suhail~AR Doi, Christian~L Althaus, et~al.
\newblock The epidemic volatility index, a novel early warning tool for
  identifying new waves in an epidemic.
\newblock \emph{Scientific reports}, 11\penalty0 (1):\penalty0 23775, 2021.

\bibitem[MacIntyre et~al.(2023)MacIntyre, Chen, Kunasekaran, Quigley, Lim,
  Stone, Paik, Yao, Heslop, Wei, et~al.]{macintyre2023artificial}
Chandini~Raina MacIntyre, Xin Chen, Mohana Kunasekaran, Ashley Quigley, Samsung
  Lim, Haley Stone, Hye-young Paik, Lina Yao, David Heslop, Wenzhao Wei, et~al.
\newblock Artificial intelligence in public health: the potential of epidemic
  early warning systems.
\newblock \emph{Journal of International Medical Research}, 51\penalty0
  (3):\penalty0 03000605231159335, 2023.

\bibitem[Scheffer et~al.(2009{\natexlab{a}})Scheffer, Bascompte, Brock,
  Brovkin, Carpenter, Dakos, Held, Van~Nes, Rietkerk, and
  Sugihara]{scheffer2009early}
Marten Scheffer, Jordi Bascompte, William~A Brock, Victor Brovkin, Stephen~R
  Carpenter, Vasilis Dakos, Hermann Held, Egbert~H Van~Nes, Max Rietkerk, and
  George Sugihara.
\newblock Early-warning signals for critical transitions.
\newblock \emph{Nature}, 461\penalty0 (7260):\penalty0 53--59,
  2009{\natexlab{a}}.

\bibitem[Thompson and Sieber(2011)]{thompson2011predicting}
J.~Michael~T. Thompson and Jan Sieber.
\newblock {Predicting climate tipping as a noisy bifurcation: a review}.
\newblock \emph{Int J Bifurcat Chaos}, 21\penalty0 (02):\penalty0 399--423,
  2011.

\bibitem[Guckenheimer et~al.(2013)Guckenheimer, Holmes, and
  Slemrod]{Guckenheimer2009a}
John Guckenheimer, Philip Holmes, and M.~Slemrod.
\newblock \emph{{Nonlinear Oscillations Dynamical Systems, and Bifurcations of
  Vector Fields}}.
\newblock Springer Sci Bus Media, 2013.

\bibitem[Ritchie and Sieber(2017)]{ritchie2017probability}
Paul Ritchie and Jan Sieber.
\newblock Probability of noise-and rate-induced tipping.
\newblock \emph{Phys Rev E}, 95\penalty0 (5):\penalty0 052209, 2017.

\bibitem[Alkhayuon and Ashwin(2018)]{Alkhayuon2018}
Hassan~M. Alkhayuon and Peter Ashwin.
\newblock Rate-induced tipping from periodic attractors: Partial tipping and
  connecting orbits.
\newblock \emph{Chaos}, 28:\penalty0 033608, 2018.

\bibitem[Slyman and Jones(2023)]{slyman2023rate}
Katherine Slyman and Christopher~K Jones.
\newblock Rate and noise-induced tipping working in concert.
\newblock \emph{Chaos}, 33\penalty0 (1), 2023.

\bibitem[Kuznetsov(2013)]{kuznetsov2013elements}
Yuri~A. Kuznetsov.
\newblock \emph{{Elements of Applied Bifurcation Theory}}, volume 112.
\newblock Springer Sci Bus Media, 2013.

\bibitem[Min et~al.(2005)Min, English, Luo, Cherayil, Kou, and
  Xie]{min2005fluctuating}
Wei Min, Brian~P English, Guobin Luo, Binny~J Cherayil, SC~Kou, and X~Sunney
  Xie.
\newblock Fluctuating enzymes: lessons from single-molecule studies.
\newblock \emph{Acc Chem Res}, 38\penalty0 (12):\penalty0 923--931, 2005.

\bibitem[Lloyd-Smith et~al.(2005)Lloyd-Smith, Schreiber, Kopp, and
  Getz]{Lloyd-Smith2005a}
J.~O. Lloyd-Smith, S.~J. Schreiber, P.~E. Kopp, and W.~M. Getz.
\newblock Superspreading and the effect of individual variation on disease
  emergence.
\newblock \emph{Nature}, 438:\penalty0 355--359, 2005.

\bibitem[Small and Tse(2005)]{Small2005}
Michael Small and Chi~K Tse.
\newblock Clustering model for transmission of the sars virus: application to
  epidemic control and risk assessment.
\newblock \emph{Physica A}, 351:\penalty0 499--511, 2005.

\bibitem[Ashwin et~al.(2017)Ashwin, Perryman, and Wieczorek]{Ashwin2017a}
Peter Ashwin, Clare Perryman, and Sebastian Wieczorek.
\newblock Parameter shifts for nonautonomous systems in low dimension:
  Bifurcation- and rate-induced tipping.
\newblock \emph{Nonlinearity}, 30:\penalty0 2185--2210, 2017.

\bibitem[Bonciolini et~al.(2018)Bonciolini, Ebi, Boujo, and
  Noiray]{Bonciolini2018}
Giacomo Bonciolini, Dominik Ebi, Edouard Boujo, and Nicolas Noiray.
\newblock {Experiments and modelling of rate-dependent transition delay in a
  stochastic subcritical bifurcation}.
\newblock \emph{Roy Soc Open Sci}, 5\penalty0 (3):\penalty0 172078, 2018.

\bibitem[Nen{\'{e}} et~al.(2012)Nen{\'{e}}, Garca-Ojalvo, and Zaikin]{Nene2012}
Nuno~R. Nen{\'{e}}, Jordi Garca-Ojalvo, and Alexey Zaikin.
\newblock {Speed-dependent cellular decision making in nonequilibrium genetic
  circuits}.
\newblock \emph{PLoS ONE}, 7\penalty0 (3):\penalty0 1--7, 2012.

\bibitem[Zou et~al.(2023)Zou, Fragonara, Qiu, and Guo]{zou2023uncertainty}
Mengbang Zou, Luca~Zanotti Fragonara, Song Qiu, and Weisi Guo.
\newblock Uncertainty quantification of multi-scale resilience in networked
  systems with nonlinear dynamics using arbitrary polynomial chaos.
\newblock \emph{Sci Rep}, 13\penalty0 (1):\penalty0 488, 2023.

\bibitem[Dai et~al.(2012)Dai, Vorselen, Korolev, and Gore]{dai2012generic}
Lei Dai, Daan Vorselen, Kirill~S Korolev, and Jeff Gore.
\newblock Generic indicators for loss of resilience before a tipping point
  leading to population collapse.
\newblock \emph{Science}, 336:\penalty0 1175--1177, 2012.

\bibitem[Golubitsky et~al.(2012)Golubitsky, Stewart, and
  Schaeffer]{Golubitsky2003a}
Martin Golubitsky, Ian Stewart, and David~G Schaeffer.
\newblock \emph{Singularities and Groups in Bifurcation Theory-II}.
\newblock Springer Sci \& Bus Media, 2012.

\bibitem[Kuznetsov and Piccardi(1994)]{kuznetsov1994bifurcation}
Yu~A Kuznetsov and Carlo Piccardi.
\newblock Bifurcation analysis of periodic seir and sir epidemic models.
\newblock \emph{J Math Biol}, 32:\penalty0 109--121, 1994.

\bibitem[Hui et~al.(2011)Hui, Haddad, and Bailey]{Hui2011}
Qing Hui, Wassim~M. Haddad, and James~M. Bailey.
\newblock Multistability, bifurcations, and biological neural networks: A
  synaptic drive firing model for cerebral cortex transition in the induction
  of general anesthesia.
\newblock \emph{Nonlin An Hybrid Sys}, 5:\penalty0 554--572, 2011.

\bibitem[Mojtahedi et~al.(2016)Mojtahedi, Skupin, Zhou, Castaño, Leong-Quong,
  Chang, et~al.]{Mojtahedi2016b}
Mitra Mojtahedi, Alexander Skupin, Joseph Zhou, Ivan~G. Castaño, Rebecca~Y.Y.
  Leong-Quong, Hannah Chang, et~al.
\newblock Cell fate decision as high-dimensional critical state transition.
\newblock \emph{PLoS Biol}, 14:\penalty0 1--28, 2016.

\bibitem[Crawford(1991)]{Crawford1991a}
John~David Crawford.
\newblock {Introduction to bifurcation theory}.
\newblock \emph{Rev Mod Phys}, 63\penalty0 (4):\penalty0 991--1037, 1991.

\bibitem[Haragus and Iooss(2010)]{Haragus2010}
M.~Haragus and G.~Iooss.
\newblock \emph{Local Bifurcation, Center Manifolds and Normal Forms in
  Infinte-Dimensional Dynamical Systems}.
\newblock Springer Sci \& Bus Media, 2010.

\bibitem[Kuehn and Bick(2021)]{kuehn2021universal}
Christian Kuehn and Christian Bick.
\newblock A universal route to explosive phenomena.
\newblock \emph{Science Advances}, 7:\penalty0 eabe3824, 2021.

\bibitem[Thom(1972)]{Thom2554}
René Thom.
\newblock \emph{Structural Stability and Morphogenesis}.
\newblock Taylor \& Francis, 1972.

\bibitem[Balamuralitharan and Radha(2018)]{balamuralitharan2018bifurcation}
S~Balamuralitharan and M~Radha.
\newblock Bifurcation analysis in sir epidemic model with treatment.
\newblock 1000\penalty0 (1):\penalty0 012169, 2018.

\bibitem[Dakos et~al.(2012)Dakos, Carpenter, Brock, Ellison, Guttal, Ives,
  Kefi, Livina, Seekell, van Nes, others, Kéfi, Livina, Seekell, van Nes, and
  Scheffer]{dakos2012methods}
Vasilis Dakos, Stephen~R. Carpenter, William~A. Brock, Aaron~M. Ellison,
  Vishwesha Guttal, Anthony~R. Ives, Sonia Kefi, Valerie Livina, David~A.
  Seekell, Egbert~H. van Nes, others, Sonia Kéfi, Valerie Livina, David~A.
  Seekell, Egbert~H. van Nes, and Marten Scheffer.
\newblock Methods for detecting early warnings of critical transitions in time
  series illustrated using simulated ecological data.
\newblock \emph{PLoS ONE}, 7:\penalty0 e41010, 2012.

\bibitem[Van~Nes and Scheffer(2007)]{van2007slow}
Egbert~H Van~Nes and Marten Scheffer.
\newblock Slow recovery from perturbations as a generic indicator of a nearby
  catastrophic shift.
\newblock \emph{The American Naturalist}, 169\penalty0 (6):\penalty0 738--747,
  2007.

\bibitem[Papoulis and Pillai(2002)]{papoulis2002probability}
Athanasios Papoulis and S~Unnikrishna Pillai.
\newblock \emph{{Probability, Random Variables, and Stochastic Processes}}.
\newblock Tata McGraw-Hill Education, 2002.

\bibitem[{Van Kampen}(1992)]{van1992stochastic}
Nicolaas~Godfried {Van Kampen}.
\newblock \emph{{Stochastic Processes in Physics and Chemistry}}, volume~1.
\newblock Elsevier, 1992.

\bibitem[Conley(1988)]{conley1988}
Charles Conley.
\newblock The gradient structure of a flow: I.
\newblock \emph{Ergodic Theory Dynam. Systems}, 8\penalty0 (Charles Conley
  Memorial Issue), 1988.

\bibitem[Plischke and Bergersen(1994)]{plischke1994equilibrium}
Michael Plischke and Birger Bergersen.
\newblock \emph{Equilibrium Statistical Physics}.
\newblock World Scientific Publishing Company, 1994.

\bibitem[H{\"a}nggi et~al.(1990)H{\"a}nggi, Talkner, and
  Borkovec]{hanggi1990reaction}
Peter H{\"a}nggi, Peter Talkner, and Michal Borkovec.
\newblock Reaction-rate theory: fifty years after kramers.
\newblock \emph{Reviews of modern physics}, 62\penalty0 (2):\penalty0 251,
  1990.

\bibitem[Dennis et~al.(2016)Dennis, Assas, Elaydi, Kwessi, and
  Livadiotis]{dennis2016allee}
Brian Dennis, Laila Assas, Saber Elaydi, Eddy Kwessi, and George Livadiotis.
\newblock Allee effects and resilience in stochastic populations.
\newblock \emph{Theo Ecol}, 9:\penalty0 323--335, 2016.

\bibitem[Bury et~al.(2020)Bury, Bauch, and Anand]{Bury2020}
T.~M. Bury, C.~T. Bauch, and M.~Anand.
\newblock Detecting and distinguishing tipping points using spectral early
  warning signals.
\newblock \emph{J Roy Soc Interface}, 17:\penalty0 20200482, 2020.

\bibitem[Carpenter et~al.(2011)Carpenter, Cole, Pace, Batt, Brock, Cline,
  et~al.]{Carpenter2011}
S.~R. Carpenter, J.~J. Cole, M.~L. Pace, R.~Batt, W.~A. Brock, T.~Cline, et~al.
\newblock {Early warnings of regime shifts: A whole-ecosystem experiment}.
\newblock \emph{Science}, 332\penalty0 (6033):\penalty0 1079--1082, 2011.

\bibitem[Kuehn(2013)]{Kuehn2013a}
Christian Kuehn.
\newblock A mathematical framework for critical transitions: Normal forms,
  variance and applications.
\newblock \emph{J Nonlin Sci}, 23:\penalty0 457--510, 2013.

\bibitem[Proverbio et~al.(2022{\natexlab{c}})Proverbio, Kemp, Magni, and
  Gonçalves]{proverbio2022performance}
Daniele Proverbio, Françoise Kemp, Stefano Magni, and Jorge Gonçalves.
\newblock Performance of early warning signals for disease re-emergence: A case
  study on covid-19 data.
\newblock \emph{PLoS Comput Biol}, 18:\penalty0 e1009958, 2022{\natexlab{c}}.

\bibitem[Boettiger and Hastings(2012{\natexlab{a}})]{Boettiger2012b}
Carl Boettiger and Alan Hastings.
\newblock Quantifying limits to detection of early warning for critical
  transitions.
\newblock \emph{J Roy Soc Interface}, 9:\penalty0 2527--2539,
  2012{\natexlab{a}}.

\bibitem[Kendall(1938)]{kendall1938new}
Maurice~G Kendall.
\newblock A new measure of rank correlation.
\newblock \emph{Biometrika}, 30\penalty0 (1/2):\penalty0 81--93, 1938.

\bibitem[Sarkar et~al.(2019)Sarkar, Sinha, Levine, Jolly, and
  Dutta]{Sarkar2019}
Sukanta Sarkar, Sudipta~Kumar Sinha, Herbert Levine, Mohit~Kumar Jolly, and
  Partha~Sharathi Dutta.
\newblock Anticipating critical transitions in epithelial-hybrid-mesenchymal
  cell-fate determination.
\newblock \emph{P Natl Acad Sci Usa}, 116:\penalty0 26343--26352, 2019.

\bibitem[Byrd et~al.(2019)Byrd, Erez, Vogel, Peterson, Vennettilli,
  Altan-Bonnet, and Mugler]{byrd2019critical}
Tommy~A Byrd, Amir Erez, Robert~M Vogel, Curtis Peterson, Michael Vennettilli,
  Grégoire Altan-Bonnet, and Andrew Mugler.
\newblock Critical slowing down in biochemical networks with feedback.
\newblock \emph{Phys Rev E}, 100:\penalty0 22415, 2019.

\bibitem[Brett and Rohani(2020)]{brett2020dynamical}
Tobias~S. Brett and Pejman Rohani.
\newblock Dynamical footprints enable detection of disease emergence.
\newblock \emph{PLoS Biol}, 18:\penalty0 e3000697, 2020.

\bibitem[Delecroix et~al.(2023)Delecroix, van Nes, van~de Leemput, Rotbarth,
  Scheffer, and Bosch]{delecroix2023potential}
Clara Delecroix, Egbert~H van Nes, Ingrid~A van~de Leemput, Ronny Rotbarth,
  Marten Scheffer, and Quirine~ten Bosch.
\newblock The potential of resilience indicators to signal the risk of disease
  outbreaks, a systematic review and guide.
\newblock \emph{PLoS Glob Pub Health}, 3\penalty0 (10):\penalty0 e0002253,
  2023.

\bibitem[Kuehn et~al.(2022)Kuehn, Lux, and Neamtu]{Kuehn2021a}
Christian Kuehn, Kerstin Lux, and Alexandra Neamtu.
\newblock {Warning Signs for Non-Markovian Bifurcations: Color Blindness and
  Scaling Laws}.
\newblock \emph{P Roy Soc A}, 478\penalty0 (2259):\penalty0 20210740, 2022.

\bibitem[Chen et~al.(2017)Chen, Kang, and Fu]{Chen2017}
Xi~Chen, Yan~Mei Kang, and Yu~Xuan Fu.
\newblock Switches in a genetic regulatory system under multiplicative
  non-gaussian noise.
\newblock \emph{J Theor Biol}, 435:\penalty0 134--144, 2017.

\bibitem[Scheffer et~al.(2009{\natexlab{b}})Scheffer, Bascompte, Brock,
  Brovkin, Carpenter, Dakos, et~al.]{Scheffer2009c}
Marten Scheffer, Jordi Bascompte, William~A. Brock, Victor Brovkin, Stephen~R.
  Carpenter, Vasilis Dakos, et~al.
\newblock {Early-warning signals for critical transitions}.
\newblock \emph{Nature}, 461\penalty0 (7260):\penalty0 53--59,
  2009{\natexlab{b}}.

\bibitem[Estes et~al.(2011)Estes, Terborgh, Brashares, Power, Berger, Bond,
  et~al.]{estes2011trophic}
James~A Estes, John Terborgh, Justin~S Brashares, Mary~E Power, Joel Berger,
  William~J Bond, et~al.
\newblock Trophic downgrading of planet earth.
\newblock \emph{Science}, 333\penalty0 (6040):\penalty0 301--306, 2011.

\bibitem[Boettiger and Hastings(2012{\natexlab{b}})]{boettiger2012early}
Carl Boettiger and Alan Hastings.
\newblock Early warning signals and the prosecutor's fallacy.
\newblock \emph{Proc Roy Soc B}, 279:\penalty0 4734--4739, 2012{\natexlab{b}}.

\bibitem[K{\'{e}}fi et~al.(2016)K{\'{e}}fi, Holmgren, and Scheffer]{Kefi2016}
Sonia K{\'{e}}fi, Milena Holmgren, and Marten Scheffer.
\newblock {When can positive interactions cause alternative stable states in
  ecosystems?}
\newblock \emph{Func Ecol}, 30\penalty0 (1):\penalty0 88--97, 2016.

\bibitem[Scheffer(2020)]{scheffer2020critical}
Marten Scheffer.
\newblock \emph{Critical Transitions in Nature and Society}, volume~16.
\newblock Princeton UP, 2020.

\bibitem[Jones and Jahanshahi(2014)]{Jones2014}
Catherine~R.G. Jones and Marjan Jahanshahi.
\newblock {Contributions of the Basal Ganglia to Temporal Processing: Evidence
  from Parkinson's Disease}.
\newblock \emph{Timing Time Percept}, 2\penalty0 (1):\penalty0 87--127, 2014.

\bibitem[Meucci(2009)]{meucci2009review}
Attilio Meucci.
\newblock Review of statistical arbitrage, cointegration, and multivariate
  ornstein-uhlenbeck.
\newblock \emph{Available at SSRN}, 2009.

\bibitem[Vatiwutipong and Phewchean(2019)]{vatiwutipong2019alternative}
Pat Vatiwutipong and Nattakorn Phewchean.
\newblock Alternative way to derive the distribution of the multivariate
  ornstein--uhlenbeck process.
\newblock \emph{Adv Diff Eq}, 2019\penalty0 (1):\penalty0 1--7, 2019.

\bibitem[Weinans et~al.(2021)Weinans, Quax, van Nes, and
  de~Leemput]{Weinans2021}
Els Weinans, Rick Quax, Egbert~H. van Nes, and Ingrid~A.van de~Leemput.
\newblock Evaluating the performance of multivariate indicators of resilience
  loss.
\newblock \emph{Sci Rep}, 11:\penalty0 1--11, 2021.

\bibitem[Morr et~al.(2023)Morr, Riechers, Gorj{\~a}o, and
  Boers]{morr2023anticipating}
Andreas Morr, Keno Riechers, Leonardo~Rydin Gorj{\~a}o, and Niklas Boers.
\newblock Anticipating critical transitions in multi-dimensional systems driven
  by time- and state-dependent noise.
\newblock \emph{arXiv}, 2023.

\bibitem[Li and Zhang(2021)]{li2021resilience}
Yan Li and Shao-Wu Zhang.
\newblock Resilience function uncovers the critical transitions in cancer
  initiation.
\newblock \emph{Brief Bioinform}, 22:\penalty0 bbab175, 2021.

\bibitem[Masuda and Kundu(2022)]{Masuda2022}
Naoki Masuda and Prosenjit Kundu.
\newblock Dimension reduction of dynamical systems on networks with leading and
  non-leading eigenvectors of adjacency matrices.
\newblock \emph{Phys Rev Research}, 4:\penalty0 023257, 2022.

\bibitem[Duan et~al.(2022)Duan, Yan, Rong, and Hou]{Duan2022}
Dongli Duan, Qi~Yan, Yisheng Rong, and Gege Hou.
\newblock Predicting the cascading failure of dynamical networks based on a new
  dimension reduction method.
\newblock \emph{Physica A}, 606:\penalty0 128160, 2022.

\bibitem[Paul and Chalup(2017)]{Paul2017}
Rahul Paul and Stephan~K. Chalup.
\newblock A study on validating non-linear dimensionality reduction using
  persistent homology.
\newblock \emph{Pattern Recog Lett}, 100:\penalty0 160--166, 2017.

\bibitem[Jiang et~al.(2018)Jiang, Huang, Seager, Lin, Grebogi, Hastings, and
  Lai]{Jiang2018}
Junjie Jiang, Zi~Gang Huang, Thomas~P. Seager, Wei Lin, Celso Grebogi, Alan
  Hastings, and Ying~Cheng Lai.
\newblock Predicting tipping points in mutualistic networks through dimension
  reduction.
\newblock \emph{P Natl Acad Sci Usa}, 115:\penalty0 E639--E647, 2018.

\bibitem[Reyes et~al.(2020)Reyes, Otero-Muras, and Petyuk]{Reyes2020}
Brandon~C Reyes, Irene Otero-Muras, and Vladislav~A Petyuk.
\newblock A general technique for the detection of switch-like bistability in
  chemical reaction networks governed by mass action kinetics with conservation
  laws.
\newblock \emph{bioRxiv}, pages 2020--11, 2020.

\bibitem[Cheng and Scherpen(2021)]{Cheng2021}
X.~Cheng and J.M.A. Scherpen.
\newblock Model reduction methods for complex network systems.
\newblock \emph{Annu Rev Control Robotics Auton Sys}, 4:\penalty0 425--453,
  2021.

\bibitem[Machta et~al.(2013)Machta, Chachra, Transtrum, and
  Sethna]{machta2013parameter}
Benjamin~B Machta, Ricky Chachra, Mark~K Transtrum, and James~P Sethna.
\newblock Parameter space compression underlies emergent theories and
  predictive models.
\newblock \emph{Science}, 342\penalty0 (6158):\penalty0 604--607, 2013.

\bibitem[Transtrum and Qiu(2014)]{transtrum2014model}
Mark~K Transtrum and Peng Qiu.
\newblock Model reduction by manifold boundaries.
\newblock \emph{Phys Rev Lett}, 113\penalty0 (9):\penalty0 098701, 2014.

\bibitem[Par{\'e} et~al.(2019)Par{\'e}, Grimsman, Wilson, Transtrum, and
  Warnick]{pare2019model}
Philip~E Par{\'e}, David Grimsman, Alma~T Wilson, Mark~K Transtrum, and Sean
  Warnick.
\newblock Model boundary approximation method as a unifying framework for
  balanced truncation and singular perturbation approximation.
\newblock \emph{IEEE T Autom Contr}, 64\penalty0 (11):\penalty0 4796--4802,
  2019.

\bibitem[Paré et~al.(2015)Paré, Wilson, Transtrum, and Warnick]{pare2015mmr}
Philip~E. Paré, Alma~T. Wilson, Mark~K. Transtrum, and Sean~C. Warnick.
\newblock A unified view of balanced truncation and singular perturbation
  approximations.
\newblock In \emph{American Control Conference}, pages 1989--1994, 2015.

\bibitem[Cohen et~al.(2022)Cohen, Leung, Legault, Gravel, Blanchet, Côté,
  Fülöp, Lee, Dufour, Liu, and Nakazato]{Cohen2022}
Alan~A. Cohen, Diana~L. Leung, Véronique Legault, Dominique Gravel,
  F.~Guillaume Blanchet, Anne-Marie~C. Côté, Tamàs Fülöp, Juhong Lee,
  Frédérik Dufour, Mingxin Liu, and Yuichi Nakazato.
\newblock Synchrony of biomarker variability indicates a critical transition:
  Application to mortality prediction in hemodialysis.
\newblock \emph{iScience}, 25:\penalty0 104385, 2022.

\bibitem[Yang et~al.(2018)Yang, Li, Tang, Liu, Zhang, Chen, and Xia]{Yang2018a}
Biwei Yang, Meiyi Li, Wenqing Tang, Weixin Liu, Si~Zhang, Luonan Chen, and
  Jinglin Xia.
\newblock Dynamic network biomarker indicates pulmonary metastasis at the
  tipping point of hepatocellular carcinoma.
\newblock \emph{Nat Commun}, 9:\penalty0 1--14, 2018.

\bibitem[Chen et~al.(2019)Chen, Chen, Chen, Zhou, and Liu]{chen2019detecting}
Pei Chen, Ely Chen, Luonan Chen, Xianghong~Jasmine Zhou, and Rui Liu.
\newblock Detecting early-warning signals of influenza outbreak based on
  dynamic network marker.
\newblock \emph{J Cell Mol Med}, 23\penalty0 (1):\penalty0 395--404, 2019.

\bibitem[Aihara et~al.(2022)Aihara, Liu, Koizumi, Liu, and Chen]{Aihara2022}
Kazuyuki Aihara, Rui Liu, Keiichi Koizumi, Xiaoping Liu, and Luonan Chen.
\newblock Dynamical network biomarkers: Theory and applications.
\newblock \emph{Gene}, 808:\penalty0 145997, 2022.

\bibitem[Yan et~al.(2021)Yan, Li, Gao, Li, and Chen]{yan2021identifying}
Jinling Yan, Peiluan Li, Rong Gao, Ying Li, and Luonan Chen.
\newblock Identifying critical states of complex diseases by single-sample
  jensen-shannon divergence.
\newblock \emph{Front Oncol}, 11:\penalty0 1824, 2021.

\bibitem[Masuda et~al.(2024)Masuda, Aihara, and
  MacLaren]{masuda2024anticipating}
Naoki Masuda, Kazuyuki Aihara, and Neil~G MacLaren.
\newblock Anticipating regime shifts by mixing early warning signals from
  different nodes.
\newblock \emph{Nat Commun}, 15\penalty0 (1):\penalty0 1086, 2024.

\bibitem[Morr et~al.(2024)Morr, Boers, and Ashwin]{morr2024internal}
Andreas Morr, Niklas Boers, and Peter Ashwin.
\newblock Internal noise interference to warnings of tipping points in generic
  multidimensional dynamical systems.
\newblock \emph{SIAM J Appl Dyn Sys}, 23\penalty0 (4):\penalty0 2793--2806,
  2024.

\bibitem[Bathiany et~al.(2024)Bathiany, Nian, Dr{\"u}ke, and
  Boers]{bathiany2024resilience}
Sebastian Bathiany, Da~Nian, Markus Dr{\"u}ke, and Niklas Boers.
\newblock Resilience indicators for tropical rainforests in a dynamic
  vegetation model.
\newblock \emph{Global Change Biol}, 30\penalty0 (12):\penalty0 e17613, 2024.

\bibitem[Jiang et~al.(2019)Jiang, Hastings, and Lai]{jiang2019harnessing}
Junjie Jiang, Alan Hastings, and Ying-Cheng Lai.
\newblock {Harnessing tipping points in complex ecological networks}.
\newblock \emph{J Roy Soc Interface}, 16\penalty0 (158):\penalty0 20190345,
  2019.

\bibitem[Ben-Yami et~al.(2023)Ben-Yami, Skiba, Bathiany, and
  Boers]{ben2023uncertainties}
Maya Ben-Yami, Vanessa Skiba, Sebastian Bathiany, and Niklas Boers.
\newblock Uncertainties in critical slowing down indicators of
  observation-based fingerprints of the atlantic overturning circulation.
\newblock \emph{Nat Commun}, 14\penalty0 (1):\penalty0 8344, 2023.

\bibitem[Boers and Rypdal(2021)]{Boers2021}
Niklas Boers and Martin Rypdal.
\newblock {Critical slowing down suggests that the western Greenland Ice Sheet
  is close to a tipping point}.
\newblock \emph{P Natl Acad Sci Usa}, 118\penalty0 (21):\penalty0 1--7, 2021.

\bibitem[Boulton et~al.(2022)Boulton, Lenton, and Boers]{boulton2022pronounced}
Chris~A Boulton, Timothy~M Lenton, and Niklas Boers.
\newblock Pronounced loss of amazon rainforest resilience since the early
  2000s.
\newblock \emph{Nat Clim Change}, 12\penalty0 (3):\penalty0 271--278, 2022.

\bibitem[Hummel et~al.(2024)Hummel, Boers, and Rypdal]{hummel2024inconclusive}
Clara Hummel, Niklas Boers, and Martin Rypdal.
\newblock Inconclusive early warning signals for dansgaard-oeschger events
  across greenland ice cores.
\newblock \emph{EGUsphere}, 2024:\penalty0 1--44, 2024.

\bibitem[Ritchie et~al.(2021)Ritchie, Clarke, Cox, and
  Huntingford]{ritchie2021overshooting}
Paul~DL Ritchie, Joseph~J Clarke, Peter~M Cox, and Chris Huntingford.
\newblock Overshooting tipping point thresholds in a changing climate.
\newblock \emph{Nature}, 592\penalty0 (7855):\penalty0 517--523, 2021.

\bibitem[Korolev et~al.(2014)Korolev, Xavier, and Gore]{Korolev2014}
K.~S. Korolev, J.B. Xavier, and J.~Gore.
\newblock {Turning ecology and evolution against cancer}.
\newblock \emph{Nat Rev Cancer}, pages 1--10, 2014.

\bibitem[Wilkat et~al.(2019)Wilkat, Rings, and Lehnertz]{Wilkat2019}
Theresa Wilkat, Thorsten Rings, and Klaus Lehnertz.
\newblock {No evidence for critical slowing down prior to human epileptic
  seizures}.
\newblock \emph{Chaos}, 29\penalty0 (9):\penalty0 2--7, 2019.

\bibitem[Horstmeyer et~al.(2018)Horstmeyer, Kuehn, and Thurner]{Horstmeyer2018}
Leonhard Horstmeyer, Christian Kuehn, and Stefan Thurner.
\newblock Network topology near criticality in adaptive epidemics.
\newblock \emph{Phys Rev E}, 98:\penalty0 1--9, 2018.

\bibitem[Alonso et~al.(2019)Alonso, Dobson, and Pascual]{Alonso2019}
David Alonso, Andy Dobson, and Mercedes Pascual.
\newblock Critical transitions in malaria transmission models are consistently
  generated by superinfection.
\newblock \emph{Philos T R Soc B}, 374:\penalty0 1--19, 2019.

\bibitem[Southall et~al.(2020)Southall, Tildesley, and
  Dyson]{southall2020prospects}
Emma Southall, Michael~J. Tildesley, and Louise Dyson.
\newblock {Prospects for detecting early warning signals in discrete event
  sequence data: Application to epidemiological incidence data.}
\newblock \emph{PLoS Comput Biol}, 16\penalty0 (9):\penalty0 e1007836, 2020.

\bibitem[Dablander et~al.(2022)Dablander, Heesterbeek, Borsboom, and
  Drake]{Dablander2021}
Fabian Dablander, Hans Heesterbeek, Denny Borsboom, and John~M Drake.
\newblock Overlapping timescales obscure early warning signals of the second
  covid-19 wave.
\newblock \emph{P Roy Soc B}, 289\penalty0 (1968):\penalty0 20211809, 2022.

\bibitem[Southall et~al.(2021)Southall, Brett, Tildesley, and
  Dyson]{Southall2021}
Emma Southall, Tobias~S. Brett, Michael~J. Tildesley, and Louise Dyson.
\newblock Early warning signals of infectious disease transitions: A review.
\newblock \emph{J. Roy Soc Interface}, 18:\penalty0 20210555, 2021.

\bibitem[Alibakhshi et~al.()Alibakhshi, Dakos, Diekert, Heyen, Nesje,
  Proverbio, et~al.]{Alibakhshi2025}
Sara Alibakhshi, Vasilis Dakos, Florian Diekert, Daniel Heyen, Frikk Nesje,
  Daniele Proverbio, et~al.
\newblock Operationalising early warning signals for effective decision-making.
\newblock \emph{npj Environ Soc Sci, under review}.

\bibitem[Pirani and Jafarpour(2023)]{pirani2023network}
Mohammad Pirani and Saber Jafarpour.
\newblock Network critical slowing down: Data-driven detection of critical
  transitions in nonlinear networks.
\newblock \emph{IEEE T Contr Net Sys}, pages 573 -- 585, 2023.

\bibitem[Ditlevsen et~al.(2007)Ditlevsen, Andersen, and
  Svensson]{ditlevsen2007climate}
Peter~D Ditlevsen, Katrine~Krogh Andersen, and Anders Svensson.
\newblock {The DO-climate events are probably noise induced: statistical
  investigation of the claimed 1470 years cycle}.
\newblock \emph{Climate of the Past}, 3\penalty0 (1):\penalty0 129--134, 2007.

\bibitem[Bizyaeva et~al.(2022)Bizyaeva, Franci, and
  Leonard]{bizyaeva2022nonlinear}
Anastasia Bizyaeva, Alessio Franci, and Naomi~Ehrich Leonard.
\newblock Nonlinear opinion dynamics with tunable sensitivity.
\newblock \emph{IEEE T Autom Contr}, 68\penalty0 (3):\penalty0 1415--1430,
  2022.

\bibitem[Cathcart et~al.(2023)Cathcart, Santos, Park, and
  Leonard]{cathcart2023proactive}
Charlotte Cathcart, María Santos, Shinkyu Park, and Naomi~Ehrich Leonard.
\newblock Proactive opinion-driven robot navigation around human movers.
\newblock In \emph{IEEE/RSJ IROS}, pages 4052--4058, 2023.

\bibitem[Soranzo et~al.(2012)Soranzo, Ramezani, Iacono, and
  Altafini]{Soranzo2012}
Nicola Soranzo, Fahimeh Ramezani, Giovanni Iacono, and Claudio Altafini.
\newblock Decompositions of large-scale biological systems based on dynamical
  properties.
\newblock \emph{Bioinformatics}, 28:\penalty0 76--83, 2012.

\bibitem[Kundu et~al.(2022)Kundu, Kori, and Masuda]{Kundu2022}
Prosenjit Kundu, Hiroshi Kori, and Naoki Masuda.
\newblock Accuracy of a one-dimensional reduction of dynamical systems on
  networks.
\newblock \emph{Phys Rev E}, 105:\penalty0 024305, 2022.

\bibitem[Thibeault et~al.(2020)Thibeault, St-Onge, Dubé, and
  Desrosiers]{Thibeault2020}
Vincent Thibeault, Guillaume St-Onge, Louis~J. Dubé, and Patrick Desrosiers.
\newblock Threefold way to the dimension reduction of dynamics on networks: An
  application to synchronization.
\newblock \emph{Phys Rev Research}, 2:\penalty0 043215, 2020.

\bibitem[Tu et~al.(2021)Tu, D'Odorico, and Suweis]{Tu2021}
Chengyi Tu, Paolo D'Odorico, and Samir Suweis.
\newblock Dimensionality reduction of complex dynamical systems.
\newblock \emph{iScience}, 24:\penalty0 101912, 2021.

\bibitem[Wu et~al.(2023)Wu, Duan, and Xiao]{Wu2023}
Chengxing Wu, Dongli Duan, and Renbin Xiao.
\newblock A novel dimension reduction method with information entropy to
  evaluate network resilience.
\newblock \emph{Physica A}, 620:\penalty0 128727, 2023.

\bibitem[Yu et~al.(2022)Yu, Cheng, Scherpen, and Xiong]{Yu2022}
Lanlin Yu, Xiaodong Cheng, Jacquelien~MA Scherpen, and Junlin Xiong.
\newblock H2 model reduction for diffusively coupled second-order networks by
  convex-optimization.
\newblock \emph{Automatica}, 137:\penalty0 110118, 2022.

\bibitem[Albert and Barab{\'a}si(2002)]{Albert2002}
R{\'e}ka Albert and Albert-L{\'a}szl{\'o} Barab{\'a}si.
\newblock Statistical mechanics of complex networks.
\newblock \emph{Rev Mod Phys}, 74\penalty0 (1):\penalty0 47--97, 2002.

\bibitem[Newman(2003)]{Newman2003}
Mark~EJ Newman.
\newblock The structure and function of complex networks.
\newblock \emph{SIAM review}, 45\penalty0 (2):\penalty0 167--256, 2003.

\bibitem[Klein and Hoel(2020)]{Klein2020}
Brennan Klein and Erik Hoel.
\newblock The emergence of informative higher scales in complex networks.
\newblock \emph{Complexity}, 2020:\penalty0 1--12, 2020.

\bibitem[Boccaletti et~al.(2014)Boccaletti, Bianconi, Criado, del Genio,
  Gómez-Gardeñes, Romance, et~al.]{Boccaletti2014}
S.~Boccaletti, G.~Bianconi, R.~Criado, C.~I. del Genio, J.~Gómez-Gardeñes,
  M.~Romance, et~al.
\newblock The structure and dynamics of multilayer networks.
\newblock \emph{Phys Rep}, 544:\penalty0 1--122, 2014.

\bibitem[Battiston et~al.(2017)Battiston, Nicosia, and Latora]{Battiston2017}
Federico Battiston, Vincenzo Nicosia, and Vito Latora.
\newblock The new challenges of multiplex networks: Measures and models.
\newblock \emph{Eur Phys J Special Topics}, 226:\penalty0 401--416, 2017.

\bibitem[Kinsley et~al.(2020)Kinsley, Rossi, Silk, and Van~der
  Waal]{kinsley2020multilayer}
Amy~C Kinsley, Gianluigi Rossi, Matthew~J Silk, and Kimberly Van~der Waal.
\newblock Multilayer and multiplex networks: An introduction to their use in
  veterinary epidemiology.
\newblock \emph{Front Vet Sci}, 7:\penalty0 596, 2020.

\bibitem[Bianconi(2018)]{bianconi2018multilayer}
Ginestra Bianconi.
\newblock \emph{Multilayer Networks: Structure and Function}.
\newblock Oxford UP, 2018.

\bibitem[Nicosia et~al.(2013)Nicosia, Bianconi, Latora, and
  Barthelemy]{Nicosia2013}
V.~Nicosia, G.~Bianconi, V.~Latora, and M.~Barthelemy.
\newblock Growing multiplex networks.
\newblock \emph{Phys Rev Lett}, 111:\penalty0 058701, 2013.

\bibitem[Menichetti et~al.(2014)Menichetti, Remondini, Panzarasa, Mondragón,
  and Bianconi]{Menichetti2014}
Giulia Menichetti, Daniel Remondini, Pietro Panzarasa, Raúl~J. Mondragón, and
  Ginestra Bianconi.
\newblock Weighted multiplex networks.
\newblock \emph{PLoS ONE}, 9:\penalty0 e97857, 2014.

\bibitem[Bick et~al.(2023)Bick, Gross, Harrington, and Schaub]{Bick2023}
Christian Bick, Elizabeth Gross, Heather~A. Harrington, and Michael~T. Schaub.
\newblock What are higher-order networks?
\newblock \emph{SIAM Rev}, 65:\penalty0 686--731, 2023.

\bibitem[Majhi et~al.(2022)Majhi, Perc, and Ghosh]{Majhi2022}
Soumen Majhi, Matjaž Perc, and Dibakar Ghosh.
\newblock Dynamics on higher-order networks: a review.
\newblock \emph{J Roy Soc Interface}, 19, 2022.

\bibitem[Muhammad and Egerstedt(2006)]{muhammad2006control}
Abubakr Muhammad and Magnus Egerstedt.
\newblock Control using higher order laplacians in network topologies.
\newblock In \emph{Proc 17th Int Symp Math T Net Sys}, pages 1024--1038.
  Citeseer, 2006.

\bibitem[Battiston et~al.(2020)Battiston, Cencetti, Iacopini, Latora, Lucas,
  Patania, et~al.]{Battiston2020}
Federico Battiston, Giulia Cencetti, Iacopo Iacopini, Vito Latora, Maxime
  Lucas, Alice Patania, et~al.
\newblock Networks beyond pairwise interactions: Structure and dynamics.
\newblock \emph{Phys Rep}, 874:\penalty0 1--92, 2020.

\bibitem[Salnikov et~al.(2018)Salnikov, Cassese, and
  Lambiotte]{salnikov2018simplicial}
Vsevolod Salnikov, Daniele Cassese, and Renaud Lambiotte.
\newblock Simplicial complexes and complex systems.
\newblock \emph{Eur J Phys}, 40\penalty0 (1):\penalty0 014001, 2018.

\bibitem[Iacopini et~al.(2019)Iacopini, Petri, Barrat, and
  Latora]{iacopini2019simplicial}
Iacopo Iacopini, Giovanni Petri, Alain Barrat, and Vito Latora.
\newblock Simplicial models of social contagion.
\newblock \emph{Nat Commun}, 10\penalty0 (1):\penalty0 2485, 2019.

\bibitem[Li et~al.(2021)Li, Deng, Han, Alfaro-Bittner, Barzel, and
  Boccaletti]{li2021contagion}
Zhaoqing Li, Zhenghong Deng, Zhen Han, Karin Alfaro-Bittner, Baruch Barzel, and
  Stefano Boccaletti.
\newblock Contagion in simplicial complexes.
\newblock \emph{Chaos Soliton Fract}, 152:\penalty0 111307, 2021.

\bibitem[Mill{\'a}n et~al.(2020)Mill{\'a}n, Torres, and
  Bianconi]{millan2020explosive}
Ana~P Mill{\'a}n, Joaqu{\'\i}n~J Torres, and Ginestra Bianconi.
\newblock Explosive higher-order kuramoto dynamics on simplicial complexes.
\newblock \emph{Phys Rev Lett}, 124\penalty0 (21):\penalty0 218301, 2020.

\bibitem[Lotito et~al.(2022)Lotito, Musciotto, Montresor, and
  Battiston]{Lotito2022}
Quintino~Francesco Lotito, Federico Musciotto, Alberto Montresor, and Federico
  Battiston.
\newblock Higher-order motif analysis in hypergraphs.
\newblock \emph{Commun Phys}, 5:\penalty0 79, 2022.

\bibitem[Ghosh et~al.(2023)Ghosh, Khanra, Kundu, Ji, Ghosh, and
  Hens]{Ghosh2023}
Subrata Ghosh, Pitambar Khanra, Prosenjit Kundu, Peng Ji, Dibakar Ghosh, and
  Chittaranjan Hens.
\newblock Dimension reduction in higher-order contagious phenomena.
\newblock \emph{Chaos}, 33, 2023.

\bibitem[Wang et~al.(2022)Wang, Zhang, Zhu, and Ma]{Wang2022}
Huan Wang, Hai-Feng Zhang, Pei-Can Zhu, and Chuang Ma.
\newblock Interplay of simplicial awareness contagion and epidemic spreading on
  time-varying multiplex networks.
\newblock \emph{Chaos}, 32, 2022.

\bibitem[Chakraborty et~al.(2024)Chakraborty, Gao, Miry, Ramazi, Greiner,
  Lewis, and Wang]{chakraborty2024early}
Amit~K Chakraborty, Shan Gao, Reza Miry, Pouria Ramazi, Russell Greiner, Mark~A
  Lewis, and Hao Wang.
\newblock An early warning indicator trained on stochastic disease-spreading
  models with different noises.
\newblock \emph{J Roy Soc Interface}, 21\penalty0 (217):\penalty0 20240199,
  2024.

\bibitem[Grassia et~al.(2021)Grassia, Domenico, and
  Mangioni]{grassia2021machine}
Marco Grassia, Manlio~De Domenico, and Giuseppe Mangioni.
\newblock Machine learning dismantling and early-warning signals of
  disintegration in complex systems.
\newblock \emph{Nat Commun}, 12\penalty0 (1):\penalty0 5190, 2021.

\bibitem[Huang et~al.(2024)Huang, Bathiany, Ashwin, and Boers]{huang2024deep}
Yu~Huang, Sebastian Bathiany, Peter Ashwin, and Niklas Boers.
\newblock Deep learning for predicting rate-induced tipping.
\newblock \emph{Nat Mach Int}, pages 1--10, 2024.

\bibitem[Salerno et~al.(2013)Salerno, Cosentino, Merola, Bates, and
  Amato]{Salerno2013}
Luca Salerno, Carlo Cosentino, Alessio Merola, Declan~G Bates, and Francesco
  Amato.
\newblock Validation of a model of the gal regulatory system via robustness
  analysis of its bistability characteristics, 2013.

\bibitem[Batt et~al.(2005)Batt, Ropers, de~Jong, Geiselmann, Mateescu, Page,
  and Schneider]{Batt2005}
Grégory Batt, Delphine Ropers, Hidde de~Jong, Johannes Geiselmann, Radu
  Mateescu, Michel Page, and Dominique Schneider.
\newblock Validation of qualitative models of genetic regulatory networks by
  model checking: Analysis of the nutritional stress response in escherichia
  coli.
\newblock \emph{Bioinformatics}, 21, 2005.

\bibitem[Paré et~al.(2020)Paré, Liu, Beck, Kirwan, and
  Basar]{Pare2020analysis}
Philip~E. Paré, Ji~Liu, Carolyn~L. Beck, Barrett~E. Kirwan, and Tamer Basar.
\newblock Analysis, estimation, and validation of discrete-time epidemic
  processes.
\newblock \emph{IEEE T Control Sys Tech}, 28:\penalty0 79--93, 2020.

\bibitem[Ascensao et~al.(2016)Ascensao, Datta, Hancioglu, Sontag, Gennaro, and
  Igoshin]{Ascensao2016}
Joao~A. Ascensao, Pratik Datta, Baris Hancioglu, Eduardo Sontag, Maria~L.
  Gennaro, and Oleg~A. Igoshin.
\newblock Non-monotonic response to monotonic stimulus: Regulation of
  glyoxylate shunt gene-expression dynamics in mycobacterium tuberculosis.
\newblock \emph{PLoS Comput Biol}, 12, 2016.

\bibitem[Wang and Su(2014)]{WangSu2014}
X~Wang and H~Su.
\newblock Pinning control of complex networked systems: A decade after and
  beyond.
\newblock \emph{Annu Rev Contr}, 38\penalty0 (1):\penalty0 103--111, 2014.

\bibitem[Blanchini et~al.(2021{\natexlab{c}})Blanchini, Devia, and
  Giordano]{BDG2021}
Franco Blanchini, Carlos~Andrés Devia, and Giulia Giordano.
\newblock Structural polyhedral stability of a biochemical network is
  equivalent to finiteness of the associated generalised petri net.
\newblock \emph{arXiv}, 2021{\natexlab{c}}.

\bibitem[Iftar and Davison(2002)]{IftarDavison2002}
A.~Iftar and E.~J. Davison.
\newblock Decentralized control strategies for dynamic routing.
\newblock \emph{Optimal Control Applications and Methods}, 23\penalty0
  (6):\penalty0 329--355, 2002.

\bibitem[Blanchini et~al.(2013)Blanchini, Franco, and Giordano]{Blanchini2013}
F~Blanchini, E~Franco, and G~Giordano.
\newblock Structured-lmi conditions for stabilizing network-decentralized
  control.
\newblock In \emph{Proc IEEE CDC}, pages 6880--6885, 2013.

\bibitem[Blanchini et~al.(2015{\natexlab{b}})Blanchini, Franco, and
  Giordano]{Blanchini2015b}
F~Blanchini, E~Franco, and G~Giordano.
\newblock Network-decentralized control strategies for stabilization.
\newblock \emph{IEEE T Automat Contr}, 60:\penalty0 491--496,
  2015{\natexlab{b}}.

\bibitem[Blanchini et~al.(2016)Blanchini, Franco, Giordano, Mardanlou, and
  Montessoro]{Blanchini2016compartmental}
F~Blanchini, E~Franco, G~Giordano, V~Mardanlou, and P~L Montessoro.
\newblock Compartmental flow control: Decentralization, robustness and
  optimality.
\newblock \emph{Automatica}, 64:\penalty0 18--28, 2016.

\bibitem[Giordano et~al.(2016{\natexlab{c}})Giordano, Blanchini, Franco,
  Mardanlou, and Montessoro]{Giordano2016smallest}
G~Giordano, F~Blanchini, E~Franco, V~Mardanlou, and P~L Montessoro.
\newblock The smallest eigenvalue of the generalized laplacian matrix, with
  application to network-decentralized estimation for homogeneous systems.
\newblock \emph{IEEE T Network Sci Eng}, 3:\penalty0 312--324,
  2016{\natexlab{c}}.

\bibitem[d’Onofrio et~al.(2023)d’Onofrio, Iannelli, Manfredi, and
  Marinoschi]{Donofrio2023}
Alberto d’Onofrio, Mimmo Iannelli, Piero Manfredi, and Gabriela Marinoschi.
\newblock Optimal epidemic control by social distancing and vaccination of an
  infection structured by time since infection: The covid-19 case study.
\newblock \emph{Siam J Appl Math}, pages S199--S224, 2023.

\bibitem[S{\'e}lley et~al.(2015)S{\'e}lley, Besenyei, Kiss, and
  Simon]{Sélley2015}
Fanni~M. S{\'e}lley, {\'A}d{\'a}m Besenyei, Istvan~Z. Kiss, and P{\'e}ter~L.
  Simon.
\newblock Dynamic control of modern, network-based epidemic models.
\newblock \emph{Siam J Appl Dyn Syst}, 14\penalty0 (1):\penalty0 168--187,
  2015.

\bibitem[Sharomi and Malik(2015)]{Sharomi2015}
Oluwaseun Sharomi and Tufail Malik.
\newblock Optimal control in epidemiology.
\newblock \emph{Annals of Operations Research}, 251:\penalty0 55--71, 2015.

\bibitem[Zaric and Brandeau(2002)]{Zaric2002}
Gregory~S. Zaric and Margaret~L. Brandeau.
\newblock Dynamic resource allocation for epidemic control in multiple
  populations.
\newblock \emph{IMA J Math Appl Med Biol}, 19:\penalty0 235--255, 2002.

\bibitem[Hansen and Day(2010)]{Hansen2010}
Elsa Hansen and Troy Day.
\newblock Optimal control of epidemics with limited resources.
\newblock \emph{J Math Biol}, 62:\penalty0 423--451, 2010.

\bibitem[Pérez et~al.(2022)Pérez, Abuin, Actis, Ferramosca, Hernandez-Vargas,
  and González]{Perez2022}
Mara Pérez, Pablo Abuin, Marcelo Actis, Antonio Ferramosca, Esteban~A.
  Hernandez-Vargas, and Alejandro~H. González.
\newblock Optimal control strategies to tailor antivirals for acute infectious
  diseases in the host: a study case of covid-19.
\newblock \emph{Feedback Control for Personalized Medicine}, NA:\penalty0
  11--39, 2022.

\bibitem[Angulo et~al.(2021)Angulo, Castaños, Moreno-Morton,
  Velasco-Hernandez, and Moreno]{Angulo2021}
Marco~Tulio Angulo, Fernando Castaños, Rodrigo Moreno-Morton, Jorge~X.
  Velasco-Hernandez, and Jaime~A. Moreno.
\newblock A simple criterion to design optimal non-pharmaceutical interventions
  for mitigating epidemic outbreaks.
\newblock \emph{Journal of the Royal Society, Interface}, 18:\penalty0
  20200803, 2021.

\bibitem[Bin et~al.(2021)Bin, Cheung, Crisostomi, Ferraro, Lhachemi,
  Murray-Smith, Myan, Parisini, Shorten, Stein, and Stone]{Bin2021}
Michelangelo Bin, Peter Y.~K. Cheung, Emanuele Crisostomi, Pietro Ferraro, Hugo
  Lhachemi, Roderick Murray-Smith, Connor Myan, Thomas Parisini, Robert
  Shorten, Sebastian Stein, and Lewi Stone.
\newblock Post-lockdown abatement of covid-19 by fast periodic switching.
\newblock \emph{PLoS Comput Biol}, 17:\penalty0 1--34, 2021.

\bibitem[Gumel et~al.(2004)Gumel, Ruan, Day, Watmough, Brauer, {van den
  Driessche}, Gabrielson, Bowman, Alexander, Ardal, Wu, and Sahai]{Gumel2004}
Abba~B. Gumel, Shigui Ruan, Troy Day, James Watmough, Fred Brauer, Pauline {van
  den Driessche}, Dave Gabrielson, Christopher~N. Bowman, Murray~E. Alexander,
  Sten Ardal, Jianhong Wu, and Beni~M. Sahai.
\newblock Modelling strategies for controlling sars outbreaks.
\newblock \emph{Proc Biol Sci}, 271:\penalty0 2223--2232, 2004.

\bibitem[Kantner and Koprucki(2020)]{Kantner2020}
Markus Kantner and Thomas Koprucki.
\newblock Beyond just “flattening the curve”: Optimal control of epidemics
  with purely non-pharmaceutical interventions.
\newblock \emph{J Math Ind}, 10:\penalty0 23, 2020.

\bibitem[Köhler et~al.(2020)Köhler, Schwenkel, Koch, Berberich, Pauli, and
  Allgöwer]{Köhler2020}
Johannes Köhler, Lukas Schwenkel, Anne Koch, Julian Berberich, Patricia Pauli,
  and Frank Allgöwer.
\newblock Robust and optimal predictive control of the covid-19 outbreak.
\newblock \emph{Annu Rev Control}, 51:\penalty0 525--539, 2020.

\bibitem[Morato et~al.(2020)Morato, Basto, Cajueiro, and
  Normey-Rico]{Morato2020}
Marcelo~Menezes Morato, Saulo~B. Basto, Daniel~O. Cajueiro, and Julio~E.
  Normey-Rico.
\newblock An optimal predictive control strategy for covid-19 (sars-cov-2)
  social distancing policies in brazil.
\newblock \emph{Annu Rev Control}, 50:\penalty0 417--431, 2020.

\bibitem[Sontag(2021)]{Sontag2021}
Eduardo~D Sontag.
\newblock An explicit formula for minimizing the infected peak in an sir
  epidemic model when using a fixed number of complete lockdowns.
\newblock \emph{Int J Robust Nonlin}, 33:\penalty0 1--4731, 2021.

\bibitem[Ishii and Zhu(2022)]{ishii2022security}
Hideaki Ishii and Quanyan Zhu.
\newblock \emph{Security and Resilience of Control Systems}.
\newblock Springer, 2022.

\bibitem[Chen and Zhu(2019)]{chen2019game}
Juntao Chen and Quanyan Zhu.
\newblock \emph{A Game-and Decision-Theoretic Approach to Resilient
  Interdependent Network Analysis and Design}.
\newblock Springer, 2019.

\bibitem[Zhu and Basar(2015)]{zhu2015game}
Quanyan Zhu and Tamer Basar.
\newblock Game-theoretic methods for robustness, security, and resilience of
  cyberphysical control systems: games-in-games principle for optimal
  cross-layer resilient control systems.
\newblock \emph{IEEE Control Systems Magazine}, 35\penalty0 (1):\penalty0
  46--65, 2015.

\bibitem[Bussell et~al.(2019)Bussell, Dangerfield, Gilligan, and
  Cunniffe]{Bussell2019}
E.~H Bussell, C.E. Dangerfield, Christopher~A. Gilligan, and Nik~J. Cunniffe.
\newblock Applying optimal control theory to complex epidemiological models to
  inform real-world disease management.
\newblock \emph{Philos T Roy Soc B}, 374\penalty0 (1776):\penalty0
  20180284--20180284, 2019.

\bibitem[Carli et~al.(2020)Carli, Cavone, Epicoco, Scarabaggio, and
  Dotoli]{Carli2020}
Raffaele Carli, Graziana Cavone, Nicola Epicoco, Paolo Scarabaggio, and
  Mariagrazia Dotoli.
\newblock Model predictive control to mitigate the covid-19 outbreak in a
  multi-region scenario.
\newblock \emph{Annu Rev Control}, 50:\penalty0 373--393, 2020.

\bibitem[Samad et~al.(2020)Samad, Bauer, Bortoff, Di~Cairano, Fagiano, Odgaard,
  Rhinehart, S{\'a}nchez-Pe{\~n}a, Serbezov, Ankersen,
  et~al.]{samad2020industry}
Tariq Samad, Margret Bauer, Scott Bortoff, Stefano Di~Cairano, Lorenzo Fagiano,
  Peter~Fogh Odgaard, R~Russell Rhinehart, Ricardo S{\'a}nchez-Pe{\~n}a, Atanas
  Serbezov, Finn Ankersen, et~al.
\newblock Industry engagement with control research: Perspective and messages.
\newblock \emph{Annu Rev Control}, 49:\penalty0 1--14, 2020.

\bibitem[Burzy{\'n}ski et~al.(2021)Burzy{\'n}ski, Machado, Aalto, Beine,
  Goncalves, Haas, et~al.]{burzynski2021covid}
Micha{\l} Burzy{\'n}ski, Joel Machado, Atte Aalto, Michel Beine, Jorge
  Goncalves, Tom Haas, et~al.
\newblock Covid-19 crisis management in luxembourg: Insights from an
  epidemionomic approach.
\newblock \emph{Econom Human Biol}, 43:\penalty0 101051, 2021.

\bibitem[Iqbal et~al.(2020)Iqbal, Abid, Hussain, Shahzad, Waqas, and
  Iqbal]{iqbal2020effects}
Muhammad~Mazhar Iqbal, Irfan Abid, Saddam Hussain, Naeem Shahzad,
  Muhammad~Sohail Waqas, and Muhammad~Jawed Iqbal.
\newblock The effects of regional climatic condition on the spread of covid-19
  at global scale.
\newblock \emph{Sci Total Environ}, 739:\penalty0 140101, 2020.

\bibitem[Lep et~al.(2020)Lep, Babnik, and Beyazoglu]{Lep2020}
{\v{Z}}an Lep, Katarina Babnik, and Kaja~Hacin Beyazoglu.
\newblock Emotional responses and self-protective behavior within days of the
  covid-19 outbreak: The promoting role of information credibility.
\newblock \emph{Front Psychol}, 11\penalty0 (11):\penalty0 1846--1846, 2020.

\bibitem[Moran et~al.(2016)Moran, Fairchild, Generous, Hickmann, Osthus,
  Priedhorsky, et~al.]{moran2016epidemic}
Kelly~R Moran, Geoffrey Fairchild, Nicholas Generous, Kyle Hickmann, Dave
  Osthus, Reid Priedhorsky, et~al.
\newblock Epidemic forecasting is messier than weather forecasting: The role of
  human behavior and internet data streams in epidemic forecast.
\newblock \emph{J Inf Dis}, 214\penalty0 (4):\penalty0 S404--S408, 2016.

\bibitem[Proverbio et~al.(2025{\natexlab{b}})Proverbio, Tessarin, and
  Giordano]{Proverbio2025}
Daniele Proverbio, Riccardo Tessarin, and Giulia Giordano.
\newblock Data informed epidemiological-behavioural modelling.
\newblock In \emph{Comput Mech Appl Math: 16, Perspectives from Young Scholars.
  Lecture Notes in Mechanical Engineering}. Springer, 2025{\natexlab{b}}.

\bibitem[Zhou et~al.(2020)Zhou, Su, Pei, Zhang, Du, Luo, Cao, Wang, Yuan, Zhu,
  et~al.]{zhou2020covid}
Chenghu Zhou, Fenzhen Su, Tao Pei, An~Zhang, Yunyan Du, Bin Luo, Zhidong Cao,
  Juanle Wang, Wen Yuan, Yunqiang Zhu, et~al.
\newblock Covid-19: challenges to gis with big data.
\newblock \emph{Geo Sustainability}, 1\penalty0 (1):\penalty0 77--87, 2020.

\bibitem[Fontaine et~al.(2011)Fontaine, Guimarães, Kéfi, Loeuille, Memmott,
  van~der Putten, et~al.]{Fontaine2011}
Colin Fontaine, Paulo~R. Guimarães, Sonia Kéfi, Nicolas Loeuille, Jane
  Memmott, Wim~H. van~der Putten, et~al.
\newblock The ecological and evolutionary implications of merging different
  types of networks.
\newblock \emph{Ecology Lett}, 14:\penalty0 1170--1181, 2011.

\bibitem[Ashwin et~al.(2025)Ashwin, Bastiaansen, von~der Heydt, and
  Ritchie]{ashwin2025early}
Peter Ashwin, Robbin Bastiaansen, Anna~S von~der Heydt, and Paul~DL Ritchie.
\newblock Early warning skill, extrapolation and tipping for accelerating
  cascades.
\newblock \emph{Proceedings of the Royal Society A}, 481\penalty0
  (2321):\penalty0 20250405, 2025.

\bibitem[Chapman et~al.(2025)Chapman, Ashwin, and Wood]{chapman2025quantifying}
Ruth~R Chapman, Peter Ashwin, and Richard~A Wood.
\newblock Quantifying tipping behavior: Geometric early warnings and
  quasipotentials for a box model of amoc.
\newblock \emph{Chaos: An Interdisciplinary Journal of Nonlinear Science},
  35\penalty0 (2), 2025.

\bibitem[Liu et~al.(2018)Liu, Ajelli, Aleta, Merler, Moreno, and
  Vespignani]{liu2018measurability}
Quan-Hui Liu, Marco Ajelli, Alberto Aleta, Stefano Merler, Yamir Moreno, and
  Alessandro Vespignani.
\newblock Measurability of the epidemic reproduction number in data-driven
  contact networks.
\newblock \emph{P Natl Acad Sci Usa}, 115:\penalty0 12680--12685, 2018.

\bibitem[Roosa and Chowell(2019)]{Roosa2019}
Kimberlyn Roosa and Gerardo Chowell.
\newblock Assessing parameter identifiability in compartmental dynamic models
  using a computational approach: Application to infectious disease
  transmission models.
\newblock \emph{Theo Biol Med Mod}, 16:\penalty0 1--15, 2019.

\bibitem[Sparks(2013)]{sparks2013challenges}
Ross Sparks.
\newblock Challenges in designing a disease surveillance plan: what we have and
  what we need?
\newblock \emph{IIE T Health Sys Eng}, 3\penalty0 (3):\penalty0 181--192, 2013.

\bibitem[Wan and Karniadakis(2006)]{wan2006multi}
Xiaoliang Wan and George~Em Karniadakis.
\newblock Multi-element generalized polynomial chaos for arbitrary probability
  measures.
\newblock \emph{SIAM J Sci Comp}, 28\penalty0 (3):\penalty0 901--928, 2006.

\bibitem[Helbing et~al.(2015)Helbing, Brockmann, Chadefaux, Donnay, Blanke,
  Woolley-Meza, et~al.]{helbing2015saving}
Dirk Helbing, Dirk Brockmann, Thomas Chadefaux, Karsten Donnay, Ulf Blanke,
  Olivia Woolley-Meza, et~al.
\newblock Saving human lives: What complexity science and information systems
  can contribute.
\newblock \emph{J Stat Phys}, 158:\penalty0 735--781, 2015.

\bibitem[Dobrushkin(2014)]{dobrushkin2014applied}
Vladimir~A Dobrushkin.
\newblock \emph{{Applied Differential Equations: the Primary Course}},
  volume~18.
\newblock CRC Press, 2014.

\bibitem[Takens(1981)]{Takens1981}
Floris Takens.
\newblock {Detecting strange attractors in turbulence}.
\newblock In \emph{Dynamical Systems and Turbulence}, pages 366--381. Springe,
  1981.

\bibitem[Levant(2003)]{levant2003higher}
Arie Levant.
\newblock Higher-order sliding modes, differentiation and output-feedback
  control.
\newblock \emph{International journal of Control}, 76\penalty0 (9-10):\penalty0
  924--941, 2003.

\bibitem[Levant et~al.(2017)Levant, Livne, and Yu]{levant2017sliding}
Arie Levant, Miki Livne, and Xinghuo Yu.
\newblock Sliding-mode-based differentiation and its application.
\newblock \emph{IFAC-PapersOnLine}, 50\penalty0 (1):\penalty0 1699--1704, 2017.

\bibitem[Levant and Livne(2020)]{levant2020robust}
Arie Levant and Miki Livne.
\newblock Robust exact filtering differentiators.
\newblock \emph{European Journal of Control}, 55:\penalty0 33--44, 2020.

\bibitem[Hanan et~al.(2021)Hanan, Levant, and Jbara]{hanan2021low}
Avi Hanan, Arie Levant, and Adam Jbara.
\newblock Low-chattering discretization of homogeneous differentiators.
\newblock \emph{IEEE Transactions on Automatic Control}, 67\penalty0
  (6):\penalty0 2946--2956, 2021.

\bibitem[Wiggins(2003)]{wiggins2003equilibrium}
Stephen Wiggins.
\newblock Equilibrium solutions, stability, and linearized stability.
\newblock \emph{Introduction to Applied Nonlinear Dynamical Systems and Chaos},
  pages 5--19, 2003.

\bibitem[Drion et~al.(2018)Drion, Franci, and Sepulchre]{drion2018cellular}
Guillaume Drion, Alessio Franci, and Rodolphe Sepulchre.
\newblock Cellular switches orchestrate rhythmic circuits.
\newblock \emph{Biol Cyber}, 113:\penalty0 71–82, 2018.

\bibitem[Yan et~al.(2016)Yan, Cao, Heffernan, McVernon, Quinn, Gruta,
  et~al.]{Yan2016}
Ada W.~C Yan, Pengxing Cao, Jane~M. Heffernan, Jodie McVernon, Kylie~M. Quinn,
  Nicole L.~La Gruta, et~al.
\newblock Modelling cross-reactivity and memory in the cellular adaptive immune
  response to influenza infection in the host.
\newblock \emph{J Theor Biol}, 413:\penalty0 34--49, 2016.

\bibitem[Dey et~al.(2018)Dey, Franci, Ozcimder, and Leonard]{Dey2018}
Biswadip Dey, Alessio Franci, Kayhan Ozcimder, and Naomi~Ehrich Leonard.
\newblock Feedback controlled bifurcation of evolutionary dynamics with
  generalized fitness.
\newblock \emph{Proc ACC}, 2018-June:\penalty0 6049--6054, 2018.

\bibitem[Drion et~al.(2012)Drion, Franci, Seutin, and Sepulchre]{Drion2012b}
Guillaume Drion, Alessio Franci, Vincent Seutin, and Rodolphe Sepulchre.
\newblock A novel phase portrait for neuronal excitability.
\newblock \emph{PLoS ONE}, 7:\penalty0 1--14, 2012.

\bibitem[Franci et~al.(2012)Franci, Drion, and Sepulchre]{franci2012organizing}
Alessio Franci, Guillaume Drion, and Rodolphe Sepulchre.
\newblock An organizing center in a planar model of neuronal excitability.
\newblock \emph{Siam J Appl Dyn Syst}, 11:\penalty0 1698--1722, 2012.

\bibitem[Franci et~al.(2014)Franci, Drion, and Sepulchre]{franci2014modeling}
Alessio Franci, Guillaume Drion, and Rodolphe Sepulchre.
\newblock Modeling the modulation of neuronal bursting: a singularity theory
  approach.
\newblock \emph{Siam J Appl Dyn Syst}, 13:\penalty0 798--829, 2014.

\bibitem[Franci et~al.(2018)Franci, Drion, and Sepulchre]{franci2018robust}
Alessio Franci, Guillaume Drion, and Rodolphe Sepulchre.
\newblock Robust and tunable bursting requires slow positive feedback.
\newblock \emph{J Neurophysiol}, 119:\penalty0 1222--1234, 2018.

\bibitem[Chekroun et~al.(2021)Chekroun, Liu, and
  McWilliams]{chekroun2021stochastic}
Micka{\"e}l~D Chekroun, Honghu Liu, and James~C McWilliams.
\newblock Stochastic rectification of fast oscillations on slow manifold
  closures.
\newblock \emph{P Natl Acad Sci Usa}, 118\penalty0 (48):\penalty0 e2113650118,
  2021.

\bibitem[Tikhonov(1948)]{tikhonov1948dependence}
Andrei Tikhonov.
\newblock On the dependence of the solutions of differential equations on a
  small parameter.
\newblock \emph{Mat Sb}, 64\penalty0 (2):\penalty0 193--204, 1948.

\bibitem[Tikhonov(1950)]{tikhonov1950systems}
Andrei~Nikolaevich Tikhonov.
\newblock On systems of differential equations containing parameters.
\newblock \emph{Mat Sb}, 69\penalty0 (1):\penalty0 147--156, 1950.

\bibitem[Tikhonov(1952)]{tikhonov1952systems}
Andrei~Nikolaevich Tikhonov.
\newblock Systems of differential equations containing small parameters in the
  derivatives.
\newblock \emph{Mat Sb}, 73\penalty0 (3):\penalty0 575--586, 1952.

\bibitem[Wasow(2018)]{wasow2018asymptotic}
Wolfgang Wasow.
\newblock \emph{Asymptotic Expansions for Ordinary Differential Equations}.
\newblock Courier Dover Publications, 2018.

\bibitem[Bersani et~al.(2020)Bersani, Borri, Milanesi, Tomassetti, and
  Vellucci]{bersani2020uniform}
Alberto~Maria Bersani, Alessandro Borri, Alessandro Milanesi, Giuseppe
  Tomassetti, and Pierluigi Vellucci.
\newblock Uniform asymptotic expansions beyond the tqssa for the
  goldbeter--koshland switch.
\newblock \emph{SIAM J Appl Math}, 80\penalty0 (3):\penalty0 1123--1152, 2020.

\bibitem[Khas'minskii(1966)]{khas1966limit}
Rafail~Z Khas'minskii.
\newblock {A limit theorem for the solutions of differential equations with
  random right-hand sides}.
\newblock \emph{Theor Probab Appl}, 11\penalty0 (3):\penalty0 390--406, 1966.

\bibitem[Kasdin(1995)]{kasdin1995discrete}
N~Jeremy Kasdin.
\newblock Discrete simulation of colored noise and stochastic processes and
  1/f/sup/spl alpha//power law noise generation.
\newblock \emph{Proc IEEE}, 83\penalty0 (5):\penalty0 802--827, 1995.

\bibitem[Namachchivaya and Leng(1990)]{namachchivaya1990equivalence}
NS~Sri Namachchivaya and Gerard Leng.
\newblock {Equivalence of stochastic averaging and stochastic normal forms}.
\newblock \emph{J App Mech}, 57\penalty0 (4):\penalty0 1011--1017, 1990.

\bibitem[Borzì(2020)]{Borzì2020}
Alfio Borzì.
\newblock \emph{Modelling with Ordinary Differential Equations: a Comprehensive
  Approach}.
\newblock CRC Press, 2020.

\bibitem[Sidney et~al.(2010)Sidney, Dunlop, and Elowitz]{Sidney2010}
Robert~Cox Sidney, Mary Dunlop, and Michael~B Elowitz.
\newblock {A synthetic three-color reporter framework for monitoring genetic
  regulation and noise}.
\newblock \emph{J Biol Eng}, 4\penalty0 (10):\penalty0 1--12, 2010.

\bibitem[Wang et~al.(2012)Wang, Li, Cheng, and Liu]{Wang2012b}
Xuesong Wang, Lijing Li, Yuhu Cheng, and Qingfeng Liu.
\newblock {Construction of gene regulatory networks with colored noise}.
\newblock \emph{Neural Comput Appl}, 21\penalty0 (8):\penalty0 1883--1891,
  2012.

\bibitem[Liu et~al.(2009)Liu, Xie, Liu, and Li]{liu2009effect}
Xue-Mei Liu, Hui-Zhang Xie, Liang-Gang Liu, and Zhi-Bing Li.
\newblock Effect of multiplicative and additive noise on genetic
  transcriptional regulatory mechanism.
\newblock \emph{Physica A}, 388\penalty0 (4):\penalty0 392--398, 2009.

\bibitem[Risken and Caugheyz(1996)]{Risken1991}
H.~Risken and T.~K. Caugheyz.
\newblock {The Fokker-Planck Equation: Methods of Solution and Application}.
\newblock \emph{J Appl Mech}, 58\penalty0 (3):\penalty0 860--860, 1996.
\newblock ISSN 0021-8936.

\bibitem[Martins and Pinto(2017)]{Martins2017}
José Martins and Alberto Pinto.
\newblock Bistability of evolutionary stable vaccination strategies in the
  reinfection siri model.
\newblock \emph{B Math Biol}, 79:\penalty0 853--883, 2017.

\bibitem[Bhattacharya et~al.(2010)Bhattacharya, Conolly, Kaminski, Thomas,
  Andersen, and Zhang]{Bhattacharya2010}
Sudin Bhattacharya, Rory~B. Conolly, Norbert~E. Kaminski, Russell~S. Thomas,
  Melvin~E. Andersen, and Qiang Zhang.
\newblock A bistable switch underlying b-cell differentiation and its
  disruption by the environmental contaminant
  2,3,7,8-tetrachlorodibenzo-p-dioxin.
\newblock \emph{Toxicol Sci}, 115:\penalty0 51--65, 2010.

\bibitem[Lebar et~al.(2014)Lebar, Bezeljak, Golob, Jerala, Kadunc, Pirš,
  Stražar, Vučko, Zupančič, Benčina, et~al.]{lebar2014bistable}
Tina Lebar, Urban Bezeljak, Anja Golob, Miha Jerala, Lucija Kadunc, Boštjan
  Pirš, Martin Stražar, Dušan Vučko, Uroš Zupančič, Mojca Benčina,
  et~al.
\newblock A bistable genetic switch based on designable dna-binding domains.
\newblock \emph{Nat Commun}, 5:\penalty0 1--13, 2014.

\bibitem[Nyman et~al.(2020)Nyman, Ashwin, and Ditlevsen]{nyman2020bifurcation}
Karl~HM Nyman, Peter Ashwin, and Peter~D Ditlevsen.
\newblock Bifurcation of critical sets and relaxation oscillations in singular
  fast-slow systems.
\newblock \emph{Nonlinearity}, 33\penalty0 (6):\penalty0 2853, 2020.

\bibitem[Stanoev et~al.(2021)Stanoev, Schröter, and Koseska]{Stanoev2021}
Angel Stanoev, Christian Schröter, and Aneta Koseska.
\newblock Robustness and timing of cellular differentiation through
  population-based symmetry breaking.
\newblock \emph{Development}, 148:\penalty0 dev197608, 2021.

\bibitem[Stamovlasis(2014)]{stamovlasis2014bifurcation}
Dimitrios Stamovlasis.
\newblock {Bifurcation and hysteresis effects in student performance: The
  signature of complexity and chaos in educational research}.
\newblock \emph{Complicity}, 11\penalty0 (2), 2014.

\end{thebibliography}

\end{document}